\documentclass[oldfontcommands,a4paper,12pt,openany]{memoir}
\usepackage{latexsym , amssymb , amsfonts, pifont , amsmath, upgreek}

\usepackage{graphicx} 

\newsubfloat{figure} 
\newsubfloat{table} 
\chapterstyle{bianchi} 

\newcolumntype{C}{>{$}c<{$}}
\newcolumntype{L}{>{$}l<{$}}
\newcolumntype{R}{>{$}r<{$}}

\newcommand{\nfrac}[2]{ {\textstyle\frac{#1}{#2}} }  

\newcommand{\lm}{\llap{-}}

\newcommand{\noneg}{\phantom{-}}

\let\TAB\tabular
\renewcommand\tabular{\noindent\TAB} 

\newcommand{\R}{\mathbb{ R}}
\newcommand{\C}{\mathbb{ C}}
\newcommand{\Z}{\mathbb{ Z}}
\newcommand{\g}{\mathfrak g}
\newcommand{\T}{\mathcal{ T}}
\newcommand{\LL}{\mathcal{ L}}

\newcommand{\so}{\mathsf{so}}
\newcommand{\SO}{\mathsf{SO}}
\renewcommand{\sp}{\mathsf{sp}} 
\newcommand{\Sp}{\mathsf{Sp}}

\newcommand{\tensor}{\otimes}
\newcommand{\del}{\partial}
\newcommand{\integral}{\int}

\newcommand{\rank}[2]{ \left( { }^{#1}_{#2}\right) }
\newcommand{\oo}{\!\rank{*}{*}}

\newcommand{\ijkequal}{ \stackrel{(ijk)}{=} }
\newcommand{\abcequal}{ \stackrel{(\alpha\beta\gamma)}{=} }

\newcommand{\DOT}{\bullet}
\newcommand{\CROSS}{ \,\mbox{\scriptsize\ding{54}}\, }
\newcommand{\Grad}{ \mbox{\ding{116}} }
\newcommand{\Curl}{\Grad\!\CROSS}
\newcommand{\Div}{\Grad\!\DOT}
\newcommand{\Triangle}{ \mbox{\ding{115}} }
\newcommand{\Square}{ \mbox{\ding{110}} }

\newcommand{\EE}{\mbox{$\rightarrow$}}
\newcommand{\AW}{\mbox{$\leftarrow$}}
\newcommand{\AN}{\mbox{$\uparrow$}}
\newcommand{\AS}{\mbox{$\downarrow$}}
\newcommand{\ANE}{\mbox{$\nearrow$}}
\newcommand{\ANW}{\mbox{$\nwarrow$}}
\newcommand{\ASE}{\mbox{$\searrow$}}
\newcommand{\ASW}{\mbox{$\swarrow$}}

\newtheorem{theorem}{Theorem}[section]
\newtheorem{lemma}[theorem]{Lemma}
\newtheorem{proposition}[theorem]{Proposition}
\newtheorem{corollary}[theorem]{Corollary}

\newtheorem{definition}[theorem]{Definition}
\newtheorem{axiom}{Axiom}
\newtheorem{conjecture}[theorem]{Conjecture}
\newtheorem{question}{Question}

\newenvironment{proof}[1][Proof]%
{\begin{trivlist}\item[\hskip \labelsep {\bfseries #1}]}%
{\qed\end{trivlist}}
{\begin{trivlist} \item[\hskip \labelsep {\bfseries #1}]}%
{\end{trivlist}}
\newcommand{\qed}{\nobreak \ifvmode \relax \else
      \ifdim\lastskip<1.5em \hskip-\lastskip
      \hskip1.5em plus0em minus0.5em \fi \nobreak
      \vrule height0.75em width0.5em depth0.25em\fi}

\begin{document}

\frontmatter

	\pagestyle{empty}

\centerline{\includegraphics[width=8cm]{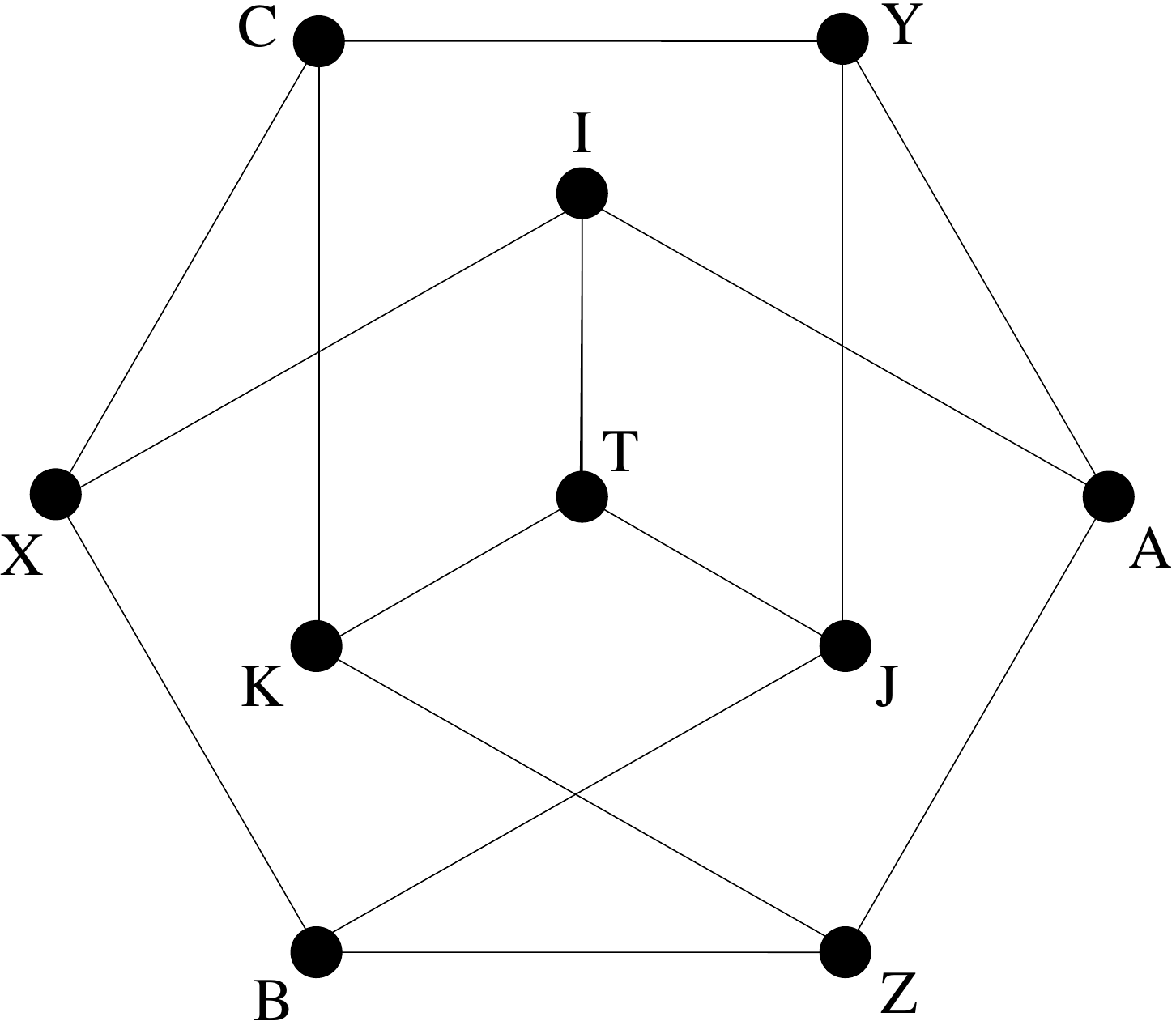}}
\vskip1cm
\centerline{ \resizebox{8cm}{!}{ \HUGE\textsc Framework} }
\vskip5mm
\centerline{\Huge\textsc The physics of $\sp(2,\R)$}
\vskip1cm
\centerline{\large\textsc by }
\vskip5mm
\centerline{\Large\textsc Ian Hawthorn}
\vskip5mm
\centerline{ \Large\textsc William Crump  \& Matthew Ussher }
\vskip5mm
\centerline{\includegraphics[width=42mm]{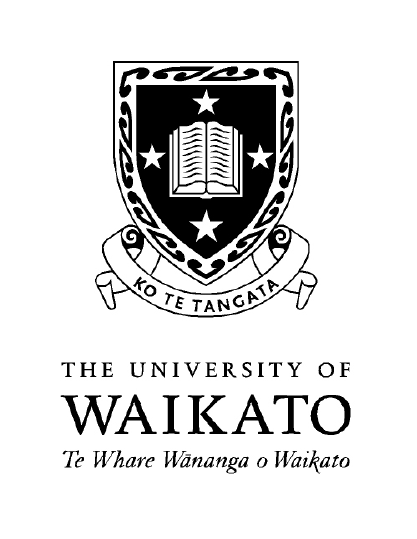}}

\clearpage
\pagestyle{plain}

\begin{abstract}
A mathematical framework for relativistic quantum mechanics 
is constructed with natural symmetry $\so(2,3)=\sp(2,\R)$. In this 
framework gravity and electromagnetism unify as aspects of the
geometry. The source equation for gravity differs from Einstein's 
equation and permits behavior that could explain dark matter.
\end{abstract}

\vfill

{\footnotesize
{\textsc{Note on the Title Page Illustration: }} The diagram on the 
title page uses the Petersen graph to represent the Lie algebra 
$\so(2,3) = \sp(2,\R)$.
The labelled vertices represent basis elements. The edges connect those which
commute. Those which are not connected are linked by a unique path of length 2. 
The other vertex connected to the middle vertex of the path gives the 
Lie bracket with sign specified by whether a physical turn
to the left (positive) or to the right (negative) is made in traversing 
the path. For example to compute $ [T,X] $ from the diagram, note that the 
path from $T$ to $X$ turns left at $I$, and the other vertex connected
to $I$ is $A$. Hence $[T,X] = A$. }

\bigskip

{ \hfill \footnotesize SVN version 434}

	\chapter*{Preface}

This is a book of original research written about physics by mathematicians. 
My task in this preface is to explain how such an odd thing came to be, and 
the reasons for publishing in this unusual form.

\medskip

I am Ian Hawthorn and the main author of this book. I am the sole author 
in a legal copyright sense. The words here are mine and any mistakes and 
omissions are my sole responsibility. However some of this work was carried
out in collaboration with two excellent students, William Crump and 
Matthew Ussher, who both made important and significant contributions at 
different stages and hence have been named as coauthors. 

I have also taken the liberty of naming aspects of the theory after them to 
ensure that they receive the recognition they deserve. Crump scalars are 
named after William Crump as the necessity for these was the main topic of 
his Masters thesis; and Ussher's equation is named after Matthew Ussher 
who came up with it in his Masters thesis.

\medskip

Why the choice of a book format? Most original research these days is 
published in journals. In this case however the article format seemed very 
confining. The scope of the work is large, so we would have needed to divide
it into a series of journal articles. Matthew and I made an attempt to 
subdivide the work in this way at one point. However the early parts seemed 
too elementary and unmotivated, while the later parts didn't make sense 
without the earlier work. 

The book format allows us the space to build gradually from a leisurely 
consideration and discussion of ideas; through the careful construction of 
an appropriate axiomatic structure to express those ideas; to the 
consideration of the mathematical theory of that structure; and finishing 
with a discussion of the physical implications. The book format 
also enabled us to include unoriginal background material in the 
introductory chapter, and allowed us to write in a more relaxed and
conversational style rather than in the terse style expected in journals. 
I believe this makes it a much easier read than if we had tried to force 
the material into the journal format.

The reader should note that the first chapter contains a lot of background 
material presented in a rather elementary way as I hope that the book will
be accessible to graduate students in both mathematics and physics. 
If you find this material too elementary I suggest you skip ahead to 
Chapter 2. The final chapter is a summary of the rest of the book. If you 
just want a quick idea of what the book is about, you could look at that. 
Otherwise the best way to read this book is to start at the beginning and 
read it through to the end.  I have tried very hard to make everything 
self-contained, readable and simple. It is my belief that correct explanations 
are simple explanations and that the best academic writing makes the subject 
seem easy.  

\bigskip
 
The rest of this preface discusses the history of the work. Those who 
are not interested in such things should skip it.

\medskip

I first started thinking about the things in this book in the 1990s. 
However it was not until 2007 that I found a way into the problem. And it 
was not until my academic leave in 2009 that significant progress was first 
made. By the end of 2009 I had a rather beautiful mathematical structure which 
did most of what I wanted. It combined curvature with \( \so(2,3) \) symmetry 
in an elegant way allowing a natural expression of the Dirac equation, 
Maxwell's equations and Einstein's equation.

However there was a small problem. Maxwell's equations appeared in this 
structure with an additional constraint which I had hoped initially was 
a gauge condition. I wish to thank Dr Yuri Litvinenko, my colleague at 
Waikato, for convincing me that it was instead an unphysical constraint 
on the electromagnetic field. This meant that after all my work, all I had 
was a broken model. It was an elegant and beautiful broken model based on 
a very small set of natural assumptions. But it was making an unphysical 
prediction so it was, as Feynman would say, wrong!

A model of this type is very specific and rigid with almost no adjustable 
parameters. It is quite the opposite in that respect to string theory which 
is a very general and flexible model with many adjustable parameters. The 
advantage of using a more rigid approach is that it tends to make strong 
predictions. The disadvantage is that if those predictions fail the test 
of reality the problem can be very hard to fix. Rigid models don't bend well. 
I therefore put the work aside and devoted my attention to other things, 
mainly administration since I had just taken over as chairperson of 
the mathematics department at that time.

In 2010 I took on a master's student, William Crump, with interests in 
mathematical physics. The task assigned to him for his Master's 
thesisi\cite{Cthesis} was to read and understand the work as it existed 
at that time, and to explore and discuss the nature of the failure with 
respect to Maxwell's equations.  In the course of working on this problem 
together we came to realise that the electromagnetic field was much more 
closely linked to the geometry of the model than I had previously thought. 
We discovered that the problem arose from a very subtle assumption that I 
had made which constrained the geometry effectively creating a zero field 
condition for the electromagnetic field.

Removing this assumption freed up the model and eliminated the unphysical 
constraint on Maxwell's equations. Furthermore our new insights about the 
link between electromagnetism and the geometry then allowed us to also 
link Maxwell's equations to the Bianchi identities for the gravitational 
field. The dynamical equations for both electromagnetism and gravity were
now unified as components of a single equation for the spinor curvature. 
These changes required a complete rewrite which I commenced in 2011. The 
added flexibility made the mathematics much more complicated however, 
which slowed progress.

In 2012 I was fortunate to have another excellent Master's student, Matthew 
Ussher. Matthew explored in his thesis the nature of the source equations for 
electromagnetism and gravity. Whereas the dynamical equations for these forces 
were unified, the source equations still looked very different.  In discussing 
this issue Matthew came up with a unified source equation~\cite{Uthesis} 
which gave the Ampere-Gauss equation for electromagnetism and another equation 
for gravity which I have called Ussher's equation. Ussher's equation is quite
different from Einstein's equation. We were not able to link the two equations 
and were not sure quite what to make of this result.

In 2014 I was again on academic leave, and during that time I looked at the 
Lagrangian analysis in more detail. Matthew had attempted a Lagrangian analysis
in his thesis with limited success. I was able to push the Lagrangian
analysis through to a full conclusion and discovered that Ussher's equation 
was only a special case; and the resulting full equations offered a possible 
explanation for dark matter. 

At that time I realised I had something pretty special. However the work
was far from complete. It was clear by the nature of the work that I would be
unable to tie off every loose end, but there were certain things I wanted
to finish up before publication. In particular I wanted to find the source 
terms for the force equations so that I would have a complete and unified
relativistic field theoretic description of electromagnetism and gravity.
Source terms for fermions at least should be obtainable by applying an 
appropriate variation to the correct Dirac Lagrangian density.

I did not expect this to be difficult, however it proved considerably harder
than I had anticipated. My Dirac Lagrangian densities were persistently
misbehaving despite the fact that I had proved equivalence to the standard 
Dirac Lagrangian density. It took me much longer than it should have done to 
realise that the problems originated with the standard Dirac Lagrangian 
density which is complex when it needs to be real, and which is not 
well defined since it has an unstated conformal dependence. 

Having discovered these issues it took time to explore them properly which 
delayed final publication to the end of 2015. And I still didn't get the
source terms I was seeking.

\medskip
I have chosen to release the work on Ar$\chi$iv while I search for a suitable 
publisher. I also intend to submit shorter summary papers to peer reviewed
journals.

\vskip0.5cm

- Ian Hawthorn

\begingroup
\setlength{\cftbeforechapterskip}{1.0em plus 0.3em minus 0.1em}
\setlength{\beforechapskip}{-2cm} 
\tableofcontents* 
\endgroup

\mainmatter
	\chapter{Spacetime Symmetry}
\label{spacetimesymmetry}

Symmetry is absolutely central to modern physics. Not only does symmetry 
play a vital and direct role in many physical theories, it also serves as 
a bridge between the various mathematical formalisms and the everyday 
reality in which our physical intuition lives. All mathematical structures 
have a natural symmetry and consequently anything defined in one context 
purely in terms of symmetry can be translated into another context having 
the same symmetry group quite easily.
 
\medskip

Of the various symmetries of use to modern physics the symmetries of 
space-time itself are the most fundamental. Whereas other symmetries 
may apply only to some objects or in some circumstances, everything 
that lives in our universe must conform to the symmetries of space-time.

The symmetry group of space-time is ten dimensional which is to say that 
ten continuous parameters are needed to situate an event. The time of 
the event gives us one parameter and three more are required to locate
its position in space. We need three parameters to describe how the event 
is oriented, and finally a further three parameters (Lorentz boost coordinates) 
are needed to specify an instantaneous reference velocity for the event. 

We may not need all ten parameters if the event we wish to locate is itself
possessed of symmetry. For example if the event is spherically symmetric 
(appears the same in all orientations), we do not need the orientation 
parameters to unambiguously locate it. If the event can be described without
reference to velocity, then we won't need the three Lorentz boost coordinates.  

It is usual in physics to regard the four dimensions required to specify
location in space and time as being more physically real (in some sense)
than the other six coordinates. That is because our everyday experience of
space-time symmetry involves its 4-D representation. However the distinction 
is not a natural one to make from the point of view of the group itself and 
in other representations these four coordinates look exactly the same as 
the other six. We will accord all ten coordinates equal status.
\medskip

The space-time symmetry group is continuous and connected. Sometimes various 
discrete symmetries; reflections and inversions; are considered as part of
the space-time symmetry group but we choose to reject this point of view.
There are known physical events, most notably involving the weak force, 
which are not conserved under these inversions. The true symmetries of 
space-time should conserve the form of all possible physical events. 

Continuous symmetries can be performed gradually and hence it is enough
to understand the symmetry processes which perform these gradual 
transformations. These can be described using a basis of ten processes;
translation through time and space, rotation about each of the three
principal axes, and Lorentz boost along the three principal directions.

Symmetries have an algebraic structure. The symmetries themselves form
a group under the binary operation of composition while the symmetry 
processes constitute a Lie algebra with three operations; an addition 
and scalar multiplication giving a vector space, and a commutator or 
Lie bracket which is bilinear, antisymmetric and obeys the Jacobi bracket 
condition. 

\medskip
Our understanding of the algebraic structure of the symmetries of space-time 
has evolved through history. The Galilean group of classical mechanics has 
been supplanted by the Poincar\'e group of special relativity. Whereas in 
the Galilean group boost operators commute, in the Poincar\'e group their 
commutators are multiples of rotation operators where the multiplying
factor is inversely proportional to the speed of light. As the speed of light
is large, it follows that these commutators are \textbf{close} to zero and 
hence the Galilean Lie algebra approximates the Poincar\'e Lie algebra in the 
regime where velocities are much less than the speed of light. 
Indeed we can regard the Galilean Lie algebra as being in some sense the
limit as \( c \rightarrow \infty \) of the Poincar\'e Lie algebra. 
This kind of limiting process is called a contraction.

The Poincar\'e Lie algebra itself can be obtained via contraction from other
Lie algebras, in particular from the DeSitter Lie algebra \( \so(1,4) \)
and the Anti-DeSitter algebra \( \so(2,3) \). For reasons we shall discuss 
later it is the second of these two possibilities that is of greatest 
interest. The contracting parameter in this case is an absolute distance 
scale \( r \) which we will colloquially refer to as the radius of the 
universe, and in this group the commutator of two translation operators
is a multiple of a rotation operator where the multiplying factor is inversely 
proportional to \( r \). If \( r \) (as measured in everyday units) is large 
then the Poincar\'e Lie algebra is very closely approximated by \( \so(2,3) \),
especially in the regime of distance and time scales much less than \( r \).

\begin{description}
\item[Obvious Question:] Which group is correct? 
\end{description}

It would be very difficult for us to answer this question directly 
by making measurements on space-time if \( r \) is significantly bigger than 
the regime of distance and time accessible to us for direct measurement.  
Furthermore on large distance scales we cannot ignore the effects of 
curvature and cosmological model.

Indeed the very notion of a global symmetry group for space-time 
becomes rather dubious in the presence of curvature. And while local 
symmetry retains its meaning it is very unclear how we might go about 
measuring it accurately enough to distinguish between the two groups.

The main problem is that the measurements that might help distinguish between 
\( \so(2,3) \) and the Poincar\'e Lie algebra are also going to be effected by 
our choice of cosmological constant or cosmological model.  

Our obvious question does not therefore have an obvious answer. It appears
unlikely that we can determine which group is correct by direct measurements
of spacetime, at least in the realm of classical physics. Can quantum 
mechanics rescue our question? 

Considering how these two groups act on spacetime we find that the problem 
of measuring the difference is just as difficult, but it manifests in 
a different way. Energy for example would be quantised in a universe based 
on \( \so(2,3) \). However as the quanta would be associated with periods of 
cosmological duration, we are not going to be able to detect this.

However this does not completely answer the question since in quantum mechanics
the symmetry group of space-time has a second action.

Rotations for example act on particles in two ways.  They act in an ordinary 
fashion on space-time which is the domain space of the wave functions, and 
they also act intrinsically on the range space of wave functions. The 
eigenvalues of the first action give us angular momentum while the eigenvalues 
of the second action give an intrinsic angular momentum or spin.

The reasons for this double action are not clear. Why should rotations act 
on this other space? We know that they do and indeed are so used to it that 
we may have come to view it as simply the natural order of things. But it is
worth remembering that mathematically speaking the existence of this second 
action on an unrelated space is an unanticipated surprise. 

One very strange feature of intrinsic actions is that we only allow part 
of the symmetry group to act in this way. Translation operators are not 
permitted to act intrinsically (equivalently we require that their intrinsic 
action is trivial). That is because in the Poincar\'e group translation 
operators would have a continuous spectrum of eigenvalues.  If the translation 
operators in the Poincar\'e group were free to act intrinsically as rotations
do, we would expect to see continuously varying intrinsic versions of momentum 
and energy, which we do not\footnote{ Rest mass is a function of 
(and constrains) the ordinary momentum and energy of a particle. 
In particular it should not be confused with intrinsic Energy}.

It is strange enough that symmetries of space-time act intrinsically  at all.
It is doubly strange that only {\em some} of the symmetries of space-time 
should act in this way, and the constraint that excludes only translations 
from acting intrinsically seems arbitrary and ugly, if not actually broken.
Are things any better in this regard if we use \( \so(2,3) \) instead?

\medskip
The Lie algebra \( \so(2,3) \) has a very different spectrum of eigenvalues. 
In particular the fundamental representation of \( \so(2,3) \) is four
dimensional with two quantum numbers, a spin (eigenvalue for rotation) 
taking values \( \pm \frac{1}{2} \), and an intrinsic energy (eigenvalue for 
translation through time)  taking values \( \pm \frac{1}{2} \) in natural 
units. How well does this spectrum of eigenvalues accord with observation?

Fermions, and in particular solutions to the Dirac equation, are 
characterised by two quantum numbers, spin and charge. If we identify
charge with intrinsic energy this fits perfectly with the intrinsic
spectrum of \( \so(2,3) \). Furthermore this identification makes a great
deal of sense. It has long been clear that time and charge are linked 
in some fashion. Positrons for example are often described as electrons
travelling backwards in time. As we will see later the Dirac equation 
is very natural when viewed from the point of view of local symmetry 
group \( \so(2,3) \). Maxwell's equations are also natural from this 
perspective.

We conclude therefore that there is good reason to expect that \( \so(2,3) \)
may fit better with observation than the Poincar\'e group particularly 
with regard to intrinsic action on wave-functions in quantum mechanics. 
It would therefore be worthwhile to explore the consequences of using 
mathematical structures which naturally express this symmetry group in 
our physics. 

\bigskip
\pagebreak[2]

\section{The Lie algebra $\so(2,3)$}
\label{sectiondetails}

Choose coordinates \( \lambda,t,x,y,z \) in \( \R^5\) and define \( O(2,3) \)
as the \( 5 \times 5 \) real matrices conserving the bilinear form
\( \lambda_1\lambda_2 + t_1t_2-x_1x_2-y_1y_2-z_1z_2 \).
For \( u,v \in \R^5 \) this form can be written as 
\( <u,v> = u^t\Omega v \) where 

\medskip
\hfill \mbox{\small $
\tabcolsep=4pt
\Omega = \left(\begin{array}{rrrrr}
    1&0&0&0&0 \\
 0&   1&0&0&0 \\
 0&0&\lm1&0&0 \\
 0&0&0&\lm1&0 \\
 0&0&0&0&\lm1 
\end{array}\right)
$} \hfill
\medskip

Hence for \( A \in \mathsf{O}(2,3) \) we have \( A^t\Omega A = \Omega \) and
thus \( |A| = \pm 1 \). The matrices with determinant one give the subgroup 
\( \SO(2,3) \) which is the connected component of the identity.

The 4-D manifold in \( \R^5 \) where 
\( \lambda^2+t^2-x^2-y^2-z^2 = 1 \) forms an invariant 4-D manifold in 
\( \R^5 \) for this group. This manifold is called Anti-DeSitter spacetime,
and \( \SO(2,3) \) is called the Anti-DeSitter group when acting on this 
manifold.

Consider a small neighbourhood of the point on the manifold where 
\( \lambda = 1 \). Near this point \( \lambda \) is a function of the 
other coordinates allowing us to identify this neighbourhood with a 
local region in space-time. Note that \( \lambda = 1 \) up to second order 
in the coordinates, and for small regions the resulting transformation group 
is almost indistinguishable from the Poincar\'e group. ADS spacetime has 
a natural scale for distance and time inherited from \( \R^5 \). Define 
\( r \) seconds = \( 1 \) natural time unit.  Then \( rc \) metres = \( 1 \) 
natural distance unit. We will call \( r \) colloquially the {\em radius of 
the universe} as measured in seconds.


\begin{table}[hbp]
\tabcolsep=4pt
\hfill\mbox{\small $
T = \left(\begin{tabular}{rr|rrr}
  0&\lm1 &  0&0&0 \\
  1& 0  &  0&0&0 \\\hline
  0& 0  &  0&0&0 \\
  0& 0  &  0&0&0 \\
  0& 0  &  0&0&0 
  \end{tabular}\right)
$}\hfill
\smallskip

\hfill\mbox{\small $
X = \left(\begin{tabular}{rr|rrr}
  0& 0  & 1&0&0 \\
  0& 0  & 0&0&0 \\\hline
  1& 0  & 0&0&0 \\
  0& 0  & 0&0&0 \\
  0& 0  & 0&0&0 
  \end{tabular}\right)
\mbox{\hskip0.4cm}
Y = \left(\begin{tabular}{rr|rrr}
  0& 0  &  0&1&0 \\
  0& 0  &  0&0&0 \\\hline
  0& 0  &  0&0&0 \\
  1& 0  &  0&0&0 \\
  0& 0  &  0&0&0 
  \end{tabular}\right)
\mbox{\hskip0.4cm}
Z = \left(\begin{tabular}{rr|rrr}
  0& 0  &  0&0&1 \\
  0& 0  &  0&0&0 \\\hline
  0& 0  &  0&0&0 \\
  0& 0  &  0&0&0 \\
  1& 0  &  0&0&0   
\end{tabular}\right)
$}\hfill
\smallskip

\hfill\mbox{\small $
A = \left(\begin{tabular}{rr|rrr}
  0&0  &  0&0&0 \\
  0&0  &  1&0&0 \\\hline
  0&1  &  0&0&0 \\
  0&0  &  0&0&0 \\
  0&0  &  0&0&0 
  \end{tabular}\right)
\mbox{\hskip0.4cm}
B = \left(\begin{tabular}{rr|rrr}
  0&0  &  0&0&0 \\
  0&0  &  0&1&0 \\\hline
  0&0  &  0&0&0 \\
  0&1  &  0&0&0 \\
  0&0  &  0&0&0 
  \end{tabular}\right)
\mbox{\hskip0.4cm}
C = \left(\begin{tabular}{rr|rrr}
  0&0  &  0&0&0 \\
  0&0  &  0&0&1 \\\hline
  0&0  &  0&0&0 \\
  0&0  &  0&0&0 \\ 
  0&1  &  0&0&0   
  \end{tabular}\right)
$}\hfill
\smallskip

\hfill\mbox{\small $
I = \left(\begin{tabular}{rr|rrr}
  0&0  &  0&0&0 \\
  0&0  &  0&0&0 \\\hline
  0&0  &  0&0& 0 \\
  0&0  &  0&0&\lm1 \\
  0&0  &  0&1& 0 
  \end{tabular}\right)
\mbox{\hskip0.4cm}
J = \left(\begin{tabular}{rr|rrr}
  0&0  &  0&0&0 \\
  0&0  &  0&0&0 \\\hline
  0&0  &  0 &0&1 \\
  0&0  &  0 &0&0 \\
  0&0  &  \lm1&0&0 
  \end{tabular}\right)
\mbox{\hskip0.4cm}
K = \left(\begin{tabular}{rr|rrr}
  0&0  &  0&0&0 \\
  0&0  &  0&0&0 \\\hline
  0&0  &  0&\!-1&0 \\
  0&0  &  1& 0&0 \\
  0&0  &  0& 0&0 
  \end{tabular}\right)
$}\hfill
\smallskip

\caption{The canonical representation of \( \so(2,3) \)}
\label{so23_5x5}
\end{table}

\begin{table}[htbp]
\begin{center} \small
\tabcolsep=5pt
\setlength{\extrarowheight}{0.5pt}
\begin{tabular}{C<{\hskip0.5em}|>{\hskip0.5em}C!{\hskip1em}CCC!{\hskip1em}CCC!{\hskip1em}CCC}
  &  T &  X& Y& Z &  A& B& C &  I& J& K \\[0.2em]\hline
  &    &   &  &   &   &  &   &   &  &   \\[-0.9em]
T &  0 &  A& B& C & \lm X&\lm Y&\lm Z &  0& 0& 0 \\[1em]
X & \lm A &  0&\lm K& J & \lm T& 0& 0 &  0& Z&\lm Y \\
Y & \lm B &  K& 0&\lm I &  0&\lm T& 0 & \lm Z& 0& X \\
Z & \lm C & \lm J& I& 0 &  0& 0&\lm T &  Y&\lm X& 0 \\[1em]
A &  X &  T& 0& 0 &  0&\lm K& J &  0& C&\lm B \\
B &  Y &  0& T& 0 &  K& 0&\lm I & \lm C& 0& A \\
C &  Z &  0& 0& T & \lm J& I& 0 &  B&\lm A& 0 \\[1em]
I &  0 &  0& Z&\lm Y &  0& C& \lm B &  0& K&\lm J \\
J &  0 & \lm Z& 0& X & \lm C& 0& A & \lm K& 0& I \\
K &  0 &  Y&\lm X& 0 &  B&\lm A& 0 &  J&\lm I& 0 
\end{tabular}
\end{center}
\caption{The Lie algebra \( \so(2,3) \) in natural units}
\label{table_so23}
\end{table}

\newcommand{\osr}{\nfrac{1}{r^2}}
\newcommand{\osc}{\nfrac{1}{c^2}}
\newcommand{\osrc}{\nfrac{1}{r^2c^2}}

\begin{table}
\tabcolsep=5pt
\setlength{\extrarowheight}{1.5pt}
\begin{center} \small
\begin{tabular}{C<{\hskip0.5em}|>{\hskip0.5em}C!{\hskip1em}CCC!{\hskip1em}CCC!{\hskip1.5em}CCC}
  &  T &  X& Y& Z &  A& B& C &  I& J& K \\[0.2em]\hline
  &    &   &  &   &   &  &   &   &  &   \\[-0.9em]
T & 0 &  \osr A&\osr B&\osr C  & \lm X&\lm Y&\lm Z  & 0&0&0 \\
X & \lm\osr A  &  0&\lm\osrc K&\osrc J  &  \lm\osc T&0&0  &  0&Z&\lm Y \\
Y & \lm\osr B  &  \osrc K&0&\lm\osrc I  &  0&\lm\osc T&0  &  \lm Z&0&X \\
Z & \lm\osr C  &  \lm\osrc J&\osrc I&0  &  0&0&\lm\osc T  &  Y&\lm X&0 \\[1em]
A &  X  &  \osc T&0&0  &  0&\lm\osc K&\osc J  &  0& C&\lm B  \\
B &  Y  &  0&\osc T&0  &  \osc K&0&\lm\osc I  &  \lm C& 0& A \\
C &  Z  &  0&0&\osc T  &  \lm\osc J&\osc I&0  &  B&\lm A& 0  \\[1em]
I &  0 & 0& Z&\lm Y & 0& C& \lm B & 0& K&\lm J \\
J &  0 & \lm Z& 0& X & \lm C& 0& A & \lm K& 0& I \\
K &  0 & Y&\lm X& 0 &  B&\lm A& 0 &  J&\lm I& 0 
\end{tabular}
\end{center}
\caption{The Lie algebra \( \so(2,3) \) in ordinary units}
\label{table_so23_ordinary}
\end{table}

Matrices for a basis of \( \so(2,3) \) are given in table~\ref{so23_5x5}.
Commutators can easily be computed for these matrices. 
Their commutators are given in table~\ref{table_so23}. These operators
can be identified by examining how they act on the invariant 4-D manifold 
identified with space-time. We discover that \( T \) describes
translation through time, \( X,Y,Z \) give translation 
through space, \( A,B,C\) are Lorentz boost operators 
changing the velocity of the inertial frame, and 
\( I,J,K\) are rotation operators. Elements of the 
group \( \SO(2,3) \) can be recovered as exponentials. For example 
the group element translating through time by \( d \) seconds
(or \( \frac{d}{r} \) natural units) is 
\( e^{\frac{d}{r}T} \).

The operators 
\( \left\{ \frac{1}{r}T, \,\frac{1}{rc}X, \,\frac{1}{rc}Y, \,\frac{1}{rc}Z, 
	\,\frac{1}{c}A, \,\frac{1}{c}B, \,\frac{1}{c}C, \,I,\,J,\,K \right\} \)
give an alternative basis for the Lie algebra using operators
defined using ordinary units. To obtain a description of the Lie algebra in 
ordinary units we should rename \( T,X,Y,Z,A,B,C,I,J,K \) to these 
scaled versions. Table~\ref{table_so23_ordinary} gives the Lie algebra in 
ordinary units. Table~\ref{table_so23} can be recovered by choosing units for
distance and time so that \( r = c = 1 \).

Looking at table~\ref{table_so23_ordinary} and letting 
\( r \rightarrow \infty \) we obtain the Lie algebra for the Poincar\'e
group.  We say that the ADS group contracts to the Poincar\'e group as 
the radius of the universe tends to infinity.
Letting \( c \rightarrow \infty \) as well gives the non-relativistic 
Galilean group of classical mechanics.  We say that the Poincar\'e group
contracts to the Galilean group as the speed of light tends to infinity.

It is also possible to allow \( c \rightarrow \infty \) while leaving 
\( r \) alone. The result describes a non-relativistic universe in which 
objects with velocity acquire distance over time (as one would expect); but
also objects which are distant acquire velocity over time. In other words
we'd have a classical universe with a cosmological constant. 

Historically the ADS group arose as a cosmological model for a space
of constant curvature. As a cosmological model it suffers from causality 
violating time-like loops. We will use the group only to describe local 
symmetry and this does not require acceptance of either ADS cosmology or 
time-like loops. The ADS group and the Poincar\'e group 
are indistinguishable as local symmetry groups for large \( r \) so
this is consistent with observation.

\bigskip
\pagebreak[2]

\section{The Lie algebra $ \sp(2,\R) $}
\label{sectionsp4r}

\( \Sp(2,\R) \) is the group of \( 4\times 4 \) symplectic 
matrices\footnote{Notation for symplectic groups varies.
We use the notation from \cite{Hall}.}. Symplectic
matrices are those which conserve a symplectic form -- in this case 
a fixed antisymmetric non-degenerate bilinear form on \( \R^4 \). 
Such a form is given with respect to a Darboux basis by 
\( <u,v> = u^t\Omega v \) where \( \Omega \) is the 
\( 4\times 4\) matrix

\centerline{
\mbox{\small $
\tabcolsep=4pt
\Omega = \left(\begin{array}{rrrr}
 0&0 & 1&0 \\
 0&0 & 0&1 \\
\lm1&0 & 0&0 \\
0&\lm1 & 0&0 
\end{array}\right)
$}}

\medskip
\pagebreak[2]

The group \( \Sp(2,\R) \) thus consists of all \( 4\times 4\) matrices
\( M \) with \( M^t\Omega M = \Omega \). The associated Lie algebra \( \sp(2,\R) \)
consists of the \( 4 \times 4 \) matrices \( D \) with 
\( D^t\Omega + \Omega D = 0 \). This Lie algebra is ten dimensional with basis
specified as in table~\ref{sp4R_basis}.

Commutators of these matrices give precisely the relationships
in table~\ref{table_so23}, hence \( \so(2,3) \) and \( \sp(2,\R) \)
are isomorphic. 

Note that each matrix \( M \) in this representation is invertible
with \( M^{-1} = \pm 4M \). The sign here is positive for 
\( \{ X,Y,Z,A,B,C \} \) and negative for \( \{ T,I,J,K \} \). We
can therefore express this in terms of the invariant bilinear form 
on the Lie algebra.


\begin{table}[htbp]
\centering
\tabcolsep=4pt
\mbox{\small $
T = \frac{1}{2}\left(\begin{tabular}{rrrr}
  0&0  &  1&0 \\
  0&0  &  0&1 \\
 \lm1&0  &  0&0 \\
  0&\lm1 &  0&0 
  \end{tabular}\right)
$}
\smallskip

\mbox{\small $
I = \frac{1}{2}\left(\begin{tabular}{rrrr}
  0&0  & 1&0 \\
  0&0  & 0&\lm1\\
 \lm1&0  & 0&0 \\
  0&1 &  0&0 
  \end{tabular}\right)
$}\hfill
\mbox{\small $
J = \frac{1}{2}\left(\begin{tabular}{rrrr}
  0&1  & 0&0 \\
 \lm1&0  & 0&0 \\
  0&0  & 0&1 \\
  0&0 & \lm1&0 
  \end{tabular}\right)
$}\hfill
\mbox{\small $
K = \frac{1}{2}\left(\begin{tabular}{rrrr}
  0&0  & 0&1 \\
  0&0  & 1&0 \\
  0&\lm1 & 0&0 \\
 \lm1&0  & 0&0 
  \end{tabular}\right)
$}
\smallskip

\mbox{\small $
X = \frac{1}{2}\left(\begin{tabular}{rrrr}
  0&0  & 0&1 \\
  0&0  & 1&0 \\
  0&1  & 0&0 \\
  1&0  & 0&0 
  \end{tabular}\right)
$}\hfill
\mbox{\small $
Y = \frac{1}{2}\left(\begin{tabular}{rrrr}
  1&0  & 0&0 \\
  0&1  & 0&0 \\
  0&0  &\lm1&0 \\
  0&0  & 0&\lm1 
  \end{tabular}\right)
$}\hfill
\mbox{\small $
Z = \frac{1}{2}\left(\begin{tabular}{rrrr}
  0&0  &\lm1&0 \\
  0&0  & 0&1 \\
 \lm1&0  & 0&0 \\
  0&1 &  0&0 
  \end{tabular}\right)
$}
\smallskip

\mbox{\small $
A = \frac{1}{2}\left(\begin{tabular}{rrrr}
  0&1  & 0&0 \\
  1&0  & 0&0 \\
  0&0  & 0&\lm1 \\
  0&0  & \lm1&0 
  \end{tabular}\right)
$}\hfill
\mbox{\small $
B = \frac{1}{2}\left(\begin{tabular}{rrrr}
  0&0  & \lm1&0 \\
  0&0  & 0&\lm1 \\
  \lm1&0  & 0&0 \\
  0&\lm1  & 0&0 
  \end{tabular}\right)
$}\hfill
\mbox{\small $
C = \frac{1}{2}\left(\begin{tabular}{rrrr}
  \lm1&0  & 0&0 \\
   0&1  & 0&0 \\
  0&0  & 1&0 \\
  0&0 &  0&\lm1 
  \end{tabular}\right)
$}
\smallskip

\caption{The canonical representation of \( \sp(2,\R) \) }
\label{sp4R_basis}
\end{table}

\bigskip
\pagebreak[2]

Consider now a matrix \( P \) with the property that 
\( P^t\Omega = \Omega P \). We will call a matrix with this property
\( \Omega \)-symmetric. If \( P \) is \( \Omega\)-symmetric and 
\( M \in \sp(2,\R) \) then \( [M,P] \) is also \( \Omega\)-symmetric. 
Hence the \( 4 \times 4 \) \( \Omega\)-symmetric matrices form a 
representation of \( \sp(2,\R) \) under the adjoint action. This 
representation is six dimensional with basis given in 
table~\ref{sp4R_action}.

\begin{table}[hbp]
\tabcolsep=4pt
\hfill\mbox{\small $
1 = \left(\begin{tabular}{rrrr}
  1&0 &  0&0 \\
  0&1 &  0&0 \\
  0&0 &  1&0 \\
  0&0 &  0&1 
  \end{tabular}\right)
$}\hfill
\smallskip

\hfill\mbox{\small $
P_\lambda = \frac{1}{2}\left(\begin{tabular}{rrrr}
  0&0  & 0&1 \\
  0&0  &\lm1&0\\
  0&1  & 0&0 \\
 \lm1&0 &  0&0 
  \end{tabular}\right)
$}\hskip2em
\mbox{\small $
P_t = \frac{1}{2}\left(\begin{tabular}{rrrr}
  0&1  & 0&0 \\
 \lm1&0  & 0&0 \\
  0&0  & 0&\lm1\\
  0&0 &  1&0 
  \end{tabular}\right)
$}\hfill
\smallskip

\hfill\mbox{\small $
P_x = \frac{1}{2}\left(\begin{tabular}{rrrr}
 \lm1&0 & 0&0 \\
  0&1 & 0&0 \\
  0&0 &\lm1&0 \\
  0&0 & 0&1 
  \end{tabular}\right)
$}\hskip2em
\hfill\mbox{\small $
P_y = \frac{1}{2}\left(\begin{tabular}{rrrr}
  0&0  & 0&1 \\
  0&0  &\lm1&0 \\
  0&\lm1 & 0&0 \\
  1&0  & 0&0 
  \end{tabular}\right)
$}\hskip2em
\mbox{\small $
P_z = \frac{1}{2}\left(\begin{tabular}{rrrr}
  0&\lm1  & 0&0 \\
  \lm1&0  & 0&0 \\
  0&0  & 0&\lm1 \\
  0&0  & \lm1&0 
  \end{tabular}\right)
$}\hfill
\smallskip

\caption{Basis for a 6-D representation of \( \sp(2,\R) \) under adjoint action}
\label{sp4R_action}
\end{table}

This 6-D representation is reducible and is the direct sum of a 1-D trivial
representation with basis \( \{ 1 \} \) and an irreducible
5-D representation with basis \( \{ P_\lambda, P_t, P_x, P_y, P_z \} \) 
where the action of \( \{ T,X,Y,Z,A,B,C,I,J,K\} \) on this basis 
generates precisely the matrices in table~\ref{so23_5x5}.

The matrices \( \{ P_\lambda, P_t, P_x, P_y, P_z \} \) can be expressed in 
terms of the matrices  \( \{ T,X,Y,Z,A,B,C,I,J,K\} \) as follows.
\begin{subequations}
\label{P_matrix_definitions}
\begin{alignat}{7}
 P_\lambda &{}={}& -2AI &{}={}&  -2IA &{}={}&  -2BJ &{}={}& -2JB &{}={}&  -2CK &{}={}& -2KC&  \\
 P_t       &{}={}&  2XI &{}={}&   2IX &{}={}&   2YJ &{}={}&  2JY &{}={}&   2ZK &{}={}&  2KZ&  \\
 P_x       &{}={}&  2TI &{}={}&   2IT &{}={}&  -2ZB &{}={}& -2BZ &{}={}&   2YC &{}={}&  2CY&  \\
 P_y       &{}={}&  2TJ &{}={}&   2JT &{}={}&  -2XC &{}={}& -2CX &{}={}&   2ZA &{}={}&  2AZ&  \\
 P_z       &{}={}&  2TK &{}={}&   2KT &{}={}&  -2YA &{}={}& -2AY &{}={}&   2XB &{}={}&  2BX&
\end{alignat}
\end{subequations} 

As all the products listed commute we can also express these matrices in 
terms of Jordan brackets; for example \( P_t = XI + IX = \{X,I\} \).

The matrices \( \{ P_\lambda, P_t,P_x,P_y,P_z, T,X,Y,Z,A,B,C,I,J,K \} \)
form a basis of the conformal Lie algebra \( \so(3,3) \). This Lie algebra 
is given in table~\ref{so33_table}

\bigskip

\begin{table}[htbp]
\tabcolsep=3pt
\setlength{\extrarowheight}{0.5pt}
\begin{center}\small
\begin{tabular}
{C<{\hskip6pt}|>{\hskip6pt}CC!{\hskip8pt}CCC!{\hskip14pt}C!{\hskip8pt}CCC!{\hskip8pt}CCC!{\hskip8pt}CCC}
& P_\lambda & P_t & P_x & P_y & P_z
  &  T &  X& Y& Z &  A& B& C &  I& J& K \\[0.2em]\hline
  &    &   &  &   &   &  &   &   &  &   \\[-0.9em]
P_\lambda & 0 & T & X & Y & Z &
\lm P_t & \lm P_x &\lm P_y &\lm P_z & 0 & 0 & 0 & 0 & 0 & 0 \\
P_t & \lm T & 0 & A & B & C &
P_\lambda & 0 & 0 & 0 & \lm P_x & \lm P_y & \lm P_z & 0 & 0 & 0 \\[0.5em]
P_x & \lm X &  \lm A & 0 & \lm K & J &
0 & \lm P_\lambda & 0 & 0 & \lm P_t & 0 & 0 & 0 & P_z & \lm P_y \\
P_y & \lm Y &  \lm B & K & 0 & \lm I &
0 & 0 & \lm P_\lambda& 0 & 0 & \lm P_t & 0 & \lm P_z & 0 & P_x \\
P_z & \lm Z &  \lm C & \lm J & I & 0 &
0 & 0 & 0 & \lm P_\lambda & 0 & 0 & \lm P_t & P_y & \lm P_x & 0 \\[1em]
T & P_t &\lm P_\lambda & 0 & 0 & 0 &
0 &  A& B& C & \lm X&\lm Y&\lm Z &  0& 0& 0 \\[0.5em]
X & P_x & 0 & P_\lambda & 0 & 0 &
\lm A &  0&\lm K& J & \lm T& 0& 0 &  0& Z&\lm Y \\
Y & P_y & 0 & 0 & P_\lambda& 0 &
\lm B &  K& 0&\lm I &  0&\lm T& 0 & \lm Z& 0& X \\
Z & P_z & 0 & 0 & 0 & P_\lambda &
\lm C & \lm J& I& 0 &  0& 0&\lm T &  Y&\lm X& 0 \\[0.5em]
A & 0 & P_x & P_t & 0 & 0 &
X &  T& 0& 0 &  0&\lm K& J &  0& C&\lm B \\
B & 0 & P_y & 0 & P_t & 0 &
Y &  0& T& 0 &  K& 0&\lm I & \lm C& 0& A \\
C & 0 & P_z & 0 & 0 & P_t &
Z &  0& 0& T & \lm J& I& 0 &  B&\lm A& 0 \\[0.5em]
I & 0 & 0 & 0 & P_z & \lm P_y &
0 &  0& Z&\lm Y &  0& C& \lm B &  0& K&\lm J \\
J & 0 & 0 & \lm P_z & 0 & P_x &
0 & \lm Z& 0& X & \lm C& 0& A & \lm K& 0& I \\
K & 0 & 0 & P_y & \lm P_x & 0 &
0 &  Y&\lm X& 0 &  B&\lm A& 0 &  J&\lm I& 0 
\end{tabular}
\end{center}
\caption{The Lie algebra $ \so(3,3) $}
\label{so33_table}
\end{table}

%

The adjoint action of \( \{ T,X,Y,Z,A,B,C,I,J,K \} \) on
\( \{ P_\lambda, P_t,P_x,P_y,P_z \} \) agrees with the usual 
physical interpretation. This suggests that the adjoint action 
of the operators \( \{ P_\lambda, P_t,P_x,P_y,P_z \} \) might
allow us to physically interpret these operators.

The operator \( P_\lambda \) for example rotates
\( P_t \) through \( T \), 
\( P_x \) through \( X \),
\( P_y \) through \( Y \) and
\( P_z \) through \( Z \). 
Hence we might view  \( P_\lambda \) as the operator for reflection or 
inversion of the \( t\),\(x\),\(y\), and \( z \) coordinates 
\label{page:Plambda_inverts}
via rotation through additional dimensions.  Similarly \( P_t \) 
can be viewed as inverting all coordinates except the \( t \) coordinate. 

The operators \( P_x \), \( P_y \) and \( P_z \) are not compact, so
would give hyperbolic rotations through additional dimensions. It is not 
clear how these should be interpreted. 

The operator \( P_x \) commutes with \( T \) and \( I \) and
\( \{ T,I,P_x\} \) is the basis of a maximally compact Cartan subalgebra 
of \( \so(3,3) \). 

\bigskip
\pagebreak[2]

\section{The Enveloping Algebra.}
\label{section_enveloping}

The enveloping algebra is the non-commutative polynomial Lie algebra 
generated abstractly by symbols \( \{ T,X,Y,Z,A,B,C,I,J,K \} \) which 
commute according to the relations in table~\ref{table_so23}.
The commutator is linear and satisfies Leibniz conditions
\begin{flalign*}
[ P,QR ] &= [ P,Q ]R + Q[ P,R ] \\
[ PQ,R ] &= [ P,R ]Q + P[ Q,R ] 
\end{flalign*}
where \( P \), \( Q \) and \( R \) are arbitrary polynomials. The 
commutator (Lie bracket) of any pair of polynomials in the enveloping 
algebra can be determined from these conditions. 

The enveloping algebra is graded by degree. The original Lie algebra 
(in our case \( \so(2,3) \)) is the subalgebra of polynomials of degree 
one. Any representation of the enveloping algebra restricts to a 
representation of \( \so(2,3) \) and  conversely any representation of 
\( \so(2,3) \) can be extended uniquely to a representation of the enveloping 
algebra. In particular the adjoint representation of the enveloping algebra 
on itself gives us a representation of \( \so(2,3) \) on the enveloping algebra. 

It is not hard to show using table~\ref{table_so23} and the Leibniz conditions
for commutators that this representation conserves the space \( V_d \) of
polynomials of degree at most \( d \). Hence \( \so(2,3) \) is also 
naturally represented on the quotient space \( V_{d+1}/V_d \).

The polynomial 
\begin{equation}
\label{2nd_degree_Casimir}
Q = -T^2 + X^2 + Y^2 + Z^2 +A^2 + B^2 + C^2 - I^2 - J^2 - K^2 
\end{equation}
commutes with all of \( \{ T,X,Y,Z,A,B,C,I,J,K \} \) and hence 
lies in the centre of the enveloping algebra. Any representation of 
\( \so(2,3) \) can be extended to a representation of the enveloping 
algebra and in particular \( Q \) is so represented. The eigenspaces 
of this operator are conserved by \( \so(2,3) \), so an irreducible 
representation of \( \so(2,3) \) consists of a single eigenspace of 
\( Q \) with eigenvalue denoted \( \mu \).

The polynomials
\begin{subequations}
\label{enveloping_algebra_P_matrices}
\begin{align}
P_\lambda &=  {\nfrac{2}{3}}(-AI - BJ - CK) \\
P_t       &=  {\nfrac{2}{3}}(XI + YJ + ZK) \\
P_x       &=  {\nfrac{2}{3}}(TI - BZ + CY) \\
P_y       &=  {\nfrac{2}{3}}(TJ - CX + AZ) \\
P_z       &=  {\nfrac{2}{3}}(TK - AY + BX)
\end{align}
\end{subequations}

generate a subspace of the enveloping algebra which is stable 
under the action of \( \so(2,3) \) and which is isomorphic to the 
canonical 5-D representation. The constant of \( \frac{2}{3} \) has
been added to ensure that when evaluated in the case of the
4-D fundamental representation, these definitions agree with those 
in equations~\ref{P_matrix_definitions} in the last 
section\footnote{
However the commutators of these operators in the enveloping algebra are 
not as described in table~\ref{so33_table}, which is valid only for the 4-D 
representation. In particular the set
 $ \{ P_\lambda,P_t,P_x,P_y,P_z, T,X,Y,Z,A,B,C,I,J,K \} $ 
does not form a basis for a copy of $ \so(3,3) $ in the enveloping algebra.}.
                
The fourth degree polynomial 
\begin{equation}
\label{4th_degree_Casimir}
R = P_\lambda^2 + P_t^2 - P_x^2 - P_y^2 - P_z^2 
\end{equation}
is invariant and lies in the centre of the enveloping algebra. The 
eigenspaces of this operator will also be conserved, hence an
irreducible representation consists of a single eigenspace with 
eigenvalue denoted \( \rho \).

The two operators \( Q \) and \( R \) are algebraically independent and 
together generate the centre of the enveloping algebra. They are known 
as the second and fourth degree Casimir operators respectively.

\bigskip
\pagebreak[2]

\section{Roots and Weights}


The Lie algebra \( \so(2,3) \simeq \sp(2,\R) \) is simple. We can therefore 
compute its roots and use those to obtain information about irreducible 
representations. The general theory of roots and weights can be found in 
any introductory Lie algebra text (for example \cite{Hall} or \cite{Knapp}). 
In this section we summarise the consequences of that general theory 
for the specific Lie algebra \( \so(2,3) \).

\medskip

Consider a finite dimensional complex irreducible representation. Such a 
representation associates matrices to elements of the enveloping algebra. 
The matrix for the Casimir element
\( -T^2 + X^2 + Y^2 + X^2 + A^2 + B^2 + C^2 - I^2 - J^2 - K^2 \) 
must commute with the matrices for all elements of the Lie algebra, 
hence eigenspaces of this matrix are invariant. Since by assumption our 
representation is irreducible we conclude that the matrix of the Casimir 
element is scalar.

The Casimir element and Killing form to which it is closely related, can be 
used to separate the Lie algebra into two subspaces. The 
\textbf{compact subspace}  \( \text{span}\{ T, I, J, K \} \) is spanned by 
basis elements with a negative coefficient in the Casimir element, while
the \textbf{non-compact subspace} \( \text{span}\{X,Y,Z,A,B,C\} \) is spanned 
by elements with a positive coefficient.

Our next task is to choose a Cartan subalgebra (CSA)- a maximal commutative 
subalgebra of \( \so(2,3) \). There are many possibilities but we single out 
two of particular utility , a \textbf{compact CSA} \( \text{span}\{ T,I \} \) 
contained in the compact subspace, and a \textbf{non-compact CSA} 
\( \text{span}\{ Y,C \} \) contained in the non-compact subspace. 

The non-compact CSA is useful in proving certain mathematical properties 
of the Lie algebra. However the compact CSA has the greatest physical utility
since compact elements have eigenvalues that are more easily interpreted. 
For example the eigenvalues of a rotation operator which is compact gives 
angular momentum or spin, whereas it is much less clear what the eigenvalues 
of a non-compact operator such as a Lorentz boost might be describing. 
\begin{table}[hbp]
\begin{center}
\begin{tabular}{C!{\hskip1em$($}R!{,}R!{$)$\hskip1em}C}          \toprule
\text{Root} &  \lambda_T & \lambda_I & \text{Root element} \\    \midrule
\cdot      &    0 &  0  &  T \\[0.2em] 
\cdot\cdot &    0 &  0  &  I \\[0.2em] 
\EE &  i & 0  &  \mbox{$\nfrac{1}{\sqrt{2}}(A+iX)$}    \\[0.2em] 
\AW & -i & 0  &  \mbox{$\nfrac{1}{\sqrt{2}}(-A+iX)$}   \\[0.2em] 
\AN &  0 & i  &  \mbox{$\nfrac{1}{\sqrt{2}}(J-iK)$}    \\[0.2em] 
\AS &  0 &-i  &  \mbox{$\nfrac{2}{\sqrt{2}}(J+iK)$}    \\[0.2em] 
\ANE&  i & i  &  \mbox{$\nfrac{1}{2}( C-Y +iB +iZ)$}   \\[0.2em] 
\ANW& -i & i  &  \mbox{$\nfrac{1}{2}( C+Y +iB -iZ )$}  \\[0.2em] 
\ASE&  i &-i  &  \mbox{$\nfrac{1}{2}(-C-Y +iB -iZ )$}  \\[0.2em] 
\ASW& -i &-i  &  \mbox{$\nfrac{1}{2}(-C+Y +iB +iZ )$}  \\[0.2em] \bottomrule
\end{tabular}
\end{center}
\caption{Roots for compact CSA}
\label{roots_compact}
\end{table}

Since the elements of a CSA commute we may seek simultaneous eigenvectors for 
the basis of the CSA. We call these \textbf{weight vectors} and the associated
pair of simultaneous eigenvalues is called a \textbf{weight}. The weights and
weight vectors for the adjoint representation of the Lie algebra have special 
names. We call these \textbf{roots} and \textbf{root elements} respectively. 
Roots and root elements for the compact CSA are given in 
table~\ref{roots_compact}.
\begin{figure}[ht]
  \begin{center}
\includegraphics[width=3cm]{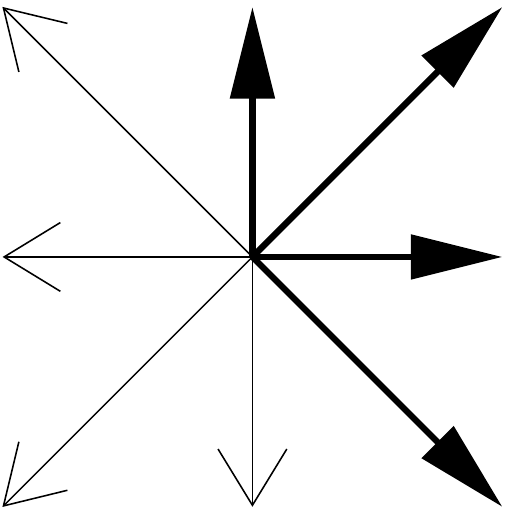}
  \end{center}
  \caption{The Root System for $ \so(2.3) $ }
\label{rootdiagram}
\end{figure}
The components of a root \( (\lambda_T,\lambda_I) \)  for the compact CSA 
are imaginary.  Ignoring the imaginary unit we can depict roots graphically 
in a root diagram as shown in figure~\ref{rootdiagram}. We have chosen to 
name our root elements by using the corresponding vectors in the root diagram.
The root elements extend the basis of the CSA to a full basis of the Lie 
algebra.  .

Suppose $ (r_1,r_2)$ is a root for root element $ R $ with respect to 
some CSA $ \text{span}(Q_1,Q_2) $. Let the Lie algebra be represented
on $ V $ and let $ v \in V $ be a weight vector for weight $ (\mu_1,\mu_2) $.
Then
\begin{equation}
Q_i(R(v)) = R(Q_i(v)) + [Q_i,R](v)
          = \mu_i R(v) + r_i R(v) \\
          = (\mu_i + r_i) R(v) 
\end{equation}

so $ R(v) $ is either zero or is a weight vector for weight 
$ (\mu_1 + r_1, \mu_2 + r_2) $.  We conclude that root elements 
map weight vectors to weight vectors, and map corresponding weights 
by adding the root.

As a corollary of this the set of all weight vectors spans an invariant 
subspace, and so any irreducible representation has a basis of weight vectors.
In such a basis the matrices for elements of the CSA are diagonal. As every 
element of \( \so(2,3) \) belongs to at least one CSA we conclude that 
all matrices in an irreducible representation are diagonalisable.

If we start with one given weight vector $v$, then the set of all weight 
vectors obtained by repeated application of root elements to $v$ spans a 
non-trivial invariant subspace and hence there is a basis of such vectors 
in the irreducible case. Hence the set of weights for an irreducible 
representation is contained in a lattice in $ \C^2 $ invariant under 
translation by root vectors. For each weight we may look at the dimension 
of the space of weight vectors belonging to that weight. We call this 
dimension the degree of the weight. A weight diagram for an irreducible
representation is a plot of the set of all weights labelled with their 
degrees.

For a finite dimensional representation only finitely many weights appear. 
The components of each weight are pure imaginary and can be ordered. We can
then order the weights themselves via the library ordering of their components
and identify a maximal weight and a minimal weight. 

The same ordering can be applied to roots by classifying the non-zero roots 
as either positive and negative depending on whether translation by this root 
constitutes an increase or a decreases under the weight ordering. 
Figure~\ref{rootdiagram} shows a depiction of the roots with positive roots
marked. 

\bigskip

The root diagram and all the weight diagrams are symmetrical with respect to 
the Weyl group. The Weyl group is given by the natural action of the Lie group 
on the Lie algebra. In particular it is the quotient of the stabiliser by the 
centraliser of the set of roots under this natural action. The Weyl group for 
\( \so(2,3) \) is the symmetry group \( D_8 \), the dihedral group of order 8, 
and is generated by reflections of the root diagram which map a root to its
negative.

In particular the minimal weight must be the negative of the maximal weight 
since inversion is an element of the Weyl group. Also the set of weights must 
be invariant under exchange of coordinates since this also is an element of 
the Weyl group.  For an irreducible representation all weights, and in 
particular the maximal weight, can be obtained by a sequence of translations 
by positive roots from the minimal weight.

We conclude that the maximal weights for an irreducible representation 
take the form \( \bigl(mi,ni\bigr) \) or 
\( \bigl((m + \nfrac{1}{2})i,(n+ \nfrac{1}{2})i\bigr) \) 
where \( m,n \in \Z \).

\bigskip

It is possible to show (we not do so here) that all finite dimensional 
irreducible complex representations of \( \so(2,3) \) have a basis in which 
the matrices for elements of \( \so(2,3) \) are all real. As a consequence 
of this result the real irreducible representations can be identified with 
the complex ones.

\begin{figure}[htp]
\hangsubcaption
\centering
\subbottom[1-dimensional $(q_0,s_0) = (0,0)$
  \label{weight0-0}]{\includegraphics[width=3cm]{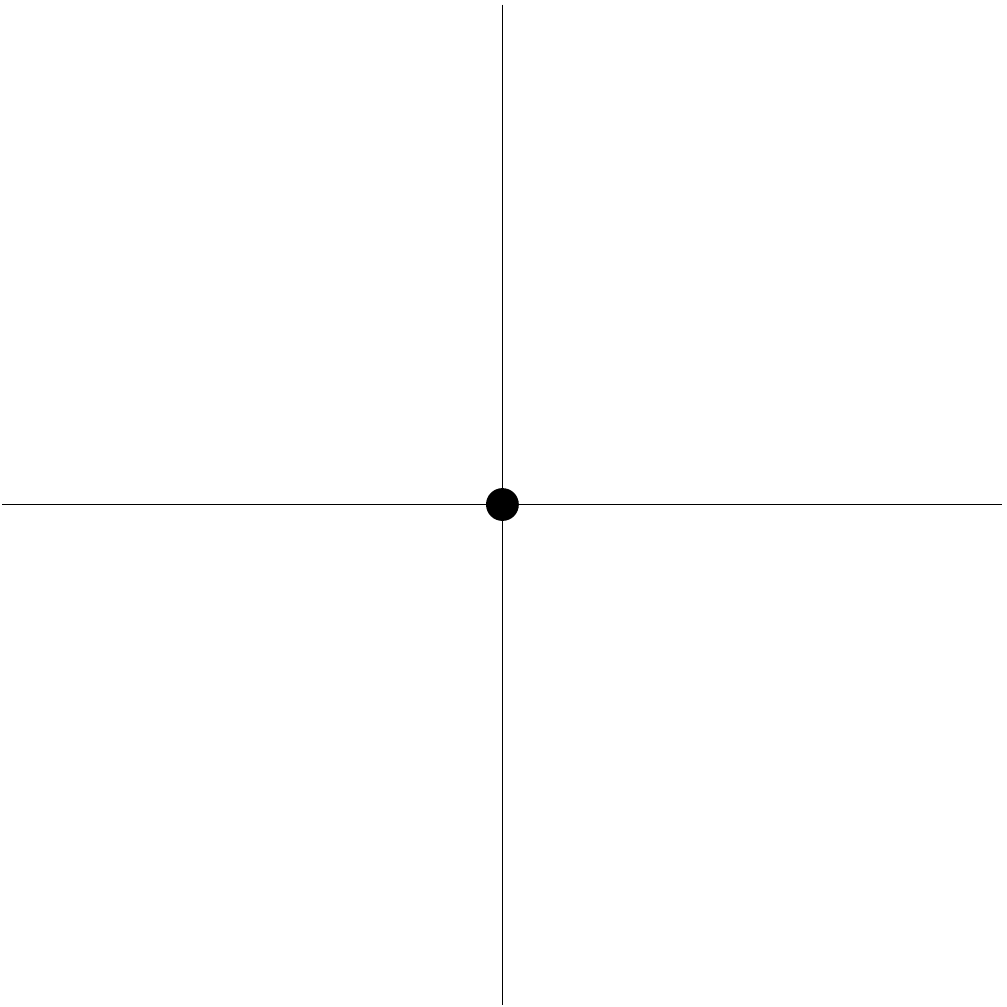} }
 \hskip1cm
\subbottom[4-dimensional $(q_0,s_0) = (\frac{1}{2},\frac{1}{2})$
  \label{weight1-1}]{\includegraphics[width=3cm]{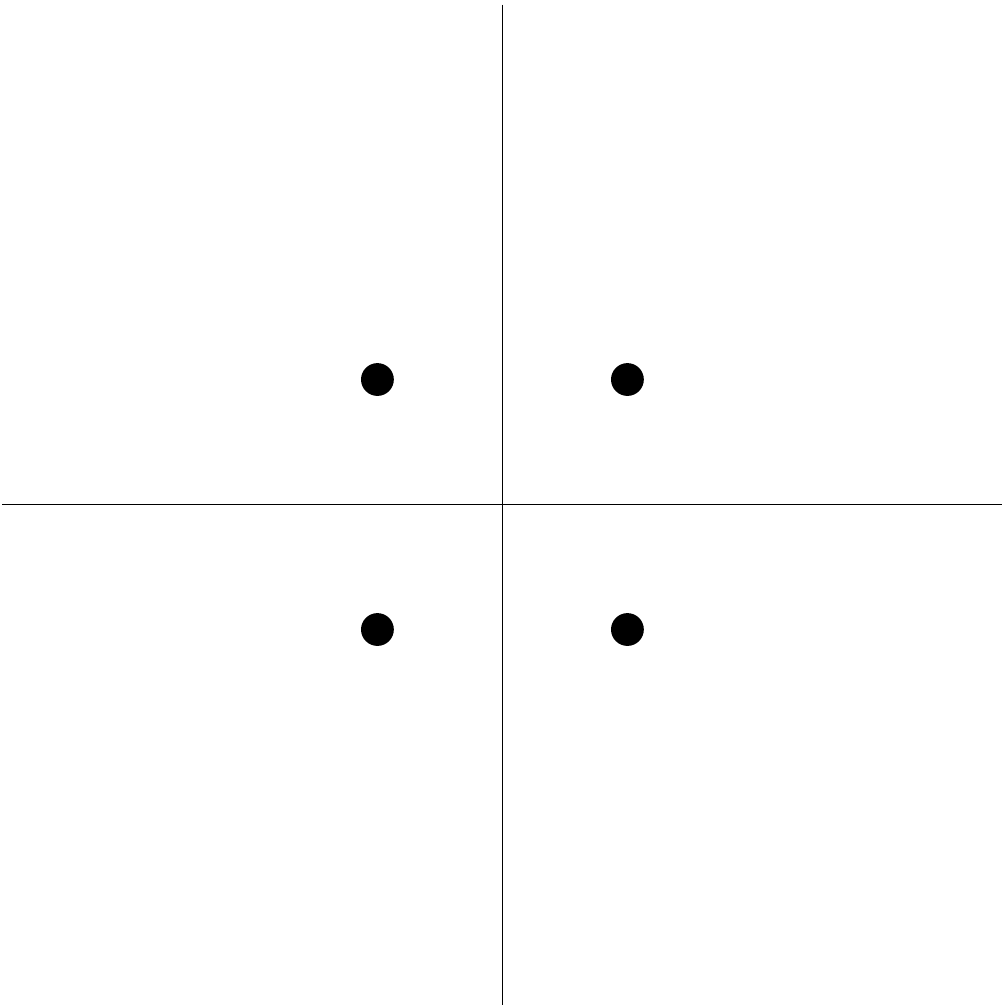} }
 \hskip1cm
\subbottom[5-dimensional $(q_0,s_0) = (1,0)$
   \label{weight2-0}]{\includegraphics[width=3cm]{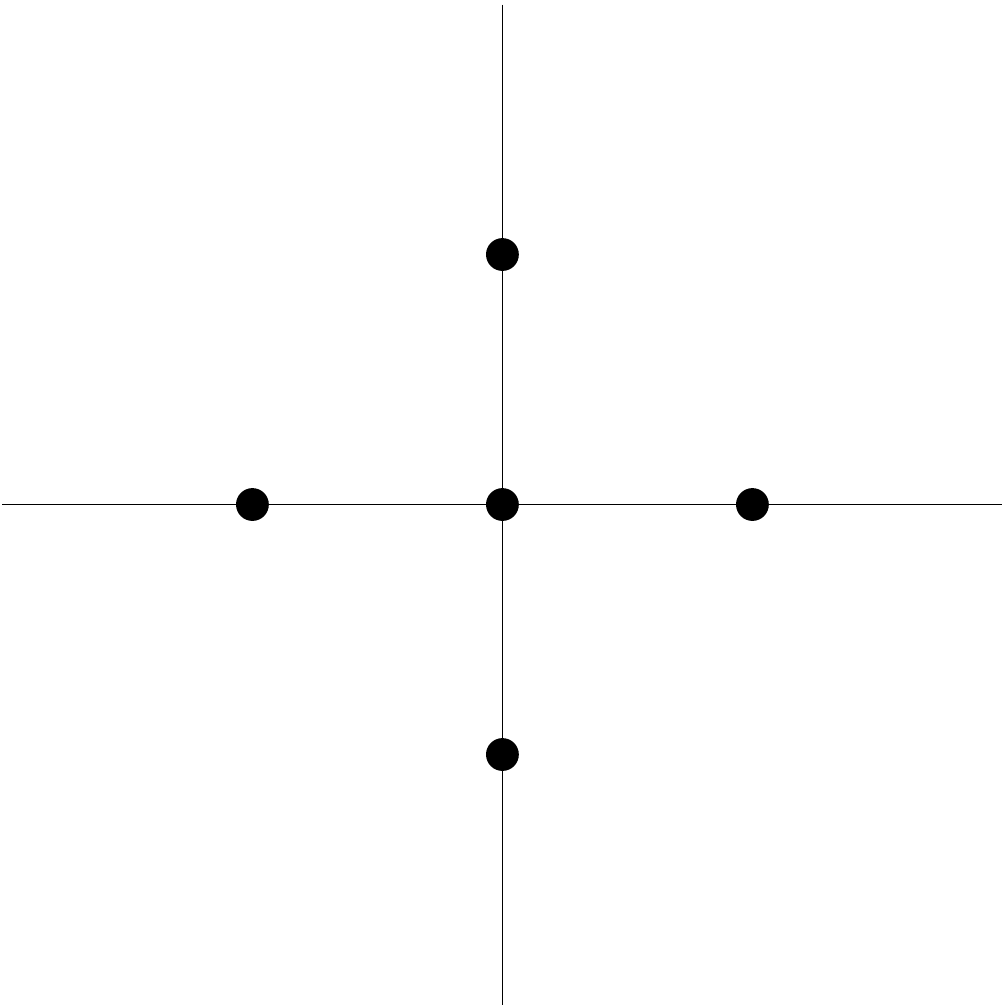} }
\\[1cm]
\subbottom[10-dimensional $(q_0,s_0) = (1,1)$
   \label{weight2-2}]{\includegraphics[width=3cm]{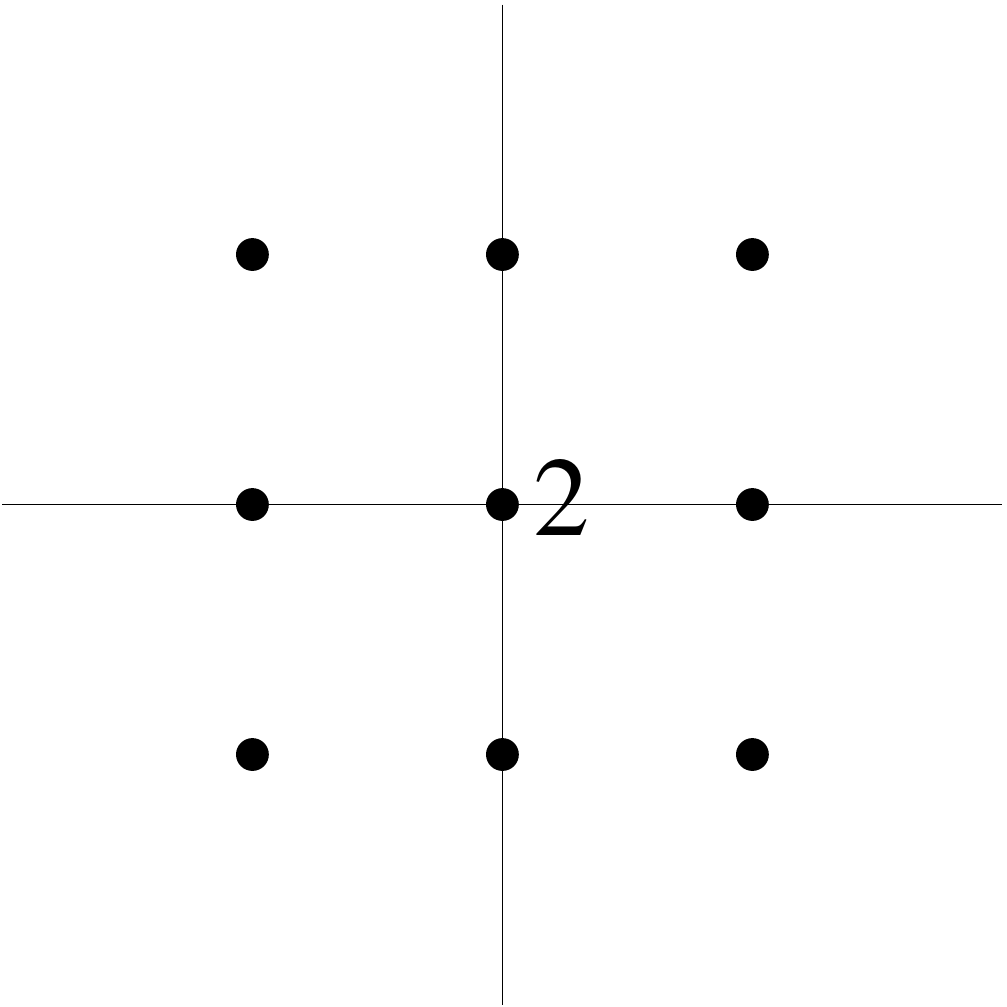} }
 \hskip1cm
\subbottom[16-dimensional $(q_0,s_0) = (\frac{3}{2},\frac{1}{2})$
   \label{weight3-1}]{\includegraphics[width=3cm]{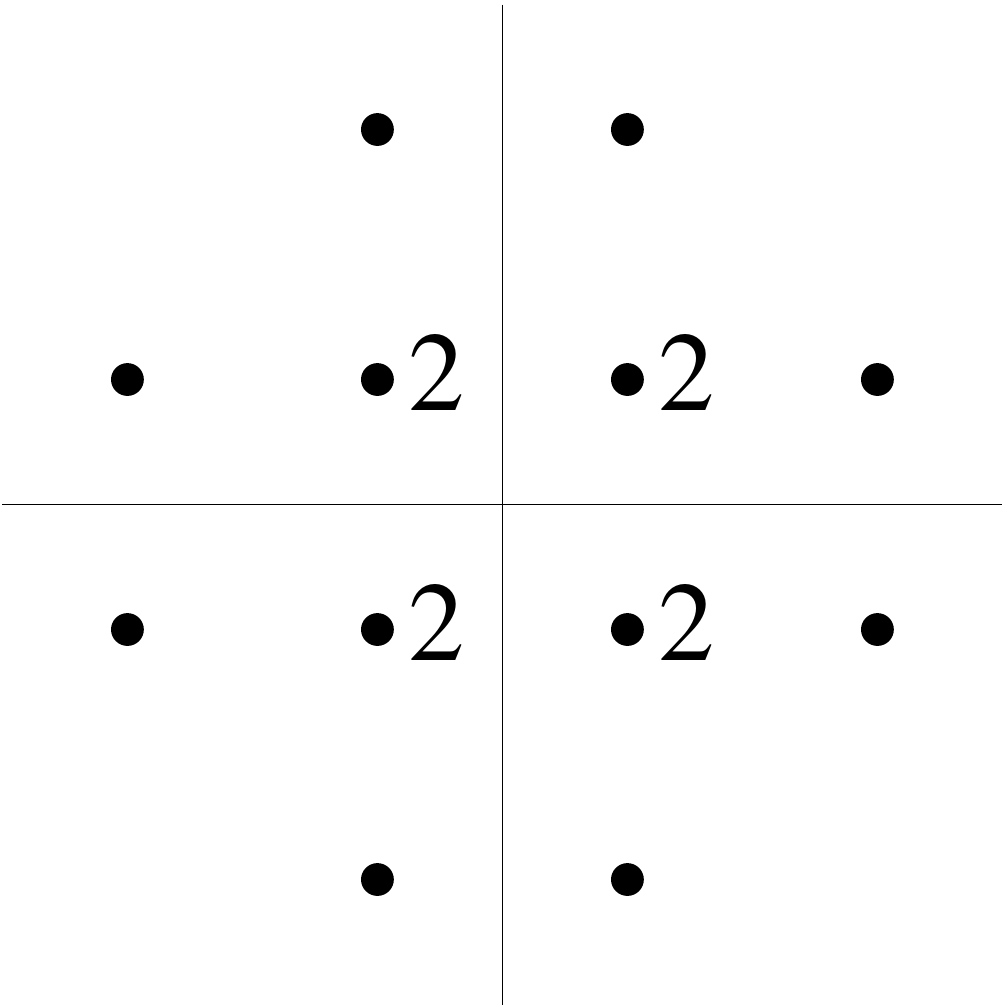} }
 \hskip1cm
\subbottom[20-dimensional $(q_0,s_0) = (\frac{3}{2},\frac{3}{2})$
   \label{weight3-3}]{\includegraphics[width=3cm]{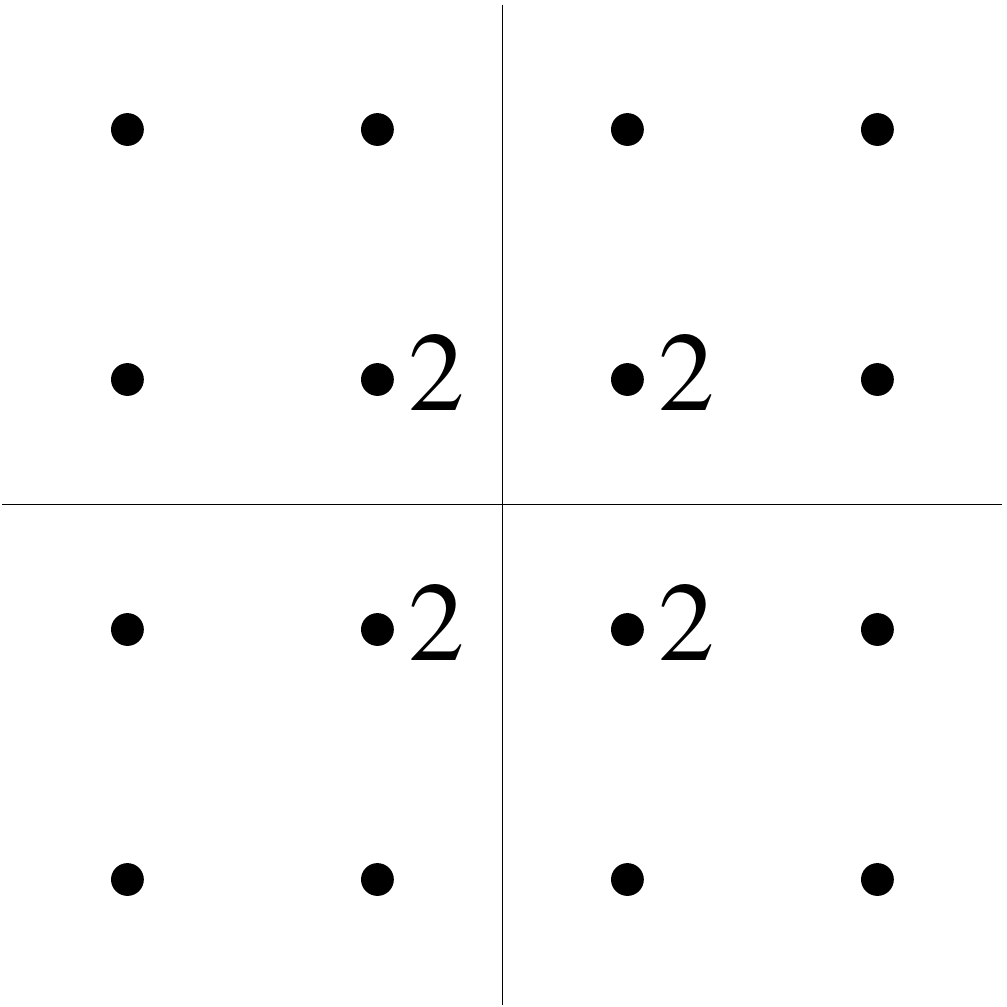} }
\\[1cm]
\subbottom[14-dimensional $(q_0,s_0) = (2,0)$
   \label{weight4-0}]{\includegraphics[width=3cm]{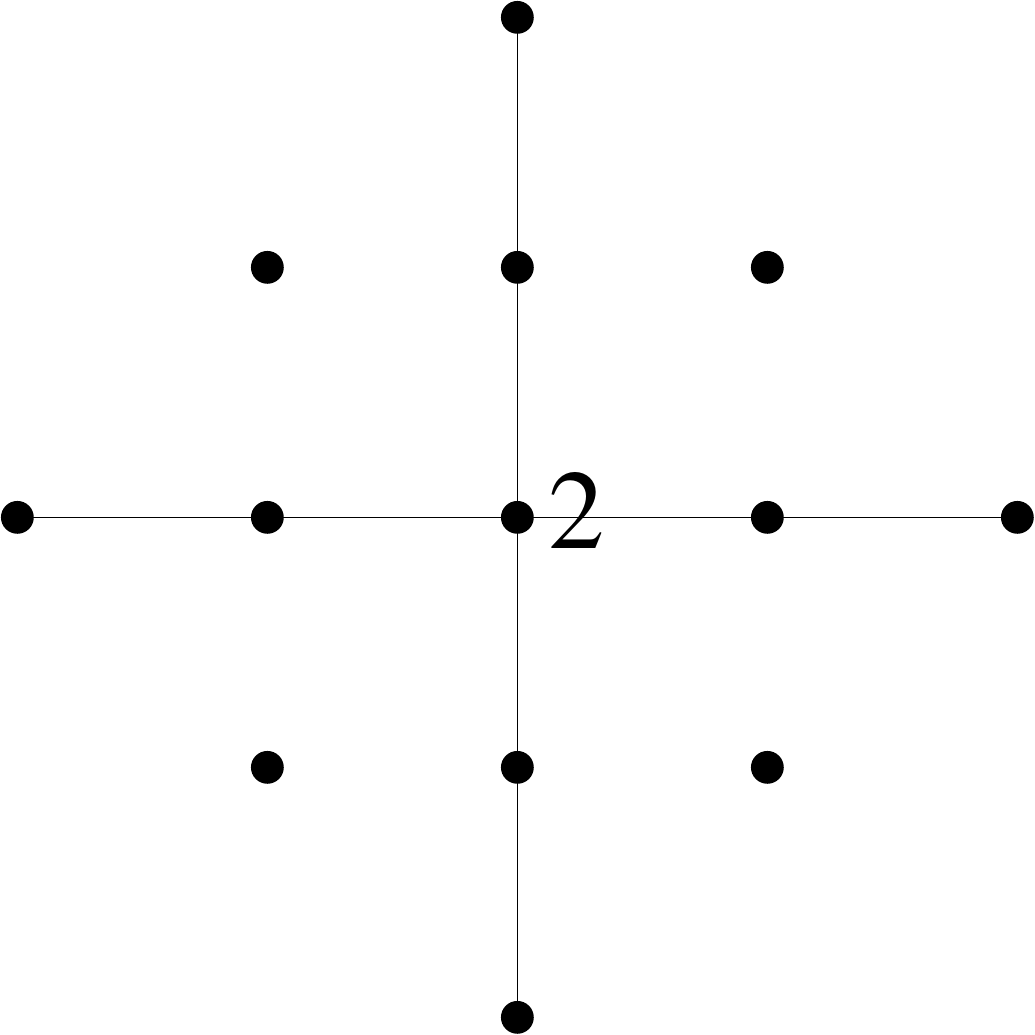} }
 \hskip1cm
\subbottom[35-dimensional $(q_0,s_0) = (2,1)$
   \label{weight4-2}]{\includegraphics[width=3cm]{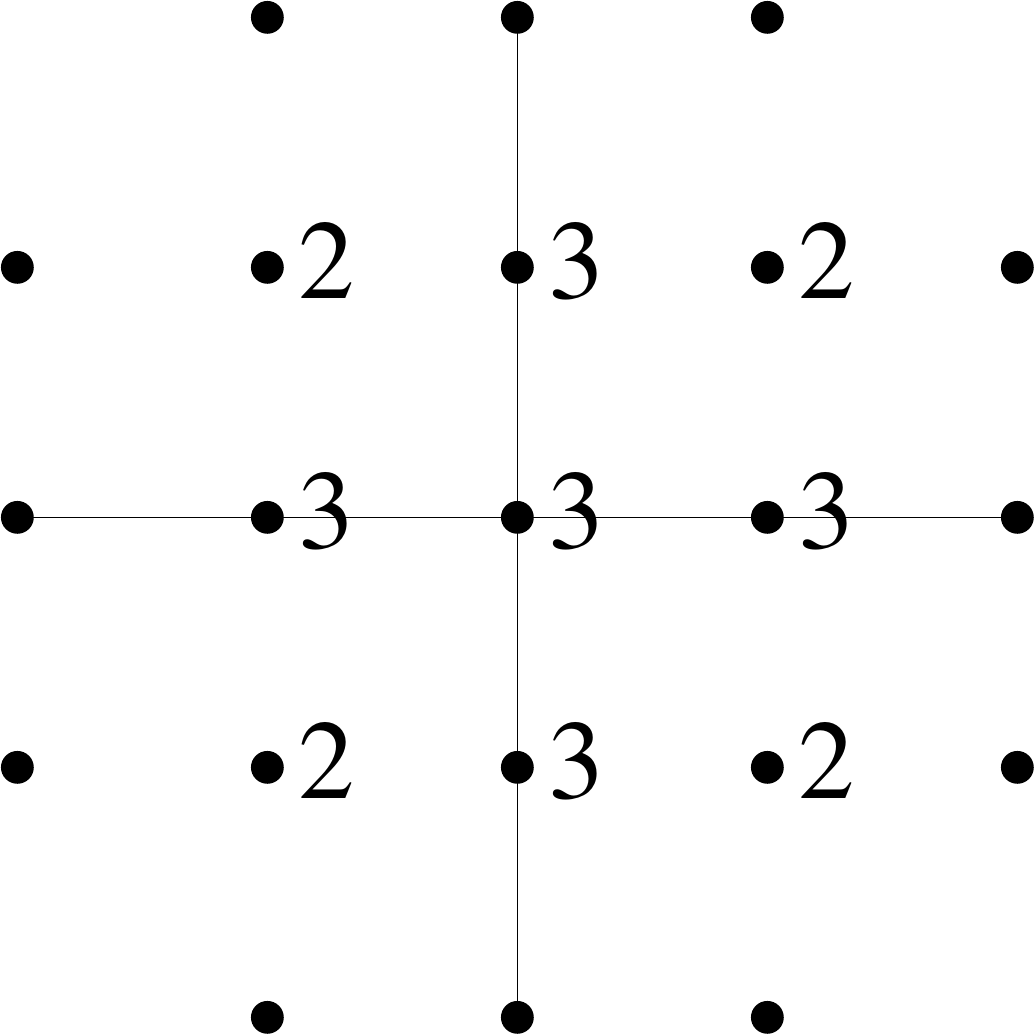} }
 \hskip1cm
\subbottom[35-dimensional $(q_0,s_0) = (2,2)$
   \label{weight4-4}]{\includegraphics[width=3cm]{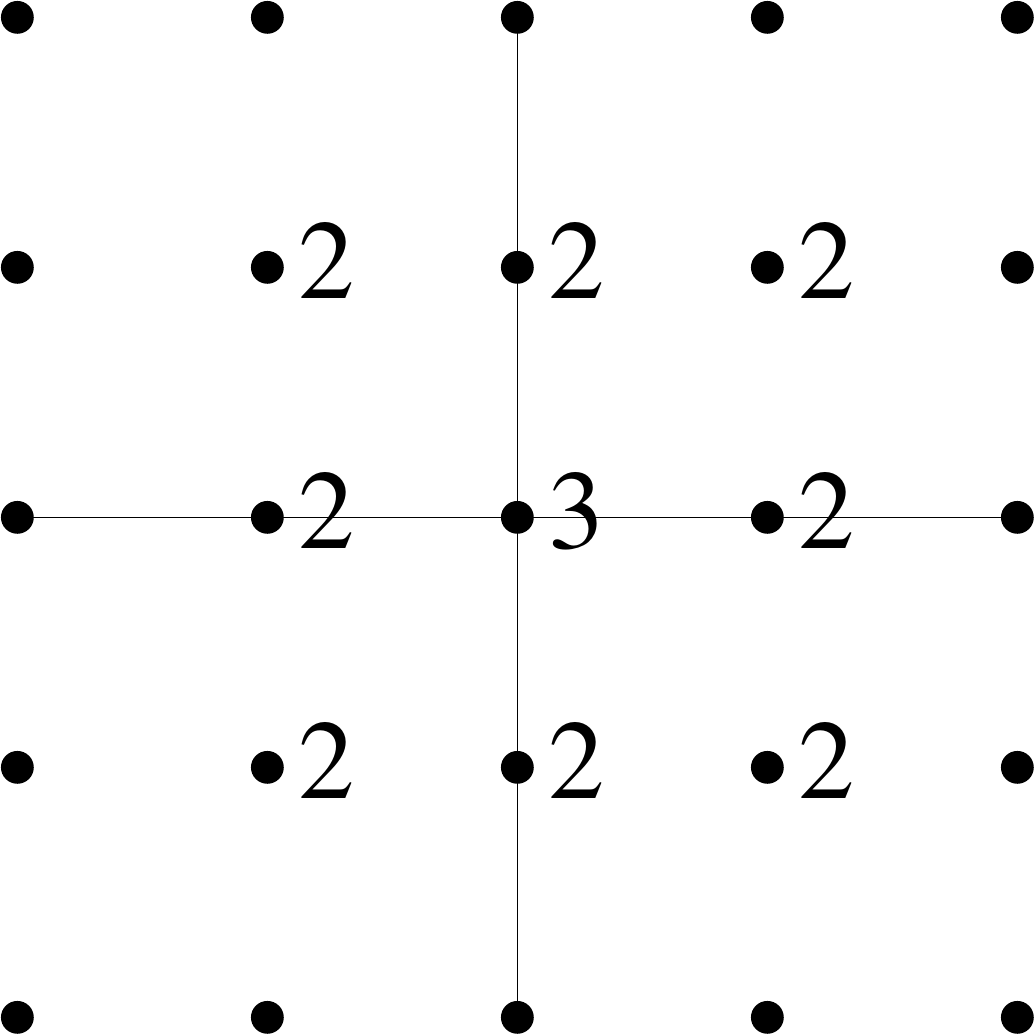} }
\\[1cm]
  \caption{Weight diagrams for small dimensional irreducible 
representations of $\so(2,3)$. Each representation is labelled by its
dimension and highest weight $(q_0,s_0)$ }
\label{weights}
\end{figure}

Any irreducible representation has a basis of weight vectors. The degree 
of each weight in the representation is defined to be the dimension of 
the corresponding subspace of weight vectors. In a finite dimensional 
representation the degrees can be calculated using Kostant's formula 
(Hall \cite{Hall} Theorem 7.42). The sum of the degrees of all weights 
is the dimension of the representation. The weights together with their 
degrees can be depicted in a weight diagram for the representation. 
Weight diagrams for some small dimensional irreducible representations 
are given in figure~\ref{weights}. The representations are labelled by
their maximal weights.

\medskip

Consider now a finite dimensional irreducible representation with 
maximal weight \( (q_0,s_0) \). Then \( Q \) and \( R \) act as 
scalar operators on this representation with eigenvalues \( \mu \) 
and \( \rho \) respectively. Clearly \( \mu \) and \( \rho \) are 
functions of \( q_0 \) and \( s_0 \). We seek now to make this explicit.

The following operators are compositions of roots which begin with a
positive root and hence act as zero operators on the maximal weight-space. 
\begin{multline*}
(Y-iZ+iB+C)(Y+iZ-iB+C) = Y^2 + Z^2 + B^2 + C^2  \\
		+ \{ Y,C\} - \{ Z,B \} - 2iI + 2iT 
\end{multline*}
\begin{multline*}
(Y+iZ+iB-C)(Y-iZ-iB-C) = Y^2 + Z^2 + B^2 + C^2 \\
		- \{ Y,C\} + \{ Z,B \} + 2iI +2iT
\end{multline*}
\begin{flalign*}
(X+iA)(X-iA) &= X^2 + A^2 + iT \\
(J+iK)(J-iK) &= J^2 + K^2 - iI
\end{flalign*}

These combine to give useful identities on the maximal weight-space
\begin{equation}
\label{Qidentity}
Q = -T^2 - I^2 - iI - 3iT
\end{equation}
\begin{equation}
\label{missingpiece}
\{ Y,C\} - \{ Z,B \} = 2iI
\end{equation}

Since \( T \) and \( I \) have eigenvalues \( iq_0 \) and \( is_0 \)
respectively while \( Q \) has eigenvalue \( \mu \), equation~\ref{Qidentity}
gives
\begin{equation}
\label{Casimirformula}
\mu = q_0(q_0+3)+s_0(s_0+1)
\end{equation}

Evaluating this at \( (q_0,s_0) = (\frac{1}{2},\frac{1}{2}) \) gives
\( \mu = \frac{5}{2} \) in agreement with the value directly computed 
using the matrices in table~\ref{sp4R_basis};
while evaluating at \( (q_0,s_0) = (1,0) \) gives
\( \mu = 4 \) which can be verified using the matrices in 
table~\ref{so23_5x5}. The 10 dimensional regular representation with 
\( (q_0,s_0)=(1,1) \) has \( \mu = 6 \).

\medskip

To compute a similar formula for the 4th degree Casimir invariant
we begin by examining the adjoint action of \( \so(2,3) \) on the
5-D space with basis \( \{ P_\lambda,P_t,P_x,P_y,P_z \} \). As 
the operators \( T \) and \( I \) commute, this space has a basis 
of simultaneous eigenvectors which gives a second set of root-like vectors 
on the weights. We call these \( P\)-roots.

\begin{table}[htb]
\begin{center}
\begin{tabular}{!{$($}R!{,}R!{$)$\hskip1cm}C} 
\toprule
\lambda_T & \lambda_I  & \text{Basis element} \\ 
\midrule
 0& 0  &       P_x \\
 i& 0  &  P_\lambda - iP_t  \\ 
-i& 0  &  P_\lambda + iP_t  \\ 
 0& i  &  P_y + iP_z  \\ 
 0&-i  &  P_y - iP_z  \\
\bottomrule
\end{tabular}
\end{center}
\caption{$P$-roots}
\label{Proots}
\end{table}
 
Products of $P$-roots which begin with a positive $P$-root must
act as zero operators on the maximal weight-space. Hence the following
operators are zero on the maximal weight-space.
\begin{flalign*}
(P_\lambda + iP_t)(P_\lambda - iP_t) &= P_\lambda^2 +P_t^2 - iT \\
(P_y - iP_z) (P_y + iP_z) &= P_y^2 + P_z^2 - iI
\end{flalign*}

giving two useful identities on the maximal weight-space
\begin{flalign}
P_\lambda^2 +P_t^2 = iT \\
P_y^2 + P_z^2 = iI
\end{flalign}

The definition of \( P_x \) together with equation~\ref{missingpiece}
gives another identity
\begin{equation}
P_x = \nfrac{2}{3}\left(TI + iI\right) 
\end{equation}

Combining these we obtain
\begin{equation}
-P_\lambda^2-P_t^2+P_x^2+P_y^2+P_z^2 = -iT +iI +\nfrac{4}{9}I^2\left(T + i\right)^2
\end{equation}

which gives a formula for the eigenvalue \( \rho \)  of the 4th order Casimir 
operator \( R \).
\begin{equation}
\rho = q_0 -s_0 +\nfrac{4}{9}s_0^2\left(q_0 + 1\right)^2
\label{Cfourformula}
\end{equation}

Evaluating this at \( (q_0,s_0) = (\frac{1}{2},\frac{1}{2}) \) gives
\( \rho = \frac{1}{4} \); while evaluating at \( (q_0,s_0) = (1,0) \) gives
\( \rho = 1 \). The 10 dimensional regular representation with 
\( (q_0,s_0)=(1,1) \) has \( \rho = \frac{16}{9} \).

\bigskip
\pagebreak[2]

\section{The Charge Multiplier}

We have suggested that we should interpret intrinsic energy as charge with 
units chosen so that \( e = q = \frac{1}{2} \).  While this works nicely for 
electrons and positrons, it works considerably less well for the other 
particles in the modern particle physics menagerie. 

In particular this interpretation of charge would require quarks to have 
representations including weights 
\( (q,s) = (\pm\frac{1}{6},\pm\frac{1}{2})) \) or 
\( (q,s) = (\pm\frac{1}{3}, \pm\frac{1}{2}) \) 
none of which appear in any finite dimensional representation of 
\( \so(2,3) \).  Even neutrinos are problematic as they would need 
the weight \( (q,s) = (0,\pm\frac{1}{2}) \) which doesn't appear either.

This problem must be addressed as we simply can't have particles with these 
kinds of problematic weights if our model is to make sense. Imagine for 
example finding a free particle with spin \( \frac{1}{4} \). How would 
rotations even act on this thing? Our current difficulty is of a similar 
order of magnitude.

To escape this trap which we have constructed for ourselves we must revisit
the identification of intrinsic energy \( q \) with charge \( Q \) which is 
what has caused this problem. We do want intrinsic energy and charge to be 
closely related, but perhaps we went too far when we asserted that they should
actually be the same.

Instead of asserting that \( Q = q \) let us instead suppose that \( Q = qm \) 
where \( m \) is some constant multiplying factor which depends on the nature 
of the particle.

All elementary fermions will fit into this revised picture using the 
fundamental \( (q_0,s_0) = (\frac{1}{2},\frac{1}{2}) \) representation with 
\( m \) in the set \( \{ 0,\frac{1}{3},\frac{2}{3},1 \} \).

The same fix works for bosons as well. Photons and gluons can be described 
using the \( (q_0,s_0) = (1,1) \) and \( (q_0,s_0) = (1,0) \) representations 
with \( m = 0 \). Weak bosons would belong to a \( (q_0,s_0) = (1,1) \) 
representation with \( m = \frac{1}{2} \), while Gravitons would presumably be 
described using one of the three \( q_0 = 2 \) representations with 
\( m = 0 \).

There are some issues remaining with the multiplicities of various states in 
this picture, particularly in the case of bosons. We also must explain the 
origin of this mysterious multiplying factor \( m \). Without trivialising 
these issues however, they are not grounds for abandoning our model at such 
an early stage. We can hope to find explanations for these things later if we 
continue. They are not fatal flaws which would cause us to stop.

	\chapter{Symmetry and Manifolds}
\label{symmetryandmanifolds}

In the first chapter we discussed why it would make a great deal of sense to
use the group \( \so(2,3) \) instead of the Poincar\'e group to describe local
symmetry. In this chapter we try to attach a precise mathematical meaning to
this vague physical statement.

It is not difficult to simply pin a Lie algebra to an arbitrary manifold. 
However this would not serve our purpose. We want the Lie algebra to be 
attached to the manifold in a natural fashion so that it arises from the 
structure of the manifold itself.  In some sense we want the Lie algebra 
at every point to describe local symmetry.

To make this work we will need to make some fairly subtle and careful 
distinctions. On a manifold with curvature, translations do not commute.
Indeed curvature itself is essentially a description of the failure of
parallel transport to commute. On the other hand Lie algebras also describe
the nature of a failure to commute. The Lie algebra we wish to attach to
the manifold involves translations which don't commute. It seems like we
have two mathematical structures competing to describe the same thing.

Resolving this apparent clash requires a deeper understanding of the
relationship between curvature and local symmetry. We will discover that we
are able to do this in very nice way provided the manifold has the same 
dimension as the Lie algebra. Hence the theory for \( \so(2,3) \) arising 
from this approach will work only when the manifold is ten dimensional. 

This is not really a problem.  The 'extra' dimensions in this case are simply 
those of rotation and Lorentz boost and the ten dimensional manifold is just 
the manifold of inertial frames. Hence, unlike some other theoretical 
approaches which introduce extra dimensions, in our case the additional six 
dimensions are quite physical and we need make no attempt to explain why they
are not observed.

Once we have built a ten dimensional curved manifold with local symmetry 
\( \so(2,3) \) the next task will be to look at how other mathematical 
structures, specifically spinor valued wave functions on the manifold
 can be added in a consistent way. Finally we show that all of this 
mathematical structure flows from a small number of very simple and quite
physical assumptions.

\bigskip
\pagebreak[2]

\section{Tensor Derivations and Manifolds}
\label{TensorDerivationsandManifolds}

We will begin by setting up some basic notation describing curvature on 
a manifold.  The approach we follow is inspired by the approach taken in 
Wald's book~\cite{Wald}.

\begin{definition}
A \textbf{tensor derivation} \( D \) on a manifold \( M \) is a map  \linebreak
\( D:\text{tensors} \longrightarrow \text{tensors} \) satisfying

\begin{itemize}
\item Linearity.
\item Leibniz condition on tensor products.
\item Commutes with contraction.
\end{itemize}

\end{definition}

We do not require tensor derivations to conserve degree.
The composition of two tensor derivations usually fails to obey the 
Leibniz condition and is therefore not a tensor derivation. 
However we can easily verify the following


\begin{proposition}
If \( D \) and \( E \) are tensor derivations, then the commutator 
\( [D,E] \) is also a tensor derivation where 
\( [D,E](X) = D(E(X)) - E(D(X)) \).
\end{proposition}

The commutator of two linear operators obeys the Jacobi identity, so we 
might expect tensor derivations to form a Lie algebra. However addition of
two arbitrary tensor definitions is problematic.  The sum \( D+E \) of two 
tensor derivations is given by \( (D+E)(X) = D(X) + E(X) \). But this
only makes sense if \( D(X) \) and \( E(X) \) have the same rank so that 
they can be added.  If this is not always the case then \( D + E \) is 
undefined. 

Since it lacks a well defined addition the set of all tensor derivations is 
not a Lie algebra.  However any subset of the set of tensor derivations
on which addition is defined will be a Lie algebra if it is closed under 
commutator, addition, and under the obvious scalar multiplication.

\medskip

Let \( D \) be a tensor derivation. A function \( f \) on \( M \) is 
a tensor of rank zero, hence \( D(f) \) is a tensor. Suppose 
\( D(f) \) has rank \( \rank{i}{j} \). Let \( T \) be a tensor of rank
\( \rank{k}{l} \). Now  \( D(fT) = D(f)T + fD(T) \) using the Leibniz 
condition. The first term has rank \( \rank{i+k}{j+l} \), hence the second 
term must have the same rank since the Leibniz condition requires that 
addition be well defined. It follows that the rank of \( D(T) \) is 
\( \rank{i+k}{j+l} \). We call \( \rank{i}{j} \) the rank of \( D \) and
express \( D \) using index notation in the form 
\( D_{\alpha_1\alpha_2\ldots\alpha_j}^{\beta_1\beta_2\ldots\beta_i} \).
We have proved

\begin{proposition} Every tensor derivation has a rank \( \rank{i}{j} \)
and maps tensors of rank \( \rank{k}{l} \) to tensors of rank 
\( \rank{k+i}{l+j} \).
\end{proposition}

\begin{proposition} If \( D \) is a tensor derivation and \( S \) is
any tensor, then \( S\tensor D \) is a tensor derivation where 
\( (S\tensor D)(T) = S\tensor D(T) \).
\end{proposition}
\begin{proof}
Just check the definition.
\end{proof}

Every ordinary derivation can be extended to a tensor derivation of rank 
\( \rank{0}{0} \) by allowing it to act on components. We use 
\( T \mapsto a^i \del_i (T ) \) to denote this tensor 
derivation where \( \del_i = \frac{\del}{\del x_i} \).

Conversely every tensor derivation \( D \) of rank \( \rank{0}{0} \) acts 
on functions as an ordinary derivation. Hence we may associate to it a 
tangent vector field so that with respect to some coordinate system we 
have \( D(f) = a^i\del_i (f) \). 

Note that the two operators need only agree in their action on functions.
It follows that the difference \( D - a^i\del_i \) is a tensor
derivation of rank \( \rank{0}{0} \) which maps all functions to the zero 
function.

\begin{proposition} Let \( E \) be a tensor derivation of rank 
\( \rank{0}{0} \) with \( E(f) = 0 \) for all functions \( f \) on \( M \).
Then there exists a tensor \( \Gamma^i_j \) of rank \( \rank{1}{1} \)
so that 
\[
E\left( 
X^{\alpha_1\alpha_2\ldots\alpha_m}_{\beta_1\beta_2\ldots\beta_n}
\right) = 
\sum_s\Gamma^{\alpha_s}_{\hat{\alpha}_s}
X^{\alpha_1\ldots\hat{\alpha}_s\ldots\alpha_m}_{\beta_1\beta_2\ldots\beta_n}
- \sum_t\Gamma^{\hat{\beta}_t}_{\beta_t}
X^{\alpha_1\alpha_2\ldots\alpha_m}_{\beta_1\ldots\hat{\beta}_t\ldots\beta_n}
\]
\end{proposition}
\begin{proof}
If \( v \) is a vector field and \( f \) is a scalar field then 
\( E(fv) = E(f)v + fE(v) = fE(v) \). If follows that \( E \) acts linearly
on the tangent vector fields on \( M \). Hence its action on vector fields
is contraction with a tensor of rank \( \rank{1}{1} \) (locally matrix 
multiplication). 

Indeed if the vector fields \( \{ e_i \} \) form a basis 
of the tangent spaces at each point so that we can write \( v = v^ie_i \) 
in terms of this basis, then we can explicitly find this tensor, since
\( E(v) = E(v^ie_i) = v^iE(e_i) = v^i\Gamma^j_i e_j \) where 
\( \Gamma^j_i \) is the \( j\)-th components of \( E(e_i) \). 
Using coordinates to describe tensors, we can thus write
\[ E(v^i) = \Gamma^i_t v^t \]

Now \( 0 = E(u_iv^i) = E(u_i)v^i + u_i\Gamma^i_t v^t \), and since this is
true for all \( v \) it follows that 
\[ E(u_i) = -\Gamma^t_iu_t \]

We can now argue inductively on the total rank of tensor \( X \) by 
contracting with either \( u_i \) or \( v^i \) to diminish the total
rank and applying \( E \) to the resulting tensor. 
\end{proof}

Conversely, let \( \Gamma \) be a tensor of rank \( \rank{1}{1} \). We
can define a tensor derivation denoted \( \Gamma\oo \) by

\begin{equation}
\label{definition_of_oo}
\Gamma\oo (
X^{\alpha_1\alpha_2\ldots\alpha_m}_{\beta_1\beta_2\ldots\beta_n} )
= \sum_s\Gamma^{\alpha_s}_{\hat{\alpha}_s}
X^{\alpha_1\ldots\hat{\alpha}_s\ldots\alpha_m}_{\beta_1\beta_2\ldots\beta_n}
- \sum_t\Gamma^{\hat{\beta}_t}_{\beta_t}
X^{\alpha_1\alpha_2\ldots\alpha_m}_{\beta_1\ldots\hat{\beta}_t\ldots\beta_n}
\end{equation}

We have proved that all rank \( \rank{0}{0} \) tensor derivations
are of the form 
\[ a^i\del_i + \Gamma\oo \] 
for some choice of vector field \( a^i \) and tensor \( \Gamma^t_s \).

Now let \( D^{\lambda_1\ldots\lambda_m}_{\mu_1\ldots\mu_n} \) be a 
tensor derivation of rank \( \rank{m}{n} \). By contracting this with 
component tensors we see that each of the operators obtained by fixing 
the values of all indices \( \lambda_i \) and \( \mu_j \) is a tensor 
derivation of rank \( \rank{0}{0} \). This leads us to the following.

\begin{proposition} Every tensor derivation of rank \( \rank{m}{n} \) takes
the form

\[ D^{\lambda_1\ldots\lambda_m}_{\mu_1\ldots\mu_n}
= a^{\lambda_1\ldots\lambda_m\,i}_{\mu_1\ldots\mu_n}\del_i
+ \Gamma^{\lambda_1\ldots\lambda_m}_{\mu_1\ldots\mu_n}\oo \]

where 
\[ \Gamma^{\lambda_1\ldots\lambda_m}_{\mu_1\ldots\mu_n}\oo 
\left(
T^{\alpha_1\ldots\alpha_p}_{\beta_1\ldots\beta_q} 
\right) = 
\sum_s 
\Gamma^{\lambda_1\ldots\lambda_m\,\alpha_s}_{\mu_1\ldots\mu_n\,\hat{\alpha}_s}
\,\,T^{\alpha_1\ldots\hat{\alpha}_s\ldots\alpha_p}_{\beta_1\ldots\beta_q} 
- \sum_t 
\Gamma^{\lambda_1\ldots\lambda_m\,\hat{\beta}_t}_{\mu_1\ldots\mu_n\,\beta_t}
\,\,T^{\alpha_1\ldots\alpha_p}_{\beta_1\ldots\hat{\beta}_t\ldots\beta_q} 
\]
\end{proposition}

\medskip

The tensor derivations of rank \( \rank{0}{1} \) are of particular importance.
These are  of the form 

\begin{equation}
\label{rank_1_tensor_derivations}
D_i = a_i^k \del_k + \Gamma_i\oo 
\end{equation}

In the special case that \( a_i^k = 1_i^k \) we call this a \textbf{covariant
derivative}. Much of the subsequent work will take place in the context of 
a manifold with a distinguished covariant derivative \( \nabla_i \).

\bigskip
\pagebreak[2]

\section{Torsion and Curvature}
\label{TorsionandCurvature}

Let \( M \) be a manifold with distinguished covariant derivative 
\( \nabla_i \).

Since the commutator of tensor derivations is a tensor derivation
\( [ \nabla_i,\nabla_j] \) is a tensor derivation of rank \( \rank{0}{2} \).
Hence we can write
\begin{equation}
\label{commutator_of_nablas}
[\nabla_i,\nabla_j] = T^k_{ij}\del_k + K_{ij}\oo 
\end{equation}

Applying both sides to a function \( f \) we see that 
\begin{equation}
\label{commuting_nablas_on_f}
T^k_{ij}\frac{\del f}{\del x_k} =
     -\Gamma^k_{ij}\frac{\del f}{\del x_k} +
      \Gamma^k_{ji}\frac{\del f}{\del x_k}
\end{equation}

Hence 
\begin{equation}
\label{torsion_defined}
T^k_{ij} = -\left(\Gamma^k_{ij} - \Gamma^k_{ji}\right) 
\end{equation}

which is the negative of the torsion tensor as usually defined.
There are reasons to believe that the torsion tensor would be
better defined with opposite sign - see for example \cite{Wald} page 31.

Similarly by applying \( [ \nabla_i,\nabla_j] \) to a vector field \( v^k \)
and comparing the terms not involving the partial derivative on both sides, we
see that 
\begin{equation}
\label{commuting_nablas_on_v}
K^x_{ijy} = 
\bigl[ \del_i\Gamma^x_{jy} - 
	\del_j\Gamma^x_{iy}\bigr] +
\bigl[ \Gamma^x_{it}\Gamma^t_{jy} - 
	\Gamma^x_{jt}\Gamma^t_{iy}\bigr]
+ T^t_{ij}\Gamma^x_{ty}
\end{equation}

In the torsion-free case the last term vanishes and \( K^x_{ijy} \) 
becomes the usual Riemann tensor. However we will predominantly not be 
working in the torsion-free case.

The commutator of operators obeys the Jacobi identity.
\begin{equation}
\label{Nabla_Jacobi}
[[ \nabla_i,\nabla_j],\nabla_k] + 
	[[ \nabla_j,\nabla_k],\nabla_i] + 
		[[ \nabla_k,\nabla_i],\nabla_j] = 0
\end{equation}

Resolving the tensor and partial derivative terms of this equation will
give identities involving the torsion and the curvature tensor which turn 
out to be the well known and very important Bianchi identities. However 
in evaluating these identities we note that expressing 
\( [\nabla_i,\nabla_j] \) in terms of \( \del_k \) 
is not very helpful when we wish to subsequently take another commutator. 
It would be much more useful to express \( [\nabla_i,\nabla_j] \) in terms of
\( \nabla_k \). We  define 
\begin{equation}
\label{torsion+curvature_defined}
[\nabla_i,\nabla_j] = T^k_{ij}\nabla_k + R_{ij}\oo
\end{equation}

Since \( \nabla_k \) and \( \del_k \) act identically on functions, 
the coefficients \( T^k_{ij} \) are as given in equation~\ref{torsion_defined} 
(the negative of the usual torsion tensor). Hence
\begin{flalign*}
T^k_{ij}\del_k + K_{ij}\oo 
	&= T^k_{ij}\nabla_k + R_{ij}\oo \\
	& =T^k_{ij}\del_k + T^k_{ij}\Gamma_k\oo + R_{ij}\oo
\end{flalign*}

and so 
\begin{equation}
\label{Riemannian_defined}
R^x_{ijy} = K^x_{ijy} -  T^k_{ij}\Gamma^x_{ky}
=
\Bigl[ \del_i\Gamma^x_{jy} - \del_j\Gamma^x_{iy}\Bigr] +
\Bigl[ \Gamma^x_{it}\Gamma^t_{jy} - \Gamma^x_{jt}\Gamma^t_{iy}\Bigr]
\end{equation}

Hence \( R^x_{ijy} \) is the usual Riemann tensor in all cases, even when
the manifold is not torsion free.

\medskip

We now derive the Bianchi identities by looking at the Jacobi identity for 
the derivations \( \nabla_i \), \( \nabla_j \) and \( \nabla_k \).
\begin{flalign*}
[[ \nabla_i,&\nabla_j],\nabla_k](v^x) 
            \,=\, [T_{ij}^t\nabla_t + R_{ij}\oo,\nabla_k](v^x) \\
& = \Bigl(T_{ij}^s\nabla_s\nabla_k(v^x) + R^x_{ijt}\nabla_k(v^t)
- R^t_{ijk}\nabla_t(v^x)\Bigr) \\
& \mbox{\hskip1cm}
     - \Bigl(\nabla_k(T^t_{ij}) \nabla_t(v^x) + T^s_{ij}\nabla_k\nabla_s(v^x)
           +\nabla_k(R^x_{ijt})v^t + R^x_{ijt}\nabla_k(v^t) \Bigr) \\
& = T^s_{ij}[\nabla_s,\nabla_k](v^x) - R^t_{ijk}\nabla_t(v^x)
           - \nabla_k(T^t_{ij})\nabla_t(v^x) - \nabla_k(R^x_{ijt})v^t \\
& = \Bigl( T^s_{ij}T^t_{sk} - \nabla_k(T^t_{ij})  - R^t_{ijk}\Bigr)\nabla_t(v^x)
     + \Bigl(T^s_{ij}R^x_{skt} - \nabla_k(R^x_{ijt})\Bigr)v^t
\end{flalign*}

If we cyclically permute the indices \( \{ i,j,k\} \) in this expression and add,
then by the Jacobi identity the result must be zero. Separating out the coefficients
of \( v^t \) and \( \nabla_t(v^x) \) we obtain the first and second Bianchi identities.

The first Bianchi identity is
\begin{equation}
\label{1st_Bianchi_original}
R_{ijk}^t + T_{kx}^tT_{ij}^x + \nabla_k(T_{ij}^t) \ijkequal 0 
\end{equation}

where we introduce the notation \( \ijkequal \) to mean that 
applying the permutation \( (ijk) \) cyclically to the left hand side and summing
gives the expression on the right.

The second Bianchi identity is
\begin{equation}
\label{2nd_Bianchi_original}
\nabla_k(R^t_{ijs}) + T^x_{ij}R^t_{kxs} \ijkequal 0 
\end{equation}

\bigskip
\pagebreak[2]

\section{Lie Groups}
\label{LieGroups}

Now that we have developed sufficient mathematical tools to allow us to
intelligently talk about the problem, let us consider the question of local 
symmetry, and in particular how we might specify that a manifold has local 
symmetry \( \so(2,3) \).

We turn for inspiration to the one manifold which we can unambiguously 
say must have local symmetry \( \so(2,3) \), namely the manifold which is
the ten dimensional Lie group \( \SO(2,3) \) itself.

There are many ways to describe the relationship between a Lie group and its
Lie algebra. For example the Lie algebra is commonly described as 
as the algebra of left invariant vector fields.  This uniquely identifies the 
Lie algebra with the tangent space at each point. Note that it also provides us
with an identification of tangent spaces at different points. Consequently 
every Lie group comes equipped with a natural connection arising from viewing 
left multiplication in the group as parallel transport.

The description of the Lie algebra as left invariant vector fields does not 
address its relationship with this natural connection which will be our main 
focus in what follows. The existence of this connection means that there is 
a natural covariant derivative \( \nabla_i \) defined on every Lie group. 
We observe the following.

\begin{theorem}
The natural connection on a Lie group satisfies
\( [\nabla_i,\nabla_i] = T^k_{ij}\nabla_k \) where \( T^k_{ij} \) are the 
structure coefficients of the associated Lie algebra on each tangent space. 
Furthermore \( \nabla_t(T^k_{ij}) = 0 \).
\end{theorem}

Hence the Lie structure is the negative of the torsion as usually defined. 
As noted in the comments following equation~\ref{torsion_defined}, there
are many reasons to believe that torsion would have been better defined with
opposite sign. Rather than live with this annoying negative sign we choose 
to correct the definition. Hence \textbf{we will call \( T^k_{ij} \) as 
defined in equation~\ref{torsion_defined} the torsion 
and not its negative.}

With this definition the torsion with respect to the natural connection on
a Lie group is precisely the Lie structure while the curvature tensor is zero.

This suggests how we might be able to have a curved space which accommodates 
a local symmetry group in which translations do not commute. We can use the 
torsion to specify the local symmetry.  Curvature and Torsion have a clearly 
defined relationship and can coexist on the same manifold. Hence we can 
specify local symmetry using the Torsion and we can do it on a curved 
manifold without any confusion between the local symmetry and the 
curvature.

\bigskip
\pagebreak[2]

\section{Lie Manifolds}
\label{LieManifolds}

\begin{definition} 
\label{local_lie_defined}
A \textbf{Lie manifold} is a connected real manifold with 
invariant torsion. That is with \( \nabla_t(T^k_{ij}) = 0 \).
\end{definition}

The covariant derivatives \( \{ \nabla_i \} \) for fixed indices \( i \)
are tensor derivations of rank \( \rank{0}{0} \), and this set of tensor 
derivations is closed under addition and commutator and so forms a Lie algebra.
The Jacobi identity now gives
\begin{equation}
[\nabla_k,[\nabla_i,\nabla_j]] (f) \ijkequal 0
\end{equation}
Expanding this out and using \( \nabla_k(T_{ij}^t) = 0 \) we obtain
\begin{equation}
\label{T_Jacobi}
T_{kx}^tT_{ij}^x \ijkequal 0 
\end{equation}
and hence the torsion obeys the Jacobi condition defining a Lie algebra 
on each tangent space. By invariance (and connectivity) these Lie algebras 
are all isomorphic, so there is essentially a unique Lie algebra associated 
with each connected Lie manifold.

Equation~\ref{T_Jacobi} implies that on a Lie manifold the first 
Bianchi identity takes the more familiar form
\begin{equation}
\label{1st_Bianchi_basic}
R_{ijk}^t \ijkequal 0
\end{equation}

A Lie algebra has a natural bilinear form, the Killing form, which is 
invariant under the adjoint action. This means every Lie manifold 
has defined on it a natural pseudometric given by the Killing form.
\begin{equation}
\label{killing_form_defined}
k_{ij} = T^a_{ib}T^b_{ja} 
\end{equation}

The Jacobi identity gives \( T_k\oo(k_{ij}) = 0 \) confirming that \( k_{ij} \)
is invariant under the adjoint action of the Lie algebra on each tangent 
space. If the Lie algebra is semisimple (or indeed simple as is the case
for \( \so(2,3) \) then \( k_{ij} \) will be non-degenerate.  Furthermore
\begin{equation}
\label{killing_form_invariance}
\nabla_t(k_{ij}) = \nabla_t(T_{ib}^aT_{ja}^b) = 
	\nabla_t(T_{ib}^a)\, T_{ja}^b + T_{ib}^a\,\nabla_t(T_{ja}^b) = 0
\end{equation}

so this pseudometric is also globally invariant under the action of the 
covariant derivative. As is common in the area of general relativity we
will stop calling it a pseudometric beyond this point and will simply 
call it the metric 

\medskip

Let us quickly examine the nature of this metric in the case that the Lie 
algebra is \( \so(2,3) \). We will use the basis \( \{T,X,Y,Z,A,B,C,I,J,K\} \)
for the Lie algebra and coordinates \( \{ t,x,y,z,a,b,c,i,j,k\} \) so that a
typical element of the Lie algebra takes the form 
\( tT + xX + yY + zZ + aA + bB + cC + iI + jJ + kK \). We use ordinary units
(because we want to see the effect of \( r \)) and compute the Killing form 
directly from table~\ref{table_so23_ordinary}. We obtain 
\begin{multline}
-\nfrac{6}{r^2}
\left(
t_1t_2 -\nfrac{1}{c^2}(x_1x_2 + y_1y_2 + z_1z_2)
\right) \\ - 6
\left(
i_1i_2 + j_1j_2 + k_1k_2 - \nfrac{1}{c^2}(a_1a_2 + b_1b_2 + c_1c_2)
\right)
\end{multline}

This gives us a Lorentz metric across the four space-time dimensions 
\( \{ T,X,Y,Z \} \) together with a spin or helicity term on the six 
Lorentz dimensions. This is precisely the kind of metric we'd expect 
to have on the manifold of frames.

\medskip
\begin{axiom}
\large
\label{axiom:local_lie_so23}
Our universe defines a Lie manifold for \( \so(2,3) \). 
We interpret this as the 10-D manifold of local inertial 
frames\footnote{or possibly a cover of the manifold of inertial frames}.
\end{axiom}
\medskip

We call this an axiom because such a statement cannot be proved mathematically. 
Instead it must be evaluated in terms of the usefulness and physicality of the 
resulting mathematical model.

\medskip

In a Lie manifold for \(  \so(2,3) \) the Lie algebra \( \so(2,3) \) can 
act on vectors in two ways. The global action is given by 
\( v^k \mapsto \nabla_i(v^k) \). There is also a local action given by 
\( v^k \mapsto T^k_{ij}v^j \) which describes the adjoint action of \( \so(2,3) \).
It is a map from the manifold into the 10-D representation.

\bigskip

At this point we have almost everything we need.  However tensors alone 
are not enough for us to do physics. We need spinors to act as the wave 
functions of fermions. This requires us to extend our ideas to allow 
other local actions.

	\chapter{Spinor Manifolds and Frameworks}

In the last chapter we looked at how local symmetry \( \so(2,3) \)
could be specified on a curved manifold. We were able to do this by 
using torsion to describe the local symmetry. A vector on such a 
manifold can be viewed as a function on which the symmetry group
has two actions; an extrinsic action specified by the covariant 
derivative, and an intrinsic action specified by the torsion.

This is exactly the kind of thing we want in order to do physics.
However tensors have representations with integral values for 
spin. We need to be able to use functions into other representations,
and in particular into spinors. We need to extend our model in a 
natural way so that this becomes possible.

\bigskip
To see how this might be done we return to the canonical example where the 
manifold is the Lie group itself. This time however we will begin with the
assumption that we are working with a matrix Lie group -- a closed subgroup 
of the group \( \text{End}(V) \) of endomorphisms some vector space \( V \).
This means our Lie algebra comes equipped with a local action on $V$. In the 
case that \( V \) is the 4-D representation of \( \sp(2,\R) \) we hope to learn
from this example how to naturally include spinor valued functions on a 
Lie manifold for \( \so(2,3) \).

\bigskip
\pagebreak[2]

\section{Matrix Lie Groups}
\label{MatrixLieGroupsandManifolds}

Consider a matrix Lie group \( M \) on a (real or complex) vector space \( V \).
Choose a basis for \( V \) and use Greek indices to denote components with 
respect to this basis. A \( V \)-field is a \( V \) valued function on \( M \)
denoted in terms of our basis as \( v^\alpha \) where it is understood that 
these coordinates are functions on \( M \).  Choose a coordinate system for 
\( M \) treating it as a real manifold.  We will use roman indices for this coordinate 
system, which gives a basis at each point of the Lie algebra. An ordinary vector 
field on \( M \) is denoted \( v^i \) where it is understood that these coordinates 
are functions on \( M \).

The elements of the Lie algebra are matrices over \( V \). Note that the Lie 
algebra is always a real Lie algebra, but the matrices that represent it may be
complex in the case that $V$ is a complex vector space. Hence the Lie algebra elements
can be written in terms of our bases as \( T^\alpha_{i\beta} \) where
the coefficients are functions of \( M \).  Since these matrices represent the 
Lie algebra we have 
\begin{equation}
\label{first_rep_equation}
T^\alpha_{i\lambda}T^\lambda_{j\beta} - T^\alpha_{j\lambda}T^\lambda_{i\beta}
= T^k_{ij}T^\alpha_{k\beta} 
\end{equation}

where \( T^k_{ij} \) is the structure constants for the Lie algebra.
We call this representation the \textbf{ local action} on \( V \)-fields.
The notation is chosen to be consistent with the local action 
on vector fields which is the adjoint action given by \( T^k_{ij} \). 
This allows us to describe the natural local action as \( T_k\oo \) 
regardless of whether vector fields or \( V \)-fields are being acted on.  
This local action can be extended naturally to tensor products and duals 
of \(V\)-fields and vector fields. 

For example a \( V^*\)-field is a map from \( M \) into the 
dual space \( V^*\) of \( V \) and can be denoted \( u_\alpha \) with 
respect to our chosen basis of \( V \). It maps \( v^\alpha \) to  
\( u_\alpha v^\alpha \) calculated under the summation convention. 
The natural local action on \( V^*\)-fields  is given by 
\begin{equation}
T_k\oo(u_\alpha) = - T^\beta_{k\alpha}u_\beta
\end{equation}
which ensures that \( T(u)(v) + u(T(v)) = T(u(v)) = 0 \) consistent with a 
trivial local action on scalar fields. 

Similarly a \( V \tensor V \)-field is a map from \( M \) to \( V \tensor V \)
denoted \( X^{\alpha\beta} \) with respect to our chosen basis of \( V \). 
The local action is defined by requiring a Leibniz condition and is 
given by 
\begin{equation}
T_k\oo(X^{\alpha\beta}) = 
    T^\alpha_{k\lambda}X^{\lambda\beta} + T^\beta_{k\lambda}X^{\alpha\lambda} 
\end{equation}
which ensures that \( T(v\tensor w) = T(v)\tensor w + v \tensor T(w) \).

The most general objects that can be formed using dual and tensor product
from \( V\)-fields and vector fields  are called \( V \)-tensors.
A typical \( V \)-tensor can be described in terms of coordinates for 
$V$ and the tangent spaces $\T$ using a combination of greek and roman, 
upper and lower indices. 
\[ a^{i_1\ldots i_m\lambda_1\ldots\lambda_p}_{j_1\ldots j_n\mu_1\ldots\mu_q} 
\in 
\underbrace{ \T \tensor \!\cdots\! \tensor \T }_{\mbox{\(m \) times }}
\tensor
\underbrace{ \T^* \tensor \!\cdots\! \tensor \T^* }_{\mbox{\(n\) times}} 
\tensor 
\underbrace{ V \tensor \!\cdots\! \tensor V }_{\mbox{\(p \) times }}
\tensor 
\underbrace{ V^* \tensor \!\cdots\! \tensor V^* }_{\mbox{\(q\) times}} 
\]

We define the rank  of such a \( V \)-tensor to be \( \rank{m,p}{n,q} \). 
Thus \( T^\alpha_{i\beta} \) is a \( V \)-tensor of rank \( \rank{0,1}{1,1} \). 

The local action on \( V \)-tensors is denoted \( T_k\oo \) and is 
calculated by contracting \( T^i_{kj} \) or \( T^\alpha_{k\beta} \) as 
appropriate with each upper and lower index in term; adding the results 
in the case of the upper indices, and subtracting in the case of the lower 
ones. This will ensure that the local action obeys a Leibniz condition on 
tensor products and respects contraction. Furthermore we have 
\begin{equation}
[T_i\oo,T_j\oo] = T^k_{ij}T_k\oo
\end{equation}
which means the local action on \( V \)-tensors of arbitrary rank represents
the Lie algebra.

We finish our consideration of the local action by looking at the local action 
on the \( V \)-tensor \( T^\alpha_{i\beta} \).  We see 
\begin{equation}
T_k\oo(T^\alpha_{i\beta}) = T^\alpha_{k\lambda}T^\lambda_{i\beta}
- T^\lambda_{k\beta}T^\alpha_{i\lambda} - T^x_{ki}T^\alpha_{x\beta} = 0 
\end{equation}

This follows directly from equation~\ref{first_rep_equation}.  We say that 
\( T^\alpha_{i\beta} \) is locally invariant.

\bigskip

We now turn to the global action.  As we noted earlier left multiplication 
by elements of \( G \) can be regarded as defining a natural parallel 
transport of vector fields on a Lie group.  In the present case where the 
elements of \( G \) are matrices over \( V \), left multiplication also defines 
a natural parallel transport of \( V \)-valued functions. This can be expressed
via a connection \( \Gamma^\alpha_{i\beta} \). This notation has been
chosen so that we can describe the natural connection as \( \Gamma_i\oo \) 
regardless of what is being acted on. This definition also extends naturally
to $V$-tensors and can be used to define a \textbf{ global action} 
\( \nabla_k = \del_k + \Gamma_k\oo  \) on $V$-tensors. This global action 
also represents the Lie algebra and we have 
\begin{equation}
\label{V_Lie_exact}
[\nabla_i,\nabla_j] = T^k_{ij}\nabla_k
\end{equation}

where \( T^k_{ij} \) is the Lie structure and also obviously the torsion.

Finally we observe that the \textbf{ local and global actions commute}. 
Hence we have 
\begin{equation}
\nabla_k(T^\alpha_{i\beta}) = 0
\end{equation}

and thus \( T^\alpha_{i\beta} \) is globally as well as locally invariant.

\bigskip
So far we have simply been exploring the structure of an ordinary matrix 
Lie group from a somewhat unorthodox perspective. It is now time to consider 
how we might we might add curvature to this picture. But that should not be 
difficult.  We just need to weaken equation~\ref{V_Lie_exact} and replace 
it with the more general equation
\begin{equation}
\label{V_Lie_curved}
[\nabla_i,\nabla_j] = T^k_{ij}\nabla_k + R_{ij}\oo
\end{equation}

where \( R_{ij}\oo \) is a linear operator on arbitrary $V$-tensors 
which we will call the curvature.

\bigskip
\pagebreak[2]

\section{$V$-tensors and $V$-tensor Derivations}

In this section we implement the idea at the end of the last section,
and clarify our notation.

 Let \( M \) be a manifold and choose a local coordinate system on \( M \).
Use the roman alphabet for the indices so that a vector field is written
\( v^i \) with respect to the coordinate system.

Let \( V \) be any vector space (real or complex) with basis 
\( \{ e_\alpha \} \) where we use greek indices to distinguish it from the 
coordinate system on \( M \). A \textbf{\( V \)-field on \( M \)} is a map 
from \( M \) to \( V \) and is denoted \( v^\alpha \) in terms of its 
components with respect to this basis. 

Other types of functions on the manifold can be constructed from 
\( V \)-fields and ordinary vector fields using the operations of dual
and tensor product. Such functions are called $V$-tensors. A typical 
\( V \)-tensor is described with respect to our coordinate system on \( M \)
and basis of \( V \) via a combination of greek and roman, upper and lower
indices. We define the \textbf{rank} of the $V$-tensor
\( a^{i_1\ldots i_m\lambda_1\ldots\lambda_p}_{j_1\ldots j_n\mu_1\ldots\mu_q}  \)
to be \( \rank{m,p}{n,q} \).

\begin{definition}
\label{Vtensor_derivation}
A \textbf{\(V\)-tensor derivation} is a map from
\( V \)-tensors to \( V \)-tensors which satisfies the following 
conditions 
\begin{itemize}
\item Linearity.
\item Leibniz condition on tensor products.
\item Commutes with contraction (trace) of Greek and Roman indices.
\end{itemize}
\end{definition}

\medskip

The results on tensor derivations obtained in 
section~\ref{TensorDerivationsandManifolds} extend to $V$-tensor 
derivations without difficulty. We summarise omitting proofs which are
essentially identical to those given earlier. 

\begin{proposition}
If \( D \) and \( E \) are \( V \)-tensor derivations, then 
\( [D,E\,] \) is also a \( V \)-tensor derivation where 
\( [D,E\,](X) = D(E(X)) - E(D(X)) \).
\end{proposition}

\begin{proposition} If \( D \) is a \( V \)-tensor derivation and \( A \)
is a \( V \)-tensor, then \( A\tensor D \) is an \( V \)-tensor derivation 
where \( (A\tensor D)(X) = A\tensor D(X) \).
\end{proposition}

\begin{proposition} Every \( V \)-tensor derivation has a rank 
\( \rank{m,i}{n,j} \) and maps \( V \)-tensors of rank 
\( \rank{p,k}{q,l} \) 
to \( V \)-tensors of rank \( \rank{m+p,k+i}{n+q,l+j} \).
\end{proposition}
Every ordinary derivation can be extended to a \( V \)-tensor derivation 
of rank \( \rank{0,0}{0,0} \) by allowing it to act on components. We use 
\( T \mapsto a^i \del_i (T ) \) 
to denote this \(V\)-tensor derivation.

Conversely every \(V\)-tensor derivation \( D \) of rank \( \rank{0,0}{0,0} \)
acts on functions as an ordinary derivation. Hence we may associate to it a 
tangent vector field so that with respect to some coordinate system on 
\( M \) we have \( D(f) = a^i\del_i (f) \). 

The difference \( D - a^i\del_i \) is a \(V\)-tensor
derivation of rank \( \rank{0,0}{0,0} \) which maps all functions to the zero 
function. 

\begin{proposition} Let \( E \) be an \(V\)-tensor derivation of rank 
\( \rank{0,0}{0,0} \) which maps all functions to the zero function. 
Then there exists a \(V\)-tensor \( \Gamma^\beta_\alpha \) of rank 
\( \rank{0,1}{0,1} \) 
and an \( V \)-tensor \( \Gamma^i_j \) of rank \( \rank{1,0}{1,0} \) so 
that \( E = \Gamma\oo \) defined in the obvious way.
\end{proposition}

\begin{proposition} All \/ \( V \)-tensors of rank \( \rank{m,i}{n,j} \)
can be expressed in the form
\[ 
a^{i_1\ldots i_m\lambda_1\ldots\lambda_p\, k}_{j_1\ldots j_n\mu_1\ldots\mu_q}
	\del_k
+ A^{i_1\ldots i_m\,\lambda_1\ldots\lambda_p}_{j_1\ldots j_n\,\mu_1\ldots\mu_q}\oo 
\]
\end{proposition}

In particular a \( V \)-tensor derivation of rank \( \rank{0,0}{1,0} \) 
takes the form

\[ a^k_i\del_k + \Gamma_k\oo \]

In the case where \( a^k_i = 1^k_i \) we call this a \textbf{covariant 
derivative}. The expressions $\Gamma^\alpha_{k\beta}$ and $\Gamma^i_{kj}$ are
called \textbf{connections} and these define a notion of parallel transport 
of $V$-fields and vector fields respectively.

\bigskip
\pagebreak[2]

\section{Matrix Lie Manifolds}

We are now in a position to define the structure we need.

\begin{definition}
A \textbf{ Matrix Lie Manifold } is 
\begin{description}
\item[(a)] a real manifold \( M \) and a real or 
	complex vector space $V$ together with
\item[(b)] a \textbf{ local action} \( T_k\oo \) of 
	tangent vectors on $V$-tensors with
	\[ [T_i\oo,T_j\oo] = T^k_{ij}\oo \]
\item[(c)] a \textbf{ global action} 
	\( \nabla_k = \del_k + \Gamma_k\oo  \) on $V$-tensors with
	\[ [\nabla_i,\nabla_j] = T^k_{ij}\nabla_k + R_{ij}\oo \]
	for some linear operator on $V$-tensors \( R_{ij}\oo \) 
\item[(d)] where local and global actions commute 
	\[ \nabla_k(T^\alpha_{i\beta}) = 0 \]
\end{description}
\end{definition}

Part (c) tells us that the Lie structure \( T^k_{ij} \) is also the 
torsion\footnote{remember that our torsion is the negative of the usual 
definition} and parts (d) and (b) give \( \nabla_t(T^k_{ij}) = 0 \).
It follows that every matrix Lie manifold is a Lie manifold with 
a natural local action on \( V \)-valued functions.

Every matrix Lie manifold defines a unique matrix Lie algebra. A tangent 
space gives the Lie algebra, the torsion defines a Lie structure, and the 
local action gives a representation on \( V \).  Furthermore the same matrix 
Lie algebra is obtained at each point on the manifold.  

Matrix Lie manifolds provide a mathematical framework for building curved
models with functions mapping into a chosen local representation. We now need
to decide which local representation of \( \so(2,3) \) to use.
 
\bigskip
\pagebreak[2]

\section{Complex matters}
\label{complex matters}

Spinors are the 4-D representation of our group \( \so(2,3) \). We want to 
be able to have spinor valued functions on our manifold so that we can use
them as wave functions for elementary particles. We can do this by looking
at a matrix Lie manifold for the appropriate representation. But do we want 
to use real spinors or complex ones? What difference does it make?

We make the following temporary definitions to aid our discussion.

\begin{definition} 
A \textbf{real spinor manifold} is defined to be a matrix 
Lie manifold for the 4-D real irreducible representation of the real Lie
algebra \( \sp(2,\R) \).

A \textbf{complex spinor manifold} is defined to be a matrix 
Lie manifold for the 4-D complex irreducible representation of the real Lie
algebra \( \sp(2,\R) \).
\end{definition}

Note that every real spinor manifold is also a complex spinor manifold since 
we can construct the required actions on complex spinors by simply taking 
tensor product with complex scalars. However the converse appears not to be 
true. There is no obvious way to build  a local and global action on real 
spinors from the actions on complex spinors. 

Let us think about how we might try to do this to understand the nature of the 
problem. The obvious approach would be to seek a separation of complex spinors 
into real and imaginary parts in a manner respected by both the local and 
global action. The actions on the real part would then meet the requirements 
for a real spinor manifold. 

This is equivalent to finding a basis in which all  \( T^\alpha_{k\beta} \) 
and \( \Gamma^\beta_{k\alpha} \) are real. While we can find a basis that 
makes all \( T^\beta_{k\alpha} \) real, there is no obvious way to ensure 
that all \( \Gamma^\beta_{k\alpha} \) are real as well.

We can also restate the problem in terms of the search for a conjugation map; 
a conjugate linear involution \( v \mapsto \overline{v} \) on complex spinors, 
which commutes with both local and global actions. The complex spinors in a 
real spinor manifold constructed by tensor product with $\C$ have such an
involution. But there is no obvious reason why complex spinors on a complex 
spinor manifold should have such a map. 

Conjugating a complex vector is not a trivial matter. Simply conjugating 
coordinates does not give a well defined map since this depends on the 
choice of basis. In order to have a well defined conjugation (or conjugate
transpose) additional structure must be added to the space. Often this is 
done via some sort of inner product. However in our case the best approach 
seems to be to simply to require that our complex spinor manifold come 
equipped with a natural invariant conjugation map. As discussed above we 
can then choose bases so that \( T^\alpha_{k\beta} \) and 
\( \Gamma^\beta_{k\alpha} \) are real. 

A complex spinor manifold equipped with such a conjugation is a real spinor 
manifold since we can use the conjugation to identify real components 
preserved by the local and global actions. Since conjugating complex spinors 
seems to be something we will need to be able to do, we will therefore choose 
a real spinor manifold to build our model. Complex spinor manifolds remain 
an obvious generalisation which is definitely worthy of investigation at a 
later date.

\medskip

\begin{definition}
We define a \textbf{framework} to be a real spinor manifold. Real spinor 
manifolds are the structure we will be looking at throughout the rest of the
book, and it is useful to have a shorted name to describe them. 
\end{definition}

\begin{axiom}
\label{axiom:matrix_lie}
\large
Our Universe is a framework.
\end{axiom}

An axiom is simply an assumption. We will be exploring the consequences of 
this assumption throughout the rest of this book. We next discuss whether 
this assumption is physically reasonable.

\bigskip
\pagebreak[2]

\section{Reality Check}
\label{Reality Check}

In ordinary general relativity the fundamental structure is the metric 
\( g_{ij} \).  This is assumed to be invariant.

A metric determines the nature of the manifold by distinguishing between 
two types of coordinate; positive and negative definite; space-like and 
time-like. Global invariance of the metric ensures that this distinction 
is conserved under parallel transport.

\bigskip

On a Lie manifold for \( \so(2,3) \) the fundamental structure is the Lie
structure or torsion \( T^k_{ij} \). This is assumed to be invariant.

A natural metric, the Killing form, is then given by  
\( g_{ij} = \frac{1}{6}T^x_{iy}T^y_{jx} \) (the constant of \( \frac{1}{6} \)
is included for consistency with later results).
Invariance of the Lie structure implies invariance of this metric.  
Hence the assumption that the Lie structure is invariant is a generalisation 
of the idea of metric invariance. 

The Lie structure identifies the physical nature of the ten coordinates 
on the manifold enabling us to distinguish between a time coordinate and a 
rotation  coordinate for example, which a metric alone cannot do. Invariance 
of the Lie structure implies that these identifications will be conserved under 
parallel transport, which seems physically reasonable. 

\bigskip

In a framework (real spinor manifold) the fundamental structure is the 
local action \(T^\beta_{i\alpha} \) which is assumed to be invariant. 

From the equation 
\( T^\beta_{i\lambda}T^\lambda_{j\beta} = T^k_{ij}T^\beta_{k\alpha} \) 
it is clear that global invariance of \( T^\beta_{i\alpha} \) implies 
global invariance of \( T^k_{ij} \) and hence of \( g_{ij} \) 
so this is a further extension of the idea of metric invariance. This
follows since
\begin{equation}
\label{invariant forms}
g_{ij} = \frac{1}{6}T^x_{iy}T^y_{jx} = T^\alpha_{i\beta}T^\beta_{j\alpha} 
\end{equation}

The requirement that \( T^\beta_{i\alpha} \) is globally invariant means
that the action of the Lie algebra on spinors is conserved under 
parallel transport. If \( \phi = X(\psi) \) for a tangent vector 
\( X \) and spinors \( \psi \) and \( \phi \) defined at a point; and
if these are parallel transported to \( X' \), \( \psi' \) and \( \phi' \) 
respectively at an adjacent point, then  we must have \( \phi' = X'(\psi') \).

Invariance of \( T^\beta_{i\alpha} \) also is equivalent to the requirement
that the local action \( T_k\oo \) and the global action \( \nabla_k \) of 
our symmetry group commute. It is not unreasonable for us to ask that the 
global action respect the local action.

\bigskip

\begin{question}\label{canweextend} 
Are axiom~\ref{axiom:local_lie_so23} and axiom~\ref{axiom:matrix_lie} 
equivalent?  Can every Lie manifold for \( \so(2,3) \) be 
extended to a framework?
\end{question}

To answer this question we need to know whether the existence of the
spinor action places any constraints on the nature of the manifold and in 
particular its curvature.  

We will later prove many interesting and significant results about the
nature of curvature on a framework which depend critically on the 
existence of the spinor action. Since we don't see how we might otherwise 
prove these results our intuition is that the answer to 
question~\ref{canweextend} is most probably ``No''.

For example the spinor action will later be used to prove 
equation~\ref{reduceddecomposition}
\begin{equation*}
R^l_{ijk} = R^t_{ij} T^l_{tk}
\end{equation*}

which tells us that on a framework,  the curvature operator 
$R_{ij}\oo $ on vectors can be always be expressed as a linear combination of 
torsion operators $T_k\oo$. A Lie manifold for $\so(2,3)$ where this was not 
true could not be a framework. However the question remains open 
awaiting the construction of an explicit example.

	\chapter{Generalised Tensors}
\label{Generalised Tensors}

A Lie manifold for $\so(2,3)$ is a manifold with the invariant Lie structure
of \( \so(2,3) \). It is a generalisation of the manifold which is the Lie 
group $\SO(2,3)$. A framework additionally has a natural local action
on spinors. It generalises the matrix Lie group $\Sp(2,\R)$ and as such comes
equipped with a local and global action on spinors. Using dual and tensor 
product we were able to extend these actions to \( V \)-tensors.

In all of this local and global actions play a central role. All
$V$-tensors have both a local and global action, and the two actions
commute. These are very physical properties. We need a global action
to discuss extrinsic properties like Energy. We need a local action to 
discuss intrinsic properties like spin. And we need the two actions to 
commute if these are to be simultaneously determined.

This leads us to define a generalised tensor, a quantity on the manifold 
with a well defined local and global action where these commute. All 
\( V \)-tensors are generalised tensors, however as we will see the converse
is not true.. In particular if we analyse a \( V \)-tensor in terms of its 
indecomposable or irreducible components, we will discover that these need 
not be \( V \)-tensors. They will however be generalised tensors.

Generalised tensors represent the end of our search for a suitable 
mathematical structure to do physics with. Tensors and \( V \)-tensors 
were steps along the journey to generalised tensors.

There appears to be a very nice mathematical theory of generalised 
tensors on a Lie manifold, and in this chapter we will look at some of the 
basics of this theory. However much more remains to be done.

\bigskip
\pagebreak[2]

\section{Definition}

We make the following  definition.

\begin{definition}
\label{generalisedtensor}
A \textbf{Generalised Tensor Space} on a Lie manifold \( M \) is 
any vector space \( V \) with  
\begin{enumerate}
\item a smooth local action defined on \( V \)-fields. 
\item a global action defined on \( V \)-fields.
\item where the local and global actions commute
\end{enumerate}
\end{definition}

A \textbf{local action} is a representation of each tangent space (as a Lie
algebra) on \( V \). Hence it can be described using matrices \( T_k\oo \) 
with 
\begin{equation*}
T_i\oo T_j\oo - T_j\oo T_i\oo = T^k_{ij}T_k\oo 
\end{equation*}

A \textbf{global action} is specified by a covariant derivative on 
\( V \)-valued functions which weakly represents the Lie algebra. The
covariant derivative \( \nabla_i = \del_i + \Gamma_i\oo \) is described
in terms of a connection \( \Gamma_i\oo \) which describes parallel transport
of \( V \)-valued functions on \( M \). It weakly presents the Lie algebra
in the sense that 
\begin{equation*}
\nabla_i\nabla_j - \nabla_j\nabla_i = T^k_{ij}\nabla_k + R_{ij}\oo
\end{equation*}
where \( R_{ij}\oo \) are linear maps on \( V \) defined on \( M \).

Clearly all tensors belong to generalised tensor spaces, as do all 
$V$-tensors on a framework.  However this does not preclude other
generalised tensors from existing on the same manifold. Indeed we will 
see that generalised tensors which are not $V$-tensors always exist.

The question of which generalised tensors exist on a Lie manifold is 
an interesting one. The general answer to this question appears 
complicated and leads us away from our objectives, so ultimately we will
provide only a partial answer in the case of a framework.
In working towards that result however we will continue to argue 
generally for as long as possible.

\bigskip
\pagebreak[2]

\section{Basics}

\bigskip

\begin{proposition}
The space of ordinary vector fields on a Lie manifold is a generalised 
tensor space which we denote \( \T \).
\end{proposition}

\begin{proposition}
The direct sum \( U \oplus V \) of two generalised tensor spaces is 
a generalised tensor space with
\begin{itemize}
\item local action \( T_k(u,v) = \left(T_k(u),T_k(v)\right) \).
\item global action 
\( \nabla_k(u,v) = \left(\nabla_k(u),\nabla_k(v)\right) \).
\end{itemize}
\end{proposition}
This can be checked directly. We also note that
\( R_{ij}(u,v) = \left(R_{ij}(u), R_{ij}(v)\right) \).

\begin{proposition}
The tensor product \( U \tensor V \) of two generalised tensor spaces is 
a generalised tensor space with
\begin{itemize}
\item local action \( T_k(u\tensor v) = T_k(u)\tensor v + u \tensor T_k(v) \).
\item global action 
\( \nabla_k(u\tensor v) = \nabla_k(u)\tensor v + u \tensor \nabla_k(v) \).
\end{itemize}
\end{proposition}

This can also be checked directly. We observe that
\( R_{ij}(u\tensor v) = R_{ij}(u)\tensor v + u \tensor R_{ij}(v) \).

\begin{proposition}
The dual \( V^* \) of a generalised tensor space is 
a generalised tensor space with
\begin{itemize}
\item local action 
\( T_k(f)\,(v) = f\left(T_k(v)\right) - T_k\left(f(v)\right) \).
\item global action 
\( \nabla_k(f)\,(v) = f\left(\nabla_k(v)\right) - \nabla_k\left(f(v)\right) \).
\end{itemize}
\end{proposition}

This can be checked directly, and in this case we have 
\( R_{ij}(f)\,(v) = R_{ij}\left(f(v)\right) - f\left(R_{ij}(v)\right) \).

Under these definitions the matrices \( T^\alpha_{k\beta} \) defining the 
local action for a generalised tensor space $V$ belong to the generalised 
tensor space \( \T^* \tensor V^* \tensor V \). Hence there is a well 
defined global action on these. The condition that local and global actions 
commute can be written
\begin{equation*}
\nabla_m\left(T^\alpha_{k\beta}\right) = 0 
\end{equation*}

\medskip

We next consider the question of isomorphisms and homomorphisms between
generalised tensor spaces. In the case of representations of a Lie algebra,
the appropriate notion is that of an intertwining map. In the case of a 
generalised tensor space we have two actions to consider, and both must be 
preserved in an isomorphism of generalised tensor spaces.

\begin{definition}
If \( U \) and \( V \) are generalised tensor spaces on the same 
Lie manifold \( M \) then an \textbf{intertwining map} between \( U \)
and \( V \) is a linear function \( \theta : U \longrightarrow V \) 
defined smoothly on \( M \) which maps \( U\)-fields to \( V\)-fields and
which satisfies
\begin{description}
\item[1.] \( T_k\left( \theta(u)\right) = \theta\left(T_k(u)\right) \) 
\item[2.] \( \nabla_k\left(\theta(u)\right) = \theta\left(\nabla_k(u)\right) \)
\end{description}
for all \( U\)-fields \( u \).
If only the first of these properties holds we call it a 
\textbf{local intertwining map}. 
\end{definition}

Two generalised tensor spaces are said to be \textbf{equivalent} if there
is an invertible intertwining map between them. They are said to be 
\textbf{locally equivalent} if there is an invertible local intertwining map
between them. 

\medskip

Let \( V \) be a generalised tensor space on \( M \), and let \( U \le V \)
be a vector subspace of \( V \). Then the local and global action on 
\( V \)-fields both act on \( U \)-fields by restriction. If the set of 
\( U \)-fields is conserved by the local action we call \( U \) 
a \textbf{locally invariant} subspace. If this set is conserved by both 
actions we say that \( U \) is \textbf{fully invariant} or just 
\textbf{invariant}. If \( U \) is an invariant subspace of \( V \) then it 
is a generalised tensor space under the restriction of the local and global 
actions on \( V \).  

\begin{proposition}
If \( \theta:U\longrightarrow V \) is an intertwining map then the kernel 
and image of \( \theta \), as obviously defined, are invariant in \( U \) 
and \( V \) respectively. If \( \theta \) is a local intertwining map then
the kernel and image are locally invariant.
\end{proposition}

If \( V \) has a non-trivial invariant subspace we call it 
\textbf{reducible}, otherwise it is called \textbf{irreducible}. 
If a non-trivial locally invariant subspace exists we call it 
\textbf{locally reducible} otherwise it is called \textbf{locally irreducible}.

If \( V \) is equivalent to a non-trivial direct sum of invariant subspaces 
we call it \textbf{decomposable}, otherwise it is called 
\textbf{indecomposable}.  If \( V \) is locally equivalent to a direct 
sum of locally invariant subspaces we call it 
\textbf{locally decomposable} otherwise it is called 
\textbf{locally indecomposable}. If the Lie algebra is semisimple then all
locally indecomposable generalised tensor spaces are locally irreducible.
We will assume that the Lie algebra is semisimple in the remainder of 
this discussion.

\medskip

We are interested in understanding the nature of irreducibles and 
indecomposables. The obvious approach is to decompose in the first 
instance into a direct sum of local irreducibles.  Since local and 
global actions commute we might hope that this decomposition is also 
respected by the global action. 

This is \emph{almost} but not quite true. The problem case involves 
generalised tensor spaces where the local action is trivial. Such spaces 
are easy to construct as we will see in the next section. If the local 
action is trivial (the zero map) then the commuting of the two actions 
does not constrain the global action which is then free to misbehave.


\bigskip
\pagebreak[2]

\section{Locally Trivial Spaces}
\label{locallytrivialspaces}

A generalised tensor space is said to be \textbf{locally trivial} if 
\( T_k\oo = 0 \) at every point so that the local action is trivial. 
These are easy to construct. Indeed choose any generalised tensor space 
with local and global action given by \( T_k\oo \) and \( \nabla_m \) and 
simply replace the local action with the trivial one. It is easy to 
verify that the result is a generalised tensor space. 

The global action still exists and represents (with curvature) the 
Lie algebra. The new (trivial) local action also represents (trivially) 
the Lie algebra. And the global and local actions commute since the zero 
map commutes with everything. Using this trick we can construct a large 
number of different locally trivial generalised tensor spaces. However we
should note that not all locally trivial generalised tensors are 
constructed in this way.  In particular the Crump scalars introduced in 
chapter~\ref{bullets and beyond} are not of this type. 

In traditional terms a locally trivial generalised tensor space would be 
interpreted as a scalar field or collection of scalar fields with non-trivial
parallel transport. The obvious question is whether these play any role in 
the physics. Crump scalars are an example of locally trivial generalised
tensors which are very physical. Indeed we will find that these are essential
for a proper description of the electromagnetic field.


\medskip
We finish this section with a useful example. Let \( U \) be a locally trivial 
irreducible generalised tensor space of dimension \( d \) and let \( V \) be a 
generalised tensor space with irreducible local action. Then the tensor product 
\( W = U \tensor V \) is a generalised tensor space. Under the local action 
we see that it consists of a direct sum of \( d \) irreducible components each 
locally equivalent to \( V \).

Consider a local intertwining map from $V$ to $W$. Then by Schur's lemma this
map gives a scalar map into each of the $d$ components. So the space $E$ of 
all local intertwining maps from $V$ to $W$ is a $d$-dimensional vector space.
There is a natural local and global action on functions between generalised 
tensors and because these maps are local intertwining maps, the natural local 
action on \( E \) is trivial. In fact it isn't hard to show that $E$
and $U$ are isomorphic, and that \( W \) is isomorphic to $E\tensor V$.

\bigskip
\pagebreak[2]

\section{Parallel Transport}

\newcommand{\transport}{\mbox{\small \textsc{Port}}}

The connection defines a notion of parallel transport 
of generalised tensors. In this section we consider parallel transport in 
more detail. We will begin with the notion of parallel transport along a 
path. 

Consider a smooth path \( \pi:[0,1]\longrightarrow M \) from point \( p \) 
to point \( q \) so that \( \pi(0) = p \) and \( \pi(1) = q \). Denote 
the tangent vectors to this path by $ \pi^k $.  Note that $\pi^k$ is only
defined at points on the path, however we could extend it to a vector field
if we wished. Alternatively we could start with a vector field \( \pi^k \) 
and choose \( \pi \) to be a trajectory. Let $\hat{X}$ be a generalised tensor. 

Then we will say that 
$\hat{X}$ \textbf{parallel transports
along \( \pi \)} if 
\begin{equation}
\label{standardtransport}
\pi^k\nabla_k(\hat{X}) = 0 
\end{equation}
at each point on the path.
This definition is the standard one from general relativity. 

\begin{proposition}
\label{translatesonpathproperties}
We list some properties of parallel transport along a path.
\begin{description}
\item[1.] Any globally invariant generalised tensor parallel transports along
all paths. In particular the generalised tensors describing the local action 
$T_k\oo$ parallel transport along all paths.
\item[2.] If \( \hat{X} \) and \( \hat{Y} \) are generalised tensors of
the same type which parallel transport along \( \pi \), then any linear 
combination $a\hat{X}+b\hat{Y}$ also parallel transports along $\pi$.
\item[3.] If \( \hat{X} \) and \( \hat{Y} \) parallel transport along \( \pi \)
then so does $\hat{X}\tensor\hat{Y}$.
\item[4.] If \( \hat{X} \) parallel transports along \( \pi \) then any 
contraction or trace of \( \hat{X} \) also parallel transports along \( \pi \).
\end{description}
\end{proposition}

Proofs are elementary. 

\medskip

Considering $\hat{X}$ evaluated at $\pi(s)$ as a function of $s$ we have

\begin{equation}
\label{DEtosolve} 
\frac{d}{ds}\hat{X} = -\pi^k\Gamma_k\oo(\hat{X})
\end{equation}

Specifying the value of \( \hat{X} \) at $p$ gives us a DE initial value 
problem which has unique solution. If the initial value at $p$ is obtained 
from another generalised tensor $X$, then the value of $\hat{X}$ at $q$ is 
the value of $X$ at $p$ parallel transported along the path $\pi$ to $q$.

Parallel transport along a path forces us to work at specific points.
In some cases it is easier to work with the related notion of parallel 
transport of a generalised tensor along a vector field which 
can be defined in a region or even on the entire manifold. If $X$ is 
a generalised tensor and $\pi^k$ is a 
suitable\footnote{the carpet under which we will sweep any technical issues} 
vector field, then \( \transport(X,\pi^k,t) \) denotes the generalised tensor 
obtained by transporting $X$ along the vector field $\pi^k$ with transport 
parameter $h$. If the path \( \pi \) is a trajectory of \( \pi^k \) then the 
value of \(  \transport(X,\pi^k,t) \) at \( \pi(t) \) is defined to be the 
value of \( X \) at \( \pi(0) \) parallel transported along the path 
\( \pi \) from \( \pi(0) \) to \( \pi(t) \). 

We state without proof a number of important properties of this parallel 
transport operator

\begin{proposition}
\label{propertiesparallel}
The parallel transport operator on generalised tensors has the 
following properties 
\begin{enumerate}
\item \( \displaystyle
	\transport(X,\pi^k,0) = X \)
\item \( \displaystyle
	\transport(\transport(X,\pi^k,t),\pi^k,s) =
	\transport(X,\pi^k,t+s) \)
\item \( \displaystyle
	\transport(X,\pi^k,-h) = 
	\transport(X,-\pi^k,h) \)
\item \( \displaystyle
	\frac{d}{dt}\left(\transport(X, \pi^k,t)\right) = 
	- \pi^k\nabla_k\left(\transport(X, \pi^k,t)\right) \)
\item \label{portlinktonabla}
	\( \displaystyle
	\left.\frac{d}{dt}\left(\transport(X, \pi^k,t)\right)\right|_{t=0} = 
	- \pi^k\nabla_k(X) \).
\item If \( X \) is globally invariant then 
	\( \transport(X,\pi^k,t) = X \) for all \( \pi^k \) and \( t \).
        In particular the generalised tensors describing the local action 
	$T_k\oo$ have this property.
\item \( \displaystyle \transport(aX+bY,\pi^k,t) =
	a\,\transport(X,\pi^k,t) + b\,\transport(Y,\pi^k,t) \)
\item \( \displaystyle \transport(X\tensor Y,\pi^k,t) =
	\transport(X,\pi^k,t)\tensor\transport(Y,\pi^k,t) \)
\item Transport preserves contraction or trace. If \/ \( Y \) is a 
	contraction of \( X \) then \(  \transport(Y,\pi^k,t) \)  
	is a contraction of \( \transport(X,\pi^k,t) \).
\item \label{transportlocalintertwining}
	\(\displaystyle T_k\oo\transport(X,\pi^k,t) = 
	\transport(T_k\oo X,\pi^k,t) \) and hence transport is a local intertwining map.
\end{enumerate}
\end{proposition}

\bigskip
\pagebreak[2]

\section{Indecomposable Generalised Tensors}
\label{IndecomposableGeneralisedTensors}

From property~\ref{transportlocalintertwining} of 
proposition~\ref{propertiesparallel} parallel transport 
preserves the decomposition of a generalised tensor into 
locally homogeneous components as this decomposition is preserved by 
intertwining maps. From property~\ref{portlinktonabla} we can then 
conclude that the global action also preserves the decomposition 
into locally homogeneous components so that this decomposition is also 
global. We have proved the following.

\begin{proposition}
\label{globalrespectshomogeneous}
The decomposition of a generalised tensor into locally homogeneous 
components is also respected by the global action.  Hence all 
indecomposable generalised tensors are locally homogeneous.
\end{proposition}

To go beyond this we must consider the question of which irreducible
representations of the Lie algebra can be realised as the local action
for a generalised tensor.

\begin{definition}
Let \( M \) be a Lie manifold for the Lie algebra \( \g \). Let
\( V \) be an irreducible representation of \( \g \). Then we say that 
\( V \) is \textbf{local on $M$} if there is a generalised tensor on
$M$ where the local action is equivalent to the action of \( \g \) on 
\( V \).
\end{definition}

Note that a framework or framework or spinor manifold $M$ is precisely a 
Lie manifold for \( \so(2,3) \) where the four dimensional irreducible 
representation is local on $M$.  

\medskip

The question of which irreducible representations are local on a given
Lie manifold is an interesting one. As we have already noted requiring
that the spinor representation is local appears to introduce non-trivial 
constraints on the nature of the curvature on $M$. The fully general situation
therefore appears complicated and beyond the scope of this work. We 
therefore restrict our attention to the case that $M$ is a framework.

Proposition~\ref{globalrespectshomogeneous} is very useful for determining
whether an irreducible $V$ is local on $M$. If we can find a generalised 
tensor on $M$ which has a single copy of $V$ in its local decomposition into
irreducibles, then the decomposition into homogeneous components will give us
a generalised tensor with local action equivalent to $V$ and we can
conclude that $V$ is local on $M$. 

All irreducible representations of $\so(2,3)$ can be ultimately formed from 
the spinor representation by taking tensor products and picking out 
irreducible components. All the small dimensional representations appearing
in figure~\ref{weights} can be formed in such a way that they arise 
with multiplicity one in a homogeneous component and hence are local on 
$M$. We will look at the details of how to do this shortly. So far no example 
of an irreducible which is not local on a framework $M$ has been found,
which leads us to make the following conjecture.

\begin{conjecture}
\label{conjecturelocalirreds} 
All irreducibles are local on a framework $M$.
\end{conjecture}

\medskip

Suppose that $ W $ is an indecomposable generalised tensor. Then $W$ is 
homogeneous. Suppose that the local irreducible appearing in the local 
decomposition of $W$ is local on $M$ and let $V$ be a generalised tensor 
with this local action.

Consider the set $F$ of all functions from $V$ to $W$. There is a well defined
local and global action on $F$ under which $F$ is a generalised tensor. If
\( f \in F \) then the local action is given by 
\begin{equation}
\left(T_k\oo(f)\right)\,(v) = T_k\oo\left( f(v)\right) - f(T_k\oo(v))
\end{equation}
while the global action is 
\begin{equation}
\left(\nabla_k(f)\right)\,(v) = \nabla_k\left( f(v)\right) - f(\nabla_k(v))
\end{equation}
Let $E$ be the set of all functions in $F$ which are local intertwining maps.
Then $E$ is precisely the homogeneous component belonging to the locally 
trivial action on $F$ and hence $E$ is invariant also under the global action 
giving a generalised tensor.

\begin{proposition}
\label{irredtensordecomp}
The generalised tensor $E$ is indecomposable and $W$ is isomorphic to 
the generalised tensor space $E\tensor V$. 
\end{proposition}

As a consequence of this proposition, if conjecture~\ref{conjecturelocalirreds} 
holds, then every indecomposable generalised  tensor on a framework 
arises as the tensor product of a  locally irreducible generalised tensor 
with an indecomposable locally trivial generalised tensor. This reduces
the problem of finding indecomposable generalised tensors on a framework
to two cases; the locally trivial case and the locally irreducible case. 


	\chapter{Specific Tensors}
\label{chapter specific tensors}

In the interests of  conservation of ink and provided there is no potential
for confusion, the word \textbf{tensor} by itself will be used to mean 
generalised tensor in what follows. Note that a tensor must have a specified 
local and global action.

Our main objection in this chapter is to show that the small dimensional 
representations of \( \so(2,3) \)  in figure~\ref{weights} are local on 
a framework $M$ by constructing tensors with these local actions.
We will also develop the algebra of tensors.

\bigskip
\pagebreak[2]

\section{Notation and General Discussion}

Consider a generalised tensor space \( V \). Given a basis of $V$ we can
describe elements of this space using index notation. We will
distinguish between different types of generalised tensor by our choice
of alphabet for the indices. We can however use repeated indices to denote
tensor product and the same alphabet with lowered indices to denote the 
dual. 

Some choices of alphabet are fixed at the outset. We will use upper indices
from the lowercase roman alphabet for vector fields, and upper indices from
the lowercase greek alphabet for spinor fields. 

For the purposes of this discussion let us use upper indices with the 
boldface lowercase roman alphabet to describe some arbitrary generalised space.
Hence a generalised tensor from this space would be written as  
\( v^\textbf{a} \) in terms of some chosen basis of $V$. Should we have 
need to consider a different basis of $V$ we would use primes to denote this. 
For example we would write a change of basis equation as
\begin{equation} 
v^{\textbf{a}'} = \delta^{\textbf{a}'}_{\textbf{a}}v^\textbf{a}
\end{equation}

where \( \delta^{\textbf{a}'}_{\textbf{a}} \) denotes the change of basis 
matrix. We will allow the change of basis matrix be a function of location 
in $M$. There is no canonical way to identify $V$ values at one point 
with those at another; the best we can do is to use parallel transport
which is path dependent. Hence a fully general treatment should allow 
basis changes of this type.

The local action is given by \( T^\textbf{b}_{k\textbf{a}} \) where
\begin{equation}
T_k\oo (v^\textbf{a}) = T^\textbf{a}_{k\textbf{b}} v^\textbf{b}
\end{equation} 

using the summation convention. The
connection for the global action is \( \Gamma^\textbf{b}_{m\textbf{a}} \)
and the global action is
\begin{equation} 
\nabla_k(v^\textbf{a}) =  \del_k(v^\textbf{a}) +
 \Gamma^\textbf{a}_{k\textbf{b}}v^\textbf{b} 
\end{equation}

Hence we have
\begin{equation}
T^\textbf{b}_{i\textbf{c}} T^\textbf{c}_{j\textbf{a}}
- T^\textbf{b}_{j\textbf{c}} T^\textbf{c}_{i\textbf{a}} =
T^k_{ij} T^\textbf{b}_{k\textbf{a}}
\end{equation}
and 
\begin{equation}
\bigl[\nabla_i,\nabla_j\bigr](x^\textbf{a}) =
T^k_{ij} \nabla_k(x^\textbf{a}) + R^\textbf{a}_{ij\textbf{b}}x^\textbf{b}
\end{equation}

where \( R^\textbf{b}_{ij\textbf{a}} \)  expresses the curvature in terms
of our bases.

\medskip

A change of coordinates is an intertwining map for both actions. Hence
\begin{flalign*}
&   T_k\oo    1^{\textbf{a}'}_\textbf{a}  = 0 \\
&   \nabla_m( 1^{\textbf{a}'}_\textbf{a}) = 0 
\end{flalign*}

These give equations describing the effect of a change of bases on
$T^\textbf{a}_{k\textbf{b}}$ and $\Gamma^\textbf{a}_{k\textbf{b}}$.
\begin{flalign}
T^{\textbf{a}'}_{k'\textbf{b}'}
&= 1^{\textbf{a}'}_{\textbf{a}} 1^{\textbf{b}}_{\textbf{b}'} 1^k_{k'}
T^{\textbf{a}}_{k\textbf{b}} \\
\Gamma^{\textbf{a}'}_{k'\textbf{b}'}
&= 1^{\textbf{a}'}_{\textbf{a}} 1^{\textbf{b}}_{\textbf{b}'} 1^k_{k'}
\Gamma^{\textbf{a}}_{k\textbf{b}} -
1^k_{k'}\del_k\left( 1^{\textbf{a}'}_{\textbf{t}} \right)
1^{\textbf{t}}_{\textbf{b}'}
\end{flalign}

We see that \( T^{\textbf{a}}_{k\textbf{b}} \) changes basis as expected 
for a generalised tensor in the space \( V\tensor \T^*\tensor V^* \). 
Indeed it is such a generalised tensor. However the connection  
\( \Gamma^{\textbf{a}}_{k\textbf{b}} \) has different behavior 
under a change of basis. Although we write it using index notation which 
would suggest that it belongs to the space \( V\tensor \T^*\tensor V^* \),
it is not a generalised tensor, and we therefore need to be careful with it 
particularly when changing basis.

\medskip
This completes our general discussion. We end this section by listing the 
tensors we have met so far.

\bigskip

\textbf{Scalars} are tensors with trivial local action and where
the global action is simply partial differentiation (the scalar connection 
is zero). Scalars are usually denoted without index. However we will
find it sometimes expedient to use a dummy index for scalars, and in such
situations a \( \circ \) may be used; for example \( f^\circ = f = f_\circ \).
Hence \( T^\circ_{k\circ} = 0 \) and \( \Gamma^\circ_{m\circ} = 0 \).
The existence of scalars shows that the \( (q_0,s_0) = (0,0) \)
representation is local on $M$.

\textbf{Spinors} are defined on a framework; they are mentioned in the
definition of a spinor manifold.  We denote them \( x^\alpha \) with a 
greek lowercase index. The local action is given by 
\( T^\beta_{k\alpha} \) and global action \( \nabla_k \) (with connection 
\( \Gamma^\beta_{k\alpha} \)).  Their existence 
shows that the \( (q_0,s_0) = (\nfrac{1}{2},\nfrac{1}{2}) \) representation
is local on $M$.

\textbf{Vectors} are vectors on the manifold in the usual sense of the word.
This means they are ten dimensional. We may from time to time need to use 
the word vector to refer to other things, but unless some other meaning is 
clear from the context, the word vector by itself will refer to a ten 
dimensional vector. 
We denote vectors \( v^i \) with a roman lower case index.  Vectors are
tensors with local action \( T^k_{ij} \) and global action \( \nabla_k \) 
(with connection \( \Gamma^k_{ij} \)).  Their existence is specified in the 
definition of a spinor manifold and shows that the \( (q_0,s_0) = (1,1) \) 
representation is local on $M$.

\bigskip
\pagebreak[2]

\section{Spinor Transformations}
\label{SpinorTransformations}

A real spinor transformation is a \( V \)-tensor of rank \( \rank{0,1}{0,1} \).
A typical tensor of this type is denoted \( X^\alpha_\beta \).  We begin by 
considering a smooth decomposition  into irreducibles under the local action. 
This will mirror the decomposition of \( 4 \times 4 \) matrices into 
irreducible representations carried out in chapter 1.
As there is no degeneracy in this decomposition (no two irreducibles are 
equivalent) it follows that each irreducible component is conserved by 
both the local and global actions and hence constitutes a type of 
generalised tensor.

The ten dimensional irreducible component is spanned by the set of 
\( V \)-tensors \( T^\alpha_{k\beta} \) for various \( k \).
If we choose \( \{ T^\alpha_{k\beta} \} \) as our basis, then the maps
\begin{flalign}
X^\alpha_\beta &\mapsto g^{ki}T^\beta_{i\alpha} = x^k \\
x^k &\mapsto T^\alpha_{k\beta}x^k = X^\alpha_\beta
\end{flalign}

represent the projection and injection maps to and from this component. We
should define local and global actions on the vector component so that these
maps are intertwining. However since  both \( T^\alpha_{k\beta} \) and 
\(g^{ij} \) are totally invariant, these maps are intertwining with respect to
the usual local and global action on vectors. We conclude that this component
gives precisely the space of vectors and not some other type of tensor. 
\medskip
The composition of projection and injection gives the totally invariant 
idempotent projection onto the vector component
\begin{equation}
X^\alpha_\beta \mapsto
	 g^{ij}T^\alpha_{i\beta}T^{\lambda}_{j\mu}X^{\mu}_{\lambda}
\end{equation}

\bigskip

The transformation \( \nfrac{1}{2} 1^\alpha_\beta \) is locally invariant
and spans a 1 dimensional irreducible component of the space of spinor 
transformations. Elements of this representation need no basis and will be 
denoted as scalars. In fact as we will see they are precisely scalars. 
To show this we need to check the nature of the local and global actions.

Projection maps to and from this component are given by 
\begin{flalign}
X^\alpha_\beta &\mapsto \nfrac{1}{2}1^\beta_\alpha . X^\alpha_\beta = 
\nfrac{1}{2} . X^\alpha_\alpha = x \\
x &\mapsto \nfrac{1}{2} 1^\alpha_\beta . x = X^\alpha_\beta
\end{flalign}

The factor of \( \nfrac{1}{2} \) was included to ensure idempotency of the 
composition 
\begin{equation}
X^\alpha_\beta \mapsto \nfrac{1}{4}1^\alpha_\beta . X^\lambda_\lambda
\end{equation}

Local and global actions are defined with respect to our chosen basis so that 
these the change of basis maps are totally invariant. This gives 
\( T_k\oo ( x ) = 0 \) and \( \nabla_m( x ) = \del_m (x) \) which we 
recognise as the usual local and global actions on scalars. We conclude that 
the one dimensional component is actually scalar and not some other kind of 
one dimensional generalised tensor.

\bigskip

We can find an idempotent projection onto the remaining 5-D component 
by subtracting the idempotent projections onto the other two components from
the identity.
\begin{equation}
X^\alpha_\beta \mapsto X^\alpha_\beta 
	- \nfrac{1}{4}1^\alpha_\beta . X^\lambda_\lambda
	- g^{ij}T^\alpha_{i\beta}T^{\lambda}_{j\mu}X^{\mu}_{\lambda}
\end{equation}

Choose a basis for this component \( \{ T^\alpha_{A\beta} \} \) which we will
denote with a capital roman index. Assume such a choice of basis has been made 
smoothly on the manifold so that we can discuss the derivative using this 
basis. 

These basis elements project trivially onto the other components giving
\begin{flalign}
& T^\alpha_{A\alpha} = 0 \\
& T^\alpha_{A\beta}T^\beta_{k\alpha} = 0 
\end{flalign}

At this point it is also useful to define
\begin{equation}
g_{AB} = T^\alpha_{A\beta}T^\beta_{B\alpha}
\end{equation}

\medskip

Projection maps to and from this component are given by 
\begin{flalign}
X^A &\mapsto  T^\alpha_{A\beta} X^A = X^\alpha_\beta \\
X^\alpha_\beta &\mapsto S^{A\beta}_\alpha . X^\alpha_\beta = X^A
\end{flalign}

A combination of these two maps gives the idempotent projection, hence
\begin{equation}
\label{idemopotentversorprojection}
T^\alpha_{A\beta} S^{A\nu}_{\mu} =
1^\alpha_{\mu}1^{\nu}_\beta 
	- g^{ij}T^\alpha_{i\beta}T^{\nu}_{j\mu}
	- \nfrac{1}{4}1^\alpha_\beta.1^{\nu}_{\mu}
\end{equation}

Contracting both sides of equation~\ref{idemopotentversorprojection}
 with \( T^\beta_{B\alpha} \) we obtain
\begin{equation}
S^{A\nu}_{\mu}g_{AB} = T^{\nu}_{B\mu}
\end{equation}

While contracting equation~\ref{idemopotentversorprojection} with 
\( S^{C\mu}_{\nu} \) gives
\begin{equation}
\left( S^{C\mu}_{\nu} S^{A\nu}_{\mu}\right)g_{AB} = 1^C_B
\end{equation}

hence \( g_{AB} \) is non-singular with inverse \( g^{AB}
= S^{A\mu}_{\nu} S^{B\nu}_{\mu} \) and
\begin{equation}
S^{A\beta}_{\alpha} = T^{\beta}_{B\alpha}g^{AB}
\end{equation}

The local and global actions in our new basis are such that these maps are
all invariant. The local action \( T^B_{iA} \) is thus given by 
\begin{equation}
T_i\oo\left(T^\alpha_{A\beta}\right) = 0
\end{equation}

and the connection \( \Gamma^B_{mA} \) giving parallel transport is defined 
by the equation
\begin{equation}
\nabla_m\left(T^\alpha_{A\beta}\right) = 0 
\end{equation}

Since \( T^\alpha_{A\beta} \) is totally invariant it follows that 
\( g_{AB} \) and \( g^{AB} \) are also totally invariant and We may use 
the metric \( g_{AB} \) and its dual \( g^{AB} \) to raise and lower these 
indices. This process will commute with both the local and global actions.

\medskip

The 5-dimensional local representation on \( \{ x^A \} \) is the canonical 
representation of \( \so(2,3) \) while the symmetric bilinear form 
\( g_{AB} \) is the canonical  metric. Hence the irreducible representation
of \( \so(2,3) \) with maximal weight \( (q_0,s_0) = (1,0) \) is local on $M$.

We call this type of tensor a \textbf{versor}. Versors are denoted \( k^A \)
with an upper-case roman index.  The word versor was originally coined by 
Hamilton to refer to a specific type of quaternion associated with an axis 
of reversal, but has fallen into disuse.  The discussion at the end of 
section~\ref{sectionsp4r} where the 5-D canonical representation was 
associated with axes of reversal makes this choice of notation seem 
appropriate. It is also a great word which deserves to be recycled.

\bigskip

We finish this section with some identities obtained by applying our 
decomposition to the product \( T^\alpha_{i\lambda}T^\lambda_{j\beta} \). 
We note that we can initially decompose this product into components
which are symmetric and antisymmetric in terms of the indices \( i \) 
and \( j \), and this decomposition will be totally invariant. The 
antisymmetric component is obtained from the commutator
\begin{equation}
T^\alpha_{i\lambda}T^\lambda_{j\beta} -
T^\alpha_{j\lambda}T^\lambda_{i\beta} =
T^k_{ij}T^\alpha_{k\beta}
\end{equation}

and is the vector component. The symmetric component is obtained from 
the Jordan bracket. As the Jordan bracket projects trivially onto the 
vector component it must decompose into the other components. The scalar 
component of the Jordan bracket is \( g_{ij} \). We denote the 5-D component
\( g^A_{ij} \) so that 
\begin{equation}
\label{JordanVersor}
T^\alpha_{i\lambda}T^\lambda_{j\beta} +
T^\alpha_{j\lambda}T^\lambda_{i\beta} =
\nfrac{1}{2}g_{ij}1^\alpha_\beta + g^A_{ij}T^\alpha_{A\beta}
\end{equation}

Since all other quantities in this equation are globally invariant it follows 
that \( g^A_{ij} \) is also globally invariant. It is symmetric in \( i \)
and \( j \) and behaves similarly to the metric in many respects. We can 
obtain an explicit expression for it by projecting the Jordan bracket.
\begin{equation}
g^A_{ij} = g^{AB}T^\beta_{B\alpha}\left(
T^\alpha_{i\lambda}T^\lambda_{j\beta} +
T^\alpha_{j\lambda}T^\lambda_{i\beta} \right)
\end{equation}

Putting these results together the product can be decomposed as follows
\begin{equation}
\label{ProductofTs}
T^\alpha_{i\lambda}T^\lambda_{j\beta}
= \nfrac{1}{2}T^k_{ij}T^\alpha_{k\beta} + 
\nfrac{1}{4}g_{ij}1^\alpha_\beta + \nfrac{1}{2}g^A_{ij}T^\alpha_{A\beta}
\end{equation}

\bigskip

In the case that the spinors are complex then as discussed in 
section~\ref{complex matters} there is a conjugation map conserved by
both the local and global actions on spinors. This extends naturally to 
give a conjugation on spinor transformations which is also conserved by
the local and global action. Complex spinor transformations are therefore 
simply the tensor product of real spinor transformations with \( \C \). 

The conjugation map will preserve the decomposition into irreducibles and 
hence the complex irreducible components are simply the tensor product with 
\( \C \) of the real irreducibles identified in this section.

\bigskip
\pagebreak[2]

\section{Action of $\so(3,3)$}

The operators \( T^\beta_{A\alpha} \) together with the operators
\( T^\beta_{i\alpha} \) define a natural local representation of the Lie
algebra \( \so(3,3) \) on the manifold. Local invariance of 
\( T^\beta_{A\alpha} \) gives
\begin{equation}
T^\beta_{i\mu}T^\mu_{A\alpha}
- T^\beta_{A\mu}T^\mu_{i\alpha} = T^B_{iA}T^\beta_{B\alpha} 
\end{equation}

We can thus regard \( T^B_{iA} \) in two ways, firstly as defining the 
action of $\so(2,3)$ on versors, and secondly as specifying some of the 
mixed index structure coefficients for the Lie algebra \( \so(3,3) \). 
Examining table~\ref{so33_table} we note that most other mixed index 
structure coefficients will be identically zero. Using compatible $T $ 
notation in particular we observe
\begin{equation}
T^k_{iA} = 0 
\mbox{ \hskip1cm}
T^A_{ij} = 0 
\mbox{ \hskip1cm}
T^C_{AB} = 0 
\end{equation}

However mixed index structure coefficients of the form \( T^k_{AB} \)
need not be zero and we can put them on our manifold by defining them via
\begin{equation}
\label{ABstructure}
T^\beta_{A\lambda}T^\lambda_{B\alpha}
- T^\beta_{B\lambda}T^\lambda_{A\alpha}
= T^k_{AB} T^\beta_{k\alpha}
\end{equation}

\medskip

The invariant metric on $\so(3,3)$ defined in the usual way to be 
\begin{equation}
g_{\Lambda\Pi} = T^\alpha_{\Lambda\beta}T^\beta_{\Pi\alpha}
\end{equation}

where the indices \( \Lambda \) and \( \Pi \) are lower-case or upper-case.
This simply evaluates to \( g_{ij} \) on lower-case indices and \( g_{AB} \) 
on upper-case indices. The mixed terms are all zero.

\medskip
Note that the totally invariant tensors \( g^{ij} \), \( g_{ij} \), 
\( g^{AB} \), \( g_{AB} \) can be used to raise and lower vector and 
versor indices in the obvious way. Contraction with these defines a 
bijective map which respects both local and global action.

\bigskip
\pagebreak[2]

\section{Tensors with Two Versor Indices}
\label{Tensors with Two Versor Indices}

Consider the space \( \{ X^{AB} \} \) of tensors with two 
versor indices. This space is 25 dimensional and decomposes into 
irreducible representations under the local action of dimensions
1 , 10 and 14. Since there is no degeneracy in this decomposition 
each of these irreducibles is preserved by the global as well as 
the local action and defines a type of tensor. The antisymmetric 
subspace is 10 dimensional and is irreducible. We can construct a 
totally invariant operator mapping this component to vectors
as follows.
\begin{equation}
X^{AB} \mapsto g^{ki}g_{AC}T^C_{iB}X^{AB}  = x^k
\end{equation}

As this map is invariant under the local and global action 
this identifies the ten dimensional component as being ordinary vectors.

The symmetric tensors with two versor indices form a 15 dimensional 
space which obviously has the trivial representation as a 
component. The intertwining operator here is
\begin{equation}
X^{AB} \mapsto g_{AB}X^{AB} = x
\end{equation}

which is invariant. Hence this 1-D component is scalar. 

\medskip
The trace free symmetric tensors with two versor components give the 
irreducible 14 dimensional representation, as can be confirmed by checking
eigenvalues for \( T \) and \( I \). This demonstrates that the representation
of \( \so(2,3) \) with \( (q_0,s_0) = (2,0) \) is local on $M$. 

We will rarely have need to talk about these tensors so we do not reserve an 
exclusive alphabet for them. Instead we will assign a temporary alphabet at 
the time of use.

\bigskip
\pagebreak[2]

\section{Tensors with one Versor and one Spinor Index}

Consider the space \( \{ X^{A\alpha} \} \) of tensors with one
versor and one spinor index. This is 20 dimensional. There are
natural invariant maps
\begin{flalign}
X^{A\alpha} &\mapsto T^\beta_{A\alpha}X^{A\alpha} = x^\beta \\
x^\beta &\mapsto g^{AB}T^\alpha_{B\beta}x^\beta = X^{A\alpha} 
\end{flalign}

which identify a 4-D spinor component. The kernel of the projection onto 
spinors is 16 dimensional. If \( v^A \) is a versor with weight 
\( (q,s) = (1,0) \) and \( x^\alpha \) is a spinor with weight
\( (q,s) = (\nfrac{1}{2},\nfrac{1}{2} \) then the tensor product 
\( v^Ax^\alpha = X^{A\alpha} \) has weight \( (\nfrac{3}{2},\nfrac{1}{2}) \).
Hence this weight must appear as a weight in one of the irreducible components.
The only way this can happen is if the entire 16 dimensional kernel
of the projection into spinors is the irreducible component with 
\( (q_0,s_0) = (\nfrac{3}{2},\nfrac{1}{2}) \). So this representation is 
also local on $M$.

We will rarely have need to talk about these tensors so we do not reserve an 
exclusive alphabet for them. Instead we will assign a temporary alphabet at 
the time of use.

\bigskip
\pagebreak[2]

\section{Symmetric 3 and 4 Component Spinors}

Consider the space of three component spinors \( \{ X^{\alpha\beta\gamma} \).
This is a 64 dimensional space of tensors which splits into a symmetric 
component of dimension 20, an antisymmetric component of dimension 4, and 
a mixed symmetry component of dimension 40. We focus our attention here 
on the totally symmetric component.

If \( v^\alpha \) has weight \( (q,s) = (\nfrac{1}{2},\nfrac{1}{2}) \) then
the tensor product of three of these \( v^\alpha v^\beta v^\gamma \) is
totally symmetric and hence lies in our space. We can compute its weight 
to be \( (q,s) = (\nfrac{3}{2},\nfrac{3}{2}) \) so this must be a weight
of one of the irreducible components of the 20 dimensional space of 
totally symmetric 3 component spinors. But the only way that can happen 
is if the space of totally symmetric 3 component spinors is irreducible
with maximal weight \( (q_0,s_0) = (\nfrac{3}{2},\nfrac{3}{2}) \).
Hence this irreducible is also local on $M$.

\bigskip

Similarly the space of symmetric 4-component spinors is 35 dimensional.
It includes the tensor \( v^\alpha v^\beta v^\gamma v^\delta \) where
\( v^\alpha \) has weight \( (q,s) = (\nfrac{1}{2},\nfrac{1}{2}) \).
This tensor has weight \( (q,s) = (2,2) \) it follows that the space of
symmetric 4-component spinors is locally irreducible with maximal weight
\( (q_0,s_0) = (2,2) \). Hence this irreducible is local on $M$.

\bigskip
\pagebreak[2]

\section{Tensors with Two Vector Indices}
\label{section twovectorindices}

Consider the space \( \{ X^{ij} \} \) of tensors with two 
vector indices. This space is 100 dimensional and decomposes into six
different irreducibles under the local action. Due to the lack of 
degeneracy each is globally invariant and we may discover a local and global
action for it.

The antisymmetric subspace is 45 dimensional and is the sum of two irreducible 
components of dimensions 10 and 35 corresponding to \( (q_0,s_0) = (1,1) \)
and \( (q_0,s_0) = (2,1) \) respectively. The tensor \( T^k_{ij} \) is 
totally invariant and defines an intertwining map 
\( X^{ij} \mapsto T^k_{ij}X^{ij} \) which projects onto the 10-dimensional 
irreducible component and identifies it as having the correct global and
local action to be the space of vectors. The antisymmetric tensors which 
lie in the kernel of this map give the 35 dimensional irreducible component.

The symmetric subspace is 55 dimensional and is the sum of irreducibles
of dimension 1, 5, 14 and 35 corresponding to \( (q_0,s_0) = (0,0) \) ,
\( (1,0) \), \( (2,0) \) and \( (2,2) \) respectively. We seek totally
invariant maps which will enable us to perform this decomposition.

\medskip

The trace map 
\begin{equation}
X^{ij} \mapsto g_{ij}X^{ij}
\end{equation} 

is a totally invariant map into the scalar component.

\medskip
A totally invariant map into the versor component can be constructed
using Jordan Brackets.

In equation~\ref{JordanVersor} we defined the totally invariant tensor
\( g^A_{ij} \) via a Jordan bracket of \( T \) operators.

The map
\begin{equation} 
X^{ij} \mapsto X^{ij}g^A_{ij} 
\end{equation}

is a fully invariant map from two component vectors to versors and
identifies the versor component. 

\medskip

An intertwining map which takes two dimensional symmetric tensors to two 
dimensional symmetric versors can be constructed as follows.
\begin{equation}
s^{AB}_{ij} = T^\alpha_{i\beta}T^\beta_{X\gamma}T^\gamma_{j\delta}T^\delta_{Y\alpha}g^{AX}g^{BY}
\end{equation}

This allows us to use our previous identification of the 14 dimensional 
irreducible as the space of two dimensional symmetric versors of trace 
zero to pick out the 14 dimensional component here.

\medskip

The 35 dimensional component of the symmetric part can be identified as the 
kernel of all the maps mentioned so far. This representation, which is 
locally irreducible with maximal weight \( (q_0,s_0) = (2,2) \) is therefore
locally equivalent to the space of fully symmetric spinors with four indices.
The global actions differ however. We will look at this in more detail
once we understand Crump scalars which are introduced in the next chapter.

\bigskip

In section~\ref{IndecomposableGeneralisedTensors} we claimed that all the 
small dimensional representations appearing in figure~\ref{weights} are local 
on $M$. We have now proved this by finding generalised tensors for each of 
these local actions. This means proposition~\ref{irredtensordecomp} applies
to indecomposable generalised tensors whose local action consists of
these irreducibles. Such generalised tensors can thus be written as a tensor
product of the generalised tensors we have found with locally trivial 
indecomposable generalised tensors.

	\chapter{Crump Scalars and Beyond}
\label{bullets and beyond}

In this chapter we continue our exploration of tensors applicable to physics.
We will also continue to develop algebraic techniques and will discuss a number
of useful identities.

Crump scalars are an important topic in this chapter. These locally 
trivial tensors arise naturally from a consideration of the invariant 
symplectic form. Crump scalars turn out to be particularly important in 
understanding electromagnetism.

We begin by looking at the Casimir identities which are useful in 
identifying locally homogeneous components of a tensor.

\bigskip
\pagebreak[2]

\section{Casimir Identities}

The quadratic and quartic Casimir invariants for representations of 
\( \so(2,3) \) were introduced in equations~\ref{2nd_degree_Casimir}
and~\ref{4th_degree_Casimir}. We can extend these to operators on the 
manifold which commute not only with the local action but also with the
global action.

The operator \( g^{ij}T_i\oo T_j\oo \) expresses the quadratic Casimir operator
\( Q \) discussed in the first chapter at every point on our manifold and is
both globally and locally locally invariant. By Schur's lemma it is scalar on 
every local irreducible. Furthermore its eigenspaces are globally invariant 
and give a decomposition of any generalised tensor. Hence this operator
is scalar on any indecomposable generalised tensor space. 

Applying this to specific spaces gives the useful \textbf{Casimir identities}.
\begin{equation}
\begin{split}
g^{ij}T^\beta_{i\lambda}T^\lambda_{j\alpha} &= \nfrac{5}{2} \,.\, 1^\beta_\alpha \\
g^{ij}T^B_{iX}T^X_{jA} &= 4\,.\, 1^B_A \\
g^{ij}T^b_{ix}T^x_{ja} &= 6\,.\, 1^b_a
\end{split}
\label{CasimirIdentities}
\end{equation}

There is a Casimir identity for every finite dimensional irreducible 
representation. In the general case, using the formula in 
equation~\ref{Casimirformula} we have 
\begin{equation}
g^{ij}T^\Delta_{i\Theta}T^\Theta_{j\Lambda} =\bigl( q(q+3)+s(s+1)\bigr) \,.\, 1^\Delta_\Lambda
\end{equation}

where the constant is a function of the maximum weight 
vector \( (q,s) \) for the irreducible local representation.

\medskip

The quartic Casimir operators are also globally invariant and give rise 
to a similar set of identities. The 4th degree Casimir invariant takes 
the form
\begin{equation}
g^{AB}g_A^{ij}g_B^{kl}T_i\oo T_j\oo T_k\oo T_l\oo
\end{equation}

and the general identity arising from it is
\begin{equation}
g^{AB}g_A^{ij}g_B^{kl}T^\Delta_{i\Theta_1} T^{\Theta_1}_{j\Theta_2} T^{\Theta_2}_{k\Theta_3} T^{\Theta_3}_{l\Lambda}
= \bigl(q -s +\nfrac{4}{9}s^2\left(q + 1\right)^2\bigr) \,.\, 1^\Delta_\Lambda
\end{equation}

where the irreducible representation has maximal weight vector \( (q,s) \).

\bigskip
\pagebreak[2]

\section{The Symplectic Form.}
\label{TheSymplecticForm}

In section~\ref{SpinorTransformations} we defined totally invariant 
bilinear forms
\begin{flalign*}
g_{ij} &= T^\beta_{i\alpha}T^\alpha_{j\beta} \\
g_{AB} &= T^\beta_{A\alpha}T^\alpha_{B\beta} 
\end{flalign*}

on vectors and versors. We now look at whether we can similarly define an
invariant  bilinear form on spinors. In the local case we have such a form, 
the canonical symplectic form defining \( \sp(2,\R) \) which is unique up to 
a choice of scalar. Consequently a locally invariant form 
\( s_{\alpha\beta} \) can be chosen smoothly on the manifold and is 
unique up to a choice of scalar field. Explicitly, if \( t_{\alpha\beta} \) 
is another such invariant form then 
\begin{equation}
t_{\alpha\beta} = f.s_{\alpha\beta}
\end{equation}

for some scalar field \( f \).  Local invariance of \( s_{\alpha\beta} \) 
gives the equation 
\begin{equation}
T^\lambda_{k\alpha}s_{\lambda\beta} +
T^\lambda_{k\beta}s_{\alpha\lambda} = 0
\end{equation}

Taking a covariant derivative we obtain
\begin{equation} 
T^\lambda_{k\alpha}\nabla_m(s_{\lambda\beta}) 
	+ T^\lambda_{k\beta}\nabla_m(s_{\alpha\lambda}) = 0 
\end{equation}

Thus \( \nabla_m(s_{\alpha\beta}) \) is also a locally invariant bilinear
form for each \( m \). We conclude that there is a vector \( A_m \) with
\begin{equation}
\nabla_m(s_{\alpha\beta}) = - A_m . s_{\alpha\beta} 
\end{equation}

The negative sign here anticipates results that will follow later.
For the dual form defined by 
\( 1^\alpha_\beta = s^{\alpha\lambda}s_{\lambda\beta}
= s_{\beta\lambda} s^{\lambda\alpha} \) we will then have
\begin{equation}
\nabla_m(s^{\alpha\beta}) = A_m .  s^{\alpha\beta}
\end{equation}

Of course the vector field \( A_m \) here depends on the choice
of \( s_{\alpha\beta} \). Choose another locally invariant 
non-degenerate form \( t_{\alpha\beta} \) and denote the associated
vector field \( B_m \). Then \( t_{\alpha\beta} \) is a multiple of 
\( s_{\alpha\beta} \) at every point by a non-zero (due to non-degeneracy) 
scalar.  Assume (with very little loss of generality) that the scalar is 
positive.  Then we may write 
\begin{equation}
\label{cobbullet}
t_{\alpha\beta} = e^f s_{\alpha\beta}
\end{equation}

for some scalar function \( f \). And
\begin{equation}
\label{DElast}
\nabla_m(t_{\alpha\beta}) = \left(\nabla_m(f)-A_m\right) . t_{\alpha\beta} 
\end{equation}

We conclude that 
\begin{equation}
\label{newA} 
B_m = A_m - \nabla_m(f) 
\end{equation}

We are interested in the question of whether or not it is possible to choose
\( s_{\alpha\beta} \) in such a way that it is globally invariant. 
We will be able to do this if we can find a scalar function \( f \) solving 
the differential equation
\begin{equation}
\label{theDEforf}
\nabla_m(f) = A_m 
\end{equation} 

\medskip
If a solution \( f \) to equation~\ref{theDEforf} exists, then 
applying \( [\nabla_i,\nabla_j] \) to \( f \) gives the necessary condition 
\begin{equation}
\nabla_i(A_j) - \nabla_j(A_i) = T^k_{ij}A_k 
\end{equation}

We define 
\begin{equation}
\label{defineF}
F_{ij} = \nabla_i(A_j) - \nabla_j(A_i) - T^k_{ij}A_k
\end{equation}

so that our condition becomes
\begin{equation}
\label{condition1}
F_{ij} = 0
\end{equation}

A differential equation like equation~\ref{theDEforf} can be solved by taking 
line integrals from a fixed point provided that this gives a well defined 
answer independent of path. For a simply connected region this will be true 
when line integrals around a small loop are zero, which is true if and only 
if condition~\ref{condition1} holds.

Hence condition~\ref{condition1} is both a necessary and sufficient condition 
that a globally invariant bilinear form \( s_{\alpha\beta} \) exists in a
simply connected region.

\medskip

Having gone as far as we can go via a purely mathematical argument it is 
tempting at this point to appeal to a physical argument and simply assert 
that a globally invariant symplectic form should exist. Such a form would
be very useful. It would for example allow us to raise and lower spinor 
indices in an invariant manner. Indeed in the initial development of this 
theory this was the path we followed. 

However the notation above was not chosen by chance. We later will find
reason to interpret \( A_i \) as the electromagnetic potential. 
This makes \( F_{ij} \) the electromagnetic field and 
condition~\ref{condition1} is thus effectively a zero field condition 
and constrains the physics in an undesirable way. 

Hence it isn't simply a case of not being able to find a globally invariant 
sympletic form. The physics is telling us that the non-existence of such a 
form is essential for a proper physical description of electromagnetism.

\bigskip

This link between the electromagnetic field and the lack of a globally 
invariant symplectic form was discovered by the lead author and 
William Crump during work on his thesis. An early account of this work 
can thus be found in~\cite{Cthesis}. Prior to this revelation we had been 
falsely assuming that a unique global action should exist for each local 
action. Since the local action on the symplectic form is trivial we had 
therefore identified it as scalar and thus globally invariant.

This error meant that the version of Maxwell's equations we had obtained 
included a non-physical constraint; effectively a zero field condition. 
We also wish to thank Dr. Yuri Litvinenko for pointing out that this extra 
constraint was indeed unphysical and was not simply a gauge condition 
as we had at first hoped.

\bigskip
\pagebreak[2]

\section{Tensors with Two Spinor Indices}
\label{TensorsWithTwoSpinorIndices}

Consider the space \( \{ X^{\alpha\beta} \} \) of tensors with two 
spinor indices. This space is 16 dimensional and, like the space of 
spinor transformations, decomposes into irreducible representations 
under the local action of dimensions 1, 5 and 10. Indeed the space of 
tensors with two spinor indices and the space of spinor transformations
are locally equivalent. Contraction with \( s_{\alpha\beta} \)
gives an isomorphism with respect to the local action. This map however
does not conserve the global action and so these two types of tensors 
are not equivalent.

Choose an element (a basis) of the trivial component smoothly at each point
and denote it \( e_\bullet^{\alpha\beta} \).
We will use a bullet as our alphabet for this trivial representation (as the 
space is one dimensional we only need an alphabet with one letter). 
A projection map onto this component is denoted \( s^\bullet_{\alpha\beta} \). 
Hence \( s^\bullet_{\alpha\beta} 
e_\bullet^{\alpha\beta} = 1^\bullet_\bullet = 1\).

We will call tensors of the same type as the trivial component 
\textbf{Crump scalars} or sometimes bullet scalars after the index used to
describe them. I have named them after my student William Crump who
demonstrated in his thesis that these are essential for a proper description 
of electromagnetism. Local and global actions on Crump scalars are then defined to ensure that 
\( e_\bullet^{\alpha\beta} \) and \( s^\bullet_{\alpha\beta} \) are 
totally invariant. 

In the case where our spinors are complex, our Crump scalars will also 
be complex. In this case we will want to extend the definition of the 
natural conjugation map on spinors discussed in section~\ref{complex matters} 
to Crump scalars. This is done by requiring that the basis elements 
\( s^\bullet_{\alpha\beta} \) are invariant under spinor conjugation
(in other words are real).  Choosing a real Crump scalar as the bullet 
basis will ensure that conjugation of Crump scalars is simply complex
conjugation of the single component.

As the local action on Crump scalars is trivial, we have 
\( T^\bullet_{i\bullet} = 0 \). Local invariance of 
\( s^\bullet_{\alpha\beta} \) then gives
\begin{equation}
   T^\lambda_{k\alpha}s^\bullet_{\lambda\beta} +
   T^\lambda_{k\beta} s^\bullet_{\alpha\lambda} = 0
\end{equation}

We conclude that \( s_{\alpha\beta} = s^\bullet_{\alpha\beta} \) is a 
locally invariant symplectic form. From 
section~\ref{TheSymplecticForm} we can write
\begin{equation}
\nabla_m(s_{\alpha\beta}) = - A_m . s_{\alpha\beta}
\end{equation}

for some vector field \( A_m \). Note however that this equation ignores 
any effect of the connection on the bullet index which is suppressed
in this equation.

By definition we must have \( \nabla_m(s^\bullet_{\alpha\beta}) = 0 \)
and indeed this condition defines the Crump connection 
\( \Gamma^\bullet_{m\bullet} \) which cannot then be zero.
We conclude that the bullet index is subject to non-trivial parallel 
transport with
\begin{equation}
\label{gammabullet}
\Gamma^\bullet_{m\bullet} = A_m
\end{equation}

Crump scalars \( x^\bullet \) are locally trivial and behave like 
scalars under the local action. However they transform 
differently under the global action. Note that this is an example of 
a locally trivial generalised tensor not formed using the trick from
section~\ref{locallytrivialspaces}.

We may take tensor products with Crump scalars to form bullet versions of
other tensors. These will have the same local action, but the global
action will be modified by the Crump connection. 

\bigskip

We can now identify the remaining irreducible components of the two 
component spinors as bullet vectors \( X^{k\bullet} \) and bullet versors 
\( X^{A\bullet} \) respectively via the invariant maps
\begin{flalign}
X^{\alpha\beta} &\mapsto
	X^{\alpha\lambda}s^\bullet_{\lambda\beta}T^\beta_{i\alpha}g^{ik} \\
X^{\alpha\beta} &\mapsto
	X^{\alpha\lambda}s^\bullet_{\lambda\beta}T^\beta_{B\alpha}g^{BA}
\end{flalign}

projecting onto each component.

\bigskip
\pagebreak[2]

\section{The Bullet Index}

Additional types of scalar can be obtained from Crump scalars using the 
dual and tensor product operations. In our notation these would be denoted
\( x_\bullet \) or \( x^{\bullet\bullet} \).  Note however that 
upper and lower bullets appearing together will cancel via what is 
essentially a process of automatic contraction. We see this for example 
in equation~\ref{gammabullet}.

We define the \textbf{bullet index} to be the number of bullets with the 
number being negative if the bullets are contravariant.  So \( x_{\bullet} \) 
will have bullet index \( -1 \) while \( x^{\bullet\bullet} \) has bullet 
index \( 2 \).  We may also combine bullets with other alphabets. So for 
example a quantity like \( v^k_{\bullet\bullet} \) would be described as a 
vector with bullet index \( -2 \). 

We can extend this definition of bullet index to tensors with a combination
of vector, versor, spinor and bullet indices by ignoring the vector and 
versor indices, and counting each spinor index as half a bullet. Hence
the bullet index of \( X_{k\bullet}^{\alpha\beta\gamma} \) is
\( \nfrac{1}{2} \). This will ensure that contraction with 
\( s^\bullet_{\alpha\beta} \) does not change the bullet index, which is 
reasonable since it is an intertwining map and does not change the nature
of the generalised tensor.
\medskip

The one element set \( \{ 1^\bullet \} \) is a basis of the 1-D trivial 
representation we have chosen in order to quantify Crump scalars. 
We shouldn't get too attached to it however as, like all bases, 
it is somewhat arbitrary.  
Consider a different basis \( 1^{\bullet'} \) with change of basis 
matrices \( \delta^{\bullet}_{\bullet'} \) and 
\( \delta^{\bullet'}_{\bullet} \).
Then the change of basis matrix is non-singular everywhere so we can write
(after adjusting the sign if needed) 
\( \delta^{\bullet'}_{\bullet} = e^f  \) for some scalar \( f \) and so
\begin{equation}
s^{\bullet'}_{\alpha\beta} 
	= \delta^{\bullet'}_{\bullet}s^\bullet_{\alpha\beta} 
 	= e^{f} . s^\bullet_{\alpha\beta} 
\end{equation}

which is equation~\ref{cobbullet}. We conclude that changing the basis for 
Crump scalars is equivalent to choosing a different symplectic form, and hence
Crump scalars can be regarded as specifying a choice of symplectic form.

\bigskip

To raise and lower spinor indices we would need
invariant bilinear forms \( s_{\alpha\beta} \) and 
\( s^{\alpha\beta} \) where
\begin{equation*}
s^{\alpha\lambda}s_{\beta\lambda} = 1^\alpha_\beta 
\end{equation*}
However as we have seen no such invariant forms exist. The 
symplectic form is in general only locally invariant.
To ensure global invariance we must instead use the bilinear form 
\( s^\bullet_{\alpha\beta} \) which maps a pair of spinors onto a 
Crump scalar and not to an ordinary scalar. The dual of this will 
then be \( s_\bullet^{\alpha\beta} \) where
\begin{equation}
\label{raiselowerspinor}
s_\bullet^{\alpha\lambda}s^\bullet_{\beta\lambda} = 1^\alpha_\beta 
\end{equation}

It follows that
\begin{equation}
\label{s_fullycontracted}
s_\bullet^{\alpha\beta}s^\bullet_{\alpha\beta} = 4
\end{equation}
and that the map \( \Pi^{\alpha\beta}_{\gamma\delta} = 
\nfrac{1}{4}s_\bullet^{\alpha\beta}s^\bullet_{\gamma\delta} \) on the set 
of generalised tensors with two spinor indices satisfies \( \Pi^2 = \Pi \)
and hence is a projection. Hence we can identify
\begin{equation}
s_\bullet^{\alpha\lambda} = 4e_\bullet^{\alpha\beta} 
\end{equation}

where  \( e_\bullet^{\alpha\beta}  \) is the basis of the bullet component
for generalised tensors with two spinor indices specified earlier. Essentially
we have adjusted our notation for \( s^\bullet_{\alpha\beta} \) and 
\( s_\bullet^{\alpha\beta} \) to avoid having the constant \( 4 \) appear
every time we raise and then subsequently lower a spinor index. 

Note that raising or lowering a spinor index in this way will leave a bullet 
behind.  Also as these forms are antisymmetric care is needed when raising 
and lowering to avoid possibly uncertainty over sign.  If, when we raise and 
lower, we place the contracted indices adjacent, then a consistent application 
of the rule 

\begin{quote}
\textbf{R}aise on the \textbf{R}ight ;
\textbf{L}ower on the \textbf{L}eft
\end{quote}

will ensure that raising and lowering are opposite operations.
\begin{equation*}
v^\gamma = \left(s^\bullet_{\beta\alpha}v^\alpha\right)s_\bullet^{\beta\gamma}
\end{equation*}

\begin{equation*}
v_\gamma = s^\bullet_{\gamma\beta}\left(v_\alpha s_\bullet^{\alpha\beta}\right)
\end{equation*}

\bigskip

In section~\ref{section twovectorindices} we noted that the 
35 dimensional irreducible \( (q_0,s_0) = (2,2) \) component of the 
symmetric two component vectors had a different global action to the 
35 dimensional irreducible \( (q_0,s_0) = (2,2) \) space of symmetric
four component spinors. Hence they are not the same kind of tensor.
We can now describe the nature of this difference which turns out to 
be simply a difference in bullet index.

We can construct intertwining maps between the space \( \{ X^{ij} \} \) 
of two component vectors and the space \( \{ 
Y^{\alpha\beta\gamma\delta}_{\bullet\bullet} \} \) of tensors with four 
spinor and two bullet indices by
\begin{flalign}
X^{ij} &\mapsto 
	X^{ij}  T^\alpha_{i\lambda}T^\gamma_{j\mu}
		s^{\lambda\beta}_\bullet s^{\mu\delta}_\bullet =
Y^{\alpha\beta\gamma\delta}_{\bullet\bullet} \\
Y^{\alpha\beta\gamma\delta}_{\bullet\bullet} &mapsto
s^\bullet_{\beta\lambda}
s^\bullet_{\delta\mu}
Y^{\alpha\beta\gamma\delta}_{\bullet\bullet}
T^\lambda_{x\alpha}T^\mu_{y\gamma}g^{xi}g^{yj} = X^{ij} 
\end{flalign}

These intertwining maps identify the two component vectors with the four
component spinors of bullet index two which are symmetric under the
exchange of the first two and last two spinor indices. Schur's lemma then 
tells us that these maps give an isomorphism between the 35 dimensional 
irreducible components.

\bigskip
\pagebreak[2]

\section{Scalars and the Zero Curvature Test}
\label{scalarsandzerocurvature}

Suppose that we come across a type of one dimensional generalised tensor and
want to know if it is scalar. Clearly any one 
dimensional tensor is locally trivial. But as Crump scalars demonstrate, 
one dimensional tensors need not be globally trivial and therefore scalar. 

One way to show that our tensor is indeed scalar would be to find a 
recognisable intertwining map to scalars. That is what we did in 
sections~\ref{Tensors with Two Versor Indices} 
and~\ref{section twovectorindices} when we identified the 1-D components 
of the symmetric tensors with two versor, or with two vector indices as 
scalar. However not being able to find such an intertwining map does not 
answer the question. What we need is a definitive test.

\medskip
Scalars of course are the ones with trivial global action. So it would seem
that all we need to do is look at the global action and see if it is trivial. 
However recognising when the global action is trivial is not easy since the 
global action is specified by the connection which changes basis in a 
complicated way. What is clearly trivial in one basis may not appear so in 
another.

The usual basis for scalars is denoted with the symbol \( \circ \), and in 
that basis we have \( \Gamma^\circ_{k\circ} = 0 \). Suppose we have a different 
basis \( \circ'\) with change of basis matrix \( \delta^\circ_{\circ'} \) 
and connection \( \Gamma^{\circ'}_{k\circ'} = B_k \).

As change of basis is an intertwining map we have 
\( \nabla_k( \delta^\circ_{\circ'} ) = 0 \) which gives
\begin{equation}
\del_k( \delta^\circ_{\circ'} ) = \Gamma^{\circ'}_{k\circ'} 
\delta^\circ_{\circ'} 
\end{equation}

Since our change of basis matrix is \( 1\times 1 \) it is a non-zero 
scalar function on the manifold and we can write it as either \( e^f \) or 
\( -e^{f} \) for some function \( f \). In both cases we have
\( \del_k( \delta^\circ_{\circ'} ) = \del_k(f).  \delta^\circ_{\circ'} \)
and can conclude
\begin{equation}
\label{Bisnablaf}
\Gamma^{\circ'}_{k\circ'} = B_k = \del_k(f) 
\end{equation}

So a one dimensional tensor will be scalar if and only if its connection 
is the gradient of a function.

A simpler test can be found by looking at the curvature which is a tensor 
and therefore changes basis in a simple fashion. The curvature of scalars
is \( R^\circ_{ij\circ} = 0 \).  It follows that it is zero in any basis
and so \( R^{\circ'}_{ij\circ'} = 0 \). This gives a necessary condition
for a one dimensional tensor to be scalar which we call the 
\textbf{zero curvature test}. We need
\begin{equation} \del_i(B_j) - \del_j(B_i) = 0 \end{equation}

if equation~\ref{Bisnablaf} is to hold. Furthermore if the manifold is simply 
connected the zero curvature test is also sufficient. It ensures that 
equation~\ref{Bisnablaf} can be solved by path integration with the result 
independent of path.

The constraint that the manifold is simply connected is a technical one that
can be satisfied by restriction to a simply connected region. Alternatively 
we could ensure that our entire manifold is simply connected by switching to 
a simply connected cover. 

However this does raise the interesting possibility of tensors that are
scalar in any simply connected region but where there is 
a global topological obstruction which prevents them from being 
scalar on the entire manifold.

\bigskip

Now that we have found a test to determine if a representation is 
scalar it is time to use it by applying it to some one dimensional
tensors. Consider the completely antisymmetric tensors with 
\begin{itemize}
\item four spinor indices \( \{ e^{\alpha\beta\gamma\delta} \} \).
\item five versor indices \( \{ k^{ABCDE} \} \).
\item ten vector indices \( \{ m^{abcdefghij} \} \).
\end{itemize}

To apply the zero curvature test to these tensors we need to understand 
the respective connections. The argument in all three cases is similar. 

Consider \( \Gamma_k\oo(e^{\alpha\beta\gamma\delta}) \).
Each spinor index can take four values which we choose to represent 
by the numbers 1 through 4. By antisymmetry it is enough to look at one 
component. Looking at the \( 1234\) component we have
\begin{equation}
\Gamma_k\oo(e^{1234}) =
\Gamma^1_{k\lambda}e^{\lambda234} +
\Gamma^2_{k\lambda}e^{1\lambda34} +
\Gamma^3_{k\lambda}e^{12\lambda4} +
\Gamma^4_{k\lambda}e^{123\lambda}
\end{equation}

By antisymmetry the coefficients of \( e^{\alpha\beta\gamma\delta} \) with 
repeated indices are zero. Hence the only non-zero components in this equation
are
\begin{equation}
\Gamma_k\oo(e^{1234}) =
\Gamma^1_{k1}e^{1234} +
\Gamma^2_{k2}e^{1234} +
\Gamma^3_{k3}e^{1234} +
\Gamma^4_{k4}e^{1234} = \Gamma^\lambda_{k\lambda}e^{1234}
\end{equation}

We conclude that
\begin{equation}
\Gamma_k\oo(e^{\alpha\beta\gamma\delta}) =
\Gamma^\lambda_{k\lambda} e^{\alpha\beta\gamma\delta}
\end{equation}

The same argument in the other two cases gives
\begin{flalign}
\Gamma_k\oo(k^{ABCDE}) &= \Gamma^M_{kM} k^{ABCDE} \\
\Gamma_k\oo(m^{abcdefghij}) &= \Gamma^t_{kt} m^{abcdefghij}
\end{flalign}

Knowing the connections allows us to compute the curvature and 
apply the zero curvature test to check whether these tensors are scalar.
The three curvatures are
\begin{flalign}
\del_i(\Gamma^\lambda_{j\lambda}) - \del_j(\Gamma^\lambda_{i\lambda})
&= R^\lambda_{ij\lambda} \\
\del_i(\Gamma^M_{jM}) - \del_j(\Gamma^M_{iM})
&= R^M_{ijM} \\
\del_i(\Gamma^t_{jt}) - \del_j(\Gamma^t_{it})
&= R^t_{ijt}
\end{flalign}

which we recognise as the traces of the spinor, versor and vector curvature 
tensors respectively.  We will be looking at the nature of these curvature 
tensors in more detail in the next chapter.  Equations~\ref{curvature_reduced}, 
\ref{reduceddecomposition} and~\ref{reducedversorcurvature} from that chapter 
show that the first of these expressions is \( F_{ij} \) (not usually zero), 
while the other two are indeed zero. 

We conclude that the completely antisymmetric tensors with four spinor 
indices are not scalar, but that the completely antisymmetric tensors with 
five versor or with ten vector indices are scalar.

\medskip

The fact that completely antisymmetric tensors with ten vector components 
are scalar has important consequences. We can use these tensors to build
differential forms and define a measure and hence an integration on the 
manifold. Because tensors of this type are scalar we can define this 
measure and integration to be invariant and to commute with both the 
local and global action. Such an invariant integration will be useful 
later when we look at Lagrangian methods. 

\medskip
This leaves us however with a question. If the completely antisymmetric 
tensors with four spinor indices are not scalar then what are they? Have
we found a new type of locally trivial tensor? In the next section we will
answer this question by finding intertwining maps that allow us to identify 
these tensors as two-bullet scalars of the form \( \{ x^{\bullet\bullet} \} \).

\section{Superalgebra Identities}


The basic idea of a Lie superalgebra is that \( T^i_{jk} \), 
\(  T^\beta_{i\alpha} \) and \( T^i_{\alpha\beta} = 
s_{\alpha\lambda}T^\lambda_{k\beta}g^{ki} \) can be viewed
as structure coefficients for a \( \Z_2 \)-graded algebra 
on vectors and spinors which satisfies a `super-Jacobi' identity.

Note however that in doing this we suppressed the bullet index. 
Putting it back to restore the proper global action breaks the 
interpretation in terms of a \( \Z_2 \)-graded algebra. 
Consequently the definition of Lie superalgebra is broken in a 
fairly fundamental way by the global action and is not natural 
on a framework.

\medskip
We can however salvage an interesting identity from the wreckage.
The super-Jacobi identity still holds and the various components 
of the super-Jacobi identity give many identities we are already 
familiar with, plus one new identity which is the main subject of 
this section.

\begin{theorem}[The Superalgebra Identity]
\begin{equation}\label{superalgebra identity}
s^\bullet_{\alpha\lambda}T^\lambda_{i\beta}T^\delta_{j\gamma}g^{ij} \abcequal 0 
\end{equation}
Where the notation \( \abcequal \) means the left hand side of the expression 
should be cyclically summed over the permutation \( (\alpha\beta\gamma) \) 
before being equated to the right.
\end{theorem}
\begin{proof}
This can be proved by direct computation using the explicit matrices
given in chapter 1. Each tensor in the sum on the left is formed as the sum
of ten 256 component tensors. Fortunately these tensors are sparse and can be 
described by simply listing non-zero components. The indices \( \alpha \), 
\( \beta \) and \( \gamma \) in this list can then be cyclically permuted and 
added. The calculation is simple enough to be done by hand, 
although it would be rather messy to present here.
\end{proof}

We can also generalise this result to obtain two further identities. 
Consider the expressions
\begin{flalign}
s^\bullet_{\alpha\beta} s^\bullet_{\gamma\delta}
\abcequal & \;\;
	E^{\bullet\bullet}_{\alpha\beta\gamma\delta} \\
 s^\bullet_{\alpha\lambda}T^\lambda_{A\beta} g^{AB} 
	T^\mu_{B\gamma} s^\bullet_{\mu\delta} 
\abcequal & \;\;
	 B^{\bullet\bullet}_{\alpha\beta\gamma\delta}
\end{flalign}

By examining the effect of transposing two of the spinor indices it is not 
hard to show that both of these expressions are completely antisymmetric.
Hence they are completely determined by one of their coefficients and in 
particular they differ by a scalar field. Furthermore both of these tensors 
are invariant. Hence they actually differ by a constant scalar. 

Note that the invariance of these tensors also means that we could use them 
to define an intertwining map between antisymmetric tensors with four spinor 
indices, and two-bullet scalars.

To understand more we must choose bases and compute. If, with respect to our 
chosen bases, we define \( \epsilon_{\alpha\beta\gamma\delta} \) to be the 
unique completely antisymmetric tensor with \( \epsilon_{1234} = 1 \) then 
we can write 
\begin{flalign}
E^{\bullet\bullet}_{\alpha\beta\gamma\delta} 
&= \; k^{\bullet\bullet}\epsilon_{\alpha\beta\gamma\delta} \\
B^{\bullet\bullet}_{\alpha\beta\gamma\delta} 
&= \; c k^{\bullet\bullet}\epsilon_{\alpha\beta\gamma\delta} 
\end{flalign}

for a suitable two-bullet scalar \( k^{\bullet\bullet} \) and 
constant scalar \( c \). Suppose we choose bases so that
at that at a specified point \( s^\bullet_{\alpha\beta} \) takes the 
form of the matrix \( \Omega \) from section~\ref{sectionsp4r} while 
the other matrices take the form of the corresponding matrices in 
chapter~\ref{spacetimesymmetry}. In these bases and at that point 
we can compute \( k^{\bullet\bullet} = -1 \) and \( c = -\nfrac{5}{4} \).

Since constant scalars are unchanged under a change of basis we will have 
\( c = -\nfrac{5}{4} \) at any point in any basis.

	\chapter{Curvature and Forces}
\label{chapter:global structure}

In this chapter we look more closely at the connection and the curvature
which describe forces in this mathematical context. In particular 
the methods of the last chapter enable us to decompose these into components 
with respect to the local action in a natural way. From this decomposition 
we will obtain dynamical equations for the components of curvature which are 
very physical. It is worth emphasising  that these equations are not assumed. 
They are mathematical identities that arise naturally and must hold on every 
framework.  There is enough structure here to potentially account for 
all observed forces however we will focus our attention in this chapter
on recognising electromagnetism and gravity.

\bigskip
\pagebreak[2]

\section{Components of Curvature}

The curvature \( R^\beta_{ij\alpha} \) can be considered as a 
set of spinor transformations indexed by \( i \) and \( j \). 
We can decompose these into components and write
\begin{equation}
R^\beta_{ij\alpha} =F_{ij}1^\beta_\alpha + R^k_{ij}T^\beta_{k\alpha} + R^A_{ij}T^\alpha_{A\beta}
\end{equation}

In fact we can do much better than this. As we will now show the third 
component in this decomposition is always zero, while the first can be 
directly computed from the curvature action on Crump scalars. Our first 
objective is to prove this result which will follow from the following 
useful Lemma.

\begin{lemma}
\label{the previous lemma}
Let \( M\oo \) be an operator defined on both spinors and Crump scalars, 
and suppose that \( M\oo\left(s^\bullet_{\alpha\beta}\right) = 0 \). Then
\begin{equation}
M^\alpha_\beta = M^k T^\alpha_{k\beta} + \frac{1}{2}M.1^\alpha_\beta
\end{equation}
Where \( M^k \) is some vector and where 
\( M = M^\bullet_\bullet \). Furthermore \( M = \nfrac{1}{2}M^\alpha_\alpha \).
\end{lemma}
\begin{proof} 
Let \( h_\bullet \) be any non-zero Crump scalar and define
\( s_{\alpha\beta} = h_\bullet s^\bullet_{\alpha\beta} \).
Applying \( M\oo \) we obtain
\begin{equation}
M\oo s_{\alpha\beta} 
	= M\oo(h_\bullet)s^\bullet_{\alpha\beta} 
	= -M s_{\alpha\beta}
\end{equation}
Note that \( 1\oo\left(s_{\alpha\beta}\right) = -2s_{\alpha\beta} \). Hence
we can rewrite this equation in the form
\begin{equation}
\left( M\oo - \nfrac{1}{2} M. 1\oo \right)s_{\alpha\beta} = 0 
\end{equation}

This means the symplectic form \( s_{\alpha\beta} \) is invariant under
the action of \( M\oo - \nfrac{1}{2}M.1\oo \) and so this operator acts 
as an element of the Lie algebra \( \sp(2,\R) \) under local action. 
Hence there are coefficients \( M^k \) so that
\begin{equation}
\label{lemmaresult}
M^\alpha_\beta - \nfrac{1}{2}M.1^\alpha_\beta = M^kT^\alpha_{k\beta}
\end{equation}
which gives the desired equation. 

We already have \( M = M^\bullet_\bullet \).  Contracting 
\( \alpha \) and \( \beta \) in equation~\ref{lemmaresult} gives
us \( M = \nfrac{1}{2}M^\alpha_\alpha \) which completes the proof.
\end{proof}

Now we consider the defining equation for torsion and curvature
\begin{equation}
\label{defining_equation}
\left[\nabla_i,\nabla_j\right] = T^k_{ij}\nabla_k + R_{ij}\oo 
\end{equation}

and in particular consider what happens when we apply this to 
the globally invariant tensor \( s^\bullet_{\alpha\beta} \). Since 
the covariant derivative acts trivially on \( s^\bullet_{\alpha\beta} \), 
the terms containing a covariant derivative all vanish leaving the 
equation
\begin{equation}
R_{ij}\oo\left(s^\bullet_{\alpha\beta}\right) = 0
\end{equation}

It follows from Lemma~\ref{the previous lemma} that we can write 
\begin{equation}
\label{curvature_reduced}
R^\theta_{ij\sigma} = 
R^t_{ij}T^\theta_{t\sigma} + F_{ij}1^\theta_\sigma
\end{equation}

Where \( F_{ij} = \nfrac{1}{2}R^\bullet_{ij\bullet}
		 = \nfrac{1}{4}R^\alpha_{ij\alpha} \) and hence
\begin{equation}
\label{curvature_bullet}
R^\bullet_{ij\bullet} = 2F_{ij}
\end{equation}

We will call the tensor \( R^t_{ij} \) the \textbf{reduced curvature tensor}.
The tensor \( F_{ij} \) giving the scalar component will be called 
the \textbf{field tensor}.  Our notation anticipates later results where we
will identify this as the electromagnetic field tensor, but we are not quite 
in a position to reach that conclusion at this stage. 

We also conclude that there is no versor component of curvature.

\medskip

The decomposition above was for the curvature on spinors. But we are also
interested in similar decompositions of the curvature \( R_{ij}\oo \) on
other types of objects. We begin by seeking a decomposition of the curvature
on vectors, that is of the ordinary Riemannian tensor.

We begin by applying the defining equation, equation~\ref{defining_equation} 
to the globally invariant tensor \( T^\alpha_{l\beta} \). Once again only
the curvature term remains giving us
\begin{equation}
R_{ij}\oo\left(T^\alpha_{k\beta}\right) = 0 
\end{equation}

This gives
\begin{equation}
R^l_{ijk} T^\alpha_{l\beta} = 
R^\alpha_{ij\lambda}T^\lambda_{k\beta} -
R^\lambda_{ij\beta}T^\alpha_{k\lambda}
\end{equation}

Contracting this equation with \( T^\beta_{m\alpha} \) gives
\begin{flalign}
R^l_{ijk} g_{lm}
&= R^\alpha_{ij\lambda}T^\lambda_{k\beta} T^\beta_{m\alpha} -
   R^\lambda_{ij\beta}T^\alpha_{k\lambda} T^\beta_{m\alpha} \\
&= R^\theta _{ij\sigma}
   \bigl( T^\sigma_{k\lambda} T^\lambda_{m\theta} 
        - T^\sigma_{m\lambda} T^\lambda_{k\theta}
   \bigr) \\
&= R^\theta _{ij\sigma}T^l_{km}T^\sigma_{l\theta} \\
&= \bigl( R^t_{ij}T^\theta_{t\sigma} + F_{ij}1^\theta_\sigma \bigr)
T^l_{km}T^\sigma_{l\theta} \\
&= R^t_{ij}g_{tl} T^l_{km} \\
&= R^t_{ij} T^l_{tk} g_{lm}
\end{flalign}

It follows that
\begin{equation}
\label{reduceddecomposition}
R^l_{ijk} = R^t_{ij} T^l_{tk}
\end{equation}

We see that the reduced curvature tensor is totally responsible for 
the curvature of the manifold. The field tensor plays no role here. 

\bigskip

We next look at the decomposition of the versor curvature.

\medskip

The tensors \( T^\alpha_{A\beta} \) are totally invariant. Applying 
equation~\ref{defining_equation} to this leaves only the curvature term 
once again
\begin{equation}
R_{ij}\oo\left( T^\alpha_{A\beta} \right) = 0 
\end{equation}

which simplifies to the equation
\begin{equation}
R^B_{ijA} T^\alpha_{B\beta} = 
R^\alpha_{ij\lambda}T^\lambda_{A\beta} -
R^\lambda_{ij\beta}T^\alpha_{A\lambda}
\end{equation}

Applying equation~\ref{curvature_reduced} we obtain
\begin{flalign}
R^B_{ijA} T^\alpha_{B\beta} &= 
\left(R^t_{ij}T^\alpha_{t\lambda} + F_{ij}1^\alpha_\lambda\right)
T^\lambda_{A\beta} -
\left(R^t_{ij}T^\lambda_{t\beta} + F_{ij}1^\lambda_\beta\right)
T^\alpha_{A\lambda} \\
&=
R^t_{ij}\left( T^\alpha_{t\lambda} T^\lambda_{A\beta} - T^\lambda_{t\beta} T^\alpha_{A\lambda} \right)
+ F_{ij}\left( 1^\alpha_{\lambda} T^\lambda_{A\beta} - 1^\lambda_{\beta} T^\alpha_{A\lambda} \right) \\
&= 
R^t_{ij}T^B_{tA}T^\alpha_{B\beta}
\end{flalign}

and hence, as the operators \( T^\alpha_{B\beta} \) are linearly independent we must have
\begin{equation}
\label{reducedversorcurvature}
R^B_{ijA} = R^t_{ij}T^B_{tA}
\end{equation}

\medskip

Decompositions for the curvature action on other types of generalised 
tensors can be obtained by similar arguments.

\bigskip
\pagebreak[2]

\section{Bianchi Identities Revisited}

Bianchi identities arise directly from the Jacobi identity for the action of 
the covariant derivative. We now understand that curvature arises from the
commutator of the covariant derivative and can act on things other than tensors.
In particular we have a description of the action of curvature on Crump scalars,
spinors and versors. We also can decompose curvature and express it in terms
of curvature components. In this section we revisit the Bianchi identities using these new ideas.
\medskip

Applying the Jacobi identity to the covariant derivative action on vectors
gives the two ordinary Bianchi identities we obtained earlier. However we can 
now express these as identities on the reduced curvature tensor.
\begin{equation}
\label{1Breduced}
R^m_{ij}T^l_{mk} \ijkequal 0 
\end{equation}
\begin{equation}
\label{2Breduced}
R^l_{im}T^m_{jk} + \nabla_i(R^l_{jk}) \ijkequal 0 
\end{equation}

\medskip

Applying the Jacobi identity to the action on Crump scalars gives
equation~\ref{1Breduced} again as the first Bianchi identity, and a new
second Bianchi identity for the field tensor
\begin{equation}
\label{2Bfield}
F_{im}T^m_{jk} + \nabla_i(F_{jk}) \ijkequal 0
\end{equation}

\medskip

Applying the Jacobi identity to the covariant derivative action on 
spinors and separating into components we obtain equation~\ref{1Breduced}
again for the first Bianchi identity. And both equations~\ref{2Breduced}
and~\ref{2Bfield} as components for the second Bianchi identity.

\medskip

Applying the Jacobi identity to the covariant derivative action on 
versors and separating into components also gives us nothing new. 
We obtain equation~\ref{1Breduced} for the first Bianchi identity and
equation~\ref{2Breduced} as the second Bianchi identity.

\medskip

We can also further expand the second Bianchi identities by expressing 
the covariant derivative in terms of the partial derivative and connection. 
The second Bianchi identity for the Field tensor is particularly interesting
in this form, giving
\begin{flalign*}
&& 
F_{im}T^m_{jk} + \nabla_i F_{jk} &\ijkequal 0 && && \\
\Rightarrow && 
F_{im}T^m_{jk} + \del_i F_{jk}
- \Gamma^m_{ij}F_{mk} - \Gamma^m_{ik}F_{jm} &\ijkequal 0 && && \\
\Rightarrow && 
F_{im}T^m_{jk} + \del_i F_{jk}
+ \Gamma^m_{jk}F_{im} - \Gamma^m_{kj}F_{im} &\ijkequal 0 && &&\\
\Rightarrow && 
F_{im}T^m_{jk} + \del_i(F_{jk})
- (- \Gamma^m_{jk} + \Gamma^m_{kj}) F_{im} &\ijkequal 0 && &&
\end{flalign*}
 
The term in brackets is the torsion allowing us to cancel and 
finally obtain the equation
\begin{equation}\label{FaradayGauss}
\del_i(F_{jk}) \ijkequal 0
\end{equation}

In this form we easily recognise the \textbf{Faraday-Gauss equations}, 
(actually a 10-D extension) the source free part of Maxwell equations. 
In our model these are purely geometric identities.

\bigskip

The second Bianchi identity for the reduced curvature tensor can also
be simplified and expressed in terms of partial derivatives. 
\begin{flalign*}
&& 
R^l_{im}T^m_{jk} + \nabla_i(R^l_{jk}) 
&\ijkequal 0  && && \\ \Rightarrow && 
R^l_{im}T^m_{jk} + \del_i(R^l_{jk})
+ \Gamma^l_{im}R^m_{jk} 
- \Gamma^m_{ij}R^l_{mk} 
- \Gamma^m_{ik}R^l_{jm} 
&\ijkequal 0  && && \\ \Rightarrow && 
R^l_{im}T^m_{jk} + \del_i(R^l_{jk})
+ \Gamma^l_{im}R^m_{jk} 
+ \Gamma^m_{jk}R^l_{im} 
- \Gamma^m_{kj}R^l_{im} 
&\ijkequal 0  && && \\ \Rightarrow && 
\del_i(R^l_{jk}) + \Gamma^l_{im}R^m_{jk} 
&\ijkequal 0 && &&
\end{flalign*}

The first Bianchi identity and the definition of curvature can be used to 
rewrite it in the form.
\begin{equation}
\del_i(R^l_{jk}) + R^m_{ij} \Gamma^l_{mk} \ijkequal 0
\end{equation}

\bigskip
\pagebreak[2]

\section{Connection Components}
\label{sectionCC}

The connection \( \Gamma^\beta_{k\alpha} \) is not a tensor. However the reason 
why it fails to be a tensor lies in the behaviour under transformation of the
index \( k \).  For fixed \( k \) each \( \Gamma^\beta_{k\alpha} \) defines
a linear transformation of spinors at each point. This can be decomposed into 
vector, versor and scalar components allowing us to write
\begin{equation}
\Gamma^\beta_{k\alpha} = 
	A_k 1^\beta_\alpha + G^i_kT^\beta_{i\alpha} + N^A_k T^\beta_{A\alpha}
\end{equation}

The three quantities \( A_k \), \( G^i_k \) and \( N^A_k \) are called the 
scalar, vector and versor \textbf{connection components} respectively. 
These do not transform as tensors. They should be viewed as facets of the 
connection.

Contracting with \( 1^\alpha_\beta \), \( T^\alpha_{j\beta} \) and 
\( T^\alpha_{B\beta} \) we obtain explicit expressions for the connection 
components as follows.
\begin{flalign}
\label{scalarconnectiondefined}
A_k &= \nfrac{1}{4}\Gamma^\alpha_{k\alpha} \\
\label{fiveconnectiondefined}
N_k^A &= \Gamma^\beta_{k\alpha}T^\alpha_{B\beta}g^{BA} \\
\label{tenconnectiondefined}
G_k^j &= \Gamma^\beta_{k\alpha}T^\alpha_{m\beta}g^{mj}
\end{flalign}

To determine the transformation properties for connection components
choose different local bases for all quantities involved.  We will denote 
that a new basis has been used by attaching a prime to its index. At every 
point on the manifold we can define change of basis matrices \( \delta^{i'}_i \), 
\( \delta^{A'}_A \) and \( \delta^{\alpha'}_\alpha \); and their inverses 
\( \delta^i_{i'} \), \( \delta^A_{A'} \) and \( \delta^\alpha_{\alpha'} \).

Then the spinor connection transforms according to the equation
\begin{equation}
\label{connectionrebased}
\Gamma^{\beta'}_{k'\alpha'} 
	= \delta^{\beta'}_\beta \delta^\alpha_{\alpha'} \delta^k_{k'} 
		\Gamma^\beta_{k\alpha} 
	- \delta^k_{k'} \delta^\lambda_{\alpha'}
\,\del_k (\delta^{\beta'}_\lambda)
\end{equation}

and this allows us to determine the transformation properties of the 
connection components \( A_k \), \( N_k^A \) and \( G_k^j \). We obtain
\begin{flalign}
\label{basischangeA}
A_{k'} &= \delta^k_{k'}A_k - 
	\nfrac{1}{4}\delta^k_{k'} \delta^\alpha_{\alpha'}
	            \,\del_k(\delta^{\alpha'}_\alpha) \\
\label{basischangeG}
G_{k'}^{j'} &= \delta^k_{k'}\delta^{j'}_j G_k^j - 
	\delta^k_{k'} \delta^\mu_{\alpha'}
                     \,\del_k(\delta^{\alpha'}_\lambda) 
                     T^\lambda_{t\mu}g^{tj}\delta^{j'}_j \\
\label{basischangeN}
N_{k'}^{A'} &= \delta^k_{k'}\delta^{A'}_A N_k^A - 
	\delta^k_{k'} \delta^\mu_{\alpha'}
                      \,\del_k(\delta^{\alpha'}_\lambda) 
                      T^\lambda_{B\mu}g^{AB}\delta^{A'}_A 
\end{flalign}

Note that if the only change is to the coordinate system of the manifold and
in particular if there is no change in the basis of spinor space then the last
terms disappear and all three transform as if they were tensors. We can 
therefore regard them as tensors with a hidden dependence on the choice of 
spinor basis. We could think of this as a choice of gauge.

\bigskip
\pagebreak[2]

\section{The Field and Potential}
\label{sectionFP}

We defined the Field tensor \( F_{ij} \) to be the scalar component of the 
spinor curvature.
\begin{equation*}
\label{curvdel}
R^\alpha_{ij\beta}
=
\Bigl[ \del_i (\Gamma^\alpha_{j\beta}) - 
       \del_j (\Gamma^\alpha_{i\beta})  \Bigr] +
\Bigl[ \Gamma^\alpha_{i\lambda}\Gamma^\lambda_{j\beta} - 
       \Gamma^\alpha_{j\lambda}\Gamma^\lambda_{i\beta} \Bigr]
\end{equation*}

We can obtain an expression for \( F_{ij} \) in terms of connection components 
by contracting this equation with \( 1^\beta_\alpha \). We obtain
\begin{flalign}
F_{ij} &= \del_i (\Gamma^\alpha_{j\beta})1^\beta_\alpha - 
       \del_j (\Gamma^\alpha_{i\beta})1^\beta_\alpha  \\
       &= \del_i (\Gamma^\alpha_{j\alpha}) -
       \del_j (\Gamma^\alpha_{i\alpha}) \\
       &= \del_i A_j - \del_j A_i 
\label{fieldpotential}
\end{flalign}

Note that we were able to move the identities \( 1^\beta_\alpha \) inside
the partial derivatives because \( \del_i(1^\beta_\alpha) = 0 \).

Equation~\ref{fieldpotential} gives the expected relationship between
the electromagnetic field and the electromagnetic potential. This adds weight 
to the notion that we should interpret the components of \( A_i \) and of 
\( F_{ij} \) indexed by the four translation dimensions as being the 
electromagnetic potential and electromagnetic field tensor respectively. 

The components of \( A_i \) and \( F_{ij} \) indexed by the Lorentz dimensions
must also describe forces and we will need to identify the nature of these 
at some point. There are several possible explanations. The most likely 
explanation is that they represent adjustments to the electromagnetic forces
to account for things like the flux of spin. However a more intriguing
possibility is that they could describe additional forces such as the weak 
force. If so then demonstrating this will not be an easy matter. 

\medskip

We would also like to find a similar equation linking \( R^k_{ij} \) to 
curvature components. This proves to be a little more complicated. Before
we do this we need to take a deeper look at connection components, 
specifically the components of connections other than the spinor connection.

\bigskip
\pagebreak[2]

\section{Connections, Curvature and Gauge}

At the end of section~\ref{sectionCC} we commented that the hidden dependence
of connection components on the choice of spinor basis could be regarded as
a choice of gauge. Now that we have discovered the equation relating the 
field \(F_{ij} \) and the potential \( A_k \), we are able to clarify that 
remark.

\medskip

In classical electromagnetism the gauge describes the extent to which
the potential $A_k$ can vary without changing the field $F_{ij}$. 
Consider for example any scalar function $f$ on the manifold and define a 
new potential ${A'}_k = A_k + \del_k(f)$. Applying equation~\ref{fieldpotential}
\begin{flalign}
{F'}_{ij} &= \del_i({A'}_j) - \del_j({A'}_i)  \\
          &= F_{ij} + \del_i\del_j(f) - \del_j\del_i(f) \\
          &= F_{ij}
\end{flalign}

we observe that this change in the potential has no effect on the field.
This describes the classical notion of gauge. However we can also view 
gauge in terms of the effect of a change of basis.

\medskip

In section~\ref{sectionCC} we noted that 
under a change of basis \( A_k \) transforms as
\begin{equation*}
A_{k'} = \delta^k_{k'}A_k - 
	\nfrac{1}{4}\delta^k_{k'} \delta^\alpha_{\alpha'}
	            \del_k(\delta^{\alpha'}_\alpha) 
\end{equation*}

and hence it has a hidden dependence on the choice of spinor basis. To 
see the effects of this hidden dependence let us assume that only the spinor 
basis is changed. This gives the equation
\begin{equation}
{A'}_{k} = A_k - \nfrac{1}{4} \delta^\alpha_{\alpha'}
	            \del_k(\delta^{\alpha'}_\alpha) 
\end{equation}

Now \( F_{ij} \) is a tensor with no dependence on the spinor basis, hence 
it will be left unchanged by this change in \( A_k \).

The classical abelian gauge is a special case which arises by choosing a 
basis change of the form 
\( \delta^{\alpha'}_\alpha = -4f.1^{\alpha'}_\alpha \) where 
\( f \) is a non-zero scalar function on the manifold. This simplifies to give
${A'}_k = A_k + \del_k(f)$ exactly as before.

\medskip
The other connection components $G^i_k$ and $N^A_k$ also have a hidden
dependence on the choice of spinor basis and we can also regard these as 
gauge symmetries as well. However these gauges for \( A_k \), for \( G^i_k \) 
and for \( N^A_k \) are all related as they simply express the consequence 
of changing the spinor basis. What we don't see here is three independent 
gauge groups each giving rise to a different force. 

In the standard model gauge groups are the key to classifying and 
understanding the different forces. In our approach the gauge groups 
do not seem to serve this function. We distinguish forces by extracting 
components of the curvature and the connection instead.

\bigskip
\pagebreak[2]

\section{Vector, Versor and Crump Connections}

The indices in the Crump connection \( \Gamma^\bullet_{i\bullet} \) 
auto-contract to give a vector-like quantity which we denote \( H_i \).
The Crump curvature obtained from the Crump connection also simplifies 
by auto-contraction to give
\begin{equation}
R^\bullet_{ij\bullet} = \del_i(H_j) - \del_j(H_i)
\end{equation}

In equation~\ref{curvature_bullet} we showed that
\( R^\bullet_{ij\bullet} = 2F_{ij} \). Hence 
\begin{equation} 
F_{ij} = \del_i\left(\nfrac{1}{2}H_j\right) 
	- \del_j\left(\nfrac{1}{2}H_i\right)
\end{equation} 

However equation~\ref{fieldpotential} states 
\begin{equation*}
F_{ij} = \del_i(A_j) - \del_j(A_i) 
\end{equation*}

this suggests a close relationship between \( A_i \) and \( H_i \). 
In particular we are tempted to conjecture that perhaps \( H_i = 2A_i \). 
Before we are tempted to proceed down this path however we need to consider
their transformation properties since neither is a tensor.

In new bases for vectors, spinors and Crump scalars we have
\begin{equation}
H_{i'} = \delta^i_{i'}H_i 
	- \del_{i'}\left(\delta^{\bullet'}_\bullet\right)
		\delta^\bullet_{\bullet'}    
\end{equation}

So \( H_i \) transforms as a vector, but has a hidden (gauge) dependence on 
the choice of basis \emph{for Crump scalars}. The basis changing behavior
of \( A_k \) was described earlier in equation~\ref{basischangeA}. We have
\begin{equation*}
A_{i'} = \delta^i_{i'}A_i - 
	\nfrac{1}{4} \del_{i'}(\delta^{\alpha'}_\alpha)\delta^\alpha_{\alpha'}
\end{equation*}

so that \( 2A_i \) is a vector with a hidden (gauge) dependence on the choice
of basis \emph{for spinors}. 

Because \( H_i \) and \( 2A_i \) transform differently they cannot be equal.  
They are related however, and describing this relationship is our next task.

\medskip

The equation \( \nabla_k(s^\bullet_{\alpha\beta}) = 0 \) gives
\begin{equation}
\del_k(s^\bullet_{\alpha\beta}) + H_k s^\bullet_{\alpha\beta}
- \Gamma^\lambda_{k\alpha} s^\bullet_{\lambda\beta}
- \Gamma^\lambda_{k\beta} s^\bullet_{\alpha\lambda}
\end{equation}

Expanding this out in terms of connection components we obtain
\begin{flalign}
\del_k(s^\bullet_{\alpha\beta}) + H_k s^\bullet_{\alpha\beta}
	&= ( A_k 1^\lambda_\alpha s^\bullet_{\lambda\beta}
		+ G_k^m T^\lambda_{m\alpha} s^\bullet_{\lambda\beta}
		+ N_k^A T^\lambda_{A\alpha} s^\bullet_{\lambda\beta} ) 
	\nonumber\\ 
	&\quad\quad + ( A_k1^\lambda_\beta s^\bullet_{\alpha\lambda}
		+ G_k^m T^\lambda_{m\beta} s^\bullet_{\alpha\lambda}
		+ N_k^A T^\lambda_{A\beta} s^\bullet_{\alpha\lambda} ) \\[3pt]
	&= 2A_k s^\bullet_{\alpha\beta}
	+ G_k^m ( T^\lambda_{m\alpha} s^\bullet_{\lambda\beta}
		   + T^\lambda_{m\beta} s^\bullet_{\alpha\lambda} )
	\nonumber \\ 
	&\quad\quad + N_k^A ( T^\lambda_{A\alpha} s^\bullet_{\lambda\beta} 
	 	   + T^\lambda_{A\beta} s^\bullet_{\alpha\lambda} ) \\[3pt]
	&= 2A_k s^\bullet_{\alpha\beta} 
		+ 2N_k^AT^\lambda_{A\alpha} s^\bullet_{\lambda\beta} 
\label{lastline}
\end{flalign}

Contracting both sides with \( s_\bullet^{\alpha\beta} \) to extract the scalar component we obtain
\begin{equation}
H_k = 2A_k - \nfrac{1}{4}\del_k(s^\bullet_{\alpha\beta})s_\bullet^{\alpha\beta}
\end{equation}

The last term can be regarded as a basis adjustment term which compensates
for the different hidden basis dependencies.

\medskip

The vector and versor components of equation~\ref{lastline}
give the identities
\begin{flalign}
\del_k(s^\bullet_{\alpha\lambda})s^{\lambda\beta}_\bullet T^\alpha_{m\beta}  
&= 0 \\
\del_k(s^\bullet_{\alpha\lambda})s^{\lambda\beta}_\bullet T^\alpha_{B\beta}  
&= 2N_k^Ag_{AB}
\end{flalign}

\bigskip

We now consider the versor connection \( \Gamma^B_{iA} \).  Since versors 
can be viewed as maps on spinors, we expect to be able to relate this to 
the spinor connection and its components.

The equation \( \nabla_i(T^\beta_{A\alpha}) = 0 \) gives
\begin{flalign}
\Gamma^B_{iA}T^\beta_{B\alpha} - \del_i(T^\beta_{A\alpha})
&= \Gamma^\beta_{i\lambda}T^\lambda_{A\alpha}
 - \Gamma^\lambda_{i\alpha}T^\beta_{A\lambda} \\[3pt]
&= ( A_iT^\beta_{A\alpha} + G^t_i T^\beta_{t\lambda}T^\lambda_{A\alpha}
+ N^B_i T^\beta_{B\lambda}T^\lambda_{A\alpha} )
	\nonumber \\ 
&\quad\quad - ( A_i T^\beta_{A\alpha} + 
		G^t_i T^\lambda_{t\alpha}T^\beta_{A\lambda}
		+ N^B_i T^\lambda_{B\alpha}T^\beta_{A\lambda} ) \\[3pt]
&= G^t_i T^B_{tA}T^\beta_{B\alpha} - N^B_iT^t_{AB}T^\beta_{t\alpha}
\end{flalign}

Contracting both sides with  \( T^\alpha_{C\beta}g^{CK} \) we
\begin{equation}
\Gamma^K_{iA} = G^t_i T^K_{tA}
+ \del_i(T^\beta_{A\alpha})T^\alpha_{C\beta}g^{CK}
\end{equation}

where the last term can be viewed as an adjustment term for the 
different behavior under a change of basis. To within such an adjustment
we see that the vector connection is also essentially determined by 
the vector component \( G^t_i \) of the spinor connection. 

\medskip

If we contract instead with the term \( T^\alpha_{m\beta}g^{mk} \) we obtain
\begin{equation}
\label{N_anotherbasisadjustment}
\del_i(T^\beta_{A\alpha}) T^\alpha_{m\beta}g^{mk} 
= N^B_iT^k_{AB}
\end{equation}
Yet another interesting equation to add to our collection of interesting
identities involving \( N^A_i \).

\bigskip

Finally we turn our attention to the vector connection \( \Gamma^k_{ij} \).
This describes curvature on the manifold and we expect it to describe gravity.
However as vectors can be viewed as maps on spinors, we expect the vector 
connection will also be closely related to the spinor connection and its 
components.

The equation \( \nabla_i(T^\beta_{j\alpha}) = 0 \) gives
\begin{flalign}
\Gamma^s_{ij}T^\beta_{s\alpha} - \del_i(T^\beta_{j\alpha})
&= \Gamma^\beta_{i\lambda}T^\lambda_{j\alpha}
    	- \Gamma^\lambda_{i\alpha}T^\beta_{j\lambda} \\[3pt]
&= ( A_iT^\beta_{j\alpha} + G^t_i T^\beta_{t\lambda}T^\lambda_{j\alpha}
    	+ N^A_i T^\beta_{A\lambda}T^\lambda_{j\alpha} )
	\nonumber \\ 
&\quad\quad - ( A_i T^\beta_{j\alpha} + 
	G^t_i T^\lambda_{t\alpha}T^\beta_{j\lambda}
	+ N^A_i T^\lambda_{A\alpha}T^\beta_{j\lambda} ) \\[3pt]
&= G^t_i T^s_{tj}T^\beta_{s\alpha} - N^A_iT^B_{jA}T^\beta_{B\alpha}
\end{flalign}

Contracting both sides with  \( T^\alpha_{m\beta}g^{mk} \) we obtain
\begin{equation}
\label{vectorconnectionfromG}
\Gamma^k_{ij} = G^t_i T^k_{tj}
+ \del_i(T^\beta_{j\alpha})T^\alpha_{m\beta}g^{mk}
\end{equation}
where the last term can be viewed as a basis adjustment term compensating
for the different hidden basis dependencies. To within such an adjustment
we see that the vector connection is essentially determined by 
the vector component \( G^t_i \) of the spinor connection. 

\medskip

If we contracted instead with the term \( T^\alpha_{C\beta}g^{CK} \) we
would obtain the equation 
\begin{equation}
\label{N_frombasisadjustment}
\del_i(T^\beta_{j\alpha}) T^\alpha_{C\beta}g^{CK} = N^A_iT^K_{jA}
\end{equation}

\bigskip

If we set \( j = k \) in equation~\ref{vectorconnectionfromG}
we obtain the interesting and useful equation
\begin{flalign}
\Gamma^k_{ik} 
&= \del_i(T^\beta_{k\alpha})T^\alpha_{m\beta}g^{mk} \\
&= \frac{1}{2} \del_i(g_{mk})g^{mk} 
\end{flalign}

we could also have obtained this equation directly.

\bigskip
\pagebreak[2]

\section{Curvature and Potentials}
\label{sectionCCC}

The curvature
is determined by the connection and so the curvature components should
be determined by connection components. 
We have already found one equation of this type, namely
equation~\ref{fieldpotential} which states 
\( F_{ij} = \del_i(A_j) - \del_j(A_i) \).
We seek a similar equation for the reduced curvature tensor.

Equation~\ref{fieldpotential} was obtained as the scalar component
of equation~\ref{curvdel}
\begin{equation*}
R^\alpha_{ij\beta}
=
\Bigl[ \del_i (\Gamma^\alpha_{j\beta}) - 
       \del_j (\Gamma^\alpha_{i\beta})  \Bigr]
+ \Bigl[ \Gamma^\alpha_{i\lambda}\Gamma^\lambda_{j\beta} - 
       \Gamma^\alpha_{j\lambda}\Gamma^\lambda_{i\beta} \Bigr]
\end{equation*}

One obvious approach might be to expand out every instance of the connection
in this equation in terms of connection components via the equation
\begin{equation*}
\Gamma^\beta_{k\alpha} = A_k 1^\beta_\alpha + 
G^i_kT^\beta_{i\alpha} + N^A_k T^\beta_{A\alpha}
\end{equation*}

Unfortunately this approach is complicated by the fact that 
partial derivatives \( \del_i \) and 
\( \del_j \) act non-trivially on the tensors \( T^\alpha_{k\beta} \) and 
\( T^\alpha_{B\beta} \) used to identify components. The resulting extra 
terms introduce further instances of the connection which must in turn be 
expanded in terms of connection components.  

\medskip

To avoid this problem it is helpful to restate equation~\ref{curvdel} using 
covariant derivatives and work with covariant derivatives as much as possible.
The connection \( \Gamma^\alpha_{j\beta} \) is not a generalised tensor, 
however we can define a formal covariant differentiation on it to be
\begin{equation}
\label{nabla_of_gamma}
\nabla_i\left(\Gamma^\beta_{j\alpha}\right) =
	\del_i\left(\Gamma^\beta_{j\alpha}\right)  
	+ \Gamma^\beta_{i\lambda}\Gamma^\lambda_{j\alpha} 
	- \Gamma^\lambda_{i\alpha}\Gamma^\beta_{j\lambda}
	- \Gamma^j_{ij}\Gamma^\beta_{k\alpha}
\end{equation}

We similarly define formal covariant differentiation on connection
components. This notation allows us to write equation~\ref{curvdel} 
in the form
\begin{equation}
\label{curvnabla}
R^\alpha_{ij\beta}
= \Bigl[ \nabla_i (\Gamma^\alpha_{j\beta}) - 
       \nabla_j (\Gamma^\alpha_{i\beta})  \Bigr]
- \Bigl[ \Gamma^\alpha_{i\lambda}\Gamma^\lambda_{j\beta} - 
       \Gamma^\alpha_{j\lambda}\Gamma^\lambda_{i\beta} \Bigr]
- T^t_{ij}\Gamma^\alpha_{t\beta}
\end{equation}

If we use this form of the equation to extract components the complications
discussed earlier will not arise.

\medskip
We already know what the scalar component should look like, but let's 
check our method by obtaining it again from equation~\ref{curvnabla}.
Contracting with \( 1^\beta_\alpha \) we obtain
\begin{equation}
F_{ij} = \nabla_i(A_j) - \nabla_j(A_i)  - T^k_{ij}A_k
\label{fieldpotentialnabla}
\end{equation}

Rewriting this using partial derivatives gives
equation~\ref{fieldpotential} as expected.

\medskip
Extracting the vector component by contracting equation~\ref{curvnabla} with 
\( T^\beta_{m\alpha}g^{mk} \) gives
\begin{equation}
\label{curvaturefrompotentials}
R^k_{ij} = 
\bigl[\nabla_i(G^k_j) - \nabla_j(G^k_i)\bigr]
- G^x_iG^y_jT^k_{xy} 
- N^A_iN^B_jT^k_{AB} 
-G^k_mT^m_{ij}
\end{equation}
 
This equation is in a form which is easy to state and use since everything
in it is well behaved. However while equation~\ref{curvaturefrompotentials} 
has many good features in some respects it is less than ideal. The covariant 
derivative in this equation involves a connection which hides additional
connection components. We might prefer an equation written in terms of 
partial derivatives where all instances of the connection in the equation 
appear explicitly.
\begin{flalign}
R^k_{ij} &= 
\bigl[\del_i(G^k_j) - \del_j(G^k_i)\bigr]
- G^x_iG^y_jT^k_{xy} 
- N^A_iN^B_jT^k_{AB}  \nonumber\\
&\hskip6cm + \bigl(\Gamma^k_{im}G^m_j - \Gamma^k_{jm}G^m_i\bigr)
\label{curvaturefrompotentialsdel}
\end{flalign}

however this equation involves the vector connection which also involves
connection components. We can relate this explicitly to connection components 
using equation~\ref{vectorconnectionfromG}, however we are then forced to 
introduce extra terms which were earlier described as basis adjustment terms. 
\begin{flalign}
R^k_{ij} &= 
\bigl[\del_i(G^k_j) - \del_j(G^k_i)\bigr]
+ G^x_iG^y_jT^k_{xy} 
- N^A_iN^B_jT^k_{AB}  \nonumber\\
&\hskip6cm + (G^m_j Q^k_{im} - G^m_i Q^k_{jm} )
\label{curvaturefrompotentialsdelQ}
\end{flalign}

Where \( Q^c_{ab} = \del_a(T^\beta_{b\alpha})T^\alpha_{s\beta}g^{sc} \) 
is a basis adjustment term. 

\medskip

In this section we have found three equations that express, in different ways,
the gravitational field in terms of potentials; 
equations~\ref{curvaturefrompotentials}, \ref{curvaturefrompotentialsdel} 
and \ref{curvaturefrompotentialsdelQ}. 
None are ideal. The relationship between field and potential is considerably 
more complicated for gravitation than it is for electromagnetism.

	\chapter{Einstein's Equation}
\label{Chapter: Einstein}

In this chapter we want to see whether Einstein's equation for gravity 
is meaningful on a framework. Einstein's equation is usually stated
in terms of the Ricci tensor, a contraction of the curvature tensor. So
first we will need to look at contractions of the curvature tensor. 

We begin however by looking at some interesting and very useful operator 
identities.

\bigskip
\pagebreak[2]

\section{Invariant Operators}

An inner product space is a vector space equipped with a generalisation of
the familiar dot product in \( \R^n \). In analogous fashion a Lie algebra
can be viewed as a vector space equipped with a generalisation of the familiar
cross product in \( \R^3 \). If viewed in this manner an invariant metric is
a generalised dot product compatible with the generalised cross product. 

A local Lie manifold has an invariant metric \( g_{ij} \) and a Lie algebra 
structure \( T^k_{ij} \). These can be interpreted as equipping the manifold
with a generalised cross product \( u \CROSS v = T^k_{ij}u^iv^j \) and a 
compatible dot product \( u \DOT v = g_{ij}u^iv^j \)  defined on tangent 
vectors.
By combining these with the contravariant derivative \( \nabla^k \), 
we can define generalisations of the operators Div Grad and Curl.

\medskip

The \textbf{generalised gradient operator} 
\begin{equation}
\Grad(X) = \nabla^k (X)
\end{equation} 
can be applied to any tensor \( X \). It is a 
tensor derivation of rank \( \rank{1,0}{0,0} \) and in 
particular it maps scalar fields to vector fields on our 
manifold.

\bigskip
The \textbf{Divergence} operator
\begin{equation}
\Div ( v ) =  g_{ij}\nabla^i(v^j) = \nabla_i(v^i)
\end{equation}
maps vector fields to scalar fields. 

A versor variant of the divergence can be defined by replacing 
\( g_{ij} \) with the related tensor \( g^A_{ij} \). 

\bigskip
The \textbf{Curl} operator 
\begin{equation}
\Curl(X) = \nabla^i\bigl(T_i\oo(X)\bigr)
\end{equation}
combines the gradient operator with the Lie structure 
(generalised cross product).
It maps tensors to tensors or the same type, 
and in particular gives a map from vectors to vectors.

\bigskip

In \( \R^3 \) the divergence of the gradient gives the Laplacian operator. 
This leads us to define a generalised Laplacian
\begin{equation}
\Triangle (X) = g_{ij}\nabla^i\nabla^j(X)
\end{equation}

A versor variant of the Laplacian can be defined by replacing \( g_{ij} \) 
with the related tensor \( g^A_{ij}  \) giving 
\begin{equation}
{\Triangle\!}^A (X) = g^A_{ij} \nabla^i\nabla^j (X)
\end{equation}

The operator 
\begin{equation}
\Square (X ) = g_{AB}\,{\Triangle\!}^A{\Triangle\!}^B (X) = 
	g_{AB}g^A_{ij}g^B_{kl}\nabla^i\nabla^j\nabla^k\nabla^l (X )
\end{equation}

is an invariant 4th order operator. The operators \( \Triangle \) and
\( \Square \) relate to the second and fourth degree Casimir 
operators respectively.

\bigskip
\pagebreak[2]

\section{Operator Identities}

In this section we examine identities for the divergence, gradient and curl
defined in the last section. Many of these identities are suggested by the 
notation and are complicated to express without it. 

\medskip

The identities we discover are not as simple as those 
for the ordinary divergence gradient and curl, as we are working on a space 
which has both torsion and curvature. Furthermore our Curl operator is more 
general than the others as it acts on tensor fields of all types, and this too 
will complicate our identities.
 
\begin{proposition}\label{CurlGrad}
If \( f \) is a scalar on the manifold then
\begin{equation}
\Curl(\Grad f) = -3\Grad f 
\end{equation}
\end{proposition}
\begin{proof}
\begin{flalign*}
\Curl(\Grad f) &= g^{is}T_{st}^k\nabla_i\bigl(g^{tj}\nabla_j(f)\bigr) \\
           &= g^{is}g^{jt}T_{st}^k\nabla_i\nabla_j(f) \\
           &= -g^{is}g^{jt}T_{st}^k\nabla_j\nabla_i(f)
\end{flalign*}
where the last follows by renaming \( i \leftrightarrow j \) and \( s \leftrightarrow t \).
Averaging these two expressions gives
\begin{flalign*}
\Curl(\Grad f ) &= \nfrac{1}{2}g^{is}g^{jt}T_{st}^k[\nabla_i,\nabla_j](f) \\
           &= \nfrac{1}{2}g^{is}g^{jt}T_{st}^kT_{ij}^x\nabla_x(f) \\
           &= \nfrac{1}{2}g^{ks}g^{jx}T_{st}^iT_{ij}^t\nabla_x(f)
\end{flalign*}
however \( T_{st}^iT_{ij}^t = -6g_{sj} \) from equation~\ref{invariant forms}
hence
\begin{flalign*}
\Curl(\Grad f ) &= -3\Grad f
\end{flalign*}
as claimed
\end{proof}

\begin{proposition} \label{DivCurla}
If \( v =  v^k \) is any vector field then
\begin{equation}
\Div(\Curl v) = -3\Div v + 6R_k v^k
\end{equation}
Where \( R_k = R^t_{kt} \)\footnote{We will later show that $R_k = 0 $.}.
\end{proposition}

\begin{proof}
\begin{flalign*}
\Div(\Curl v) &= g^{ij}T^k_{is}\nabla_k\nabla_j (v^s) \\
	&= -g^{ik}T^j_{is}\nabla_k\nabla_j (v^s) \\
	&= -g^{ij}T^k_{is}\nabla_j\nabla_k (v^s)
\end{flalign*}
where the last follows by renaming \( j \leftrightarrow k \). Averaging gives
\begin{flalign*}
\Div(\Curl v) &= \nfrac{1}{2}g^{ij}T^k_{is}[\nabla_k,\nabla_j] (v^s) \\
	&= \nfrac{1}{2}g^{ij}T^k_{is}T^t_{kj}\nabla_t(v^s) 
	+ \nfrac{1}{2}g^{ij}T^k_{is}R^t_{kj}T^s_{tm}v^m \\
	&= -\nfrac{1}{2}\bigl( 6 . 1^t_s\bigr)\nabla_t(v^s) 
	- \nfrac{1}{2}g^{ij}T^k_{is}\bigl( R^t_{jm}T^s_{tk} + R^t_{mk}T^s_{tj} \bigr)v^m \\
	&= -3 \nabla_k(v^k) 
	- \nfrac{1}{2}g^{ij} \bigl( T^k_{is}T^s_{tk} \bigr) R^t_{jm} v^m
	+ \nfrac{1}{2}\bigl( g^{ij}T^k_{is}T^s_{jt}\bigr) R^t_{mk} v^m \\
	&= -3 \nabla_k(v^k) 
	- \nfrac{1}{2}g^{ij}6g_{it}R^t_{jm} v^m
	+ \nfrac{1}{2} 6.1^k_t R^t_{mk} v^m \\
	&= -3 \Div v + 6R_kv^k
\end{flalign*}
using the first Bianchi and Casimir identities.
\end{proof}

The combination of the last two propositions is extremely powerful as, for any
scalar field \( f \) on our manifold, we can write
\begin{equation}
\Div\left(\Curl\left(\Grad (f)\right)\right) = 
\Div\left(-3\Grad(f)\right) = -3\Triangle(f)
\end{equation}

using proposition~\ref{CurlGrad}. But by proposition~\ref{DivCurla} we have
\begin{equation}
\Div\left(\Curl\left(\Grad (f)\right)\right) = -3\Triangle(f)) + 6R_k\nabla^k(f)
\end{equation}

We conclude that \( R_k\nabla^k(f) = 0  \) for all scalar functions 
\( f \) on our manifold. In particular we can choose \( f \) to be a coordinate 
function in the neighbourhood of any point, and it follows that all coordinates
of \( R_k \) are zero at every point. We have proved the useful identity 
\begin{equation}
\label{Riszero}
R^m_{km} = 0
\end{equation}

We can use this to simplify proposition~\ref{DivCurla}.

\begin{proposition} \label{DivCurl}
If \( v^k \) is any vector field then
\begin{equation}
\Div(\Curl v) = -3\Div v
\end{equation}
\end{proposition}

\bigskip
\pagebreak[2]

\section{Contractions of the Curvature Tensor}

In this section we look at various contractions of the reduced curvature 
tensor \( R^k_{ij} \) with a view to obtaining something like Einstein's 
equation. We have already seen that the obvious contraction 
\( R_k = R^t_{kt} = - R^t_{tk} \) is identically zero. However other 
contractions can be obtained by using the fully invariant tensors 
\( g_{ij} \), \( g^{ij} \) and \( T^k_{ij} \). 

 We define
\begin{flalign}
&\mbox{The \textbf{Ricci Tensor} }& R_{ij} &= R^x_{iy}T^y_{xj}
&\mbox{\hskip1cm}&\mbox{\hskip1cm} \\
&\mbox{The \textbf{Curvature Scalar} }& R &= R_{ij}g^{ij}
&\mbox{\hskip1cm}&\mbox{\hskip1cm}
\end{flalign}

Other contractions are also possible although most reduce to one of these. 
Note that the Ricci tensor is the usual Ricci tensor since
\begin{equation}
R^k_{ikj} = R^x_{ik}T^k_{xj} = R_{ij}
\end{equation}

\medskip

The curvature scalar can also be written as 
\begin{equation}
R = R^k_{ij}T^c_{ab} g^{ia}g^{jb}g_{kc} = R^k_{ij}T_k^{ij}
\end{equation}
where \( T_k^{ij} = T^c_{ab}g^{ia}g^{jb}g_{kc} \) is the Lie structure with
all indices raised/lowered.

The Bianchi identities give identities on these contractions. 

\begin{proposition}
\label{Contract_Bianchi}
Contractions of the Bianchi identities give
\begin{equation} \label{Ricci_symmetry}
R_{ij} = R_{ji} 
\end{equation}
\begin{equation} \label{DivR}    
\nabla_k(R^k_{ij}) = 0           
\end{equation}
\begin{equation}\label{PseudoBianchi}
6R^m_{ij}g_{mk} + R_{im}T^m_{jk} \ijkequal 0
\end{equation}
\begin{equation} \label{GradR}   
\nabla_k(R) = 2\nabla^t(R_{tk}) 
\end{equation}
\end{proposition}

\begin{proof} The first Bianchi identity gives
\begin{equation*} 
R_{ij} = R^a_{ib}T^b_{aj} 
       = - R^a_{bj}T^b_{ai} - R^a_{ji}T^b_{ab}
       = R^a_{jb}T^b_{ai}
       = R_{ji}
\end{equation*}
hence the Ricci tensor is symmetric as claimed in 
equation~\ref{Ricci_symmetry}.

\medskip
Contracting the second Bianchi identity we have
\begin{equation*} 
\nabla_k(R^k_{ij}) + \nabla_i(R^k_{jk}) + \nabla_j(R^k_{ki}) +
R^k_{km}T^m_{ij} +
R^k_{im}T^m_{jk} +
R^k_{jm}T^m_{ki}
 = 0           
\end{equation*}
This simplifies since \( R^k_{ik} = 0 \) to give
\begin{equation*} 
\nabla_k(R^k_{ij}) - R_{ij} + R_{ji} = 0 
\end{equation*}
and the last two terms cancel since the Ricci tensor is symmetric
proving equation~\ref{DivR}.

\medskip
To prove equation~\ref{PseudoBianchi} we start with the expression on the left
and simplify by expanding \( 6g_{mk} \) and \( R_{ij} \). This allows us to 
usefully apply the first Bianchi and Jacobi identities.
\begin{flalign*}
6R^m_{ij}g_{mk} + R_{im}T^m_{jk} &=
R^m_{ij}T^a_{mb}T^b_{ka} + R^a_{ib}T^b_{am}T^m_{jk} 
\\&=
\left(
R^m_{jb}T^a_{mi}T^b_{ak} + R^m_{bi}T^a_{mj}T^b_{ak}
\right) + \left(
R^a_{bi}T^b_{jm}T^m_{ka} + R^a_{bi}T^b_{km}T^m_{aj}
\right) 
\\&=
R^z_{jx}T^x_{ky}T^y_{iz} - R^z_{ix}T^x_{jy}T^y_{kz}
\end{flalign*}
Cycling the indices \( i \), \( j \) and \( k \) and adding 
gives equation~\ref{PseudoBianchi}.

\medskip
Finally we have 
\begin{flalign*}
\nabla_k(R) &= g^{ij}T^a_{bj}\nabla_k(R^b_{ia}) \\[3pt]
            &= -g^{ij}T^a_{bj}\nabla_i(R^b_{ak}) 
               -g^{ij}T^a_{bj}\nabla_a(R^b_{ki}) \\
	    &\quad\quad\quad\quad
               -g^{ij}T^a_{bj}R^b_{ks}T^s_{ia}
               -g^{ij}T^a_{bj}R^b_{is}T^s_{ak}
               -g^{ij}T^a_{bj}R^b_{as}T^s_{ki} \\[3pt]
            &= g^{ij}\nabla_i( R^b_{ka} T^a_{bj})
               +g^{aj}\nabla_a( R^b_{ki} T^i_{bj})  \\
	    &\quad\quad\quad\quad
               +g^{is}T^a_{bj}T^j_{ia} R^b_{ks} 
               -g^{ij}T^a_{bj}R^b_{is}T^s_{ak}
               -g^{aj}T^i_{bj}R^b_{as}T^s_{ik} \\[3pt]
            &= 2g^{ij}\nabla_i(R_{kj}) 
               +6g^{is}g_{bi} R^b_{ks} 
               -2g^{ij}T^a_{bj}R^b_{is}T^s_{ak} \\
            &= 2\nabla^j(R_{kj}) 
               +6 R^s_{ks} 
               +2g^{aj}T^i_{bj}R^b_{is}T^s_{ak} \\
            &= 2\nabla^j(R_{kj}) 
               -2g^{aj}R_{sj}T^s_{ak}
\end{flalign*}
However \( -2g^{aj}R_{sj}T^s_{ak} = 2g^{as}R_{js}T^j_{ak} = 
2g^{aj}R_{sj}T^s_{ak} \) by relabelling \( s \) and \( j \), hence
the last term is identically zero and we conclude
\begin{flalign*}
\nabla_k(R) &= 2\nabla^j(R_{kj})
\end{flalign*}
proving equation~\ref{GradR}.
\end{proof}

An immediate consequence of equation~\ref{GradR} is the following 
Corollary which states that the Einstein tensor, if defined on our
manifold in the obvious way, is divergence free.

\begin{corollary}
\begin{equation}
\label{equationCC}  \nabla^t \bigr(R_{tk} - \frac{1}{2}g_{tk}R\bigr) = 0 
\end{equation}
\end{corollary}

\bigskip
\pagebreak[2]

\section{Einstein's Equation}

Equation~\ref{equationCC} can be rewritten as
\begin{equation}
\label{EinsteinsEquation}
R_{ij} - \frac{1}{2}g_{ij}R = \Theta_{ij} 
\end{equation}

Where \( \Theta_{ij} \) is a divergence free symmetric tensor.

\medskip

If we choose to identify \( \Theta_{ij} \) so that the non-Lorentz
components are the stress energy tensor we would then have the 
Einstein equation, or rather a ten dimensional extension of it. 

\medskip
We must decide whether this identification is physically reasonable.
This is not something that can be determined mathematically. One must 
verify that the equation is consistent with observation. 

Einstein's equation represents the current paradigm for gravitation so 
is generally regarded as consistent with observation. However there are 
some issues of concern with regard to anomalous galactic rotation and the 
need for dark matter and energy to fit observations to the equation. 

Also of concern to us is the absence of the anticipated unified equation 
with electromagnetism.

\medskip

Before proceeding we note that our equation~\ref{EinsteinsEquation} differs 
from Einstein's equation in several respects. 

The most obvious difference is that our version of the equation is ten 
dimensional. The manner in which we decide to deal with the additional 
dimensions then becomes important. Note that our model is explicitly 
quantum mechanical - we built it in order to accommodate wave functions 
on a curved manifold. Those wave functions can be viewed as quantum 
mechanical trajectories. 

The Lorentz components are unimportant for waves which are constant across 
Lorentz dimensions. Classical point-like particles with well defined energy
and momentum are approximated by such wave functions. Hence
the Lorentz components can be ignored when passing to the classical 
(non-quantum) approximation. 

Another difference is that the curvature in Einstein's equation uses as its 
connection the Christoffel symbols, the unique \textbf{symmetric} connection 
conserving the metric. Our connection also conserves the metric but has torsion.

The covariant derivative \( \nabla_k +\nfrac{1}{2}T_k\oo \) is torsion free 
and conserves the metric. Hence the associated  connection 
\begin{equation}
\label{Christoffel}
C^k_{ij} = \Gamma^k_{ij} +\nfrac{1}{2}T^k_{ij} 
\end{equation}

is the Christoffel symbols for the metric \( g_{ij} \) as it conserves the 
metric and is torsion free.  If we use these Christoffel symbols to construct 
a curvature tensor we obtain a torsion free version of the Riemannian.
\begin{flalign}
\hat{R}_{ij}\oo 
&= \bigl[\nabla_i +\nfrac{1}{2}T_i\oo\, , \, \nabla_j +\nfrac{1}{2}T_j\oo ] \\
&= \bigl[\nabla_i,\nabla_j\bigr]
 + \bigl[ \nfrac{1}{2}T_i\oo , \nabla_j \bigr]
 + \bigl[ \nabla_i , \nfrac{1}{2}T_j\oo ]
 + \bigl[ \nfrac{1}{2}T_i\oo , \nfrac{1}{2}T_j\oo ] \bigr] \\
&= \bigl[\nabla_i,\nabla_j\bigr]
 - \nfrac{1}{2}T^k_{ij}\nabla_k
 - \nfrac{1}{2}T^k_{ij}\nabla_k
 + \nfrac{1}{4}\bigl[ T_i\oo , T_j\oo ] \bigr] \\
&= R_{ij}\oo + \nfrac{1}{4}T^k_{ij}T_k\oo
\end{flalign}

As expected there is no torsion term and we have 
\begin{equation}
\hat{R}^l_{ijk} = R^l_{ijk} + \nfrac{1}{4}T^m_{ij}T^l_{mk}
\end{equation}

where 
\( \hat{R}^l_{ijk} \) is the Riemannian obtained from the Christoffel symbols.
Contracting to obtain torsion free versions of the Ricci tensor and scalar
curvature we obtain
\begin{flalign}
\hat{R}_{ij} &= R_{ij} - \nfrac{3}{2}g_{ij} \\
\hat{R}      &= R - 15 \\
\end{flalign}

If we now rewrite equation~\ref{EinsteinsEquation} in terms of 
the torsion free Ricci tensor we obtain
\begin{equation}
\label{EinsteinsCosmo}
\hat{R}_{ij} - \frac{1}{2}g_{ij}\hat{R} -6g_{ij} = \Theta_{ij} 
\end{equation}

Which is Einstein's equation with a cosmological constant of \( -6 \).
Of course this is \( 6 \) in natural units. The usual units for the 
cosmological constant are \( m_{-2} \). Switching to those units we
obtain 
\( \Lambda = -\frac{6} {r^2c^2}\).

If equation~\ref{EinsteinsCosmo} is assumed then measurements of the 
cosmological constant will give information about \( r \). As of 
2015 the cosmological constant is estimated to be on the order of 
\( \left|\Lambda\right| \approx 10^{-52} m^2 \)\footnote{thanks to 
Prof. Malte Henkel for inviting me to update the estimate used in an earlier draft}. This would give 
\begin{equation}
\label{r_estimate}
r \approx 2.6 \times 10^{17} \text{ seconds } \simeq 8.2 \text{ billion years }
\end{equation}

as our theory has been based on the assumption that \( r \) is large
it is reassuring that our first estimate of \( r \) is indeed large.

Note that \( r \) is a structure constant for the geometry. It measures 
the scale at which torsion becomes significant. In particular although 
we assigned the catchy name 'radius of the universe' to this quantity,
we did not intend to relate it to the universe currently estimated at 
\( 13.772 \pm 0.059 \) billion years.  

\medskip
This estimate of $r$ depends on assuming that equation~\ref{EinsteinsCosmo} is 
the correct source equation for gravity. And that would seem on the face of 
it to be a very reasonable assumption to make. Einstein's equation after all 
is the current paradigm for gravitation, so we would expect to find that the
gravitational equation in the context of a framework should resemble it.
On these grounds the extended Einstein equation seems to be exactly what 
we are looking for. However in the next chapter we will discover another 
quite different looking equation which also seems to be a source equation 
for gravitation. If we believe that this alternative equation is correct then
the calculation of $r$ we have just performed is invalidated.

	\chapter{The Ampere-Gauss equation and Ussher's equation}

We begin this chapter by looking for an appropriate source equation for 
electromagnetism. Much of this work can be found in~\cite{Uthesis}. 

\medskip
Maxwell's equations in a relativistic context are 
given by two equations; a dynamical equation called the Faraday-Gauss
equation, and a source equation called the Ampere-Gauss equation.

In our new context a generalised 10-D version of the Faraday-Gauss 
equation~\ref{FaradayGauss} arises directly from the Bianchi identities, 
and hence is a geometric identity on our manifold. So half the work of
finding a suitable form of Maxwell's equations has already been done.

To find the Ampere-Gauss equation however we will need a different argument.
The standard Ampere-Gauss equation equates the divergence of the field 
tensor \( F_{ij} \) to a charge-current tensor \( J_k \). So we might 
guess that the form of the equation should be something like
\begin{equation}
\label{AGguess}
\nabla^iF_{ik} = J_k
\end{equation}

However this is really just a wild guess which we are not very confident about.
Should we be using the partial derivative here instead of the covariant one for
instance? We are not sure.

In the usual context the source term \( J_k \) is divergence free and this
hints at a method for finding the correct form of this equation. We
need merely look for a divergence free expression on the left hand side.
In fact this is precisely the way that we obtained the extended Einstein 
equation. 

We seek an identity of the form \( \nabla^k (\textsf{expression}_k) = 0 \), 
where the expression is some function of the field tensor. We can then write 
this identity in the form \( \textsf{expression}_k = J_k \) where \( J_k \) 
is divergence free. If we interpret \( J_k \) as a source term this will give
us a source equation.

This kind of argument will still require us to make a physical assumption, 
namely we must decide on the physical interpretation of \( J_k \). But having 
an argument of this type is at least a little bit more convincing than just 
guessing at the form of the equation. 

An alternative approach would be to use a Lagrangian method.  We will look at 
the Lagrangian approach in chapter~\ref{ChapterLagrangianMethods}.

\bigskip
\pagebreak[2]

\section{Ussher's Identity}
\label{section:Ussher}


In this section we seek an identity which we can use to construct a source
equation for the field tensor \( F_{ij} \). We this identity to construct 
the same kind of argument as was used in chapter~\ref{Chapter: Einstein} to 
obtain the generalised Einstein equation, equation~\ref{EinsteinsEquation}.
We begin with some lemmas.

\begin{lemma} 
\begin{equation}
g^{is}g^{jt}R_{st}\oo\left(F_{ij}\right) = 0 
\end{equation}
\end{lemma}

\begin{proof}
\begin{flalign*}
g^{is}g^{jt}R_{st}\oo\left(F_{ij}\right) &=
g^{is}g^{jt}R^x_{st}T_x\oo\left(F_{ij}\right) \\
&=
- g^{is}g^{jt}R^x_{st}T^a_{xi}F_{aj}
- g^{is}g^{jt}R^x_{st}T^a_{xj}F_{ia}  \\
&=
g^{ia}g^{jt}R^x_{st}T^s_{xi}F_{aj}
+ g^{is}g^{ja}R^x_{st}T^t_{xj}F_{ia}  \\
&=
R^{aj}F_{aj}
+ R^{ia}F_{ia} = 0 
\end{flalign*}

as each term consists of a contraction of an antisymmetric tensor with 
a symmetric one. 
\end{proof}

An analogous result also holds for the other component of curvature 
\( R^k_{ij} \) and this is the content of our next lemma.

\begin{lemma} 
\begin{equation}
g^{is}g^{jt}R_{st}\oo\left(R^k_{ij}\right) = 0 
\end{equation}
\end{lemma}

\begin{proof}
\begin{flalign*}
g^{is}g^{jt}R_{st}\oo\left(R^k_{ij}\right) &=
g^{is}g^{jt}R^x_{st}T_x\oo\left(R^k_{ij}\right) \\
&= 
g^{is}g^{jt}R^x_{st}
T^k_{xa}R^a_{ij}
- g^{is}g^{jt}R^x_{st}
T^a_{xi}R^k_{aj}
- g^{is}g^{jt}R^x_{st}
T^a_{xj}R^k_{ia}
\end{flalign*}
We now show that each of these three terms is zero and the result with
then follow.

\begin{description}
\item[Case 1.] Consider the term \( g^{is}g^{jt}R^x_{st}T^k_{xa}R^a_{ij} \). 

By renaming fully contracted indices \( s \leftrightarrow i \) , 
\( t \leftrightarrow j \) , and \( x \leftrightarrow a \) we can make
use of the antisymmetry of \( T \) to equate this term to its own negative.
Hence it is identically zero.

\item[Case 2.] Consider the term \( - g^{is}g^{jt}R^x_{st} T^a_{xi}R^k_{aj} \).

Applying the first Bianchi identity we obtain
\begin{flalign*}
- g^{is}g^{jt}R^x_{st} T^a_{xi}R^k_{aj} 
&=
  g^{is}g^{jt}R^x_{is} T^a_{xt}R^k_{aj} +
  g^{is}g^{jt}R^x_{ti} T^a_{xs}R^k_{aj} 
\\ &=
  g^{is}g^{jt}R^x_{is} T^a_{xt}R^k_{aj} -
  g^{as}g^{jt}R^x_{ti} T^i_{xs}R^k_{aj} 
\\ &=
  g^{is}g^{jt}R^x_{is} T^a_{xt}R^k_{aj} +
  g^{as}g^{jt}R_{ts}R^k_{aj} 
\end{flalign*}

We now observe that the first term is zero as it involves a contraction 
of the symmetric tensor \( g^{is} \) with the antisymmetric tensor 
\( R^x_{is} \); while the second term is zero as it includes a contraction 
of the symmetric tensor \( R^{aj} \) with the antisymmetric tensor 
\( R^k_{aj} \).

\item[Case 3.] Consider the term \( - g^{is}g^{jt}R^x_{st}T^a_{xj}R^k_{ia} \).

By renaming fully contracted indices \( i \leftrightarrow j \) and
\( s \leftrightarrow t \) this reduces to the previous case.
\end{description}

Hence all three terms are zero completing the proof.
\end{proof}

\begin{lemma}
\label{RofR}
\begin{equation}
g^{si}g^{yj} R_{st}\oo\bigl(R^\beta_{ij\alpha}\bigr) = 0 
\end{equation}
\end{lemma}

\begin{proof}
Simply write 
\( R^\beta_{ij\alpha} = F_{ij}1^\beta_\alpha + R^k_{ij}T^\beta_{k\alpha} \) 
and apply the tensor derivation \( R_{st}\oo = R^x_{st}T_x\oo \). The
result follows from the previous two lemmas and from the local invariance
of \( T^\beta_{k\alpha} \) and \( 1^\beta_\alpha \).
\end{proof}

\begin{theorem}[Ussher's Identity]
\label{UI}
\begin{equation}
\label{UI1}
\nabla^i \bigl( \nabla^jR^\beta_{ji\alpha} - 
	\nfrac{1}{2}g^{ts}T^r_{it}R^\beta_{rs\alpha}\bigr) = 0 
\end{equation}
hence in particular
\begin{flalign}
\label{UI2}
\nabla^i \bigl( \nabla^jR^k_{ji} - 
	\nfrac{1}{2}g^{ts}T^r_{it}R^k_{rs}\bigr) &= 0 \\
\label{UI3}
\nabla^i \bigl( \nabla^jF_{ji} - 
	\nfrac{1}{2}g^{ts}T^r_{it}F_{rs}\bigr) &= 0  
\end{flalign}
\end{theorem}

\begin{proof}
Consider the double divergence of the curvature tensor 
\( \nabla^i\nabla^j\bigl(R^\beta_{ij\alpha} \bigr) \).
By renaming \( i \leftrightarrow j \) and using the antisymmetry of \( R \)
we obtain
\begin{flalign}
2\nabla^i\nabla^j
\bigl(R^\beta_{ij\alpha} \bigr) &= 
	[\nabla^i,\nabla^j] \bigl(R^\beta_{ij\alpha} \bigr) \\
&= 	g^{is}g^{jt}T^x_{st}\nabla_x\bigl(R^\beta_{ij\alpha}\bigr)
+ 	g^{is}g^{jt}R_{st}\oo\bigl(R^\beta_{ij\alpha}\bigr) 
\end{flalign}
where the second term is zero by lemma~\ref{RofR}. 
Relabelling and rearranging we can write this in the form
\begin{equation*}
\nabla^i \bigl( \nabla^jR^\beta_{ji\alpha} - 
	\nfrac{1}{2}g^{ts}T^r_{it}R^\beta_{rs\alpha}\bigr) = 0 
\end{equation*}
thereby proving equation~\ref{UI1}. Equations~\ref{UI2} and~\ref{UI3} now 
follow by taking components.
\end{proof}

As an immediate consequence of Ussher's identity we can write 

\begin{corollary}
\begin{equation}
\label{FieldEquation}
\nabla^jR^\beta_{ji\alpha} - \nfrac{1}{2}g^{ts}T^r_{it}R^\beta_{rs\alpha}
= \Theta^\beta_{i\alpha}
\end{equation}

Where the divergence \( \nabla^i \Theta^\beta_{i\alpha} = 0 \).
\end{corollary}

Ussher's identity is a purely geometric identity. It must hold on any 
Lie manifold obeying the very weak assumptions of the model. We 
cannot choose to disregard it but we can choose how to interpret it. In 
particular we must decide whether it is reasonable to interpret 
\( \Theta^\beta_{i\alpha} \) as a source term determined from the 
distribution of matter.

\bigskip
\pagebreak[2]

\section{Extended Ampere-Gauss Equation}

Ussher's identity is the right kind of identity to give source equations.
In particular we hope to find the Ampere-Gauss equation by looking at the 
scalar component. So we need to start by separating this equation into 
component equations. We begin by writing \( \Theta^\beta_{i\alpha} \) 
in terms of its components.
\begin{equation}
\Theta^\beta_{i\alpha} = J_i1^\beta_\alpha + K^k_iT^\beta_{k\alpha}
\end{equation} 

Note that equation~\ref{FieldEquation} implies there is no versor 
component. The equation \( \nabla^i \Theta^\beta_{i\alpha} = 0 \) now
gives two component equations \( \nabla^i J_i = 0 \) and 
\( \nabla^i K^k_i = 0 \). 
\medskip

The scalar component of equation~\ref{FieldEquation} is
\begin{equation}
\label{ScalarFieldEquation}
\nabla^jF_{ji} - \nfrac{1}{2}g^{ts}T^r_{it}F_{rs} = J_i 
\end{equation}

which is in the right form to be a suitable generalisation of the Ampere-Gauss
equation and to act as a source equation for the field tensor \( F_{ij} \). 
We will call this the \textbf{extended Ampere-Gauss equation}.

Note that this equation seems to have an extra term when compared to our 
initial guess, equation~\ref{AGguess}. This extra term is not a problem 
for three reasons.

\begin{enumerate}
\item \( T^i_{kb} \) is zero when all three indices \( i \), \( k \) and 
\( b \) are non-Lorentz. Hence this extra term will have no effect on the 
behavior of the translation components which give the classical fields. 

\item All spacetime components of the extra term vanish in the Poincar\'e 
limit and are insignificant over distance and time scales less than \( r \).

\item If we rewrite the equation in terms of the torsion free 
covariant derivative  \( \hat{\nabla}_k = \nabla_k + \nfrac{1}{2}T_k\oo \)
it becomes simply
\begin{equation}
\label{extendedAGnotorsion}
g^{ij}\hat{\nabla}_iF_{jk} = J_k 
\end{equation}
So the extra term is simply a consequence of writing the equation using a 
covariant derivative with torsion.
\end{enumerate}

This achieves our objective of obtaining a plausible source equation for 
electromagnetism. However we obtained it as a component of a larger equation,
and we now have a leftover component which we must decide what to do with.
We need to take a very hard look at that component equation because if we 
accept the Ampere-Gauss equation it will be difficult to argue that we should 
not also accept this other component.

\bigskip
\pagebreak[2]

\section{Ussher's equation.}

The vector component of equation~\ref{FieldEquation} is 
\begin{equation}
\label{UsshersEquation}
\nabla^j R^k_{ji} - \nfrac{1}{2} T^j_{is} g^{st} R^k_{jt} = K^k_i
\end{equation}

which looks like a source equation for the gravitational field described 
given by the reduced curvature tensor \( R^k_{ij} \). That is unexpected 
and is very interesting. This equation is claiming to do the same job as
Einstein's equation; it can be viewed as a source equation for gravity. 
Is that reasonable? Which source equation should we be using?

\medskip
In our picture, the gravitational and electromagnetic fields are 
components of the overall field specified by the curvature tensor 
\( R^\alpha_{ij\beta} \). It is entirely reasonable therefore to expect
that the source equations for electromagnetism and gravity should be the 
components of a combined source equation for the curvature tensor. 

Equations~\ref{UsshersEquation} and~\ref{ScalarFieldEquation} are
indeed  components of equation~\ref{FieldEquation} exactly as 
we would expect.   But Einstein's equation~\ref{EinsteinsEquation} 
does not seem to unify with a source equation for electromagnetism in 
this way. 

\medskip
The form of equation~\ref{UsshersEquation} 
\begin{equation*}
\nabla^j R^k_{ji} - \nfrac{1}{2} T^j_{is} g^{st} R^k_{jt} = K^k_i
\end{equation*}

where the divergence of a field is related to the presence of a source, is 
actually typical of what we would expect the source equation for a field to 
look like. Einstein's equation is the odd one out. It doesn't fit this pattern.

\medskip
If we were to ignore history and imagine that we were meeting
these equations for the first time, we would probably regard 
equation~\ref{UsshersEquation} as a stronger candidate for a source 
equation for gravity than Einstein's equation for these reasons.

Of course we still do have Einstein's equation, \ref{EinsteinsEquation} 
\begin{equation*}
R_{ij} - \frac{1}{2}g_{ij}R = \Theta_{ij} 
\end{equation*}
as a geometric identity on our manifold. Indeed
equations~\ref{EinsteinsEquation}~and~\ref{UsshersEquation} 
are both geometric identities on our manifold.  It therefore isn't a matter 
of choosing one of these equations over the other equation. Rather the issue
is how we decide how to interpret the two divergence free quantities 
$\Theta_{ij}$ and $K^k_i$. The question is which of these quantities, 
if any, is to be interpreted as a source term determined by the 
distribution of matter. 

\medskip

The ultimate test is of course whether equation~\ref{UsshersEquation}
is consistent with the experimental tests of general relativity which have been 
carried out to date. This may not be as unlikely as one would think given that 
Ussher's equation and Einsteins equation look very different.  Einstein's 
equation is notoriously difficult to test. Most experimental verifications 
have involved looking at small fields in the absence of sources, so we need 
simply show that the two equations are compatible in those limited 
circumstances. 

One way to validate equation~\ref{UsshersEquation} would be to relate
it in some way to Einstein's equation, the existing paradigm for gravity. 
But if the two equations make slightly different predictions we won't 
succeed in linking them exactly via a mathematical identity. We would 
need to use some kind of approximation argument and these are often difficult
to construct.

Furthermore there are instances such as anomalous galactic rotation where 
Einstein's equation does not seem to do very well. Assorted dark quantities 
have been conjectured to explain the difference between prediction and 
observation in these instances. Our inability to directly observe these dark 
things or indeed to understand their nature is however embarrassing, and gives
us reason to suspect that Einstein's equations might not be the final answer 
to the question of gravitation, even in the weak field case. For this reason 
perhaps trying to link Ussher's equation directly to Einstein's equation is 
not the best approach.

It would instead be simpler and indeed more certain to seek direct validation 
from observation. That would require us to construct the equivalent of 
Schwarzschild and Kerr solutions for equation~\ref{UsshersEquation} which would
then allow direct validation from experimental data.
 
\medskip

We will not attempt to construct specific solutions, however we will 
comment briefly on the issues that arise in attempting to do so. We begin
by noting the similarity between the equations for gravity in the absence 
of sources
\begin{flalign*}
&R^m_{ia}T^a_{jk} + \nabla_iR^m_{jk} \ijkequal 0 \\
&\nabla^iR^m_{ik} - \nfrac{1}{2}R^m_{ij}g^{jt}T^i_{kt} = 0 
\end{flalign*}

and the equations for electromagnetism in the absence of sources
\begin{flalign*}
&F_{ia}T^a_{jk} + \nabla_iF_{jk} \ijkequal 0 \\
&\nabla^iF_{ik} - \nfrac{1}{2}F_{ij}g^{jt}T^i_{kt} = 0 
\end{flalign*}

For each fixed value of \( m \) in the first set of equations the resulting
tensor satisfies the second set of equations. Hence solutions to the 
gravitational equations consist of ten solutions to Maxwells equations, 
one for each value of \( m \). These are not independent but must be 
compatible with the first Bianchi identity.
\begin{equation*}
R^m_{ij}T^n_{mk} \ijkequal 0 
\end{equation*}

This suggests that we might be able to construct a Schwartzchild type 
solution for the gravitational equations from central charge solutions 
to Maxwell's equations.

However things are not as simple as they seem. The problem of course is 
that the covariant derivative in the above equations involves the connection 
which also defines the curvature. Hence the equations above are considerably
more complicated than they appear and finding specific solutions is likely to 
be difficult. This looks like a job for an applied mathematician.

\medskip

Finally of course there is the issue of how we deal with the extra six 
dimensions. How these are perceived will depend on the nature of the 
matter on which they act. We do not envisage a universe filled with 
pointlike particles whose trajectories are described by geodesics across 
all ten dimensions of the manifold. Instead matter will be described by a 
spinor wave function satisfying a wave equation.

As we will see in chapter~\ref{Chapter MoM}, electrons seem to be described 
by wave functions that are constant (at least locally) across the Lorentz 
dimensions.  For this reason we expect that electrons will be insensitive to 
what is happening there. We also expect protons to be similar to electrons in 
this respect.  Hence it seems likely that components of curvature involving 
the Lorentz dimensions will exert only a second order effect on ordinary
matter. To first order we we should be able to ignore the extra dimensions 
and just look at what is happening in the first four. 

\bigskip
\pagebreak[2]

\section{Further Identities}

The tensors in equation~\ref{UsshersEquation} satisfy a number of 
interesting algebraic constraints. These may also prove useful in 
our later work. To facilitate  exploring these identities 
we define 
\begin{equation} 
\label{Udefined}
U^k_m = \nabla^j R^k_{jm} -\nfrac{1}{2} T^j_{ms} g^{st} R^k_{jt}
\end{equation}

Hence Ussher's equation~\ref{UsshersEquation} takes the form 
\( U^k_m = K^k_m \) where the left hand side is a tensor related 
to curvature and constraining gravity as described in equation~\ref{Udefined};
while the right hand side is a tensor related to the distribution of energy 
and matter, and acting as a source. 

The tensor \( U^k_m \) satisfies a number of algebraic identities which must
therefore also constrain \( K^k_m \).  

\begin{lemma}
\label{Ulemmas}
The tensor \( U^k_m \) satisfies
\begin{flalign}
&\nabla^mU^k_m = 0 
\label{Uidentity1}
\\
&U^k_k = \nfrac{1}{2}R 
\label{Uidentity2}
\\
& T^m_{ik} U^k_m = - \nabla^jR_{ji}
\label{Uidentity3}
\\
&T^m_{ik} U^k_m = - \nfrac{1}{2} \nabla_i R 
\label{Uidentity4}
\\
&T^m_{ik} U^k_m = - \nabla_i U^k_k
\label{Uidentity5}
\\
&\nabla_kU^k_m + T^a_{mb}U^b_a = g^{ij}R^l_{ik}T^s_{lm}R^k_{js}
\label{Uidentity6}
\end{flalign}
\end{lemma}
\begin{proof}
Equation~\ref{Uidentity1} is simply Ussher's identity~\ref{UI2}. 
\medskip

Contracting the two indices in equation~\ref{Udefined} gives the equation
\begin{equation*}
U^k_k = \nabla^j R_j - \nfrac{1}{2}g^{ts}T^r_{kt}R^k_{rs}
\end{equation*}

The curvature vector \( R_j \) is identically zero while the other 
term on the right hand side simplifies to give
equation~\ref{Uidentity2}
\medskip

Next contract equation~\ref{Udefined} with \( T^a_{kb} \) to give
\begin{equation*}
\nabla^jR^k_{ji}T^a_{kb} - 
\nfrac{1}{2}g^{ts}T^r_{it}\bigl(R^k_{rs}T^a_{kb}\bigr)
= U^k_iT^a_{kb}
\end{equation*}

Applying the 1st Bianchi identity to the term in brackets we obtain
\begin{equation*}
\nabla^jR^k_{ji}T^a_{kb}
+ \nfrac{1}{2}g^{ts}T^r_{it} R^k_{sb}T^a_{kr} 
+ \nfrac{1}{2}g^{ts}T^r_{it} R^k_{br}T^a_{ks}
= U^k_iT^a_{kb}
\end{equation*}

By swapping the contracted indices \( s \) and \( r \) in the last term on the
left, we see that this term is the same as the one before it giving the 
equation
\begin{equation*}
\nabla^jR^k_{ji}T^a_{kb}
+ g^{ts}T^r_{it} R^k_{sb}T^a_{kr} 
= U^k_iT^a_{kb}
\end{equation*}

contracting the indices \( a \) and \( i \) in this equation we thus obtain
\begin{equation*}
\nabla^jR_{jb}
= U^k_aT^a_{kb}
\end{equation*}

which rearranges to give~\ref{Uidentity3}.  Equation~\ref{equationCC} 
can now be applied to obtain equation~\ref{Uidentity4}

\medskip
Equation~\ref{Uidentity5} follows directly from 
equations~\ref{Uidentity2} and~\ref{Uidentity4}.

\medskip
Starting again from the definition we have
\begin{equation*}
\nabla_kU^k_m = g^{ij}\nabla_k\nabla_iR^k_{jm} 
			- \nfrac{1}{2}T^j_{ms}g^{st}\nabla_kR^k_{jt}
\end{equation*}

The last term vanishes as \(\nabla_kR^k_{jt} = 0 \) by equation~\ref{DivR}.
We would also like to exploit this fact in the first term which requires us 
to commute the two covariant derivatives.  We obtain
\begin{equation*}
\nabla_kU^k_m = g^{ij}[\nabla_k,\nabla_i]R^k_{jm} 
+ g^{ij}\nabla_i\nabla_kR^k_{jm} 
\end{equation*}
where the last term will vanish once again by equation~\ref{DivR}. Hence
\begin{flalign*}
\nabla_kU^k_m &= g^{ij}T^l_{ki}\nabla_l R^k_{jm} 
	+ g^{ij}R^l_{ki}
	\bigl(T^k_{ls}R^s_{jm} - T^s_{lj}R^k_{sm} - T^s_{lm}R^k_{js} \bigr)
\\ 
&= g^{il}\nabla_l (T^j_{ik} R^k_{jm}) 
	+ g^{ij}(R^l_{ki} T^k_{ls})R^s_{jm} 
	+ g^{sj}(R^l_{ki} T^i_{lj})R^k_{sm} 
	- g^{ij}R^l_{ki} T^s_{lm}R^k_{js} 
\\
&=  g^{il}\nabla_l R_{im}
	- g^{ij}R_{is}R^s_{jm} 
	+ g^{sj}R_{kj}R^k_{sm} 
	- g^{ij}R^l_{ki} T^s_{lm}R^k_{js} 
\\
&=  \nabla^i R_{im} + g^{ij}R^l_{ik} T^s_{lm} R^k_{js} 
\end{flalign*}
Applying \ref{Uidentity3} and rearranging we obtain \ref{Uidentity6}.
\end{proof}


	\chapter{Lagrangian Methods and Forces}
\label{ChapterLagrangianMethods}

In this chapter we use Lagrangian methods to find source equations 
for the forces. Since electromagnetism and gravity are simply components
of curvature we might expect to find a single source equation for curvature
that gives equations for electromagnetism and gravity as its components.
We begin by looking at the mathematical machinery required to apply a
Lagrangian method.

\pagebreak[2]

\section{Integration and Stokes' Theorem}

Lagrangian methods depend on Stokes theorem which relates the integral in 
an oriented simply connected\footnote{assuming simply connected will
simplify the discussion. We will only have need to integrate on regions
of this type} compact region \( \Omega \) to an integral on 
its boundary \( \del\Omega \).
\begin{equation}
\label{Stokes}
\integral_\Omega \nabla_k v^k \, dx^\circ = 
	\integral_{\del\Omega} v^k \, dx^\circ_k
\end{equation}

Stokes' theorem is a more general result, but this special case written in 
terms of the natural measure and integration on our framework is 
sufficient to our purpose.

\medskip
Integration is the measure weighted summation of values at different points. 
To add values at different points we must be able to compare them. To compare 
values at different points we need a notion of parallel transport, which 
we do have on a framework.

The integral on a framework is naturally defined by partitioning the 
region of integration into small pieces; computing a value for the integrand 
at each piece by combining the quantity to be integrated with the measure of 
the piece; and then parallel transporting these values to a single location 
to be added.  The integral is a limit obtained by refining the partition. 
We omit the technical details of the limiting process.

For this procedure to be well defined, parallel transport of values from 
one point to another must give the same answer regardless of path. 
Equivalently parallel transport of these values around a loop must be trivial. 
If the region of integration is simply connected it is sufficient that this 
be the case for small loops. Hence the associated curvature is zero and via the 
same arguments that were used in section~\ref{scalarsandzerocurvature} there 
must exist a basis in which the global action on these values is trivial and 
the associated connection is zero.

So this kind of natural integration is only well defined where the integrand
parallel transports in the same way as a scalar. Note that both integrands 
in equation~\ref{Stokes} are scalar.

\bigskip

A measure on the manifold is defined by specifying its value on a small 
parallelepiped. Consider the 10-D parallelepiped at a point \( p \) 
spanned by the ordered set of small vectors 
\( \{ dx_0^{i_0},\cdots, dx_9^{i_9} \} \) (indexing from 0 to 9
to avoid double digits). We define its measure to be
\begin{equation}
d{x_0}\wedge\cdots\wedge d{x_9} 
 = \eta_{i_0 \cdots i_9}\,\, d{x_0}^{i_0}\!\cdots d{x_9}^{i_9} 
\end{equation}

where \( \eta_{i_0 \cdots i_9} \) is a completely antisymmetric tensor 
with ten indices.  Different measures will arise by choosing different 
antisymmetric tensors.  One obvious choice is to pick 
\( \eta_{i_0 \cdots i_9} \) so that all non-zero entries are \( \pm 1 \). 
This has the advantage that it is easy to specify given a basis and that 
\( \del_k\eta_{i_0 \cdots i_9}  = 0 \). However this is basis dependent.

The completely antisymmetric tensors with ten indices are scalar and 
can be written as scalar functions via a change of basis. Conversely 
any scalar function gives a completely antisymmetric tensor with ten 
indices and will define a measure. If we start with a constant function 
then the measure we obtain will be invariant under parallel transport. 
Such a measure is called a Haar measure and will be denoted \( dx^\circ \).

Every framework is equipped with a Haar measure which is unique up to a 
choice of scale.  The scale is unimportant for Lagrangian calculations.
To find the Haar measure we use the fact that all completely antisymmetric
ten dimensional tensors differ by a scalar function. Hence we can write 
\begin{equation}
\label{explicit_measure}
dx^\circ
 = \epsilon\,\, \eta_{i_0 \cdots i_9}\,\, d{x_0}^{i_0}\!\cdots d{x_9}^{i_9} 
\end{equation}

where \( \eta_{i_0 \cdots i_9} \) has entries \( \pm 1 \) and where 
\( \epsilon \) is a scalar function chosen so that 
\( \eta^\circ = \epsilon\,\, \eta_{i_0 \cdots i_9} \) is constant.
Since  \( \eta^\circ \) is constant we must have 
\( \nabla_k\eta^\circ = 0  \). Since \( \del_k \eta_{i_0 \cdots i_9} = 0 \)
this simplifies to give
\begin{equation}
\del_k\epsilon = \epsilon\Gamma^a_{ka}
\label{determine_epsilon}
\end{equation}
which determines \( \epsilon \) uniquely to within a choice of scale.

\medskip

The Haar measure can also be defined from the metric. Choose oriented 
bases so that \( T^\alpha_{i\beta} \) look like the matrices in 
table~\ref{sp4R_basis} on page~\pageref{sp4R_basis}. Then 
the Haar measure of a parallelepiped spanned by small vectors \( dx_k \)
along the axes in this basis can be defined as the product of their lengths
computed from the metric. The invariance of the metric means this 
construction will be invariant. The \( \sqrt{g} \) which appears in some 
notations for the integral refers to this procedure.

\bigskip
We now turn to the right hand side of equation~\ref{Stokes} which involves
integration on the nine dimensional boundary of our region. A measure 
on a nine dimensional surface is defined as
\begin{equation}
d{x_1}\wedge\cdots\wedge d{x_9} 
 = \eta_{i_1 \cdots i_9}\,\, d{x_1}^{i_1}\!\cdots d{x_9}^{i_9} 
\end{equation}

where \( \eta_{i_1 \cdots i_9} \) is a completely antisymmetric tensor
with nine components. The completely antisymmetric tensors with nine components
are vectors. Indeed we can use a completely antisymmetric tensor with ten
components as an intertwining map
\begin{equation}
m^k \mapsto m^k\,\,\eta_{k\,\, i_1 \cdots i_9} = 
\eta_{i_1 \cdots i_9} 
\end{equation}

If we use the completely antisymmetric tensor chosen to construct the Haar
measure to do this then the measure given by \( \eta_{i_1\cdots i_9} \) will 
parallel transport as a vector. We denote this measure \( dx^\circ_k \) 
where the superscript reminds us that we used the scalar Haar measure to 
define it, and the subscript reminds us that it transforms as a vector.
 We can also write
\begin{equation}
dx^\circ_k
      = \epsilon\,\, 
	\eta_{k\,\, i_1 \cdots i_9}\,\, 
	d{x_1}^{i_1}\!\cdots d{x_9}^{i_9} 
\end{equation}
where \( \epsilon \) is determined by equation \ref{determine_epsilon}, 
and \( \eta_{i_0 \cdots i_9} \) is the completely antisymmetric tensor 
with values \( \pm 1 \).
\medskip

Integration is only well defined where the integrand parallel transports
as a scalar. Hence the vector measure on the boundary must combine with 
other quantities in the integral to form a scalar integrand. We observe
this in equation~\ref{Stokes}. The parallel transport in the integral 
can be around the boundary. However we will also get the same result by 
parallel transporting across the interior of the region.

\bigskip

Apart from the unusual notation which was chosen to relate to the natural 
structures on a framework, equation~\ref{Stokes} is simply the 
standard Stokes' theorem.  

\bigskip
\pagebreak[2]

\section{Lagrange Methods: Electromagnetism}
\label{section_EMLagrangian}

A Lagrangian approach uses a scalar function \( \LL \) known as the 
Lagrangian density. Obtaining this function is an art. The Lagrangian
density is a function of quantities on the manifold for which 
we are seeking some kind of dynamical equation. The Lagrangian 
hypothesis is that these quantities will arrange themselves so that 
the integral of the Lagrangian density is extremal under any small
variation on a compact region.

Consider for example the Lagrangian density
\begin{equation}
\LL_e = F_{ij}F^{ij} = g^{ia}g^{jb}F_{ij}F_{ab}
\end{equation}

which is the usual Lagrangian density for electromagnetism, albeit with a 
few extra dimensions. This is a function of the field tensor which in turn 
is a function of the potential \( A_i \). Equation~\ref{fieldpotentialnabla} 
on page~\pageref{fieldpotentialnabla} gives
\begin{equation*}
F_{ij} = \nabla_iA_j - \nabla_jA_i  - T^k_{ij}A_k
\end{equation*}

We prefer this rather than equation~\ref{fieldpotential} as we prefer
to work with the covariant derivative.

Assume now that \( A_i \) undergoes a variation 
\( A_i \mapsto A_i + \updelta A_i \), where \( \updelta A_i \) is 
small and is zero outside the interior of a compact region \( \Omega \). 
In particular \( \updelta A_i = 0 \) on the boundary \( \del\Omega \). 
Then other quantities that depend on \( A_i \) will also vary and 
in particular \( F_{ij} \) and the Lagrangian density \( \LL_e \) will also 
vary. Note that other quantities like \( g_{ij} \) which do not depend on 
\( A_i \) will not vary. To first order the variation in \( F_{ij} \) is 
\begin{equation}
\updelta F_{ij} 
= \nabla_i\updelta A_j - \nabla_j\updelta A_i  - T^k_{ij}\updelta A_k
\end{equation}

while the variation in \( \LL_e \) is 
\begin{equation}
\updelta \LL_e = 
g^{ia}g^{jb}\updelta F_{ij}F_{ab} + g^{ia}g^{jb}F_{ij}\updelta F_{ab}
=2g^{ia}g^{jb}\updelta F_{ij}F_{ab}
\end{equation}

the Lagrangian principle now states that 
\begin{equation}
L_e = \integral_\Omega \LL_e \, {dx}^\circ 
\end{equation}

is extremal under this variation. That is we have \( \updelta L_e = 0 \) 
to first order in \( \updelta A_i \). But
\begin{flalign*}
\updelta L_e 
& = \integral_\Omega 2 F^{ij}\updelta F_{ij} \, dx^\circ \\
&= \integral_\Omega 
	  2 F^{ij} \nabla_i\updelta A_j 
	- 2 F^{ij} \nabla_j\updelta A_i 
	- 2 F^{ij} T^k_{ij}\updelta A_k 
  \,\,dx^\circ \\
&= \integral_\Omega 
	  4 F^{ij} \nabla_i\updelta A_j 
	- 2 F^{ij} T^k_{ij}\updelta A_k 
  \,\,dx^\circ \\
&= \integral_\Omega 
	- 4 \nabla_i F^{ij} \updelta A_j 
	- 2 F^{ij} T^k_{ij}\updelta A_k 
	+ 4 \nabla_i\left( F^{ij} \updelta A_j \right) 
  \,\,dx^\circ \\
&= \integral_\Omega 
	- 4 \nabla_i F^{ij} \updelta A_j 
	- 2 F^{ij} T^k_{ij}\updelta A_k 
  \,\,dx^\circ 
 + 4 \integral_{\del\Omega}
	F^{ij} \updelta A_j
  \,\,dx_i^\circ 
\end{flalign*}

The integral on the boundary \( \del\Omega \) is zero because the variation 
is zero there. Thus
\begin{equation}
\updelta L_e = \integral_\Omega 
	\left(-4\nabla_iF^{ik} -2 F^{ij}T^k_{ij}\right)
	\updelta A_k dx^\circ = 0
\label{last_Le_equation}
\end{equation}

and since this must be true for any variation, we must have
\begin{equation}
2\nabla_iF^{ik} + F^{ij}T^k_{ij} = 0 
\end{equation}

This is interpreted as the dynamic equation for electromagnetic fields
in the absence of sources. When sources are present the right hand side
will be replaced by a source term.

We would prefer an equation in terms of \( F_{ij} \). Rearranging we
obtain
\begin{equation}
\label{extendedAG}
\nabla^iF_{ik} - \nfrac{1}{2}F_{ia}g^{ab}T^i_{kb} = J_k 
\end{equation}

where \( J_k \) is a source term. This is equation~\ref{ScalarFieldEquation},
the scalar component of Ussher's equation, which reinforces our belief that
this is the correct extension of the Ampere-Gauss equation in our geometry.

\medskip
The advantage of obtaining this equation via a Lagrangian approach rather than
as a consequence of Ussher's identity is that it gives us a method of finding
the source term. In particular the source term \( J_k \) can be determined for 
a particle by varying its Lagrangian with respect to \( A_i \). 

\medskip
The disadvantage of a Lagrangian approach is that we must guess at the form 
of the Lagrangian.The Lagrangian that we used in this case was the obvious 
extension of the usual electromagnetic Lagrangian into our model, formed
by contracting \( F_{ij} \) with its dual \( F^{ij} \). Contracting a tensor 
with its dual is one of the most obvious and natural ways to obtain a scalar 
from it. 

We call this scalar the \textbf{norm} of \( F_{ij} \) , and denote it 
\( ||F_{ij}||^2 \). Such a norm is defined whenever we can define a unique 
dual. 

\bigskip
\pagebreak[2]

\section{Lagrange Methods: Gravity}
\label{section:Lagrangegravity}

Equation~\ref{extendedAG} is the dynamic equation for \( F_{ij} \) 
obtained via a Lagrangian principle by variation of \( A_k \). 
But \( F_{ij} \)  is merely a component of the overall spinor curvature 
tensor \( R^\alpha_{ij\beta} \) while \( A_k \) is merely a component 
of the spinor connection \( \Gamma^\alpha_{k\beta} \). 

This suggests that we should seek a dynamical equation for the entire 
spinor curvature tensor by varying the spinor connection. To do this we
need a Lagrangian constructed from the curvature \( R^\alpha_{ij\beta} \).

\medskip
One natural way to do this is to contract it with its dual
\begin{equation}
R_\alpha^{ij\beta} = 
s^\bullet_{\alpha\lambda} R^\lambda_{ab\mu} s_\bullet^{\mu\beta}
g^{ia}g^{jb} 
\end{equation}

Note the use of the convention for spinor indices that they are lowered
on the left and raised on the right. This ensures that raising and 
lowering are opposite operations.  This gives
\begin{equation}
\LL	= ||R^\alpha_{ij\beta}||^2 
	= R^\alpha_{ij\beta}R^\lambda_{ab\mu}g^{ia}g^{jb} 
	  s^\bullet_{\alpha\lambda}s_\bullet^{\mu\beta}
\end{equation}

Resolving the spinor curvature into its components we obtain
\begin{equation}
\LL = R^k_{ij}R^c_{ab}g^{ia}g^{jb}g_{kc} - 4F_{ij}F_{ab}g^{ia}g^{jb}
	= ||R^k_{ij}||^2 - 4||F_{ij}||^2 
	= \LL_g -4\LL_e
\end{equation}

Where \( \LL_e = ||F_{ij}||^2 \) is the electromagnetic Lagrangian we considered
previously, and \( \LL_g = ||R^k_{ij}||^2 \) is a Lagrangian constructed from 
the reduced curvature tensor responsible for the gravitational field. 

\medskip
This is very promising. It is a natural Lagrangian constructed from the 
curvature tensor, and it extends the Lagrangian for electromagnetism in 
exactly the way that we wanted. It contains two terms, the expected  
electromagnetic term and a new and very natural gravitational term which 
we might hope to use to obtain a source equation for gravity.

Unfortunately at this point the weakness in the Lagrangian method is about 
to bite us, because there is another natural way in which we can form a 
scalar from the curvature tensor \( R^\alpha_{ij\beta} \). Since this 
tensor has one upper and one lower spinor index, these can be contracted 
without raising or lowering to obtain the Lagrangian
\begin{equation}
\LL = R^\alpha_{ij\beta} g^{ia}g^{jb}R^\beta_{ab\alpha} = \LL_g +4\LL_e
\end{equation}

The gravitational and electromagnetic components \( \LL_g \) and \( \LL_e \) 
are the same as before, but there is a difference in sign in the way they 
are combined. 

We expect this to lead to a measurable difference in physical 
phenomenology. In particular we expect the resulting gravitational field 
equations will have a different source term for the electromagnetic 
contribution to the gravitational field. In principle we should eventually 
be able to rule out one of these Lagrangians on that basis. 

Meanwhile we continue the analysis by working with the components 
\( \LL_g \) and \( \LL_e \) which will enable us to defer the question of 
how these should combine until later. It will also have the benefit of 
minimising algebraic complexity. We define Lagrangians
to go along with these Lagrangian densities
\begin{flalign}
\begin{split}
L_g &=  \int_\Omega \LL_g \, dx^\circ \\
L_e &=  \int_\Omega \LL_e \, dx^\circ
\end{split}
\end{flalign}

The overall Lagrangian will be \( L = L_g \pm 4 L_e \), with the 
sign determined by our choice of Lagrangian density \( \LL \).

\medskip

Having chosen a Lagrangian we must decide what to vary and how to vary it.
The obvious approach is to simply vary the connection. However 
the extent to which the connection by itself can vary without breaking the 
structure of our framework is very limited. For example if the 
local action \( T^\alpha_{i\beta} \) does not vary then this fixes both 
the torsion \( T^k_{ij} \) and the metric \( g_{ij} \). The metric 
determines the Christoffel symbols so both the antisymmetric and 
symmetric parts of \( \Gamma^k_{ij} \) are fixed which implies 
\( R^l_{ijk} = R^t_{ij}T^l_{tk} \) is also fixed. Hence no variation 
in the gravitational field is possible. 

\medskip

This implies that any non-trivial variation which does not break the structure
of the framework must vary both the local and global actions.
In considering the general situation the following Lemma will be useful.

\begin{lemma}
\label{lemma:variation not versor} 
Consider a tensor derivation \( M\oo \) acting on spinors and vectors,
given by \( M^\alpha_\beta \) and \( M^i_j \) respectively. 

Then \( M\oo (T^\alpha_{i\beta}) = 0 \) if and only if there is a scalar 
\( a \) and vector \( b^k \) with
\begin{flalign}
M^\alpha_\beta &= a 1^\alpha_\beta + b^k T^\alpha_{k\beta} \\
M^i_j &= b^kT^i_{kj} 
\end{flalign}
\end{lemma}
\begin{proof}
Firstly, since \( M\oo(T^\alpha_{i\beta}) = 0 \) we have
\begin{equation}
M^\alpha_\lambda T^\lambda_{i\beta} - T^\alpha_{i\lambda}M^\lambda_\beta =
M^k_iT^\alpha_{k\beta} 
\end{equation}
the left hand side here can be viewed as the commutator \( [M,T_i] \) of
two matrices.
\medskip
Next we write \( M^\alpha_\beta  \) in terms of the decomposition of 
spinor transformations into components to obtain
\begin{equation}
M^\alpha_\beta 
= a 1^\alpha_\beta + b^k T^\beta_{k\alpha} + c^AT^\beta_{A\alpha}
\end{equation}
putting these things together we must have 
\begin{equation}
b^kT^m_{ki}T^\alpha_{m\beta} - c^AT^B_{iA}T^\alpha_{B\beta} 
	= M^m_iT^\alpha_{m\beta}
\end{equation}

Comparing versor components on each side of this equation we must have 
\( c^AT^B_{iA} = 0 \) for all \( i \) and \( B \) which gives
\( c^A = 0 \). Comparing vector components we have 
\( b^kT^m_{ki} = M^m_i \) completing the proof.
\end{proof} 

\medskip
Now suppose that both the local and global actions vary in a manner 
consistent with the structure of a framework.  
\begin{flalign}
T^\beta_{i\alpha} 
	&\mapsto T^\beta_{i\alpha} + \updelta T^\beta_{i\alpha} \\
\Gamma^\beta_{i\alpha} 
	&\mapsto \Gamma^\beta_{i\alpha} + \updelta \Gamma^\beta_{i\alpha}
\end{flalign}

Where the expressions on the right specify valid local and global actions
for the Lie algebra \( \so(2,3) \) and where these commute. 

On any framework bases can be chosen so that the matrices 
\( T^\beta_{i\alpha} \) at a chosen point take the form of the matrices on 
page~\pageref{sp4R_basis}. Since after variation we have a valid spinor
manifold we must therefore be able to describe the variation in 
\( T^\alpha_{i\beta} \) as if it were a small change of basis. 

Note that we are not saying here that the variation is simply a change of 
basis. A variation is a real change whereas a change of basis is only a 
change in description. However these two things will have the same algebraic 
form. We have met this kind of thing before in basic linear algebra where a 
rotation is a transformation having the same algebraic description as a change 
of orthonormal basis.

Hence there are small tensors \( \delta^a_b \) and \( \delta^\alpha_\beta \) 
so that to first order 
\begin{equation}
T^\beta_{i\alpha} + \updelta T^\beta_{i\alpha}
= (1^k_i - \delta^k_i)T^\mu_{k\nu}
	(1^\beta_\mu + \delta^\beta_\mu)(1^\nu_\alpha - \delta^\nu_\alpha)
\end{equation}

where the negative signs arise from taking the inverse assuming 
\( \delta^a_b \) and \( \delta^\alpha_\beta \) are small. Expanding out 
we find that to first order we have 
\begin{equation}
\updelta T^\beta_{i\alpha}
= \delta^\beta_\mu T^\mu_{i\alpha}
- \delta^k_iT^\beta_{k\alpha}
- \delta^\nu_\alpha  T^\beta_{i\nu} = \delta\oo (T^\beta_{i\alpha})
\end{equation}

Hence the variation in the local action can be described via
small tensor derivation \( \delta\oo \). 

The variation in the global action given by \( \updelta \Gamma_k\oo \)
must be consistent with the variation in the local action. Requiring that 
\( \nabla_kT^\alpha_{i\beta} = 0 \) after variation gives (to first order)
\begin{equation}
\bigl[\updelta \Gamma_k\oo + \nabla_k(\delta)\oo\bigr]T^\alpha_{i\beta} = 0
\end{equation}

Hence by lemma~\ref{lemma:variation not versor} we can write
\begin{flalign}
\updelta \Gamma^\alpha_{k\beta} + \nabla_k(\delta^\alpha_\beta) &=
a_k 1^\alpha_\beta + b^t_k T^\alpha_{t\beta} \\
\updelta \Gamma^i_{kj} + \nabla_k(\delta^i_j) &= b^t_kT^i_{tj} 
\label{equation:varygammmab}
\end{flalign}

for scalar fields \( a_k \) and vector fields \( b^t_k \). Using this 
description a general variation in the local and global actions will be 
written in the form
\begin{flalign}
\begin{split}
\updelta T^\alpha_{k\beta} &= \delta\oo T^\alpha_{k\beta} \\
\updelta T^i_{kj} &= \delta\oo T^i_{kj} \\
\updelta \Gamma^\alpha_{k\beta} &= 
\delta^\alpha_\lambda \Gamma^\lambda_{k\beta}
- \delta^\lambda_\beta \Gamma^\alpha_{k\lambda}
-\del_k(\delta^\alpha_\beta) + a_k1^\alpha_\beta + b^t_kT^\alpha_{t\beta} \\
\updelta \Gamma^i_{kj} &= 
\delta^i_m \Gamma^m_{kj}
- \delta^m_j \Gamma^i_{km}
-\del_k(\delta^i_j) + b^t_kT^i_{tj} 
\end{split}
\label{general_variation}
\end{flalign}

Such variations preserve the requirements
that the local action represent \( \so(2,3) \) and that the global action 
commute with it. We now need to require that the global action 
represent (albeit with curvature) the Lie algebra even after variation.
This is true if the variation in the torsion agrees with the variation in the 
Lie structure. Hence we must have 
\begin{equation}
\updelta \Gamma^i_{kj} - \updelta \Gamma^i_{jk} 
	= \updelta T^i_{jk} 
	= \delta\oo T^i_{jk} 
\end{equation}

But from equation~\ref{equation:varygammmab} we also have
\begin{equation}
\updelta \Gamma^i_{kj} - \updelta \Gamma^i_{jk} 
	= \left( b^t_kT^i_{tj}   - b^t_jT^i_{tk} \right) 
	+ \left( \nabla_j\delta^i_k - \nabla_k\delta^i_j \right)
\end{equation}

Putting these together gives a constraint relating \( b^t_k \) to 
\( \delta\oo \).
\begin{flalign}
\label{equation_constraint}
\begin{split}
\left( b^t_kT^i_{tj}   - b^t_jT^i_{tk} \right) 
&= \left(\nabla_k\delta^i_j - \nabla_j\delta^i_k \right) 
	+ \delta\oo T^i_{jk} \\
&= \left(\del_k\delta^i_j - \del_j\delta^i_k \right)  
	- \left( \delta^t_k\Gamma^i_{tj} - \delta^t_j\Gamma^i_{tk} \right) 
\end{split}
\end{flalign}

Variations which satisfy equation~\ref{equation_constraint} are called 
\textbf{physical variations} as the variation conserves all axioms of 
a framework.  The physical variations form an abelian group under 
addition. Our task is to require that the Lagrangian is invariant 
under an arbitrary physical variation. 

\bigskip

First however we discuss the effect of these variations on the measure 
which we will need to know about since the variation \( \updelta L \)
in the Lagrangian depends on both the variation \( \updelta\LL \) in the 
Lagrangian density and on the variation \( \updelta dx^\circ \) in the 
measure.

The measure is specified by equation~\ref{explicit_measure}  to be
\begin{equation*}
dx^\circ
 = \epsilon\,\, \eta_{i_0 \cdots i_9}\,\, d{x_0}^{i_0}\!\cdots d{x_9}^{i_9} 
\end{equation*}

Its variation \( \updelta dx^\circ \) can be described in terms of the 
variation \( \updelta \epsilon \) in the multiplying factor.
\begin{equation*}
\updelta dx^\circ
 = \updelta \epsilon\,\, 
	\eta_{i_0 \cdots i_9}\,\, d{x_0}^{i_0}\!\cdots d{x_9}^{i_9} 
\end{equation*}

This can be simplified by writing \( \updelta \epsilon \) as a multiple 
of epsilon
\begin{equation}
\updelta \epsilon = \delta . \epsilon 
\end{equation}
where \( \delta \) is a small scalar function specifying the variation. 
We obtain
\begin{equation}
\updelta dx^\circ
 = \delta . \epsilon\,\, 
	\eta_{i_0 \cdots i_9}\,\, d{x_0}^{i_0}\!\cdots d{x_9}^{i_9} 
= \delta . dx^\circ
\end{equation}

Applying the variation to equation~\ref{determine_epsilon} we must also have
\begin{flalign}
    & \del_k(\updelta \epsilon) 
	= (\updelta \epsilon) \Gamma^a_{ka} + \epsilon(\updelta \Gamma^a_{ka})
\nonumber \\ \Rightarrow \quad &
    (\del_k\delta).\epsilon + \delta.(\del_k\epsilon) 
	= (\updelta \epsilon) \Gamma^a_{ka} + \epsilon(\updelta \Gamma^a_{ka})
\nonumber \\ \Rightarrow \quad &
    \del_k\delta = \updelta \Gamma^a_{ka}
\end{flalign}

and~\ref{general_variation} then can be applied to obtain
\begin{equation}
    \del_k\delta = -\del_k(\delta^a_a)
\end{equation}

It follows that \( \delta = -\delta^a_a + c \) for some constant \( c \),
and since both are zero on the boundary \( \del\Omega \) the constant of
\( c \) must be zero. We have proved that 
\begin{equation}
\label{varydx}
\updelta dx^\circ = -\delta^a_a dx^\circ
\end{equation} 

Now that we understand the effect of the variation on the measure we are 
at last ready to look at the Lagrangian problem. 

\bigskip

We break the problem up into three cases, each involving a different type
of physical variation. We may consider these separately since an arbitrary 
physical variation can be written as a sum of variations from these three cases.

\label{page3cases}
\begin{description}
\bigskip
\item[Case 1.] Variations given by \( a_k \) only with 
\( b^t_k = 0 \), \( \delta^\alpha_\beta = 0 \) and \( \delta^i_j = 0 \).
\begin{equation}
\updelta T^\alpha_{k\beta} = 0 
\quad\quad\,
\updelta T^i_{kj} = 0 
\quad\quad\,
\updelta \Gamma^\alpha_{k\beta} = a_k 1^\alpha_\beta 
\quad\quad\,
\updelta \Gamma^i_{kj} = 0 
\quad\quad\,
\label{VariationCase1}
\end{equation}

These are physical. There is no variation on the measure and since they
only affect the \( F_{ij} \) component of the spinor curvature only the
\( \LL_e \) component of the Lagrangian density can change under a variation 
of this type. The Lagrangian method in this cases reduces to the 
situation in section~\ref{section_EMLagrangian} giving 
equation~\ref{ScalarFieldEquation}, the extended Ampere-Gauss equation.
\end{description}
 
\bigskip
\begin{description}
\item[Case 2.] Variations given by \( \delta^\beta_\alpha \) only with
\( a_k = 0 \), \( b^t_k = 0 \) and \( \delta^i_j = 0 \).
\begin{flalign}
\begin{split}
\updelta T^\alpha_{k\beta} &= 
	\delta^\alpha_\lambda T^\lambda_{k\beta} 
	- \delta^\lambda_\beta T^\alpha_{k\lambda}  \\
\updelta T^i_{kj} &= 0 \\
\updelta \Gamma^\alpha_{k\beta} &= 
	  \delta^\alpha_\lambda\Gamma^\lambda_{k\beta}
 	- \delta^\lambda_\beta\Gamma^\alpha_{k\lambda} 
	- \del_k\delta^\alpha_\beta \\
\updelta \Gamma^i_{kj} &= 0 
\end{split}
\label{VariationCase2}
\end{flalign}

All variations of this type are physical. They fix the measure, and as
we will now show they also fix the Lagrangian density. Hence the Lagrangian
does not change under a variation of this type and they play no role in 
determining dynamic equations via a Lagrangian method.
\end{description}
\medskip
 To determine the variation of the Lagrangian we need to know 
the variations in \( R^\alpha_{ij\beta} \), \( g^{ij} \),
\( s_\bullet^{\alpha\beta}  \) and \( s^\bullet_{\alpha\beta} \).   But 
\begin{equation}
\updelta g^{ij} = 0 = \delta\oo g^{ij} 
\label{Case2g}
\end{equation}
as it is determined by \( T^k_{ij} \) and \( \updelta T^i_{kj} = 0 \). 
The last part follows since \( \delta^i_j = 0 \).

The symplectic form \( s^\bullet_{\alpha\beta} \) is only determined up to 
a scalar by \( T^\alpha_{i\beta} \), so we don't expect to find  a unique 
variation for it. However the variations
\begin{flalign}
\updelta s^\bullet_{\alpha\beta} &= \delta\oo s^\bullet_{\alpha\beta}  \\
\updelta s_\bullet^{\alpha\beta} &= \delta\oo s_\bullet^{\alpha\beta} 
\end{flalign}

where \( \delta\oo \) is defined in terms of \( \delta^\alpha_\beta \) and
an arbitrary  scalar \( \delta^\bullet_\bullet \), ensure the proper 
relationships are conserved under variation and so must be correct. Note
\label{Case2arbitrarybulletNote}
that the scalar \( \delta^\bullet_\bullet \) can take any value and we 
can if we wish choose it to be zero. It represents a local variation 
only in our choice of symplectic form.
To find \( \updelta R^\alpha_{ij\beta} \) we write
\begin{flalign}
\updelta R^\alpha_{ij\beta} = 
\bigl(\del_i\updelta \Gamma^\alpha_{j\beta}
	- \del_j\updelta \Gamma^\alpha_{j\beta} \bigr)
&+ \bigl( \updelta \Gamma^\alpha_{i\lambda} \Gamma^\lambda_{j\beta}
	- \Gamma^\alpha_{j\lambda} \updelta \Gamma^\lambda_{i\beta} \bigr)
\nonumber\\
&+ \bigl( \Gamma^\alpha_{i\lambda} \updelta \Gamma^\lambda_{j\beta}
	- \updelta \Gamma^\alpha_{j\lambda} \Gamma^\lambda_{i\beta} \bigr)
\end{flalign}

We next substitute for the variations using~\ref{VariationCase2}. This 
gives an equation with 24 terms. Fortunately 16 of these cancel and
the remaining 8 terms can be collected up to show
\begin{equation}
\updelta R^\alpha_{ij\beta} = 
	  \delta^\alpha_\lambda R^\lambda_{ij\beta}  
	- \delta^\lambda_\beta R^\alpha_{ij\lambda}  
\label{EvalVarR}
\end{equation}
and since \( \delta^i_j = 0 \) we can write this in the form
\begin{equation}
\updelta R^\alpha_{ij\beta} = 
\delta\oo R^\alpha_{ij\beta}
\end{equation}

\medskip
At this point we can conclude that the variation acts as the tensor 
derivation \( \delta\oo \) on \( R^\alpha_{ij\beta} \), \( g^{ij} \),
\( T^\alpha_{k\beta} \), \( s^{\alpha\beta}_\bullet \),
\( s_{\alpha\beta}^\bullet \) and \( 1^\alpha_\beta \). Hence it 
will act as the tensor derivation \( \delta\oo \) on all tensors constructed
from these via tensor product and contraction. 

As our Lagrangian densities \( \LL_e \) and \( \LL_g \) are constructed in 
precisely this way, the variation will act as \( \delta\oo \) on these too. 
But the tensor derivation \( \delta\oo \) acts trivially on scalars. 
Hence variations of this type have no effect on \( \LL_e \) and \( \LL_g \)
as claimed.

\bigskip

\begin{description}
\item[Case 3.] Variations given by \( b^t_k \) and \( \delta^i_j \) only
with \( a_k = 0 \) and \( \delta^\alpha_\beta = 0 \).
\begin{flalign}
\begin{split}
\updelta T^\alpha_{k\beta} &= \delta\oo T^\alpha_{k\beta} \\ 
\updelta T^i_{kj} &= \delta\oo T^i_{kj} \\
\updelta \Gamma^\alpha_{k\beta} &= b^t_kT^\alpha_{t\beta} \\
\updelta \Gamma^i_{kj} &= b^t_kT^i_{tj}  - \nabla_k\delta^i_j \\
\end{split}
\label{VariationCase3}
\end{flalign}

Not all such variations are physical so we must also require
equation~\ref{equation_constraint}. This is the interesting and
difficult case.
\end{description}
\medskip

Since the variation in the connection \( \Gamma^\alpha_{k\beta} \) is
expressed only in terms of \( b^t_k \), we can obtain an expression
for the variation in the curvature purely in terms of \( b^t_k \).
After a little algebraic heroism we obtain
\begin{equation}
\updelta R^\alpha_{ij\beta} = 
\Bigl[ \bigl( \del_ib^m_j - \del_jb^m_i \bigr) +
       \bigl( b^t_j\Gamma^m_{it} - b^t_i\Gamma^m_{jt} \bigr) 
\Bigr]T^\alpha_{m\beta} 
\end{equation}
which, when expressed in terms of the covariant derivative, gives
\begin{equation}
\updelta R^\alpha_{ij\beta} = 
\Bigl[ \nabla_ib^m_j - \nabla_jb^m_i
- T^t_{ij}b^m_t \Bigr]T^\alpha_{m\beta}
\end{equation}

separating out components we obtain
\medskip
\begin{equation}
\begin{split}
\updelta F_{ij} &= 0 \\
\updelta R^k_{ij} 
	&= \nabla_ib^m_j - \nabla_jb^m_i - T^t_{ij}b^m_t +R^m_{ij}\delta^k_m 
\end{split}
\end{equation}

The variations in the Lagrangian densities \( \updelta \LL_e \) and 
\( \updelta \LL_g \) depend on the variations \( \updelta g_{ij} \) 
and \( \updelta g^{ij} \) in the metric, while the variation in the 
Lagrangian depends also on the variation \( \updelta dx^\circ \) in 
the measure. The metric variation is
\begin{flalign}
\begin{split}
\updelta g_{ij} &= \delta\oo g_{ij} \\
\updelta g^{ij} &= \delta\oo g^{ij} \\
\end{split}
\end{flalign}


We can now determine \( \updelta \LL_e \) and \( \updelta \LL_g \). We get
\begin{flalign}
\updelta \LL_e  
 &= F_{ij}F_{ab}\bigl(\updelta g^{ia}g^{jb} + g^{ia}\updelta g^{jb}\bigr) \\
 &= \bigl[4F_{ij}F_{an}g^{ia}g^{jm}\bigr]\delta^n_m
\end{flalign}
for the variation of the electromagnetic Lagrangian density, where the symmetry
properties of the tensors have been exploited to show all four terms are equal.
The gravitational Lagrangian density gives
\begin{flalign}
\begin{split}
\updelta \LL_g  
 &= 2\updelta R^k_{ij}R^c_{ab}g_{kc}g^{ia}g^{jb}
 +2R^k_{ij}R^c_{ab}g_{kc}\updelta g^{ia}g^{jb}
 +R^k_{ij}R^c_{ab}\updelta g_{kc}g^{ia}g^{jb} \\
 &= 2\updelta R^k_{ij}R^c_{ab}g_{kc}g^{ia}g^{jb}
 +4R^k_{ij}R^c_{ab}g_{kc}g^{ia}g^{jm} \delta^b_m
 -2R^k_{ij}R^c_{ab}g_{mc}g^{ia}g^{jb}\delta^m_k \\[2pt]
 &= 2\bigl(\nabla_ib^k_j - \nabla_jb^k_i 
	- T^t_{ij}b^k_t +R^m_{ij}\delta^k_m\bigr) 
	R^c_{ab}g_{kc}g^{ia}g^{jb} \\
&\quad\quad
 +4R^k_{ij}R^c_{ab}g_{kc}g^{ia}g^{jm} \delta^b_m
 -2R^m_{ij}R^c_{ab}g_{kc}g^{ia}g^{jb}\delta^k_m \\[2pt]
 &= 2\bigl(\nabla_ib^k_j - \nabla_jb^k_i - T^t_{ij}b^k_t\bigr) 
	R^c_{ab}g_{kc}g^{ia}g^{jb}
 +4R^k_{ij}R^c_{ab}g_{kc}g^{ia}g^{jm} \delta^b_m \\
 &= 2\bigl(2\nabla_ib^k_j - T^t_{ij}b^k_t\bigr) 
	R^c_{ab}g_{kc}g^{ia}g^{jb}
 +4R^k_{ij}R^c_{ab}g_{kc}g^{ia}g^{jm} \delta^b_m 
\end{split}
\end{flalign}

We can now integrate these to obtain the variations in the
Lagrangians \( L_e \) and \( L_g \). In doing so we must also use 
$ \updelta dx^\circ = -\delta^a_a dx^\circ $ from 
equation~\ref{varydx}.
\medskip

\begin{flalign}
\updelta L_e &= \nonumber 
\int_\Omega\updelta \LL_e \, dx^\circ + \int_\Omega \LL_e \updelta dx^\circ \\
&= \int_\Omega 
	4F_{ij}F_{an}g^{jm} g^{ia}\delta^n_m 
	- F_{ij}F_{ab}g^{jb}g^{ia}\delta^m_m \, dx^\circ \\
&= \int_\Omega 
  F_{ij}F_{ab} g^{ia} 
	\bigl( 4 g^{jm} 1^b_n  - g^{jb} 1^m_n \bigr) 
		\delta^n_m \, dx^\circ 
\label{emlagrangiandensity}
\end{flalign}

The integrand here involves the obvious extension of the stress energy 
tensor for the electromagnetic field to our ten dimensional context. 
This is the expected source term for the contribution to gravity 
of the energy in an electromagnetic field.

The gravitational Lagrangian \( L_g \) is as follows. Note the use of 
Stoke's theorem in the third step.
\begin{flalign}
\updelta L_g &= \nonumber 
\int_\Omega\updelta \LL_g \, dx^\circ + \int_\Omega \LL_g \updelta dx^\circ 
\\[2pt] &= \nonumber \int_\Omega
4\nabla_ib^k_j R^c_{ab}g_{kc}g^{ia}g^{jb}
   - 2T^m_{ij} R^c_{ab}g_{nc}g^{ia}g^{jb} b^n_m 
\\ \nonumber &\quad\quad\quad 
    +4R^k_{ij}R^c_{an}g_{kc}g^{ia}g^{jm}\delta^n_m
    -R^k_{ij}R^c_{ab}g_{kc}g^{ia}g^{jb}\delta^a_a \,\,dx^\circ \\[2pt]
&= \nonumber \int_\Omega
\bigl(-4\nabla_iR^c_{ab}g^{mb}
- 2T^m_{ij} R^c_{ab}g^{jb}\bigr)g_{nc}g^{ia}\, b^n_m 
\\ \nonumber &\quad\quad\quad
 +\bigl(4R^k_{ij}R^c_{an}g^{jm}
-R^k_{ij}R^c_{ab}g^{jb}1^m_n\bigr)g_{kc}g^{ia}\,\delta^n_m \,\,dx^\circ \\
&= \nonumber \int_\Omega
-4\bigl(\nabla^jR^k_{ji} - \nfrac{1}{2}T^j_{is} g^{st} R^k_{jt} \bigr)
g^{mi} g_{nk} \, b^n_m 
\\ \nonumber &\quad\quad\quad
 + R^k_{ij}R^c_{ab} g^{ia} g_{kc}
	\bigl(4 g^{jm} 1^b_n -g^{jb}1^m_n\bigr)
		\,\delta^n_m \,\,dx^\circ \\
&= \int_\Omega
-4U^k_i g^{mi} g_{nk} \, b^n_m 
\\ \nonumber &\quad\quad\quad
 + R^k_{ij}R^c_{ab} g^{ia} g_{kc}
	\bigl(4 g^{jm} 1^b_n -g^{jb}1^m_n\bigr)
		\,\delta^n_m \,\,dx^\circ 
\end{flalign}

Where \( U^k_i \) is Ussher's tensor from equation~\ref{Udefined}.
Note the similarity of the last term to equation~\ref{emlagrangiandensity}.
This term describes the energy held in the gravitational field and 
suggests the presence of non-linear phenomena where gravity acts 
as a source of gravity. 

Combining these two into the Lagrangian \( L = L_g \pm 4L_e \) we
obtain
\begin{flalign}
\updelta L 
&= \nonumber \int_\Omega
-4U^k_i g^{mi} g_{nk} \, b^n_m 
\\&\quad\quad\quad
+ \bigl( R^k_{ij}R^c_{ab} g^{ia} g_{kc} \pm 4 F_{ij}F_{ab} g^{ia} \bigr)
	\bigl( 4 g^{jm} 1^b_n - g^{jb} 1^m_n \bigr) 
		\delta^n_m \, dx^\circ 
\label{fulllagrangian}
\end{flalign}

To go beyond this we will need to use an equation relating \(b^n_m \) 
and \( \delta^n_m \).  The algebra involved is long and complicated; 
we will work our way through it in section~\ref{section:exactequation}.  

\medskip

However we are now in a position to obtain a good approximate solution
in the small field\footnote{calling it the `weak field case' might be 
misunderstood} case where \( F_{ij} \) and \( R^k_{ij} \) are small.
Under these conditions we can drop all terms which are second order in 
\( F_{ij} \) or \( R^k_{ij} \). Equation~\ref{fulllagrangian}
then simplifies to 
\begin{equation}
\updelta L_{\text{small field}} 
= \int_\Omega -4U^k_i g^{mi} g_{nk} \, b^n_m \,\,dx^\circ
\label{approximatelagrangian}
\end{equation}

which we can now solve quite easily since all the terms involving 
\( \delta^n_m \) have been eliminated. If this is to be zero for all 
\( b^n_m \) we must have 
\begin{equation}
-4U^k_ig^{mi}g_{nk} = 0
\end{equation}
Rewriting this in terms of \( R^k_{ij} \), adding in a source term
and rearranging we obtain
\begin{equation*} 
\nabla^j R^k_{ji} -\nfrac{1}{2} T^j_{is} g^{st} R^k_{jt} = K^k_i
\end{equation*}

Which is precisely Ussher's equation~\ref{UsshersEquation} for 
gravitation which we obtained earlier from Ussher's identity. We have now
obtained this equation again via a Lagrangian approach. However we obtained
it as an approximate solution in the small field case. Fortunately this
is likely to be sufficient for many applications.

\bigskip

We end this section with some words of caution about the small field 
solution we have just obtained. 
\begin{enumerate}
\item \label{paragraph:smallfield}
 Eliminating \( \delta^n_m \) from the Lagrangian problem by making 
the small field assumption does not relieve us of our obligation to use 
only physical variations. To the extent that this 
constrains the possible values of \( b^n_m \), \emph{additional solutions 
may be possible in the small field case}. We will revisit this very subtle 
technicality once we have solved the exact equation.
\item \label{paragraph:UssherSourceProblem} We have not yet specified the
 nature of the source term in Ussher's equation. We were hoping to find 
this by varying the Lagrangian for matter. However the small field 
approximation we have just made will make this 
difficult as we have no reason to expect that such a variation will 
depend only on \( b^n_m \). We would need to express any \( \delta^n_m \) 
term as a function of \( b^n_m \) and to the extent that \( \delta^n_m \) 
is not fully determined by \( b^n_m \) this may not be possible. We may 
have to seek a best approximation.
\end{enumerate}

\bigskip
\pagebreak[2]

\section{The Exact Equation}
\label{section:exactequation}

We resume our efforts to obtain an exact solution for the full
Lagrangian problem continuing on from equation~\ref{fulllagrangian}.
\begin{flalign*}
\updelta L 
&= \nonumber \int_\Omega
-4U^k_i g^{mi} g_{nk} \, b^n_m 
\\&\quad\quad\quad
+ \bigl( R^k_{ij}R^c_{ab} g^{ia} g_{kc} \pm 4 F_{ij}F_{ab} g^{ia} \bigr)
	\bigl( 4 g^{jm} 1^b_n - g^{jb} 1^m_n \bigr) 
		\delta^n_m \, dx^\circ 
\end{flalign*}

The variation is described in terms of two quantities 
\( \delta^a_b \) and \( b^m_n \) which are small in 
the interior of a region \( \Omega \) and zero on its boundary;
if the variation is to be physical these are constrained by 
equation~\ref{equation_constraint} which states. 
\begin{flalign*}
\begin{split}
\left( b^t_kT^i_{tj}   - b^t_jT^i_{tk} \right) 
&= \left(\nabla_k\delta^i_j - \nabla_j\delta^i_k \right) 
	+ \delta\oo T^i_{jk} \\
&= \left(\del_k\delta^i_j - \del_j\delta^i_k \right)  
	- \left( \delta^t_k\Gamma^i_{tj} - \delta^t_j\Gamma^i_{tk} \right) 
\end{split}
\end{flalign*}

However in this form the constraint is difficult to use.
We would prefer a constraint with \( b^m_n \) or \( \delta^m_n \) as the 
subject which might allow us to substitute. We begin this section by 
seeking an equivalent constraint in this form.

\medskip

The vector connection \( \Gamma^k_{ij} \) can be written in terms of 
the Christoffel connection and the torsion as given in 
equation~\ref{Christoffel}. Furthermore the Christoffel connection 
can be written purely in terms of the metric and its partial derivatives. 
Hence the connection can be expressed in terms of the torsion,
the metric and its partial derivatives.
\begin{equation}
\label{viametricandtorsion}
\Gamma^k_{ij} = \frac{1}{2}\Bigl[
	\bigl( \del_ig_{jm}+\del_jg_{im} -\del_mg_{ij} \bigr)g^{mk}
	- T^k_{ij}\Bigr]
\end{equation}

Since the variations of the quantities on the right hand side only
depend on \( \delta^m_n \) we can use this to write 
\( \updelta\Gamma^k_{ij} \) in terms of \( \delta^m_n \) and its partial 
derivatives. We can then use this to write \( b^m_n \) as a function 
of \( \delta^m_n \).

\medskip

We begin by writing equation~\ref{viametricandtorsion} in a form
which minimises contractions.
\begin{equation}
\label{viametricandtorsionsimplified}
2\Gamma^m_{ij}g_{mk} 
	= \del_ig_{jk}+\del_jg_{ik} -\del_kg_{ij} - T^m_{ij}g_{mk}
\end{equation}

We do this in order to minimise the number of terms in the equation after 
variation. We can also simplify the algebra by working using lowered index 
versions, \( \delta_{ab} = \delta^t_ag_{tb} \) and
\( b_{ab}=b^t_ag_{tb} \), of our variation.  In particular we have
\begin{flalign}
\updelta g_{ij} &= - \delta_{ij} - \delta_{ji} \\
b_i^tT^m_{tj}g_{mk} &= b_{im}T^m_{jk} \\
- \updelta (T^m_{ij}g_{mk}) &= 
	T^m_{ij}\delta_{mk} + T^m_{jk}\delta_{mi} + T^m_{ki}\delta_{mj} 
\end{flalign}

Applying the variation we obtain in the first instance
\begin{equation}
\begin{split}
& 2 b_{im}T^m_{jk} 
- 2\nabla_i(\delta^m_j)g_{mk} 
- 2\Gamma^m_{ij}\delta_{mk}  - 2\Gamma^m_{ij}\delta_{km}  
\\[3pt]&\quad\quad\quad\quad = \quad
  - \del_i\delta_{jk} - \del_i\delta_{kj} 
  - \del_j\delta_{ik} - \del_j\delta_{ki} 
  + \del_k\delta_{ij} + \del_k\delta_{ji} 
\\&     \quad\quad\quad\quad\quad\quad \quad\quad\quad\quad
  + T^m_{ij}\delta_{mk} + T^m_{jk}\delta_{mi} + T^m_{ki}\delta_{mj} 
\end{split}
\end{equation}

The next step is to replace all of the partial derivatives with covariant 
derivatives, and to cancel and collect up all the resulting terms 
involving the connection. These pair up nicely to give torsion terms
yielding the equation
\begin{equation}
\begin{split}
2 b_{im}T^m_{jk} &= 
(  \nabla_i\delta_{jk} 
- \nabla_i\delta_{kj}) 
+ (\nabla_k\delta_{ij} 
- \nabla_j\delta_{ik})
+ (\nabla_k\delta_{ji}
- \nabla_j\delta_{ki}) 
\\&\quad
+ 2\delta_{mi}T^m_{jk}
+ \delta_{im}T^m_{jk}
+ (\delta_{jm}T^m_{ik}
- \delta_{km}T^m_{ij})
\end{split}
\end{equation}

Note that the equation is antisymmetric with respect to the indices
\( j \) and \( k \) and has been written to illustrate this. We now contract
with the tensor \( \nfrac{1}{2} T^{jk}_t = \nfrac{1}{2} T^j_{ts}g^{sk} \) which is also 
antisymmetric with respect to the indices \( j \) and \( k \) to obtain 
the equation
\begin{equation}
\begin{split}
  6 b_{it} &= 
   \nabla_i\delta_{jk}  T^j_{ts}g^{sk}
+  \nabla_k\delta_{ij}  T^j_{ts}g^{sk}
+  \nabla_k\delta_{ji}  T^j_{ts}g^{sk}
\\&\quad
+ 6 \delta_{ti} + 3\delta_{it} 
+  \delta_{jm} T^m_{ik} g^{sk} T^j_{ts}
\end{split}
\end{equation}

We finish up by renaming contracted indices in a more systematic fashion.
\begin{equation}
\begin{split}
  6 b_{mn} &= 
   \nabla_m\delta_{ij}  T^i_{ns}g^{sj}
+  \nabla_j\delta_{mi}  T^i_{ns}g^{sj}
+  \nabla_j\delta_{im}  T^i_{ns}g^{sj}
\\&\quad
+ 3\delta_{mn} + 6 \delta_{nm} 
+  \delta_{ij} T^j_{ms} g^{st} T^i_{nt}
\end{split}
\label{substituter}
\end{equation}

\medskip

We now prepare equation~\ref{fulllagrangian} for substitution by rewriting 
it in the form
\begin{equation}
\updelta L 
= \int_\Omega -4U^n_k g^{km} b_{mn} + W^m_p g^{pn} \delta_{mn}
 \,\,dx^\circ
\label{lagrangianUW}
\end{equation}
where \( U^n_k \) is Ussher's tensor from equation~\ref{Udefined}, and
\begin{equation}
W^m_p = \bigl(R^s_{ai} R^t_{bj} g_{st}g^{ab} \pm 4 F_{aj}F_{bi}g^{ab} \bigr)
\bigl(4 g^{jm} 1^i_p - g^{ij} 1^m_p \bigr)
\end{equation}

\bigskip

We now at last are ready to substitute equation~\ref{substituter} into 
equation~\ref{lagrangianUW} to obtain a Lagrangian we can solve. We obtain 
the following equation. Note the use of Stokes' theorem to simplify the first
three terms.
\begin{equation}
\begin{split}
3\updelta L 
&= \int_\Omega 
2 \nabla^k U^n_k T^i_{ns} g^{sj} \delta_{ij}  
+2 \nabla^s U^n_k T^i_{ns} g^{km} \delta_{mi}  
+2 \nabla^s U^n_k T^i_{ns} g^{km} \delta_{im}  
\\ & \quad \quad \quad
-6U^n_k g^{km} \delta_{mn} 
-12U^n_k g^{km} \delta_{nm} 
-2U^n_k g^{km} T^j_{ms} g^{st} T^i_{nt} \delta_{ij} 
\\[5pt] & \quad \quad \quad \quad \quad
 + 3W^m_p g^{pn} \delta_{mn}
 \,\,dx^\circ
\end{split}
\end{equation}

We next rename contracted indices and rewrite it in terms of \( \delta^n_m \)
\begin{equation}
\begin{split}
\updelta L 
&= \int_\Omega 
\nfrac{2}3 \nabla^k U^i_k T^m_{in} \delta^n_m  
+\nfrac{2}{3} \nabla^s U^i_k T^p_{is} g^{km} g_{pn} \delta^n_m  
+\nfrac{2}{3} \nabla^s U^i_n T^m_{is} \delta^n_m
\\ & \quad \quad \quad
-2U^p_k g^{km} g_{pn} \delta^n_m
-4U^m_n \delta^n_m
+\nfrac{2}{3}U^j_k g^{ki} T^t_{in} T^m_{jt} \delta^n_m
\\[5pt] & \quad \quad \quad \quad \quad
 + W^m_n \delta^n_m
 \,\,dx^\circ
\end{split}
\end{equation}
 This puts the equation in the form
\begin{equation}
\label{expressionequation}
\updelta L  = 
\int_\Omega [\text{expression}]^m_n\delta^n_m\,dx^\circ
\end{equation} 

which is zero for all variations if and only if 
\( [\text{expression}]^m_n = 0 \). We can now solve the Lagrangian problem
to obtain
\begin{multline}
\nabla^k U^i_k T^m_{in} 
+ \nabla^s U^i_k T^p_{is} g^{km} g_{pn} 
+ \nabla^s U^i_n T^m_{is} \\
-3U^p_k g^{km} g_{pn} -6U^m_n 
+U^j_k g^{ki} T^t_{in} T^m_{jt} 
 + \nfrac{3}{2}W^m_n = 0
\end{multline}

Ussher's identity~\ref{Uidentity1} gives \( \nabla^k U^i_k =0 \) 
allowing us to drop the first term. We collect up the remaining terms 
and write it in the form of a differential equation in Ussher's tensor.
\begin{multline}
\bigl( - T^t_{in} g^{km} + 1^k_n T^m_{is} g^{ts} \bigr) 
	\nabla_t U^i_k \\
- \bigl( 3 g^{km} g_{in} +6,1^m_i1^k_n - g^{kj} T^t_{jn} T^m_{it}\bigr)
	U^i_k 
 + \nfrac{3}{2}W^m_n = 0
\label{exactequation}
\end{multline}
 
This is the long sought exact solution to the Lagrangian problem. We can 
simplify the form of this equation by defining the coefficient tensors
\begin{flalign}
\label{coefficientA}
A^{tkm}_{in} &= T^t_{in} g^{km} - 1^k_n T^m_{is} g^{ts} \\ 
\label{coefficientB}
B^{km}_{in} &= 3 g^{km} g_{in} +6.1^m_i1^k_n - g^{kj} T^t_{jn} T^m_{it}
\end{flalign}

In the presence of sources it will take the form
\begin{equation}
A^{tkm}_{in} \nabla_t U^i_k + B^{km}_{in} U^i_k =  \nfrac{3}{2}(S^m_n + W^m_n)
\label{exactequationsources}
\end{equation}

where the source term \( S^m_n \) will be the coefficient of 
\( \delta^n_m \) under a type 3 variation of the Lagrangian for matter.

\medskip

Consider the small field approximation in the absence of sources obtained 
by dropping the second order term \( W^m_n \) and the source term $S^m_n$
from equation~\ref{exactequationsources}.
\begin{equation}
A^{tkm}_{in} \nabla_t U^i_k + B^{km}_{in} U^i_k = 0
\label{smallfieldequation}
\end{equation}

At the end of section~\ref{section:Lagrangegravity} we obtained
Ussher's equation~\ref{UsshersEquation} as a small field approximate
equation for gravity in the absence of sources. Now we have a different 
small field approximation.  How can we reconcile the two?

Ussher's equation (in the absence of sources) takes the form  \( U^i_k = 0 \), 
so clearly any solution to Ussher's equation will also be a solution to 
equation~\ref{smallfieldequation}. But the converse will not be true. 
Equation~\ref{smallfieldequation} will have additional solutions.

In the discussion on page~\pageref{paragraph:smallfield} we noted that 
additional small field solutions might be possible if the variation 
\( b^n_m \) is partially constrained by the requirement that the variation 
be physical. That does indeed seem to be the origin of the additional 
solutions in equation~\ref{smallfieldequation}.

\medskip

In the presence of sources the small field equation takes the form of 
differential equation in Ussher's tensor. 
\begin{equation}
A^{tkm}_{in} \nabla_t U^i_k + B^{km}_{in} U^i_k =  \nfrac{3}{2}S^m_n
\label{smallfieldsources}
\end{equation}

This equation is not quite as simple as it looks since there are connections
hiding in the covariant derivative. Neverthless it invites us to view the 
situation in terms of a two stage process, where 
equation~\ref{smallfieldsources} determines the value of Ussher's tensor, 
and Ussher's equation then determines the curvature.

\bigskip
\pagebreak[2]

\section{The Dark side of the Force}

The title of this section is perhaps a temptation that should have been 
resisted. However one of the things we are particularly interested in 
looking at is whether there is room in these equations for dark phenomena --- 
behavior which could give the appearance of the presence of dark matter 
and/or energy. The answer seems to be yes. These equations do indeed seem 
to predict dark phenomena as we shall now explain.

\medskip

We will use an analogy with the suspension system of a car to help describe the 
situation. 

Ussher's equation in the presence of sources is \( U^k_m = S^k_m \) where 
\( S^k_m \) is a source term determined by the distribution of matter.  It 
is an equation of type \( u = s \), where the value of \( s \) directly 
determines the value of  \( u \). A step change in \( s \) will thus result 
in a step change in \( u \). We may think of it as a car with no suspension. 
Every bump in the road has immediate effect.

The small field equation~\ref{smallfieldsources} introduces a derivative term
to this relationship. It gives an equation of type \( \ddot u = -k(u - s) \). 
The value of \( s \) still determines the value of \( u \), but in a less 
direct fashion. Our allegorical car is now equipped with a suspension system
smoothing out the bumps in the road.

This suspension analogy is a useful one. In a car with suspension the car
no longer directly `sees' the road. Instead its behavior is determined by the 
state of the suspension. It is the suspension that `sees' the road.

\medskip

Hence gravity, in the weak field case, can be described by two equations. 
Firstly we have the \textbf{dark equation}, which is essentially just the 
small field equation~\ref{smallfieldsources} 
\begin{equation}
A^{tkm}_{in} \nabla_t X^i_k + B^{km}_{in} X^i_k =  
\nfrac{3}{2}S^m_n
\label{darkequation}
\end{equation}

This in principle determines $X^i_k$, which is then used as the effective 
source in Ussher's equation $U^i_k = X^i_k$.
\begin{equation}
\label{UsshersEquationX}
\nabla^j R^k_{ji} - \nfrac{1}{2} T^j_{is} g^{st} R^k_{jt} = X^k_i
\end{equation}

Ussher's equation describes what we might think of as the ordinary behavior 
of gravity. We expect to obtain Newton's law from it for example as a 
non-relativistic approximation. But since Ussher's equation 'sees' the 
effective source \( X^i_k \) and not the usual source \( K^i_k \)
it will seem to us at times as if gravity is responding to a different 
distribution of energy and matter than we can actually observe.

Sometimes the gravitational force predicted by these equations may respond as 
though matter were present in places where there isn't any, because while
$K^i_k = 0 $  there, $X^i_k$ is not. We might interpret this as dark matter.

It may also be possible for the effective source to take values which cannot
be achieved by the apparent source. For example \( X^i_k \) might take values 
which could only be achieved by \( K^i_k \) if negative mass were 
present.  We might interpret this as dark energy.

\bigskip

We need a better name for $X^i_k$. Because it controls the dark side of the 
force we propose to call it the \textbf{sith}. While this is a little cheesy 
at first sight, it is a good short word and you rapidly get used to it. 
The sith is effectively the suspension system of the universe. It mediates 
between the source terms and Ussher's equation which describes the more 
usual behaviors of gravity. 

\medskip

If we want to talk only about Dark matter, we could define it to be 
the difference between the sith $X^i_k$ and the usual source term $K^i_k$ 
for Ussher's equation determined directly from 
the distribution of matter.

\begin{equation}
D^i_k = X^i_k - K^i_k 
\end{equation}

Of course in order to do this we would need to know how to compute $K^i_k$
which is something we have not yet determined. We noted in the discussion 
on page~\pageref{paragraph:UssherSourceProblem} that we can't expect to 
find this via a Lagrangian argument. However we might be able to find it 
by simplifying the exact equation.  To do this we need to look more 
closely at the nature of the relationship between Ussher's 
equation and the exact one. 

\medskip

Ussher's equation is algebraic in $U^i_k$. It is local in the sense that 
Ussher's tensor at each point is a function of the source at that point. 
The dark equation relates $U^i_k$ to the source via a differential equation.
It is not local since Ussher's tensor depends also on what is happening 
at other points.  The derivative terms in the equation induce this non-local 
behavior.

We expect Ussher's equation to approximate the exact one. This suggests we 
should look at the best (algebraic) local approximation to the exact equation.
The easiest way to obtain an algebraic equation from the exact equation is
to simply drop all the terms involving derivatives, since it is these that 
induce the non-local behavior.

Is this reasonable? Returning to our simple suspension analogy, we are 
effectively trying to remove the suspension from the car. We want to start
with the equation for the behavior with suspension, and deduce the equation
without suspension. Will simply dropping derivative terms do the trick? 
If we drop the derivative term from the equation \( \ddot u = -k(u - s) \) 
we get \( 0 = -k(u-s) \) which simplifies to \( u = s \). This is indeed
the original equation for a car without suspension. So yes the procedure 
does seem to be a reasonable one.

This suggests the apparent source term for Ussher's equation should be 
the solution to the linear equation
\begin{equation}
B^{km}_{in} K^i_k = \nfrac{3}{2}S^m_n
\label{apparentequation}
\end{equation}

where \( S^m_n \) is the source term obtained from the Lagrangian for matter
which we hope to be able to compute in chapter~\ref{Chapter MoM}. Note that 
in the source free case where \( S^m_n = 0 \) that \( K^i_k = 0 \) 
is a solution. 

\bigskip

We have now `found' dark matter and energy in the sense of having 
discovered the requisite phenomena within the equations for gravity.
So does dark matter exist and if so what is it?" 

Whether dark matter can be said to \emph{exist} or not is really a question 
of philosophy; what does it mean for anything to exist? Such questions generate
more headaches than answers. It is much more productive to ask what we would 
expect to find if we were to second quantise these equations.  What quanta 
would we get and could we detect them as particles?

As dark matter and energy are aspects of the gravitational field, 
the associated particles would have to be gravitons of some sort, since
gravitons are by definition the quanta of the gravitational field.  
Similarly the magnetic field is an aspect of the electromagnetic field 
and the particles of magnetism are simply photons. 

\bigskip

Having considered the small field case we now briefly discuss the large
field situation. When the fields are large we cannot ignore the second 
order term \( W^m_n \). In fact for very large fields this term will 
dominate. 

We can write down a large field approximation by dropping linear terms.
Ussher's tensor thus becomes \( U^i_k = \nabla^pR^i_{pk} \). We will also 
drop the electromagnetic field terms from \( W^m_n \) under the assumption 
that gravitational fields completely dominate in this situation. This gives 
the second order differential equation
\begin{multline}
\bigl( T^n_{is} g^{km}  g^{ts} + T^m_{is} g^{kn} g^{ts} 
\bigr) \nabla_t\nabla^pR^i_{pk} \\
+ \bigl( 3 g^{km} g_{in} +6.1^m_i1^k_n - g^{kj} T^t_{jn} T^m_{it}\bigr)
\nabla^pR^i_{pk}  \\
+ \nfrac{3}{2}R^s_{ai} R^t_{bj} g_{st}g^{ab}
\bigl(4 g^{jm} g^{ip} - g^{ij} g^{mp} \bigr)g_{pn} = 0 
\label{largefieldapproximation}
\end{multline}

This equation will control behavior around black holes and any place where the 
curvature is very large. It should determine whether singularities are 
possible in such situations and will describe how they form if they arise. We 
have become used to the notion that singularities exist in black holes. 
However in most physical theories singularities are a symptom that our 
model has broken down. We should be skeptical about the existence 
of singularities in any physical theory. If singularities do not form, 
something else must happen when a very large concentration of matter forms.
Equation~\ref{largefieldapproximation} is where we should look for answers.

\bigskip

We started this section by noting that the exact equations for gravity are 
unlikely to give exact solutions. We can however simplify them by looking at 
\textbf{component equations}. 

Our equations constrain tensors with two vector components which
are 100 dimensional. Hence these equations are really systems of 
100 equations in 100 variables. That is a lot of equations and a lot of 
variables. However all the operations in our equations are natural in 
the context of a framework and hence in particular they respect 
decomposition into irreducibles.  By looking at irreducible components 
of our 100 dimensional tensors we can obtain much simpler component 
equations. 

Tensors with two vector components decompose into irreducibles of dimension
1, 5, 10, 14 and 35 (two different types). For example the scalar component 
of the dark equation is
\begin{equation}
\nabla^t X^i_kT^k_{it} -3 X^k_k = -\nfrac{3}{2}S^m_m
\label{scalardarkequation}
\end{equation}

The dark equation separates into component equations very naturally
since the constant tensors in the equation are linear combinations of 
component projection maps. Hence they will disappear or simplify when 
looking at a single component.

\bigskip

We end with a word of caution. The equations in these last two chapters 
are not geometric identities arising directly from the structure of the 
framework as most of the equations in the earlier chapters were. 
They depend on an additional assumption of some sort, whether it be a 
choice of Lagrangian or a decision to attach a particular physical 
interpretation to a certain divergence free quantity. These equations are
only as good as the assumptions used to generate them.

	\chapter{Matter}
\label{Chapter MoM}

We expect matter to be described by tensors giving the wave function, 
which we will regard as a matter field.
\medskip

The particles which mediate forces have, as their wave functions, tensors 
associated with the geometry; specifically the connection, which acts as 
the potential, and the curvature which gives the associated field. Dynamical 
equations for these tensors arise from the structure of the geometry, and 
source equations can be obtained via a Lagrangian argument by varying the 
geometry. Second quantisation would be necessary for a complete theory of 
these particles, but that is beyond the scope of this work which focuses on 
the field theory before second quantisation.

\medskip

Other elementary particles do not arise from the geometry. Elementary 
fermions will have spinor wave functions. We can therefore think of a spinor
on our framework as describing a fermion field. We seek dynamical and
source equations for the fermion field. Once again second quantisation would 
be needed to obtain a complete multiparticle theory, but this is beyond the 
scope of this work. Spinors will thus be our main focus in this Chapter.

\medskip
The fundamental physical equation governing the dynamics of fermions is the 
Dirac equation. We therefore begin by finding what the Dirac equation looks 
like in our ten dimensional context. Initially we will simply state an 
appropriate equation and explore some of its properties. Later we will look 
at deriving it from a natural Lagrangian.

\bigskip
\pagebreak[2]

\section{The Extended Dirac Equation}
\label{section:Dirac}

Let  \( \psi = \psi^\alpha \) be a 4-D vector complex spinor 
defined on our 10-D spinor manifold which we have identified as 
the manifold of frames. Consider the equation
\begin{equation}
\label{QDirac}
\Curl\psi = \lambda\psi
\end{equation}

We can also write in the form
\begin{equation}
\label{QDirac2}
T^\mu_{k\nu}\nabla^k\psi^\nu = \lambda\psi^\mu
\end{equation}

We will call this the \textbf{extended Dirac equation}.
The constant \( \lambda \) is left without interpretation at this 
point, although we expect it to relate to rest mass and 
possibly charge. Note that this equation uses natural operations 
in our context and is quite obviously covariant.
We now consider how this equation  relates to the standard Dirac equation. 

\bigskip
\pagebreak[2]

We firstly claim that the extended Dirac equation becomes the standard
Dirac equation in the Poincar\'e limit as \( r \rightarrow \infty \). 
Ignoring curvature and using a standard basis with natural units,
the Curl operator on spinors is
\begin{multline}
\Curl(\ ) = -T\del_t + X\del_x + Y\del_y + Z\del_z \\
	+ A\del_a + B\del_b + C\del_c - I\del_i - J \del_j - K \del_k
\end{multline}

where \( T,X,Y,Z,A,B,C,I,J,K \) are the \( 4 \times 4 \) matrices 
\( T^\beta_{i\alpha} \).  In ordinary units this becomes
\begin{multline}
\Curl(\ ) = -rT\del_t + rcX\del_x + r`cY\del_y + rcZ\del_z \\
	+ cA\del_a + cB\del_b + cC\del_c - I\del_i - J \del_j - K \del_k
\end{multline}

Assuming that \( r \) is very large, dropping insignificant terms gives
\begin{equation} 
\frac{1}{rc}\Curl(\ ) = -\frac{1}{c}T\del_t + X\del_x + Y\del_y + Z\del_z
\end{equation}

and equation \ref{QDirac} becomes
\begin{equation} 
\label{Dirac}
\left(-\frac{1}{c}T\del_t + X\del_x + Y\del_y + Z\del_z\right)\psi^\alpha = rc\lambda\psi^\alpha
\end{equation}

which for an appropriate choice of \( \lambda \) is in the right form to be 
the standard Dirac equation. Let us do a more detailed comparison.

\medskip

The usual Dirac equation (in ordinary units) can be written as
\begin{equation*}
\bigl(\frac{1}{c}\gamma^0\del_0 + \gamma^1\del_1 
	+ \gamma^2\del_2 + \gamma^3\del_3\bigr)\psi 
		= \frac{mc}{i\hbar}\psi
\end{equation*}
where 
\begin{itemize}
\item \( \gamma^\mu\gamma^\nu + \gamma^\nu\gamma^\mu = 2\eta^{\mu\nu} \) 
\item \( \eta^{\mu\nu} = 0 \) for all \( \mu \ne \nu \).
\item \( \eta^{00} = 1 \) and \( \eta^{11} = \eta^{22} = \eta^{33} = -1 \).
\item \( \gamma^0 \) is Hermitian and Unitary.
\item \( \gamma^i \) is anti-Hermitian and Unitary for \( i = 1,2,3 \). 
\end{itemize}

\medskip

If we make the identifications
\begin{flalign} &&
\gamma^0 &= -2iT &
\gamma^1 &=  2iX &
\gamma^2 &=  2iY &
\gamma^3 &=  2iZ 
&& \label{gammaidentification}\end{flalign}

then the matrices so defined satisfy all the properties above and multiplying 
equation \ref{Dirac} on both sides by \( 2i \) gives precisely the Dirac
equation provided that
\begin{equation}
\label{requation}
-2\lambda r= \frac{m}{\hbar}
\end{equation}

So the standard Dirac equation can be regarded as the Poincar\'e limit of 
equation~\ref{QDirac}. However while mention of the Lorentz dimensions 
disappears from the extended Dirac equation in the Poincar\'e limit, 
those additional dimensions have not gone away. The equations just ignore 
how the wave functions behave along those directions as in the Poincar\'e 
limit that behavior becomes insignificant.

\bigskip

A better way to view the relationship between the extended and standard 
Dirac equations is to consider spinors  \( \psi^\alpha \) which are functions 
of space and time only, and have no dependence on the Lorentz coordinates. 
Spacetime and Lorentz coordinates are mixed by large translations so this 
distinction between types of coordinates is not possible globally. However
it is quite meaningful on laboratory scales. Under this restriction the 
partial derivatives with respect to the Lorentz coordinates in the extended
Dirac equation will be zero giving the standard Dirac equation.

Consequently we can identify solutions to the standard Dirac equation with
solutions of the extended Dirac equation which are locally constant on the 
Lorentz coordinates. We expect that the extended Dirac equation will also 
have solutions that vary along the Lorentz coordinates. The identification 
of these additional solutions is something that we will leave for later.

\bigskip
%

Equation~\ref{QDirac} has many advantages over the standard Dirac equation. 
We pause to list a few of them here.

For one thing equation \ref{QDirac} is obviously invariant since we built 
it from the invariant Curl operator. Furthermore it is already expressed 
in coordinate free fashion on a space with curvature. Contrast this with 
usual approaches to the Dirac equation where a demonstration of relativistic 
invariance and the correct transformation properties of Dirac matrices 
often requires a laborious search for S-matrix transformations and is 
a fairly non-trivial exercise. 

The most important difference however is that every quantity in equation 
\ref{QDirac} including the Dirac matrices themselves comes equipped with 
a direct physical interpretation.
Contrast this with the usual approach to the Dirac equation where the
Dirac matrices are initially chosen purely for their algebraic properties 
and physical interpretations are obtained later and only with considerable 
effort. Indeed the Dirac matrices themselves still lack a generally accepted 
physical interpretation.

Our gamma matrices \( T,X,Y,Z \) are matrices representing the intrinsic 
action of translation by one natural unit along the \( t,x,y,z \) directions. 
We don't have to interpret them as such. That is what they {\em are}. We thus 
know what their eigenvalues should represent. They should give us 
\textit{intrinsic} energy and momentum in natural units.  We expect these 
to relate to ordinary energy and momentum in the same way that spin relates 
to angular momentum. We may think of them as properties internal to the 
particle.  Note that intrinsic energy is unrelated to rest mass which is a 
function of ordinary energy and momentum. 

Since \( T,X,Y,Z \) anti-commute, intrinsic energy and momentum are not 
simultaneously observable. The eigenvalues of \( T \) are 
\( \pm \frac{1}{2}i \) while those of \( X \), \( Y \) and \( Z \) are 
\( \pm \frac{1}{2} \). The lack of a factor of \( i \) in the eigenvalues
of \( X \) \( Y \) and \( Z \) is of note. This occurs because  \( X \) 
for example is Hermitian while \( T \) is anti-Hermitian. In other words
the difference arises because the associated representation of \( \SO(2,3) \) 
is not unitary. Disregarding any factors of \( i \) that arise, the 
intrinsic energy or momentum as measured along any chosen axis is 
always \( \pm \frac{1}{2} \) in natural units. In ordinary units the intrinsic 
energy is \(\pm \frac{\hbar}{2r} \) and the intrinsic momentum measured along 
an axis is \(\pm \frac{\hbar}{2rc} \).  Note that only one of these can be 
observed at a time. 

\medskip

Intrinsic velocity along a given direction, were we interested in defining 
such a thing, would be most naturally defined as the quotient of intrinsic 
energy with intrinsic momentum and we could therefore obtain it from 
appropriate quotients of the Dirac matrices.  The operators so obtained
will not commute so intrinsic velocity could only be measured along one 
direction at a time.  The intrinsic velocity measured along any direction 
would be either \( c \) or \( -c \).  

The operators \( (\gamma^0)^{-1}\gamma^i \) have in fact long been recognised 
as velocity operators. Their eigenvalues of \( \pm c \) have thus been
regarded as problematic since the electron, having finite mass, should not 
be moving at the speed of light. The usual explanation for this paradox 
envisages the instantaneous electron velocity as oscillating very rapidly 
between \( +c \) and \( -c \) in such a way that the observed overall 
velocity is finite, an effect commonly known as the Zitterbewegung or 
``shaking motion'' of the electron.  Getting this explanation to work has 
however proved difficult.

The realisation that the velocity measured by these operators should be 
interpreted as {\em intrinsic} solves the problem completely. An intrinsic 
velocity of \( c \) can coexist with an ordinary finite velocity quite
easily since we no longer have to reconcile the two. The intrinsic
velocity is free to do whatever it likes independent of the ordinary
velocity just as spin is independent of angular momentum. The intrinsic 
velocity need not even oscillate. The problem of the Zitterbewegung is 
thus resolved.

\medskip

As can be seen in table~\ref{table_so23} commutators of the Dirac matrices 
\( \{ X,Y,Z \} \) of translations are precisely the operators 
\( \{ I,J,K \} \) of rotation. We do not need to {\em interpret} them this
way. That is what they {\em are}. Hence it is immediate and trivial that 
their eigenvalues should describe spin. Justification that these operators
should be interpreted as spin operators and that the equation therefore
describes particles with spin \( \frac{1}{2} \) is much more difficult in
the standard approach to the Dirac equation. 

\medskip

As another example consider Dirac's \( \gamma_5 \) matrix. We can compute 
this now directly from the matrices in table~\ref{sp4R_basis}. Making 
appropriate adjustments for the differing conventions with regard to 
the factor \( \frac{1}{2} \) and the complex unit we obtain the matrix
\begin{equation*}
8TXYZ = \frac{1}{2}\left(\begin{array}{rrrr}
 0 & 0 & 0 & 1\\
 0 & 0 &-1 & 0\\
 0 & 1 & 0 & 0\\
-1 & 0 & 0 & 0
\end{array}\right)
\end{equation*}

This is precisely the matrix \( P_\lambda \) in table~\ref{sp4R_action}.
As we noted on page~\pageref{page:Plambda_inverts}, this matrix can be 
interpreted as a matrix for an inversion of the \( t \), \( x \), \( y \) 
and \( z \) coordinates, which is consistent with the usual understanding 
of the \( \gamma_5 \) operator.

\medskip

Generally using equation~\ref{QDirac} is a much more transparent process
than using the standard Dirac approach. The quantities are well defined with
clear physical meanings, their coordinate dependence is clear, their
transformation properties are known, and we have a rigorous algebraic
framework for working with them.

\bigskip

We finish this section by considering the probability current. We define
\begin{equation}
J^\bullet_k = s^\bullet_{\mu\alpha}\overline{\psi}^\alpha T^\mu_{k\beta}\psi^\beta
\label{bulletJ}
\end{equation}

This is real since 
\begin{equation}
\begin{split}
\overline{J}^\bullet_k
 &= s^\bullet_{\mu\alpha}\psi^\alpha T^\mu_{k\beta}\overline{\psi}^\beta \\
 &= s^\bullet_{\mu\beta}\psi^\alpha T^\mu_{k\alpha}\overline{\psi}^\beta 
	= J^\bullet_k
\end{split}
\end{equation}

Furthermore if \( \psi^\alpha \) is a solution to 
equation~\ref{QDirac2} then
\begin{equation}
\begin{split}
\nabla^k J^\bullet_k &= s^\bullet_{\mu\alpha}\nabla^k\overline{\psi}^\alpha T^\mu_{k\beta}\psi^\beta
			+ s^\bullet_{\mu\alpha}\overline{\psi}^\alpha T^\mu_{k\beta}\nabla^k\psi^\beta \\
                     &= s^\bullet_{\mu\beta} \Bigl( T^\mu_{k\alpha} \nabla^k\overline{\psi}^\alpha \Bigr) \psi^\beta
			+ s^\bullet_{\mu\alpha}\overline{\psi}^\alpha \Bigl( T^\mu_{k\beta}\nabla^k\psi^\beta \Bigr) \\
                     &= \lambda s^\bullet_{\mu\beta} \overline{\psi}^\mu \psi^\beta
			+ \lambda s^\bullet_{\mu\alpha}\overline{\psi}^\alpha \psi^\mu \\
                     &=  0
\end{split}
\end{equation}

Hence \( J^\bullet_k \) is conserved and it would therefore seem reasonable
to identify it as the probability current. It is however a bullet vector and 
not a true vector.  We therefore call \( J^\bullet_k \) the 
\textbf{bullet probability vector}. 

The fact that \( J^\bullet_k \) is a bullet vector and not a vector makes 
it difficult to see how we might use it as the source term for our field 
equations.  In particular this cannot be the \( J_i \) of 
equation~\ref{ScalarFieldEquation}. In our theory bullet vectors and true
vectors are quite distinct and behave differently under parallel transport
and we cannot simply substitute one for the other.

The presence of a bullet index indicates an object with a functional 
dependence on the choice of symplectic form. We can obtain an object 
without bullet indices by specifying a symplectic form. 

A symplectic form is specified by choosing a non-zero Crump scalar 
\( h_\bullet \) which we will call the \textbf{Crump factor} allowing us 
to define \( s_{\alpha\beta} = h_\bullet s^\bullet_{\alpha\beta} \). 
The corresponding probability current vector with respect to the
symplectic form specified by the Crump factor will then be 
\( J_k = h_\bullet J^\bullet_k \).

Of course things are not quite that simple. A probability current vector 
\( J_k \) defined in this way from a solution \( \psi^\alpha \) to 
equation~\ref{QDirac2} will usually not be conserved since
\begin{equation}
\nabla^kJ_k = \nabla^k(h_\bullet J^\bullet_k) = H^kJ_k 
\end{equation}

where \( H^k \) is given by \( \nabla^k(h_\bullet) = H^k h_\bullet \).
This is not surprising since \( J_k \) is a function of the Crump factor
\( h_\bullet \), and \( h_\bullet \) does not appear in equation~\ref{QDirac2}. If we want \( J_k \) to be conserved we would need to modify 
equation~\ref{QDirac2} to explicitly include terms involving \( h_\bullet \).

\bigskip
\pagebreak[2]

\section{The Dirac Lagrangian}

We next seek to obtain the Dirac equation via a Lagrangian argument. This
will provide a further check on the form of the equation. A Lagrangian density
for fermions should also give fermion source terms for our force equations 
when we subject it to variations in the geometry.

\medskip

Most texts on relativistic field theory present a derivation of the 
Dirac equation from a Lagrangian density and, as we shall see, the standard 
Dirac Lagrangian density translates quite naturally into our framework. 
Indeed we will find that some things work considerably better; for one thing
our Dirac Lagrangian is an obviously well defined scalar whereas showing that 
the this is the case for the standard Dirac Lagrangian density requires 
considerable effort.

Another issue with the standard Dirac Lagrangian density is that it is complex.
The translation of this Lagrangian into our context will also be complex.
This is awkward since a complex Lagrangian density makes physical 
interpretation of the Lagrangian density unclear and threatens us with the 
prospect of imaginary source terms for our force equations.

The problem is not specific to our version and afflicts standard Dirac 
Lagrangian theory as well. The obvious cure would be to simply take the 
real or imaginary part. Unfortunately one component generates the Dirac 
equation in the Lagrangian problem while the other is responsible for
generating the electromagnetic source terms under the appropriate variation 
of the geometry. Hence both are needed for a full treatment.

Because this issue is critical for a proper understanding of source terms
we will look at the two components separately to assist us in fully 
understanding these issues.

\bigskip

We begin by seeking a Lagrangian density function \( \LL_D \) for the fermion 
field given by a complex spinor \( \psi^\alpha \). We want our Lagrangian 
density to be a real scalar. 

\bigskip

To obtain a real scalar from a complex spinor we will need to use a 
bilinear form and a conjugation map. The natural invariant bilinear 
form for spinors is the symplectic form \( s^\bullet_{\alpha\beta} \). 
A framework also has a natural conjugation map as discussed in 
section~\ref{complex matters}. This enables us to separate a complex spinor
into real and imaginary parts in a way which is respected by the local
and global actions. 

If we choose a spinor basis consisting of spinors 
which have no imaginary part then natural conjugation corresponds to 
conjugation of coordinates with respect to this basis; in particular in 
such a basis  the components of \( T^\alpha_{i\beta} \) and 
\( \Gamma^\alpha_{k\beta} \) will all be real. 

The natural conjugation map is defined on Crump scalars by requiring that 
the symplectic form be invariant. If we use a real Crump scalar as 
our bullet basis then the natural conjugation map on Crump scalars will
correspond to ordinary conjugation; and in particular the components of 
\( s^\bullet_{\alpha\beta} \) will be real. Using these tools we
can construct
\begin{equation}
s^\bullet_{\alpha\beta}\overline{\psi^\beta}\psi^\alpha
\end{equation}

Conjugating this we obtain
\begin{flalign*}
s^\bullet_{\alpha\beta}\psi^\beta\overline{\psi}^\alpha
= - s^\bullet_{\beta\alpha} \overline{\psi}^\alpha \psi^\beta
= - s^\bullet_{\alpha\beta} \overline{\psi}^\beta\psi^\alpha
\end{flalign*}

hence this is imaginary. Multiplying by \( - i \) to extract the imaginary 
part gives a real Crump scalar.

\bigskip

For a dynamical term, the form of the standard Dirac Lagrangian suggests that 
we should look at 
\begin{equation}
s^\bullet_{\alpha\beta} T^\alpha_{k\lambda} 
\overline{\psi}^\beta \nabla^k \psi^\lambda
\label{proposeddynamic}
\end{equation}

As we will see later this is indeed equivalent to the dynamical term in the
standard Dirac Lagrangian density.  Unfortunately it is complex, just 
like the dynamical term in the standard Dirac Lagrangian density.

Physics texts that comment on this issue may mention that the 
problem can be addressed by taking the real\footnote{or imaginary, 
depending on how things have been defined} part. They will also typically 
explain that this is unnecessary since when the Lagrangian problem is 
solved the integral of the unwanted component reduces via Stokes theorem to 
an integral on the boundary -- a surface term -- which will vanish. 

This reasoning seems odd. Why persist then in using a complex Lagrangian 
density if an adequate real alternative is available? In fact as we will 
see physicists have a very good reason for wanting to hang on to the 
apparently superfluous component since dropping it causes a problem later 
when seeking source terms.

\medskip

We will will use a real Lagrangian density which we will obtain by taking an 
appropriate component -- real or imaginary -- of~\ref{proposeddynamic}. 
Just as it does in the standard approach we will find that this causes
an issue later on when we seek source terms. We could try to sidestep that 
issue at this point by arguing (insincerely) that complex Lagrangian densities 
are not so bad after all, but it is more illuminating to proceed in a 
completely straightforward manner, allow ourselves to run head on into the 
problem, and deal with the issue then. 

To see which component of~\ref{proposeddynamic} we want 
we first look at the conjugate.
\begin{flalign}
s^\bullet_{\alpha\beta} T^\alpha_{k\lambda} 
\psi^\beta \nabla^k \overline{\psi}^\lambda
&= s^\bullet_{\alpha\lambda} T^\alpha_{k\beta}
\nabla^k \overline{\psi}^\lambda \psi^\beta  \\
&= s^\bullet_{\alpha\beta} T^\alpha_{k\lambda}
\nabla^k \overline{\psi}^\beta \psi^\lambda  
\end{flalign}

Hence the real part is
\begin{flalign}
\text{Re}\bigl(s^\bullet_{\alpha\beta} T^\alpha_{k\lambda} 
\overline{\psi}^\beta \nabla^k \psi^\lambda\bigr)
&= \nfrac{1}{2}s^\bullet_{\alpha\beta} T^\alpha_{k\lambda}
\bigl( \overline{\psi}^\beta \nabla^k \psi^\lambda
+ \nabla^k \overline{\psi}^\beta \psi^\lambda   \bigr) \\
&= \nfrac{1}{2}\nabla^k \bigl( 
s^\bullet_{\alpha\beta} T^\alpha_{k\lambda} \overline{\psi}^\beta \psi^\lambda \bigr)
\end{flalign}

This is a divergence and when integrating, Stokes' theorem will allow us to 
convert it to an integral on the boundary\footnote{As it is a Crump 
scalar however, there is a residual term which does not vanish.}. 
This therefore isn't the component we are interested in; we want the other one.
\begin{flalign}
\text{Im}\bigl(s^\bullet_{\alpha\beta} T^\alpha_{k\lambda} 
\overline{\psi}^\beta \nabla^k \psi^\lambda\bigr)
&= - \nfrac{1}{2} i s^\bullet_{\alpha\beta} T^\alpha_{k\lambda}
\bigl( \overline{\psi}^\beta \nabla^k \psi^\lambda
- \psi^\beta \nabla^k\overline{\psi}^\lambda   \bigr) \\
&= - \nfrac{1}{2} i s^\bullet_{\alpha\beta} T^\alpha_{k\lambda}
\bigl( \overline{\psi}^\beta \nabla^k \psi^\lambda
- \nabla^k \overline{\psi}^\beta \psi^\lambda   \bigr)
\end{flalign}

We thus consider the expression
\begin{equation}
\LL_D^\bullet = \text{Im}\bigl(
	s^\bullet_{\lambda\beta} T^\lambda_{k\alpha} 
\overline{\psi^\beta}\nabla^k\psi^\alpha
+ \lambda s^\bullet_{\alpha\beta}
\overline{\psi^\beta}\psi^\alpha
			  \bigr) 
\end{equation}

where \( \lambda \) is a real constant relating the two terms.
This is certainly real, however it is a Crump scalar and a 
Lagrangian density must be a true scalar.

\medskip

To obtain a true scalar from a Crump scalar we must contract with a 
distinguished non-zero Crump scalar \( h_\bullet \) which we call the 
Crump factor.  We can think of \( h_\bullet \) as fixing the choice of 
invariant symplectic form \( h_\bullet s^\bullet_{\alpha\beta} \) at 
every point on the manifold. We will assume for now that the Crump factor
\( h_\bullet \) is real\footnote{the consequences of allowing it to be 
complex should be explored}.
Hence we will use for our Lagrangian density the scalar function
\begin{equation}
\LL_D = \text{Im}\bigl(
	h_\bullet s^\bullet_{\lambda\beta} T^\lambda_{k\alpha} 
		\overline{\psi^\beta}\nabla^k\psi^\alpha
      + h_\bullet \lambda s^\bullet_{\alpha\beta}
		\overline{\psi^\beta}\psi^\alpha
		 \bigr) 
\end{equation}

If we define
\begin{equation} \psi_\nu = h_\bullet s^\bullet_{\nu\beta}\psi^\beta
\end{equation}

this will take the form
\begin{equation}
\LL_D = \text{Im} \bigl(
	\overline{\psi}_\mu T^\mu_{k\nu}\nabla^k\psi^\nu
  + \lambda \overline{\psi}_\nu\psi^\nu 
                  \bigr)
\label{DLagrangian}
\end{equation}

We now can compare this to the usual Dirac Lagrangian. We take the real part
since we wish to compare it to \( \LL_D \) which is real.
\begin{equation}
\label{ULagrangian}
\LL_U =  \text{Re}\bigl( 
	i \psi_\nu \gamma^k\del_k \psi^\nu - m \psi_\nu \psi^\nu 
                \bigr) 
\end{equation}

Care is needed when working with this equation since the notation 
differs from the notation in the rest of the book. With allowances for the 
notation the two expressions appear quite similar. We note the following 
differences.
\begin{enumerate}
\item The standard Lagrangian density is defined on spacetime 
and not the ten dimensional manifold of frames. However the additional 
dimensions in equation~\ref{DLagrangian} disappear if we restrict our 
attention to wave functions with no Lorentz coordinate dependence.
\item Equation~\ref{ULagrangian} uses partial derivatives while 
equation~\ref{DLagrangian} uses covariant ones. This simply means 
that equation~\ref{DLagrangian} is equipped with interaction terms.  
\end{enumerate}

Otherwise the two expressions would appear to differ only in the placement 
of imaginary units and the choice of constant \( \lambda \). To see if 
indeed this is the case we must unwrap the definition of \( \psi_\nu \) 
in equation~\ref{ULagrangian}. This is typically defined by
\begin{equation}
\label{theirdual} 
\psi_\nu = \left(\psi^\nu\right)^\dagger \gamma^0 
\end{equation}

where the dagger operation is the conjugate transpose
and where \( \gamma^0 \) is the Dirac matrix associated with time. 

This definition is problematic for a number of reasons. The conjugate 
transpose of a complex vector depends on the choice of basis and so is not 
well defined. We really should be using the adjoint with respect to an 
invariant inner product of some kind, but there isn't one. Also the 
\( \gamma^0 \) matrix is explicitly linked to the time coordinate. 
Hence \( \LL_U \) would seem very unlikely to be scalar.

Demonstrating that \( \LL_U \) is indeed a well defined scalar is not easy. 
Lengthy arguments in support of this can be found in some of the better physics
texts (for example \cite{Robinson}). Note that by contrast \( \LL_D \) is 
manifestly a well defined scalar. Not only is \( \LL_U \) difficult to define, 
it is also awkward to work with. Fortunately we need only relate it to 
\( \LL_D \).

\medskip
In equation~\ref{gammaidentification} we noted that \( \gamma^0 = - 2iT \) 
where \( T \) is the local action on spinors for translation through time. 
In section~\ref{sectionsp4r} an explicit real matrix is given for \( T \) 
and also for \( \Omega \), the matrix of the symplectic form. Comparing 
these matrices we note that \( 2T = \Omega \). This equality is coincidental 
and basis dependent since \( \Omega \) is the matrix of a bilinear form while 
\( T \) is the matrix of a linear transformation. Nevertheless in this basis 
we can write \( \gamma^0 = - i\Omega \) and have it be true in the sense that 
the matrices are the same. Updating our notation for the symplectic form, 
\( \Omega \) becomes \( h_\bullet s^\bullet_{\alpha\beta} \). Putting all 
this together gives \( \gamma^0 = - ih_\bullet s^\bullet_{\alpha\beta} \). 
Hence
\begin{equation}
\underbrace{
      \psi_\nu = \left(\psi^\nu\right)^\dagger\gamma^0
           }_{\text{standard notation}} 
= 
\underbrace{
	  - i h_\bullet s^\bullet_{\nu\alpha}\overline{\psi}^\alpha  
	= - i \overline{\psi}_\nu
	   }_{\text{our notation}}
\end{equation}

We will also need to substitute for \( \gamma^k \). 
Rewriting~\ref{gammaidentification} in terms of the local action we 
obtain (for the four translation dimensions) 
\begin{equation}
\label{gamma_identification_updated}
\gamma^k = 2i g^{km}T_m\oo
\end{equation}

We can now rewrite \( \LL_U \) in terms of our notation and directly compare 
it to \( \LL_D \). We have
\begin{equation}
\label{NULagrangian}
\begin{split}
\LL_U &=  \text{Re}\bigl( 
	 2i \overline{\psi}_\nu g^{km}T^\nu_{m\lambda} \del_k \psi^\lambda 
      + i m \overline{\psi}_\nu \psi^\nu 
                  \bigr) \\
     &=  \text{Im}\bigl(
       -2 \overline{\psi}_\nu g^{km}T^\nu_{m\lambda} \del_k \psi^\lambda 
         - m \overline{\psi}_\nu \psi^\nu \bigr) \\
     &= -2 \, \text{Im} \bigl(
	\overline{\psi}_\alpha T^\alpha_{k\beta}\del^k\psi^\beta
       + \nfrac{m}{2} \overline{\psi}_\alpha\psi^\alpha \bigr) \\
\end{split}
\end{equation}

Provided therefore that we choose \( \lambda = \frac{m}{2} \), ignore 
interactions by identifying \( \del^k \) with \( \nabla^k \), and restrict 
our attention to spinors with no Lorentz coordinate dependence, we have shown 
\begin{equation} \LL_U = -2\LL_D  \end{equation}

While our proof of this depended on a coincidence in a specific basis,
if both sides are scalar they should be equal in any basis. 

There are some subtleties here however. Our Lagrangian has an explicit 
dependence on the Crump factor \( h_\bullet \) which defines the choice of 
symplectic form. Changing \( h_\bullet \) gives a different Lagrangian and 
since there is no preferred symplectic form there is really no way to choose 
among them. This raises the awkward question of which one of these is the 
one that is supposed to be equivalent to the standard Dirac Lagrangian.

The answer would seem to be that they all are. It looks like the rather
awkward definition of the standard Dirac Lagrangian has a hidden conformal 
degree of freedom arising from the choice of spinor basis and Dirac matrices.
In mathematical technical terms we would say that the standard Dirac 
Lagrangian is not well defined since the definition has a hidden dependence 
on a conformal factor. This could be corrected by making the dependence 
explicit. 

This is not a problem created by our approach. It is a problem with the 
way that the definition of the standard Dirac Lagrangian is stated which 
our notation has simply made more apparent.

We will not consider this issue further here. We simply note at this point
that our Lagrangian is consistent with the standard one and we can therefore 
expect to obtain comparable solutions for the Lagrangian problem, which we 
now seek.

\bigskip

Let \( \Omega \) be a compact set and define 
\( L_D = \int_\Omega \LL_D\,\,dx^\circ \). Consider a variation 
\( \psi^\alpha \mapsto \psi^\alpha + \updelta \psi^\alpha \)
defined on the interior of \( \Omega \) and zero on the boundary, and let 
\( \updelta  L_D = \int_\Omega \updelta \LL_D\,\,dx^\circ \) be the resulting 
change. Then
\begin{equation}
\updelta \LL_D = \text{Im} \bigl(
  \updelta\overline{\psi}_\mu T^\mu_{k\nu} \nabla^k\psi^\nu
+ \overline{\psi}_\mu T^\mu_{k\nu} \nabla^k(\updelta \psi^\nu)
+ \lambda \updelta\overline{\psi}_\nu \psi^\nu
+ \lambda \overline{\psi}_\nu \updelta\psi^\nu
\bigr)
\end{equation}

Using Stokes' theorem we obtain
\begin{equation}
\updelta L_D = \int_\Omega \text{Im} \Bigl(
\bigl( 
	T^\mu_{k\nu} \nabla^k\psi^\nu 
	+ \lambda \psi^\mu 
\bigr) \updelta\overline{\psi}_\mu 
- \bigl( 
        T^\mu_{k\nu} \nabla^k(\overline{\psi}_\mu) 
	- \lambda \overline{\psi}_\nu 
\bigr) \updelta \psi^\nu
\Bigr) \,\,dx^\circ
\label{equation100}
\end{equation}

Now \( \updelta\overline{\psi}_\mu \) is a function\footnote{In no
way are these independent}
of \( \updelta \psi^\nu \). Expressing it as such we have
\begin{equation}
\psi^\mu \updelta\overline{\psi}_\mu 
= \psi^\mu h_\bullet s^\bullet_{\mu\nu}\updelta\overline{\psi}^\nu 
= - \psi_\nu \updelta\overline{\psi}^\nu
\label{term101}
\end{equation}

A similar treatment of the covariant derivative term is complicated by the 
fact that, since \( \nabla^kh_\bullet \ne 0 \)\footnote{else $h_\bullet$ 
would be a true scalar and not a Crump scalar}, spinor index lowering does 
not commute with the covariant derivative. This will generate an extra term. 
We define a vector \( H^k \) by writing \( \nabla^kh_\bullet \) as a multiple 
of \( h_\bullet \)
\begin{equation}
\nabla^k h_\bullet = H^k h_\bullet 
\end{equation}

Note that if we choose the Crump factor \( h_\bullet \) to have a constant 
numerical value in our bullet basis, then \( H^k \) is simply the Crump 
connection. We can then express the extra term cleanly in terms of \( H^k \). 
\begin{equation}
   T^\mu_{k\nu} \nabla^k\psi^\nu \updelta\overline{\psi}_\mu 
= T^\mu_{k\lambda} \nabla^k\psi_\mu \updelta\overline{\psi}^\lambda 
- H^k T^\mu_{k\lambda} \psi_\mu \updelta\overline{\psi}^\lambda 
\label{term102}
\end{equation}
Substituting equations~\ref{term101} and~\ref{term102} 
into equation~\ref{equation100} gives 
\begin{flalign}\nonumber
\updelta L_D = \int_\Omega \text{Im} \Bigl( 
\,\,\, 
\bigl( 
	T^\nu_{k\mu} \nabla^k\psi_\nu 
	- H^k T^\nu_{k\mu} \psi_\nu 
	- \lambda \psi_\mu 
\bigr)&\, \updelta\overline{\psi}^\mu  \\[-4pt]
- \, \bigl( 
	T^\nu_{k\mu} \nabla^k\overline{\psi}_\nu 
	- \lambda \overline{\psi}_\mu 
\bigr)&\, \updelta \psi^\mu
\,\,\, \Bigr) \,\,dx^\circ
\end{flalign}

The imaginary part complicates this expression making it difficult to see
what we must equate to zero.  We use 
	\( \text{Im}(z) = \frac{1}{2i}(z - \overline{z}) \) 
to rewrite it in terms of the conjugate.
\begin{flalign}
\nonumber
\updelta L_D = \nfrac{1}{2i} \int_\Omega  \quad
& \bigl( 
	  2 T^\nu_{k\mu} \nabla^k\psi_\nu 
	-   H^k T^\nu_{k\mu} \psi_\nu 
	- 2 \lambda \psi_\mu 
\bigr)  \updelta\overline{\psi}^\mu  
\\
- & \bigl( 
	  2 T^\nu_{k\mu} \nabla^k\overline{\psi}_\nu 
	-   H^k T^\nu_{k\mu} \overline{\psi}_\nu 
	- 2 \lambda \overline{\psi}_\mu 
\bigr) \updelta \psi^\mu 
\quad
dx^\circ
\label{expressionsinbrackets}
\end{flalign}

and this must be zero for all variations \( \updelta \psi^\alpha \). 
By considering the cases where the variations are real and imaginary 
respectively we see that this can only occur if both of the expressions 
in brackets are zero. As these are conjugate we obtain from this the
single equation
\begin{equation}
T^\nu_{k\mu} \nabla^k\psi_\nu 
	-   \nfrac{1}{2}H^k T^\nu_{k\mu} \psi_\nu 
		- \lambda \psi_\mu = 0
\end{equation}

The equation is expressed in terms of \( \psi_\nu \). Writing it in terms of 
\( \psi^\nu \) gives
\begin{equation}
\label{getDirac}
T^\nu_{k\mu} \nabla^k\psi^\mu 
	+ \nfrac{1}{2}H^k T^\nu_{k\mu} \psi^\mu 
		+ \lambda \psi^\mu  = 0
\end{equation}

Which is the extended Dirac equation, albeit with an extra term. To
understand the nature of this extra term it is helpful to consider 
the more general equation
\begin{equation}
T^\nu_{k\mu} \nabla^k\psi^\mu = 
	\alpha H^k T^\nu_{k\mu} \psi^\mu - \lambda \psi^\nu 
\label{protoDirac}
\end{equation}

which becomes equation~\ref{getDirac} when \( \alpha = -\frac{1}{2} \). By
varying \( \alpha \) we can see what effect this term has. We are 
particularly interested seeing if there is any physical reason why we would
want to have \( \alpha = -\frac{1}{2} \).

The key turns out to be the probability current. We need this to be 
divergence free to ensure that matter is conserved. So how should the
probability current be defined? The usual definition is something like
\begin{equation}
j^k = \psi_\nu \gamma^k \psi^\nu
\label{usualJk}
\end{equation}

which when translated into our notation suggests
\begin{equation}
J_k = h_\bullet s^\bullet_{\lambda\alpha}
	T^\lambda_{k\beta} 
		\overline{\psi}^\alpha\psi^\beta
\label{Jkdefined}
\end{equation}

This is the lowered index version. The upper index version is defined by
\(J^k = g^{ki}J_i \).  We will also find the following related quantities 
useful
\begin{flalign}
J &= h_\bullet s^\bullet_{\alpha\beta}
	\psi^\alpha \overline{\psi}^\beta
\label{Jdefined} \\
J_A &= h_\bullet s^\bullet_{\lambda\alpha}
	T^\lambda_{A\beta} 
		\overline{\psi}^\alpha\psi^\beta
\label{JAdefined}
\end{flalign}

The fact that the Crump factor \( h_\bullet \) appears in the definition of 
\( J_k \) is interesting. It seems that the probability current can only 
be defined relative to a choice of symplectic form, which \( h_\bullet \) 
encodes. The implications of this are unclear and will need to be considered
later. For now let's just see if it is conserved. Assume that \( \psi^\alpha \)
is a solution to equation~\ref{protoDirac} and consider the divergence.
\begin{flalign}
\nabla^kJ_k &=
	\nabla^kh_\bullet 
	s^\bullet_{\lambda\alpha} T^\lambda_{k\beta} 
		\overline{\psi}^\alpha \psi^\beta
+	h_\bullet 
	s^\bullet_{\lambda\alpha} T^\lambda_{k\beta} 
		\nabla^k\overline{\psi}^\alpha \psi^\beta
+	h_\bullet 
	s^\bullet_{\lambda\alpha} T^\lambda_{k\beta} 
		\overline{\psi}^\alpha \nabla^k\psi^\beta
\nonumber \\ = & \,\,
	H^kJ_k
+	h_\bullet s^\bullet_{\lambda\beta} 
	\left( \alpha H^k T^\lambda_{k\alpha} \overline{\psi}^\alpha
		-\lambda \overline{\psi}^\lambda 
	\right) \psi^\beta
+	h_\bullet s^\bullet_{\lambda\alpha} \overline{\psi}^\alpha 
	\left( \alpha H^k T^\lambda_{k\beta} \psi^\beta
		-\lambda \psi^\lambda 
	\right)
\nonumber \\ = & \,\, H^kJ_k 
	+ (  \alpha H^k J_k - \lambda J)
	+ (  \alpha H^k J_k + \lambda J)
\nonumber \\ = & \,\, 	(1+2\alpha) H^kJ_k 
\end{flalign}

So the probability current \( J_k \) is conserved if and only if 
\( \alpha = -\frac{1}{2} \). The extra term in equation~\ref{getDirac} 
therefore ensures that the probability current is conserved, 
which is a very good physical reason for including it.

\bigskip

There are clearly mysteries in this approach. What physical interpretation
should we place on the Crump factor \( h_\bullet \) for example? We included 
this term purely in order to account for the explicit choice of symplectic 
form needed to obtain a true scalar Lagrangian. However our solutions depend 
on it as does the probability current. This suggests it is something more than 
simply an arbitrary choice of coordinates. Things with physical consequences 
are by definition physical things.

\medskip

If \( h_\bullet \) is a physical field of some sort we should look at the 
consequences of varying it on the interior of a compact region \( \Omega \). 
Any such variation \( \updelta h_\bullet \) can be expressed in terms of 
\( h_\bullet \) in the form \( \updelta h_\bullet = \delta . h_\bullet  \) 
where \( \delta \) is a small real scalar field. This will then give
\( \updelta\LL_D = \delta . \LL_D \) and we will have \( \updelta L_D = 0 \) 
for all such \( \delta \) if and only if \( \LL_D = 0 \). But if \( \psi^\mu \)
is a solution to our Dirac equation then
\begin{equation}
\begin{split}
\LL_D &= \text{Im} \Bigl( \overline{\psi}_\mu 
		\bigl( T^\mu_{k\nu}\nabla^k\psi^\nu + \lambda \psi^\mu \bigr)
                  \Bigr) \\
     &= \text{Im} \Bigl( \overline{\psi}_\mu 
		\bigl( 
- \nfrac{1}{2}H^k T^\mu_{k\nu} \psi^\nu \bigr)
                  \Bigr) 
	= 0
\end{split}
\label{LDis}
\end{equation}

so solutions to the Dirac equation minimise the Lagrangian under variation of 
of the Crump factor \( h_\bullet \) as well as under variation of 
\( \psi^\mu \). This is somewhat reassuring, but it doesn't give us a
dynamical equation for \( h_\bullet \). However perhaps we should not
expect this since it is doubtful that \( \LL_D \) is a full and complete 
Lagrangian for \( h_\bullet \) anyway. We note in particular the absence 
of a dynamical term involving \( \nabla_k(h_\bullet) \) which we would 
expect to see in such a Lagrangian. 

\medskip
If The Crump factor \( h_\bullet \) is physical, what physics is being 
described? More precisely, after second quantisation what particle would 
we expect to get from \( h_\bullet \)?  Since \( h_\bullet \) is a (Crump)
scalar it is tempting to think that we might get a Higgs\footnote{hence
the rather hopeful choice of the letter `h'}, but as there is 
no symmetry breaking or Higgs mechanism in sight at the moment this would 
be extremely premature.  In fact it would be a wild guess.

\bigskip

Setting aside these interesting questions, we note that we have now fully 
solved the Lagrangian problem for the extended Dirac Lagrangian. Solutions
are functions which obey equation~\ref{getDirac}, which is the extended
Dirac equation from the previous section with an extra term. This extra
term is necessary to ensure that the probability current is divergence free.
The solution did however require us to introduce a Crump scalar field 
\( h_\bullet \) of uncertain physical significance.

\medskip
  
So far everything has been going reasonably well. In the next section 
the wheels will fall off.

\bigskip
\pagebreak[2]

\section{The Electromagnetic Source Term.}

In this section we look at how the Dirac Lagrangian behaves under the
variations to the geometry discussed starting on page~\pageref{page3cases}.
We expect this to yield source contributions from the fermion field towards 
the field equations for electromagnetism and gravity.

\pagebreak[2]
\bigskip

We begin by considering a Case 1 variation of the type given in 
equation~\ref{VariationCase1}.  Variations of this type gave us the 
Ampere-Gauss equation so we expect to be able to obtain a source term
for that equation from the fermion field. The variation takes the form
\begin{equation*}
\updelta T^\alpha_{k\beta} = 0 
\quad\quad\,
\updelta T^i_{kj} = 0 
\quad\quad\,
\updelta\Gamma^\alpha_{k\beta} = a_k 1^\alpha_\beta 
\quad\quad\,
\updelta\Gamma^i_{kj} = 0 
\quad\quad\,
\end{equation*}

This variation in the geometry does not cause any variation in the 
symplectic form. We are also not varying the matter fields. Hence we 
can add the equations
\begin{equation*}
\updelta \psi^\alpha = 0 
\quad\quad
\updelta h_\bullet = 0 
\quad\quad
\updelta s^\bullet_{\alpha\beta} = 0 
\end{equation*}

We require that \( \nabla_k s^\bullet_{\alpha\beta} = 0 \) is
still true after variation. This determines the variation in the 
Crump connection and hence the variation in \( H^k \).
\begin{equation*}
\updelta\Gamma^\bullet_{k\bullet} = 2a_k
\quad\quad
\updelta H_k =  \Gamma^\bullet_{k\bullet} = 2a_k
\quad\quad
\updelta H^k =  2a^k
\end{equation*}

Under such a variation the only part of the Dirac Lagrangian
to vary is the connection in the covariant derivative term which gives
\begin{equation}
\updelta \nabla^k\psi^\lambda = a^k\psi^\lambda
\end{equation}

Hence the variation in the Dirac Lagrangian is
\begin{equation}
\updelta\LL_D = \text{Im} \bigl(a^k J_k \bigr) = 0
\end{equation}

What?\footnote{The sound of wheels falling off.} 

\bigskip
The problem, which also afflicts the standard Lagrangian, is that while
the imaginary part generates the Dirac equation in the Lagrangian problem, 
it is the real part that gives the expected source terms.  This is the 
reason why physicists are so keen to hang on to both components despite 
the difficulties of being forced to work with a complex Lagrangian density. 
Let's look then at the real part which previously we discarded.
\begin{equation}
\LL_R = \text{Re} \bigl(
	\overline{\psi}_\mu T^\mu_{k\nu}\nabla^k\psi^\nu
  + \lambda \overline{\psi}_\nu\psi^\nu 
                  \bigr)
\label{RLagrangian}
\end{equation}

If \( \lambda \) is real then \( \lambda\overline{\psi}_\nu\psi^\nu \) 
is imaginary. Hence
\begin{equation}
\begin{split}
\LL_R &= \text{Re} 
	\bigl( \overline{\psi}_\mu T^\mu_{k\nu}\nabla^k\psi^\nu \bigr) \\
	&= \nfrac{1}{2}\bigl(\overline{\psi}_\mu T^\mu_{k\nu}\nabla^k\psi^\nu
		+ \psi_\mu T^\mu_{k\nu}\nabla^k\overline{\psi^\nu}\bigr) \\
	&=\nfrac{1}{2}\bigl(\overline{\psi}_\mu T^\mu_{k\nu}\nabla^k\psi^\nu 
		+ \nabla^k\overline{\psi_\mu} T^\mu_{k\nu} \psi^\nu\bigr) \\
	&= \nfrac{1}{2}\bigl(\nabla^k J_k - H^kJ_k\bigr) 
\end{split}
\label{Real_vs_NablaJ}
\end{equation}

Applying a type 1 variation to this we obtain
\begin{equation}
\updelta\LL_R = -a^kJ_k
\end{equation}

which would give us the desired source term of \( J_k \) to within a constant.

\medskip

So should we add the real part back in and follow the crowd by using the 
complex Lagrangian
\begin{equation}
\LL = \overline{\psi}_\mu T^\mu_{k\nu}\nabla^k\psi^\nu
  + \lambda \overline{\psi}_\nu\psi^\nu 
\label{complexlagrangian}
\end{equation}%

That would certainly give us the right source term thus solving our problems 
in this section. But it would break our work in the last section since the 
real part in our construction is not simply a divergence. It is close to a 
divergence as can be seen in equation~\ref{Real_vs_NablaJ}, but the extra 
\( - H^kJ_k \) term is going to change the answer to our Lagrangian problem. 
To see how we must consider what happens to this term under a variation in 
\( \psi^\alpha \).
\begin{equation}
\begin{split}
\updelta J_k &= \updelta\overline{\psi_\mu}T^\mu_{k\nu}\psi^\nu  
             + \overline{\psi_\mu}T^\mu_{k\nu}\updelta\psi^\nu  \\
             &= \psi_\mu T^\mu_{k\nu} \updelta\overline{\psi^\nu}
             + \overline{\psi_\mu}T^\mu_{k\nu}\updelta\psi^\mu 
\end{split}
\end{equation}
and hence
\begin{equation}
\begin{split}
\updelta \LL_R 
	&= \nfrac{1}{2}\bigl(\nabla^k (\updelta J_k) - H^k\updelta J_k\bigr) \\
	&= 
\nfrac{1}{2}
\nabla^k (\updelta J_k) 
- \nfrac{1}{2} H^k \psi_\mu T^\mu_{k\nu} \updelta\overline{\psi^\nu}
- \nfrac{1}{2} H^k \overline{\psi_\mu}T^\mu_{k\nu}\updelta\psi^\mu 
\end{split}
\end{equation}

When we solve the Lagrangian problem for a complex Lagrangian the variations
in both the real and imaginary parts must be zero. The variation in the 
imaginary part gives us the Dirac equation as in the last section. Setting 
the variation in the real part to zero however gives

\begin{equation}
\begin{split}
\updelta L_R &= \int_\Omega \updelta\LL_R \,\, dx^\circ \\
&= \int_\Omega 
- \nfrac{1}{2} H^k \psi_\mu T^\mu_{k\nu} \updelta\overline{\psi^\nu}
- \nfrac{1}{2} H^k \overline{\psi_\mu}T^\mu_{k\nu} \updelta\psi^\mu 
\,\, dx^\circ 
\end{split}
\end{equation}

where the divergence term vanishes using Stoke's theorem. This must be zero 
for all variations. Hence the coefficients of both 
\( \updelta\overline{\psi^\nu} \) and \( \updelta\psi^\mu  \) must be zero.
As these are conjugate we obtain the single equation
\begin{equation} 
H^k \psi_\mu T^\mu_{k\nu} = 0
\end{equation}

which when rewritten in terms of \( \psi^\nu \) gives 
\begin{equation} 
H^k h_\bullet s^\bullet_{\mu\nu} T^\mu_{k\lambda} \psi^\lambda 
= 0
\end{equation}

This is a problematic and severely constraining equation. Since the 
symplectic form \( h_\bullet s^\bullet_{\alpha\beta} \) is non-singular 
this simplifies to the equation
\begin{equation} 
H^k T^\mu_{k\lambda} \psi^\lambda = 0
\label{annoyance}
\end{equation}

which asks that at each point \( \psi^\lambda \) lies in the kernel 
of \( H^kT_k\oo \). Hence in particular at any point where the wave 
function can take non-zero values we must have the vector \( H^k \) 
pointing in a direction where the corresponding local action 
on spinors is singular. While this doesn't rule out the possibility of 
solutions altogether, it comes close. It certainly seems to be a severe
and apparently unphysical constraint on the possible values of both 
\( H^k \) and \( \psi^\alpha \). 

The solution to the Lagrangian problem for the complex
Lagrangian density specified in equation~\ref{complexlagrangian} is
thus any solution to the two equations
\begin{equation}
\begin{split}
& T^\nu_{k\mu} \nabla^k\psi^\mu + \lambda \psi^\mu  = 0 \\
&	H^k T^\nu_{k\mu} \psi^\mu =0
\end{split}
\label{CDiracsolutionsystem}
\end{equation}

Note that for such a solution we have \( \LL = 0 \). While the first of these
equations seems quite reasonable, the second does not. 

The second equation involves the vector \( H_k \) arising from the Crump factor
\( h_\bullet \). Since we doubt that our Lagrangian density is a full 
and complete Lagrangian density for \( h_\bullet \) (it lacks a dynamical 
term) perhaps the cure might be to simply correct this by adding on an
additional term. This would add extra terms to our solution which could help. 
How would this help? Note that equation~\ref{annoyance} would not be nearly 
so problematic were it to include an extra term involving \( \nabla_k H^k \). 

So what might this missing dynamical term look like? Note that since 
\( \nabla^k(h_\bullet) = H^kh_\bullet \) dynamical terms can be written in 
terms of \( H^k \).  The following seem plausible and worthy of investigation.
\begin{flalign}
& H^k\overline{\psi}_\nu T^\nu_{k\mu}\psi^\mu \\
& H^k\overline{\psi}_\nu\nabla_k\psi^\nu 
\end{flalign}

As both are complex we would also need to also worry about 
them generating imaginary source terms for our force equations.
And as we are treating the Crump factor as a physical quantity 
we should also consider the effect of varying \( h_\bullet \).

Would such terms break the link to the standard Dirac Lagrangian?
The standard Dirac Lagrangian does not mention the Crump factor 
\( h_\bullet \) (effectively setting it equal to 1), and ignores 
the conformal degree of freedom that it represents.  Hence terms 
which involve \( H^k \) in our formulation are invisible to the 
standard approach. That means we are free to include such terms
without having to worry about breaking the link to the standard 
Dirac Lagrangian.

\bigskip

Our efforts to derive both a Dirac equation and source terms for our forces 
from a plausible Dirac Lagrangian have thus run into difficulty. We have not 
run out of options, but there are no clear answers. What can we conclude 
from this?

Firstly we should note that the problem we have run into here is not a 
consequence of our method. The same issues exists also with the standard 
Dirac Lagrangian although it is considerably harder to see them in that 
context.  Our more rigorous definitions and clearer notation have simply 
exposed an issue with the Dirac Lagrangian that was previously obscure.

Secondly we should note that being unable to find a Lagrangian that 
behaves as we would like is not a fatal flaw. The great weakness of 
the Lagrangian method has always been the lack of a rigorous procedure 
for determining the Lagrangian. We don't know what the Lagrangian 
should be and have no method other than guessing to find it. We don't 
even know that a correct Lagrangian exists; the physics we are seeking 
to describe need not arise from a Lagrangian at all.

\bigskip
\pagebreak[2]

\section{The Gravitational Source Term.}

In this section we will apply type 2 and type 3 variations to the complex 
Lagrangian density \( \LL \) from equation~\ref{complexlagrangian} to see 
what sort of source terms it generates. This Lagrangian is equivalent to 
the usual one and delivers the correct source term for electromagnetism. 
However solutions to the Lagrangian problem for this Lagrangian density must
satisfy two equations~\ref{CDiracsolutionsystem}. And while the first of 
these is the expected generalisation of the Dirac equation, the second 
appears unphysical. 

Consequently we are not confident that \( \LL \) is the correct Lagrangian 
density. That means we cannot place much weight on the answers obtained in 
this section. They should be viewed as qualitative indicators of the 
possible nature of the source term for gravity only.

\bigskip

We first look at a Case 2 variation to the geometry. The Lagrangians \( L_g \)
and \( L_e \) were both found to be invariant under this this type of variation 
and no force equations were obtained as a result. Since we don't have any 
equations in need of source terms we rather hope we don't obtain any by 
applying this type of variation to the complex Dirac Lagrangian \( L \) 
since we wouldn't know what to do with them.  

\medskip

Case 2 variations are specified by  \( \delta^\alpha_\beta \) with 
\( \delta^a_b = 0 \) and take the form
\begin{equation*}
\updelta T^\alpha_{k\beta} = \delta\oo T^\alpha_{k\beta}
\quad\quad\,
\updelta T^i_{kj} = 0  = \delta\oo T^i_{kj}
\quad\quad\,
\updelta\Gamma^\alpha_{k\beta} = - \nabla_k(\delta^\alpha_\beta)
\quad\quad\,
\updelta\Gamma^i_{kj} = 0 
\end{equation*}

The discussion starting on page~\pageref{Case2g} adds the equations
\begin{equation*}
\updelta g^{ij} = 0 = \delta\oo g^{ij} 
\quad\quad
\updelta s^\bullet_{\alpha\beta} = \delta\oo s^\bullet_{\alpha\beta} 
\end{equation*}

where an arbitrary  \( \delta^\bullet_\bullet \) is used to define
the variation in \( s^\bullet_{\alpha\beta} \). As noted on 
page~\pageref{Case2arbitrarybulletNote}, \( \delta^\bullet_\bullet \) 
can take any value here including zero. 

\medskip
Assume for a moment that in addition to the variation in the geometry we 
apply compatible variations to the wave function \( \psi^\alpha \) and 
the Crump factor \( h_\bullet \) given by
\begin{equation}
\updelta \psi^\alpha = 
	\delta\oo \psi^\alpha = 
		\delta^\alpha_\lambda \psi^\lambda
\quad\quad
\updelta h_\bullet = \delta\oo h_\bullet
\label{Case2compatible}
\end{equation}

Then this suite of variations gives
\begin{equation}
\updelta\bigl(\nabla_k\psi^\alpha\bigr) = 
	\delta\oo\bigl( \nabla_k\psi^\alpha\bigr)
\end{equation}

All components in the complex Dirac Lagrangian \( \LL \) vary according to the
tensor derivation \( \delta\oo \), and as \( \LL \) is scalar it follows
that \( \updelta\LL = \delta\oo \LL = 0 \).

\medskip

So variations of this type are non-trivial only to the extent that the
variations in \( \psi^\alpha \) and \( h_\bullet \) differ from the 
compatible variations specified in equation~\ref{Case2compatible}. And
to the extent that they do differ from those compatible variations they 
will simply give solutions to the Dirac Lagrangian problem. In particular, 
as expected, no source terms arise from applying this type of variation 
to the Dirac Lagrangian.

\bigskip

Finally we consider a Case 3 variation to the geometry. Case 3 variations 
are the most difficult and are responsible for the gravitational equations. 
They are specified by \( b^i_j \) and \( \delta^i_j \) and give
\begin{equation*}
\updelta T^\alpha_{k\beta} = -\delta^m_k T^\alpha_{m\beta} 
\quad\quad
\updelta \Gamma^\alpha_{k\beta} = b^t_kT^\alpha_{t\beta} 
\quad\quad
\updelta g^{ij} = \delta^i_mg^{mj} + \delta^j_mg^{mi}
\quad\quad
\updelta s^\bullet_{\alpha\beta} = 0 
\end{equation*}

Assuming also that 
\( \updelta \psi^\alpha = 0 \) and \( \updelta h_\bullet = 0  \)
we obtain
\begin{equation}
\updelta \LL =
\overline{\psi}_\nu T^\nu_{i\lambda}g^{ij}T^\lambda_{k\mu} \psi^\mu
b^k_j +
\overline{\psi}_\nu T^\nu_{i\mu}g^{ia}\nabla_b \psi^\mu
 \delta^b_a
\end{equation}

We can now use equation~\ref{substituter} to replace \( b^k_j \) with 
a function of \( \delta^b_a \). Equation~\ref{substituter} is written 
in terms of quantities with lowered indices because that form was easiest
to obtain. Putting indices back in standard position we have
\begin{equation}
\begin{split}
  b^k_j &= \nfrac{1}{6} \bigl(
  1^m_j T^a_{nb} g^{kn} - 1^a_j T^m_{nb} g^{kn} + T^a_{nt}g^{tm} g^{kn} g_{bj} 
	                \bigr)  \nabla_m\delta^b_a 
\\&\quad
+ \bigl( \nfrac{1}{2}.1^k_b1^a_j + g^{ak} g_{bj} 
- \nfrac{1}{6}\ T^t_{jb} T^a_{nt} g^{kn} \bigr) \delta^b_a
\end{split}
\label{resubstituter}
\end{equation}

We are interested in \( \updelta L \) which 
is related to \( \updelta\LL \) via
\begin{equation}
\updelta L = \int_\Omega \updelta\LL\,dx^\circ 
	+ \int_\Omega \LL \, \updelta dx^\circ 
\end{equation}

When \( \psi^\alpha \) is a solution to the Dirac Lagrangian problem then
\( \LL = 0 \) as noted on page~\pageref{CDiracsolutionsystem}. Hence the last
integral is zero and can be ignored. Hence, substituting for \( b^k_j \),
we obtain

\begin{flalign*}
\updelta L &= 
\nfrac{1}{6} \int_\Omega 
\overline{\psi}_\nu T^\nu_{i\lambda}g^{ij}T^\lambda_{k\mu} \psi^\mu
\bigl( 
	1^m_j T^a_{nb} g^{kn} - 1^a_j T^m_{nb} g^{kn} 
	+ T^a_{nt}g^{tm} g^{kn} g_{bj} 
\bigr)  
\nabla_m\delta^b_a 
\,dx^\circ 
\\ & \quad +
\int_\Omega 
\overline{\psi}_\nu T^\nu_{i\lambda}g^{ij}T^\lambda_{k\mu} \psi^\mu
\bigl( 
	\nfrac{1}{2}.1^k_b1^a_j + g^{ak} g_{bj} 
	- \nfrac{1}{6}\ T^t_{jb} T^a_{nt} g^{kn} 
\bigr) 
\delta^b_a
\,dx^\circ 
\\ & \quad +
\int_\Omega 
\overline{\psi}_\nu T^\nu_{i\mu}g^{ia}\nabla_b \psi^\mu
 \delta^b_a
\,dx^\circ 
\end{flalign*}

The next step is to apply Stoke's theorem to the first integral.  We obtain

\begin{flalign*}
\updelta L &= 
- \nfrac{1}{6} \int_\Omega 
T^\nu_{i\lambda} T^\lambda_{k\mu}
\bigl( 
    g^{im} T^a_{nb} g^{kn} 
  - g^{ia} T^m_{nb} g^{kn} 
  + 1^i_b T^a_{nt}g^{tm} g^{kn}
\bigr)  
\nabla_m \bigl( \overline{\psi}_\nu \psi^\mu \bigr)
\delta^b_a 
\,dx^\circ 
\\ & \quad +
\int_\Omega 
T^\nu_{i\lambda} T^\lambda_{k\mu}
\bigl( 
    \nfrac{1}{2} g^{ai} 1^k_b 
  + g^{ak} 1^i_b 
  - \nfrac{1}{6} g^{ij} T^t_{jb} T^a_{nt} g^{kn} 
\bigr) 
\overline{\psi}_\nu \psi^\mu
\delta^b_a
\,dx^\circ 
\\ & \quad +
\int_\Omega 
\overline{\psi}_\nu T^\nu_{i\mu}g^{ia}\nabla_b \psi^\mu
 \delta^b_a
\,dx^\circ 
\end{flalign*}

The coefficient of \( \delta^b_a \) will be our source term. 

\begin{flalign*}
S^a_b &=
- \nfrac{1}{6} 
T^\nu_{i\lambda}T^\lambda_{k\mu}
\bigl( 
    g^{im} T^a_{nb} g^{kn} 
  - g^{ia} T^m_{nb} g^{kn} 
  + 1^i_b T^a_{nt}g^{tm} g^{kn}
\bigr)  
\nabla_m \bigl( \overline{\psi}_\nu \psi^\mu \bigr)
\\ & \quad  +
T^\nu_{i\lambda}T^\lambda_{k\mu}
\bigl( 
    \nfrac{1}{2} g^{ai} 1^k_b 
  + g^{ak} 1^i_b 
  - \nfrac{1}{6} g^{ij} T^t_{jb} T^a_{nt} g^{kn} 
\bigr) 
\overline{\psi}_\nu \psi^\mu
\\ & \quad +
\bigl(\overline{\psi}_\nu T^\nu_{i\mu} \nabla_b \psi^\mu\bigr) g^{ia}
\end{flalign*}

We next apply equation~\ref{ProductofTs}, which states
\begin{equation*}
T^\nu_{i\lambda}T^\lambda_{k\mu} = 
\nfrac{1}{2}T^c_{ik}T^\nu_{c\mu} + 
\nfrac{1}{4}g_{ik}1^\nu_\mu + \nfrac{1}{2}g^A_{ik}T^\nu_{A\mu}
\end{equation*}

to obtain the equation

\begin{flalign}
S^a_b &= \nonumber
- \nfrac{1}{6} 
\nabla_m\bigl(\overline{\psi}_\nu \psi^\mu \bigr)
\bigl(\nfrac{1}{2}T^c_{ik}T^\nu_{c\mu} + 
\nfrac{1}{4}g_{ik}1^\nu_\mu + \nfrac{1}{2}g^A_{ik}T^\nu_{A\mu}\bigr)\\
& \quad\quad\quad\quad\quad\quad\quad\quad\quad\quad\nonumber
\bigl( 
g^{im} T^a_{nb} g^{kn} - g^{ia} T^m_{nb} g^{kn} 
+ 1^i_b T^a_{nt}g^{tm} g^{kn} 
\bigr)  
\\ & \quad \nonumber  +
\overline{\psi}_\nu \psi^\mu
\bigl(\nfrac{1}{2}T^c_{ik}T^\nu_{c\mu} + 
\nfrac{1}{4}g_{ik}1^\nu_\mu + \nfrac{1}{2}g^A_{ik}T^\nu_{A\mu}\bigr)
\bigl( 
\nfrac{1}{2} g^{ai} 1^k_b + g^{ak} 1^i_b 
- \nfrac{1}{6} g^{ij} T^t_{jb} T^a_{nt} g^{kn} 
\bigr) 
\\ & \quad \label{19terms} +
\bigl(\overline{\psi}_\nu T^\nu_{i\mu} \nabla_b \psi^\mu\bigr) g^{ia}
\end{flalign}

We now expand this out (there will be nineteen terms) and simplify. 
The Dirac equations~\ref{CDiracsolutionsystem} can be used to simplify 
some of the terms. We hope to express our answer in terms of
the probability currents
\begin{flalign}
J   &= \overline{\psi}_\nu \psi^\nu \\
J^k &= g^{ki}\overline{\psi}_\nu T^\nu_{i\mu} \psi^\mu \\
J_A &= g^{AB}\overline{\psi}_\nu T^\nu_{B\mu} \psi^\mu 
\end{flalign}
and their covariant derivatives. Note that \( J \) and \( J^A \) are 
imaginary while \( J^k \) is real. Since a source term for gravity should be
real we would hope that our imaginary terms cancel.

\bigskip

The first term in equation~\ref{19terms} is
\begin{equation*}
- \nfrac{1}{12} \nabla_m\bigl(\overline{\psi}_\nu \psi^\mu \bigr)
\bigl(T^c_{ik}T^\nu_{c\mu} + \nfrac{1}{2}g_{ik}1^\nu_\mu 
+ g^A_{ik}T^\nu_{A\mu}\bigr)\\
\bigl( g^{im} T^a_{nb} - g^{ia} T^m_{nb} + 1^i_b T^a_{nt}g^{tm}
\bigr)  g^{kn} 
\end{equation*}

Expanding will give nine terms, which we will simplify and write 
in terms of \( J \), \( J^A \) and \( J^k \). We will also use
\( g^a_{Ab} = g_{AK}g^K_{ib}g^{ia} \) to simplify some of the answers. 
The nine terms are
\begin{flalign}
- \nfrac{1}{12} \nabla_m\bigl(\overline{\psi}_\nu \psi^\mu \bigr)
T^c_{ik}T^\nu_{c\mu} g^{im} T^a_{nb} g^{kn} &=
\nfrac{1}{12} T^a_{bn} T^n_{ik} \nabla^i J^k \\
\nfrac{1}{12} \nabla_m\bigl(\overline{\psi}_\nu \psi^\mu \bigr)
T^c_{ik}T^\nu_{c\mu} g^{ia} T^m_{nb} g^{kn} &=
\nfrac{1}{12} T^a_{kn} T^n_{bi} \nabla^i J^k \\
- \nfrac{1}{12} \nabla_m\bigl(\overline{\psi}_\nu \psi^\mu \bigr)
T^c_{ik}T^\nu_{c\mu} 1^i_b T^a_{nt}g^{tm} g^{kn}  &=
\nfrac{1}{12} T^a_{in} T^n_{kb} \nabla^i J^k  \\
- \nfrac{1}{24} \nabla_m\bigl(\overline{\psi}_\nu \psi^\mu \bigr)
g_{ik}1^\nu_\mu g^{im} T^a_{nb} g^{kn} &=
\nfrac{1}{24} T^a_{bi} \nabla^i J \\
\nfrac{1}{24} \nabla_m\bigl(\overline{\psi}_\nu \psi^\mu \bigr)
g_{ik}1^\nu_\mu g^{ia} T^m_{nb} g^{kn} &=
\nfrac{1}{24} T^a_{bi} \nabla^i J \\
- \nfrac{1}{24} \nabla_m\bigl(\overline{\psi}_\nu \psi^\mu \bigr)
g_{ik}1^\nu_\mu 1^i_b T^a_{nt}g^{tm} g^{kn}  &=
\nfrac{1}{24} T^a_{ib} \nabla^i J \\
- \nfrac{1}{12} \nabla_m\bigl(\overline{\psi}_\nu \psi^\mu \bigr)
g^A_{ik}T^\nu_{A\mu} g^{im} T^a_{nb} g^{kn} &=
\nfrac{1}{12} T^a_{bn} g^n_{Ai} \nabla^i J^A \\
\nfrac{1}{12} \nabla_m\bigl(\overline{\psi}_\nu \psi^\mu \bigr)
g^A_{ik}T^\nu_{A\mu} g^{ia} T^m_{nb} g^{kn}  &=
\nfrac{1}{12} T^n_{bi} g^a_{An} \nabla^i J^A \\
- \nfrac{1}{12} \nabla_m\bigl(\overline{\psi}_\nu \psi^\mu \bigr)
g^A_{ik}T^\nu_{A\mu} 1^i_b T^a_{nt}g^{tm} g^{kn} &=
\nfrac{1}{12} T^a_{in} g^n_{Ab} \nabla^i J^A
\end{flalign}

The first three terms cancel via the Jacobi identity. Collecting up
the remaining terms we get
\begin{equation}
\nfrac{1}{24} T^a_{bi} \nabla^i J 
+ \nfrac{1}{12} 
\bigl( T^a_{bn} g^n_{Ai} + T^n_{bi} g^a_{An} + T^a_{in} g^n_{Ab} \bigr)
\nabla^i J^A
\end{equation}

These are all imaginary.

\bigskip

\medskip
The second term in equation~\ref{19terms} is
\begin{equation*}
 \overline{\psi}_\nu \psi^\mu
\bigl(\nfrac{1}{2}T^c_{ik}T^\nu_{c\mu} + 
\nfrac{1}{4}g_{ik}1^\nu_\mu + \nfrac{1}{2}g^A_{ik}T^\nu_{A\mu}\bigr)
\bigl( 
\nfrac{1}{2} g^{ai} 1^k_b + g^{ak} 1^i_b 
- \nfrac{1}{6} g^{ij} T^t_{jb} T^a_{nt} g^{kn} 
\bigr) 
\end{equation*}

Expanding will give nine terms which simplify as follows.
\begin{flalign}
\nfrac{1}{4}T^c_{ik}T^\nu_{c\mu} g^{ai} 1^k_b 
(\overline{\psi}_\nu \psi^\mu) &=
- \nfrac{1}{4} T^a_{kb} J^k
\\
\nfrac{1}{2}T^c_{ik}T^\nu_{c\mu} g^{ak} 1^i_b 
(\overline{\psi}_\nu \psi^\mu) &=
\noneg\nfrac{1}{2} T^a_{kb} J^k 
\\
-\nfrac{1}{12}T^c_{ik}T^\nu_{c\mu} g^{ij} T^t_{jb} T^a_{nt} g^{kn} 
(\overline{\psi}_\nu \psi^\mu) &=
\noneg\nfrac{1}{12}T^j_{ik} T^t_{bj} T^k_{nt} g^{an} J^i \\ 
%
%
\nfrac{1}{8}g_{ik}1^\nu_\mu g^{ai} 1^k_b 
(\overline{\psi}_\nu \psi^\mu) &=
\noneg\nfrac{1}{8} 1^a_b J \\
\nfrac{1}{4}g_{ik}1^\nu_\mu g^{ak} 1^i_b 
(\overline{\psi}_\nu \psi^\mu) &= 
\noneg\nfrac{1}{4} 1^a_b J \\
- \nfrac{1}{24}g_{ik}1^\nu_\mu g^{ij} T^t_{jb} T^a_{nt} g^{kn} 
(\overline{\psi}_\nu \psi^\mu) &=
- \nfrac{1}{4} 1^a_b J \\
\nfrac{1}{4}g^A_{ik}T^\nu_{A\mu} g^{ai} 1^k_b 
(\overline{\psi}_\nu \psi^\mu) &=
\noneg\nfrac{1}{4}g^a_{Ab} J^A \\
\nfrac{1}{2}g^A_{ik}T^\nu_{A\mu} g^{ak} 1^i_b 
(\overline{\psi}_\nu \psi^\mu) &=
\noneg\nfrac{1}{2}g^a_{Ab} J^A \\
- \nfrac{1}{12}g^A_{ik}T^\nu_{A\mu} g^{ij} T^t_{jb} T^a_{nt} g^{kn} 
(\overline{\psi}_\nu \psi^\mu) &=
\noneg\nfrac{1}{12} T^t_{jb} T^k_{nt} g^{an} g^j_{Ak} J^A 
%
\end{flalign}

Collecting up terms we obtain two real terms
\begin{equation}
\nfrac{1}{4} T^a_{kb} J^k
+ \nfrac{1}{12}T^j_{ik} T^t_{bj} T^k_{nt} g^{an} J^i 
\end{equation}

and three imaginary terms
\begin{equation}
\nfrac{1}{8} 1^a_b J 
+ \nfrac{3}{4}g^a_{Ab} J^A 
+ \nfrac{1}{12} T^t_{jb} T^k_{nt} g^{an} g^j_{Ak} J^A 
\end{equation}

\bigskip

Finally the  last term of equation~\ref{19terms} gives
\begin{equation}
\bigl(\overline{\psi}_\nu T^\nu_{i\mu} \nabla_b \psi^\mu\bigr) g^{ia}
\end{equation}

The real part of this term can be written in terms of current vectors in 
the form \( \nfrac{1}{2}\nabla_bJ^a - H_bJ^a \).  However the imaginary
part cannot be so simply expressed.

\bigskip

Assembling all the components, the real part of the source term is 
\begin{equation}
\nfrac{1}{4} T^a_{kb} J^k
+ \nfrac{1}{12}T^j_{ik} T^t_{bj} T^k_{nt} g^{an} J^i 
+ \nfrac{1}{2}\nabla_bJ^a - H_bJ^a 
\end{equation}

and the imaginary part is 
\begin{multline}
\nfrac{1}{24} T^a_{bi} \nabla^i J 
+ \nfrac{1}{12} 
\bigl( T^a_{bn} g^n_{Ai} + T^n_{bi} g^a_{An} + T^a_{in} g^n_{Ab} \bigr)
\nabla^i J^A \\
+ \nfrac{1}{8} 1^a_b J 
+ \nfrac{3}{4}g^a_{Ab} J^A 
+ \nfrac{1}{12} T^t_{jb} T^k_{nt} g^{an} g^j_{Ak} J^A  \\
+ \text{Im}\bigl(\overline{\psi}_\nu T^\nu_{i\mu} \nabla_b \psi^\mu\bigr) g^{ia}
\end{multline}

The gravitational equations we are using expect a real source term so this
imaginary part is problematic. It seems very unlikely to be zero. The terms 
are sufficiently diverse in nature that it would be extraordinary if via some 
prodigious feat of algebra-fu we could get them all to cancel. The existence 
of these imaginary gravitational source terms is one more reason to doubt
that our Lagrangian density is correct.

\bigskip
\pagebreak[2]

\section{Discussion}

This has been an interesting chapter. Some things have worked very well. 
Some things have proved more difficult. It is therefore worthwhile to sum 
up the situation with regard to the Dirac equation in our model.

\medskip

We began the chapter by observing that the Dirac equation could be very simply
and nicely incorporated into our framework. The Dirac matrices themselves 
turned out to be simply intrinsic translation operators and the Dirac operator 
could therefore be viewed as a Curl operator. 

Adopting this viewpoint simplifies and clarifies many things about the 
Dirac matrices that are otherwise obscure. For example the $S$-matrix 
transformations of the Dirac matrices become simply the expected change 
of coordinate behavior for the intrinsic translation operators. That products
of Dirac matrices give rotation operators is now a simple observation that 
can be directly calculated. The Zitterbewegung also has a simple explanation
since the velocity being described is intrinsic.

Our extended Dirac equation involved wave functions defined on all ten 
dimensions of the manifold of frames. For those functions which are constant
across the six Lorentz dimensions the equation reduces to the ordinary four 
dimensional Dirac equation, particularly in the Poincar\'e limit. Presumably
these functions then represent electrons.

Functions which are not constant across the six Lorentz dimensions must
therefore represent other types of fermion. Since the three rotation 
dimensions are compact we expect solutions across these dimensions to be 
discretely quantised. This opens up the possibility of a fermion field
describing all fermions with a single spinor function.

\medskip

The Dirac Lagrangian also was much more easily understood in our context.
The version in our framework was also easier to work with. In particular 
it was obviously a well defined scalar, something which requires considerable 
work in the standard approach. We were able to show that, under the assumption 
that the standard Dirac Lagrangian density is scalar, it was equal to our 
version for functions \( \psi^\alpha \) constant across the Lorentz 
dimensions.

However problems started to appear at this point. Our clearer notation 
revealed two difficulties with the Dirac Lagrangian which exist also within
the standard theory but which are usually obscured by the notation.

\medskip

Firstly, our Lagrangian depended on a Crump scalar \( h_\bullet \) 
encoding the choice of symplectic form 
\( s_{\alpha\beta} = h_\bullet s^\bullet_{\alpha\beta} \). This prompted 
us to ask, for which choice of \( h_\bullet \) was our Lagrangian equal 
to the standard one? Our proof, which assumed only that the standard 
Lagrangian was a well defined scalar, seemed to work for all of them. Looking 
more closely at the effect of changing the Crump factor \( h_\bullet \) on 
the correspondence we found that a conformal degree of freedom equivalent to 
a choice of \( h_\bullet \) also exists within the standard Dirac Lagrangian 
arising from a combination of choice of Dirac matrices and a conformal change 
of spinor basis. Mathematically therefore the standard Dirac Lagrangian 
density is not well defined. To correct this problem the conformal degree 
of freedom in its definition should be made explicit.

In a sense we can view uncovering this issue as a positive for our 
framework. While we might not be happy about discovering a problem, 
we should credit the added clarity of our notation and approach that 
enabled us to discover it.

\medskip

The second obvious problem with the Dirac Lagrangian density is the fact that 
it is complex. A Lagrangian density strictly speaking should be real. Allowing
it to be complex makes it difficult to attach to it a reasonable physical 
interpretation, and threatens to generate complex source terms for our force
equations. This is a thoroughly explored issue with the standard Dirac 
Lagrangian where the use of a complex Lagrangian density is usually justified 
on the basis that it gives the expected answers\footnote{physics generally has 
an `ends justifies the means' attitude toward mathematics}. Complex 
Lagrangian densities are therefore now commonly accepted in modern physics.

The way the theory is normally presented the imaginary component gives the 
Dirac equation when solving the Lagrangian problem but makes a zero 
contribution towards source terms for forces, which is fortunate since we
we wouldn't know what to do with an imaginary source term. The real 
component on the other hand is a divergence and has no effect on the 
Lagrangian problem, but gives the correct real source terms. The complex 
Lagrangian combines both components and thus gives both a Dirac equation 
and also the expected source terms.

Unfortunately when we translate all this to our framework things don't work 
out quite so well. The Crump factor \( h_\bullet \), needed to explicitly 
deal with the conformal degree of freedom gives us a real part that is no 
longer a divergence. Hence including the real part changes the solution to 
the Lagrangian problem by introducing an additional constraint, 
equation~\ref{annoyance}, which appears unphysical. It is possible that 
this extra constraint may be more palatable if a dynamical term for 
\( h_\bullet \) is included in our Lagrangian. However we did not pursue 
this idea further.

This same issue must exist even in standard Dirac Lagrangian theory and should 
become apparent there if the conformal degree of freedom described in our 
notation by \( h_\bullet \) is made explicit and treated rigorously.
Hence it is not simply the case that Dirac Lagrangian theory failed to 
work in our framework. It looks like it has never actually worked properly 
even in the standard case. It only seems to work because the parts that 
are broken are hidden by the use of inadequate notation.

\medskip

Finally we found that the much more complicated variations which generate 
gravity in our framework do not act trivially on the imaginary component
of the Dirac Lagrangian. This means they generate imaginary source terms 
for gravity which would break our gravitational equations. This is one more 
reason to doubt that we have the correct Lagrangian density.

\bigskip

Not being able to find the correct Lagrangian density does not invalidate
or disprove the framework approach. Indeed what our approach has done
is clarified the situation revealing problems that were previously hidden and
pointing the way towards their solution. So where should we look to solve
these issues?

One bright light in the darkness here is that using the mathematical tools 
in our framework makes it very easy to construct scalar functions to use as 
candidate Lagrangians. We already have noted possible dynamical terms for
\( h_\bullet \) and the effect of adding these should be explored.

Furthermore the curvature acts on the spinors which opens up the possibility 
of Lagrangians that combine both curvature and the wave functions in more 
interesting ways. Consider for example the following constructions.

\begin{flalign}
& 
g^{ia}g^{jb} 
\overline{\psi}_\nu 
R^\nu_{ij\lambda}
 R^\lambda_{ab\mu}\psi^\mu 
\\ &
h_\bullet s^\bullet_{\alpha\beta} 
g^{ia}g^{jb}
R^\alpha_{ij\mu}\overline{\psi}^\mu
R^\beta_{ab\nu}\psi^\nu 
\\ &
g^{ia}g^{jb} 
[\nabla_i,\nabla_j]\overline{\psi}_\nu\,
[\nabla_a,\nabla_b]\psi^\nu 
\\ &
h_\bullet s^\bullet_{\alpha\beta} 
g^{ia}g^{jb}
[\nabla_i,\nabla_j]\overline{\psi}^\nu\,
[\nabla_a,\nabla_b]\psi^\nu 
\end{flalign}

All are trivially scalar and look like reasonable candidates for applying
the Lagrangian method. They are complex, but one could correct that by
taking their real or imaginary components. Indeed our problem is 
not that it is difficult to construct interesting looking scalar functions, 
rather it is all too easy.  What we really need is more thought as to the 
physical meaning of the Lagrangian density which might restrict our choices
and help guide us to the appropriate expression.

	\chapter{Summary}

It is time for this book to end, although obviously there is a huge amount 
remaining to be done. In this chapter we will reflect on what we set out to
do; on what we actually achieved and its implications; and on the loose ends
and unanswered questions that should be the subject of future research.

\medskip
\pagebreak[2]

\section{Objectives}

We began by noting, as a great many others have noted, that the anti-deSitter 
group \( \SO(2,3) \) seems to work better in some physical theories than the 
Poincar\'e group. Since the one contracts into the other, so long as the 
contraction parameter \( r \) is large enough we cannot distinguish between the
two groups from the point of view of classical physics and the extrinsic 
action. However the intrinsic action of \( \so(2,3) \) gives realistic quantum 
numbers which the Poincar\'e Lie algebra does not.

If indeed \( \SO(2,3) \) is the correct symmetry group for physics then we have 
a problem because all the mathematical tools that we use to do physics 
were built originally in Euclidean space and naturally express Euclidean 
symmetry. A manifold for example is defined via an atlas of homeomorphisms 
into Euclidean space; our calculus and linear algebra were all initially 
defined on Euclidean space; and our notions of curvature measure 
the departure from the properties of flat Euclidean space.

Of course the tools of  modern mathematics are very flexible and it is 
certainly possible to add \( \so(2,3) \) symmetry in a variety of ways 
to any given mathematical construct.  However it isn't clear how to do 
this consistently for all the different types of mathematical structures 
we need in modern physics. 

What we really need are mathematical tools that have the symmetry group 
\( \so(2,3) \) baked into them from the start, so that every natural 
structure expressed using these tools immediately manifests this symmetry 
group, both extrinsically and intrinsically, in a natural and consistent way.
The construction of such a mathematical toolkit was our primary objective.

\medskip

Wigner~\cite{Wigner} and others talk about what they call 
``the unreasonable effectiveness of mathematics in the natural sciences'',
however I prefer to think that mathematics is just unreasonably effective. 
What mathematics particularly excels at is in revealing the unexpected 
logical consequences of a small set of assumptions or axioms. Our hope was 
that natural axioms incorporating \( \so(2,3) \) symmetry would not only allow 
the physics to be expressed more clearly, but might also allow mathematics to
demonstrate a bit of that unreasonable effectiveness for us by means of some 
unexpected consequences.

\medskip

At the outset we needed to decide what type of mathematical structures we 
were looking to define. We chose to direct our efforts to developing the 
mathematical tools required in order to do relativistic quantum field theory 
in curved space. We therefore sought to define curved manifolds with a natural 
extrinsic  \( \so(2,3) \) symmetry; along with wave functions from those 
manifolds into a complex spinor space with intrinsic \( \so(2,3) \) symmetry.
That seemed to us to be sufficient to enable a description of most
physical phenomena.  We stopped short of second quantisation however. Hence 
forces would be described in terms of fields, potentials and curvature; and 
fermions in terms of a spinor wave function which we can think of as 
specifying the fermion field. 

We decided not to consider second quantisation for several reasons. Firstly
we simply didn't have time. Secondly we lacked the background knowledge needed 
to tackle this difficult topic. And finally second quantisation is notorious 
for its mathematical abuses and as mathematicians we therefore thought it 
prudent to steer clear of it.

\medskip
\pagebreak[2]

\section{The Mathematics}

If you have read through the rest of the book I hope you will agree 
that our efforts to create a consistent axiomatic mathematical structure
with the required symmetry were successful. The axioms of a framework are 
natural, as all good axiom systems should be. Indeed if you start from the 
assumption that \( \so(2,3) \) is the natural symmetry group for physics 
(which I do admit is a big assumption) then it is hard not to believe that 
it can be described using a framework.

\medskip

The approach we ended up taking was to use as our guide the matrix Lie group 
\( \Sp(2,\R) \) itself, which already has most of the features we need. The
matrix Lie  algebra \( \sp(2,\R) \) acts on this manifold as the collection 
of left invariant vector fields. These vector fields define derivations 
and the identification of the Lie algebra with each tangent space defines 
a consistent notion of parallel transport, and hence a connection which 
defines a covariant derivative. The torsion for this covariant derivative 
is (up to an annoying sign which we adjusted the definition of torsion to 
discard) the Lie algebra structure constants. The curvature is zero. 

The adjoint action of the Lie algebra on itself gives a local action on
vector fields which is conserved by the covariant derivative. 
Furthermore the \( 4 \times 4 \) matrices of the Lie algebra defines a
local action of each tangent space on \( \R^4 \). Extension of this local 
spinor action to an action on \( \C^4 \) is trivial. Left multiplication by 
matrices in \( \Sp(2,\R) \) also defines a natural parallel transport and 
hence a covariant derivative on spinor valued functions. The global action
once again respects the local action. One can of course extend this basic
structure using tensor product and dual to other types of tensor, each of
which has a local and global action which commute.

This example therefore has all of the features we are looking for except 
curvature.  Our approach was to generalise from the properties of this 
specific example by simply adding curvature to the extrinsic action in a 
consistent fashion. This does however restrict us to manifolds with the 
same dimension as the Lie algebra. Hence we must talk always about the 
10-D manifold of inertial frames and not 4-D spacetime.

\medskip

Having constructed our axiomatic system we then spent some time discussing 
its features and in particular its mathematical properties. We have only 
barely scratched the surface. It is in fact a bit embarrassing to leave so 
many questions unanswered, but if we had stopped to explore all of these 
mathematical issues properly we would never have had a chance to look at 
the physics. And we needed to look at the physics to justify the utility 
of our axiom system.

The biggest mathematical omission is the lack a complete classification 
theorem for generalised tensors. We have only a partial classification in the
small dimensional case. In particular we would really like to know what the 
locally trivial generalised tensors all look like. We would also like to know
whether a compatible global action exist for every local action. This seems 
likely to be true and indeed we have proved that it is true for all the small 
dimensional representations we have looked at in this book. We have most of 
a proof in the general case; however the proof is incomplete so this is not 
yet a theorem.

There are also some interesting open questions with regard to the extent to 
which the algebraic properties of our structure constrain the geometry and 
vice versa. And there are also some obvious generalisations which should be 
explored, the most obvious one being the generalisation to other Lie algebras.

The generalisation of greatest potential physical relevance would be to 
use complex spinor manifolds instead of real ones. This is equivalent to 
dropping the requirement than an invariant conjugation map on spinors exists, 
which is the most tenuous of our axioms. If we dropped this requirement we 
would need to consider the possibility that invariant conjugation maps did 
not exist, and the set of all local conjugation maps would then constitute 
a new kind of locally trivial tensor, possibly involving a non-trivial 
curvature. This would have consequences particularly for the description of 
fermions. 

A more radical mathematical generalisation would be to drop the requirement 
that spinors exist from our set of axioms altogether. 

In applied mathematical terms however the most glaring omission is probably 
the lack of specific solutions in this book. In our defence these do seem 
rather difficult to construct.

\bigskip
\pagebreak[2]

\section{The Physics of Forces}

The first test of our new framework was in how it coped with the 
physics of forces. And in that regard I would say that it has been
extremely successful. 

Forces should arise from the connection and from curvature. The 
connection is the most fundamental description, but is not a tensor
which can make working with it difficult. The curvature is a less 
immediate description of the geometry but has the advantage of being
a tensor. The action on spinors is fundamental, and hence the spinor
connection and curvature are the most important objects for describing
curvature.

The mathematical process of decomposition into irreducibles allowed
us to separate both the spinor connection and the curvature into components 
in a very natural way. 
\begin{flalign*}
\Gamma^\beta_{k\alpha} &= 
	A_k 1^\beta_\alpha + G^i_kT^\beta_{i\alpha} + N^A_k T^\beta_{A\alpha} \\
R^\beta_{ij\alpha} &=
F_{ij}1^\beta_\alpha + R^k_{ij}T^\beta_{k\alpha}
\end{flalign*}

Hence the unified force described by the spinor connection and curvature 
could be separated into components forces.

\medskip

With regard to the scalar components we identified the four translation 
components of \( A_k \) as the electromagnetic potential, and the translation
components of \( F_{ij} \) as the electromagnetic field tensor. The 
connection determines the curvature, which in terms of scalar components 
gave
\begin{equation*}
F_{ij} = \del_i A_j - \del_j A_i 
\end{equation*}

which is the expected relationship between the potential and the field and
the scalar component of the Jacobi identity for the global action was
\begin{equation*}
\del_i(F_{jk}) \ijkequal 0
\end{equation*}

which we recognise as the Faraday-Gauss equation. This is almost miraculous. 
These equations of electromagnetism are just there naturally in the geometry 
and we don't have to do anything to obtain them. Furthermore when we later 
come to look at matter we find that the covariant derivative in the Dirac 
equation naturally inserts \( A_k \) into the expected place. Hence the 
action of the field on matter as specified by the connection also seems 
correct. We lack only a source equation to have a complete set of equations
for electromagnetism.

The potential \( A_k \) is a connection and not a vector, and has a hidden
dependence on the choice of spinor basis, which seems to be how gauge 
symmetries enter into this picture. However we caution that gauge groups
are not the main event here and it would a mistake to overemphasise them
and attempt to force the mathematics to go along that path.

Of course we were only looking at the four translation components. What
of the other six? Presumably these describe other forces. We might hope 
that they describe the weak and perhaps also the strong force. The other 
possibility is that these additional dimensions simply provide minor 
adjustments to the electromagnetic force to account for things like 
the effect on spin.

We chose not to investigate these interesting possibilities further for 
several reasons. Firstly we were not sure that we could recognise the 
field equations for the weak or strong force if we tripped over them, which 
in a sense we may just have done. Typically these short range forces are not 
described in terms of fields and we don't know what their field 
equations ought to look like. These forces are more typically described 
in terms of their interactions which requires an understanding of matter
and probably also requires second quantisation. And of course we just 
didn't have time to go down this road. 

We have left this therefore as unexplored territory. However we are not
aware of anything here that would have the effect of invalidating our approach.

\medskip

The vector components of the curvature and potential seem to describe
gravity, although once again there are extra components which could 
suggest either additional forces or adjustment factors to take account
of such things as spin.  The vector component of the Jacobi identity 
gives the first and second Bianchi identities
\begin{flalign*}
R^m_{ij}T^l_{mk} &\ijkequal 0  \\
R^l_{im}T^m_{jk} + \nabla_i(R^l_{jk}) &\ijkequal 0 
\end{flalign*}

and while we have written these in terms of the reduced curvature tensor,
if you rewrite them in terms of the Riemann tensor they are indeed precisely
the usual first and second Bianchi identities, albeit with extra dimensions. 

The potential \( G^i_k \) determines the vector connection \( \Gamma^k_{ij} \)
via the equation
\begin{equation*}
\Gamma^k_{ij} = G^t_i T^k_{tj}
+ \del_i(T^\beta_{j\alpha})T^\alpha_{m\beta}g^{mk}
\end{equation*}

where the last term can be viewed as a basis adjustment term compensating
for the different hidden basis dependencies. Hence \( G^i_k \) 
determines the geometry of the manifold and should also therefore
determine the gravitational field \( R^k_{ij} \).

However when we sought such an equation explicitly we obtained
\begin{flalign*}
R^k_{ij} &= 
\bigl[\del_i(G^k_j) - \del_j(G^k_i)\bigr]
+ G^x_iG^y_jT^k_{xy} 
- N^A_iN^B_jT^k_{AB}  \\
&\hskip6cm + (G^m_j Q^k_{im} - G^m_i Q^k_{jm} )
\end{flalign*}

Where \( Q^c_{ab} = \del_a(T^\beta_{b\alpha})T^\alpha_{s\beta}g^{sc} \).

This would suggest that the gravitational field depends on both
the vector potential \( G^i_k \) and the versor potential \( N^A_k \). 

There is an inconsistency here that we have yet to fully resolve. We suspect
that the versor potential \( N^A_k \), which as a connection component 
has a hidden basis dependence on the spinor basis, can be made to vanish 
by a suitable choice of spinor basis.

The lack of an associated versor component of curvature tells us that 
the effect on spinors of parallel translation around a loop has no versor 
component. Hence we suspect that it is possible, perhaps by using path 
integration, to redefine our spinor bases consistently across the manifold 
in such a way that it eliminates the versor component of parallel translation.

If this can indeed be done then the gravitational field \( R^k_{ij} \)
will indeed be determined only by the gravitational potential \( G^i_k \), 
albeit in a rather more complicated fashion than we might have hoped for.

\medskip

With regards to source equations for our forces it is not difficult using
our toolkit to write down versions of the Ampere-Gauss and Einstein equations.
However these equations are very different and in particular they clearly
are not simply components of an overall equation for \( R^\alpha_{ij\beta} \).
This suggests that perhaps one of these equations is not correct. 

We then considered the equation
\begin{equation*}
\nabla^jR^\beta_{ji\alpha} - \nfrac{1}{2}g^{ts}T^r_{it}R^\beta_{rs\alpha}
= \Theta^\beta_{i\alpha}
\end{equation*}

where \( \nabla^i\Theta^\beta_{i\alpha} = 0 \); which has components
\begin{flalign*}
\nabla^jF_{ji} - \nfrac{1}{2}g^{ts}T^r_{it}F_{rs} &= J_i  \\
\nabla^j R^k_{ji} - \nfrac{1}{2} T^j_{is} g^{st} R^k_{jt} &= K^k_i
\end{flalign*}

The scalar component is the Ampere-Gauss equation for electromagnetism
while the vector component, which we call Ussher's equation, is an 
alternative source equation for gravity. 

It is rather bold to suggest such a different looking equation as an 
alternative to Einstein's equation. However the close parallel to
the equations for electromagnetism makes us confident that Ussher's 
equation will approximate Newtonian gravitation in the weak field 
non-relativistic limit.

The only loose end at this point seemed to be determining the nature of 
the source terms. In preparation for this we looked at an alternative 
Lagrangian approach to the source equations which would enable us to 
later calculate source terms from the Lagrangians for matter.

\medskip

The Lagrangian approach however turned out to be very interesting indeed.
Finding a suitable Lagrangian density for the curvature 
\( R^\alpha_{ij\beta} \) was not difficult. The function
\begin{equation*}
	\LL	= ||R^\alpha_{ij\beta}||^2  = \LL_g -4\LL_e
\end{equation*}

for example is an obvious choice and separates naturally into components 
\( \LL_e = ||F_{ij}||^2 \) and \( \LL_g = ||R^k_{ij}||^2 \) for the
 electromagnetic and gravitational fields respectively.

The interesting part was describing all possible variation to the geometry 
which leave the axiomatic properties of our framework intact. This
complicated our analysis. 

We were able to categorise the variations into three cases. Case 1 variations
only involved the scalar component of the field and minimising the Lagrangian
for such variations gave the Ampere-Gauss equation. Case 2 variations were
variations with the same form as a change of basis and left the Lagrangian
invariant. Case 3 variations were the most complicated. 

We found that these variations could be expressed in terms of two small 
but not independent parameters \( b^n_m \) and \( \delta^n_m \), and 
that the variation in the Lagrangian in terms of these was 
\begin{flalign*}
\updelta L 
&= \nonumber \int_\Omega
-4U^k_i g^{mi} g_{nk} \, b^n_m 
\\&\quad\quad\quad
+ \bigl( R^k_{ij}R^c_{ab} g^{ia} g_{kc} \pm 4 F_{ij}F_{ab} g^{ia} \bigr)
	\bigl( 4 g^{jm} 1^b_n - g^{jb} 1^m_n \bigr) 
		\delta^n_m \, dx^\circ 
\end{flalign*}

where $U^k_i $ is Ussher's tensor 
\begin{equation*} 
U^k_i = \nabla^j R^k_{ji} -\nfrac{1}{2} T^j_{is} g^{st} R^k_{jt}
\end{equation*}

In the weak field case the second term of the integral is insignificant
compared to the first. Hence $U^k_i = 0 $ which is
Ussher's equation in the absence of sources, solves the Lagrangian problem
in the weak field case.

This however is deceptive because the dependence between \( b^n_m \) and 
\( \delta^n_m \) prevents the former from varying arbitrarily and hence 
additional weak field solutions may be possible.

\bigskip
Solving the Lagrangian problem fully we obtained the equation
\begin{equation*}
A^{tkm}_{in} \nabla_t U^i_k + B^{km}_{in} U^i_k =  \nfrac{3}{2}(S^m_n + W^m_n)
\end{equation*}

where $U^i_k$ is Ussher's tensor and $W^m_n$ is a second order term. In the 
weak field case we can conceptualise this in terms of a two step process;
firstly solving the dark equation
\begin{equation*}
A^{tkm}_{in} \nabla_t X^i_k + B^{km}_{in} X^i_k =  \nfrac{3}{2}S^m_n
\end{equation*}

to find the sith $X^i_k$; and then solving Ussher's equation
$ U^k_i = X^k_i $, or equivalently
\begin{equation*}
\nabla^j R^k_{ji} -\nfrac{1}{2} T^j_{is} g^{st} R^k_{jt} = X^k_i 
\end{equation*}

to find the field.

\medskip

What these equations are telling us is that the influence of matter on
gravity is less direct than we had supposed. In the weak field case we 
can conceptualise this as a two stage process whereby the distribution of 
matter determines an intermediate quantity we call the sith; and the sith 
then acts as the effective source of gravity.  It is quite possible for the 
sith to differ from the source, and when this occurs we are likely to 
interpret it as dark matter. However there really is no matter present.

We have stumbled upon a possible explanation for one of the greatest 
mysteries of modern physics, the nature of dark matter. It seems that 
gravity can indeed behave in some circumstances as though matter were 
present when it is not, And we even have a dynamical equations to describe 
this behaviour. 

Of course this is only a possible explanation. It must be tested against
reality. That means we are going to need to actually simulate those rather 
horrible looking differential equations and compare the results to 
observation. Only once this is done if the results agree, can we truly
say that we have found and explained dark matter. We might also hope to 
explain dark energy in the same way.

\medskip

In the spirit of optimism we end this section with a conjecture. The
sith satisfies the same identities as Ussher's tensor, hence
\( \nabla^kX^i_k = 0 \) and sith is conserved.  We conjecture that the 
amount of extra sith in galaxies (dark matter) exactly balances the 
deficit of sith (dark energy) in intergalactic space.  Hence dark energy
is the result of sith that `should' be present in intergalactic space 
going off and hanging around a galaxy pretending to be dark matter; 
leaving behind a deficit.

\bigskip
\pagebreak[2]

\section{The Physics of Matter}

We next turned to the equations of matter and in particular the Dirac equation 
which should describe fermions. It was largely the Dirac equation and its 
natural compatibility with \( \so(2,3) \) that prompted this work in the first 
place. We therefore expected to have no difficulty adapting the Dirac equation 
to our framework.

Indeed the Dirac equation does adapt very well. The Dirac
operator itself becomes simply a Curl operator; the Dirac matrices are 
intrinsic translation operators; and everything works beautifully. Indeed 
because our framework affords an interpretation to the Dirac matrices it 
is considerably easier to work with the equation in this context than it is
normally. 

For wave functions which are constant across the six Lorentz dimensions the
parallel to the standard Dirac equation is exact, and we can therefore 
describe electrons in this way. Wave functions that vary across the Lorentz 
dimensions need explanation. We did not investigate these further, but it 
is possible that they describe additional fermions. 

An analysis of these more general solutions should start by using some kind of
separation of variables on the equation to isolate the behaviour across the 
Lorentz dimensions from the behaviour across the translation dimensions. 
Since the rotation coordinates are compact we might hope that the behaviour 
across Lorentz dimensions can be described using special functions with 
discrete parameters, which could be regarded as determining particle type. 
However we did not attempt to do this.

\medskip

We also desired a Lagrangian derivation of the Dirac equation in order
to extract from the Lagrangian density source terms for gravity. This turned
out to be considerably more complicated than we had anticipated. 

The problem was not with our framework which actually made the work 
considerably clearer. It was much easier for example to check that our 
Lagrangian density was scalar. However this added clarity revealed 
difficulties with the standard Lagrangian approach to the Dirac equation 
that are typically obscured in the notation.

In particular as a result of our analysis we now doubt that the standard
Dirac Lagrangian density is a well defined scalar. It seems to have an 
undeclared conformal degree of freedom. In our notation this degree
of freedom is described by the Crump factor \( h_\bullet \), which 
specifies the choice of invariant symplectic form and is needed to 
define the probability vector. This Crump factor seems to be a physical 
field of some sort.

Crump scalars differ from true scalars in their parallel transport 
properties and are linked to the electromagnetic field. Hence we can't
simply drop this factor or set it to a constant and ignore it. The 
non-trivial parallel transport will introduce an extra term involving 
the Crump connection \( H^k \) every time we take a covariant derivative.

Solving the Lagrangian problem properly with the Crump factor  present 
introduces an additional equation involving \( H^k \) which seems 
severely constraining and non-physical. The cure seems to be to introduce 
a dynamical term for \( h_\bullet \) which the standard Dirac Lagrangian 
lacks. This will convert this additional equation into a more physical 
looking wave equation for \( h_\bullet \). The precise nature of this 
extra term and the interpretation of \( h_\bullet \) (is it a Higgs?) 
were not investigated further.

Also causing difficulties with the Dirac Lagrangian density was
the fact that it is complex. This is a well known annoyance with the
standard theory justified on the basis that the complex Lagrangian 
density can be made to give the expected real answers. However the more 
complicated gravitational variations in our framework make the issue 
more critical since it appears that the complex Lagrangian gives
imaginary source terms for gravity under these variations. Since we 
can't have that, this is yet another reason for thinking that we need
a better Lagrangian density.

\bigskip
\pagebreak[2]

\section{Conclusion}

I hope you will agree that the mathematical tools developed in this book
have proved their worth by clarifying much that is obscure and by making 
unexpected predictions. Possibly the most striking prediction is the
new description of gravitation and in particular the explanation of dark 
matter.  What is most striking is the minimal nature of the assumptions 
from which these were generated. Our description of dark matter assumed only
the axioms of a framework which naturally encode the $\so(2,3)$ symmetry, 
and the fairly natural Lagrangian \( ||R^\alpha_{ij\beta}||^2 \). That was 
all. The power of the predictions arising from such minimal assumptions 
demonstrates the value of this approach. 

\medskip

We now hope many others will come and play in this playground since
there is much work remaining to be done. In particular second quantisation 
or similar is obviously needed. However in doing so we ask that an effort
be made to respect the mathematics as much as possible. Mathematics is a 
powerful tool for predicting the unexpected consequences of a set of 
assumptions, but it can only do its job if the mathematics is not abused. 
The abuse of mathematics in modern mathematical physics destroys 
the utility of what should be one of its most powerful tools.



\backmatter
    \bibliographystyle{amsplain}
    \bibliography{../Bibliography/Ourbib}

\providecommand{\bysame}{\leavevmode\hbox to3em{\hrulefill}\thinspace}
\providecommand{\MR}{\relax\ifhmode\unskip\space\fi MR }
\providecommand{\MRhref}[2]{%
  \href{http://www.ams.org/mathscinet-getitem?mr=#1}{#2}
}
\providecommand{\href}[2]{#2}
\begin{thebibliography}{1}

\bibitem{Cthesis}
William Crump, \emph{Maxwell’s equations on a 10-dimensional manifold with
  local symmetry {$\mathsf{so}(2,3)$}}, Master's thesis, University of Waikato,
  2012.

\bibitem{Hall}
Brian~C. Hall, \emph{Lie groups, lie algebras and representations: An
  elementary introduction}, Springer, 2003.

\bibitem{Knapp}
Anthony~W. Knapp, \emph{Lie groups beyond an introduction, 2nd ed.},
  Birkhauser, 2002.

\bibitem{Robinson}
M.~Robinson, \emph{Symmetry and the standard model}, Springer, 2011.

\bibitem{Uthesis}
Matthew Ussher, \emph{Investigating electromagnetism and gravity on a
  10-dimensional manifold with local symmetry {$\mathsf{so}(2,3)$}}, Master's
  thesis, University of Waikato, 2013.

\bibitem{Wald}
Robert~M. Wald, \emph{General relativity}, Chicago Press, 1984.

\bibitem{Wigner}
Eugene~P. Wigner, \emph{The unreasonable effectiveness of mathematics in the
  natural sciences}, Communications on Pure and Applied Mathematics
  \textbf{XIII} (1960), 001--14.

\end{thebibliography}

\end{document}